\documentclass[12pt]{article}
\newcommand{\twocolbreak}{}
\newcommand{\twocolbreakamp}{}

\newenvironment{widetext}{}{}

\usepackage{amssymb}
\usepackage[latin1]{inputenc}
\usepackage[T1]{fontenc}
\usepackage[english]{babel}
\usepackage{amsmath}
\usepackage{amsthm}
\usepackage{amsfonts,textcomp}
\usepackage{xcolor}
\usepackage{array}
\usepackage{hhline}
\usepackage{hyperref}
\usepackage{nameref}
\usepackage{soul}
\usepackage{relsize}
\usepackage{graphicx}
\usepackage{enumitem}
\usepackage{caption}
\usepackage[labelformat=simple]{subcaption}

\DeclareMathAlphabet{\mathpzc}{OT1}{pzc}{m}{it}
\hypersetup{pdftex, colorlinks=true, linkcolor=blue, citecolor=blue, filecolor=blue, urlcolor=blue, pdftitle=}

\newcommand{\lsbeq}{\stackrel{\textnormal\small{lsb}}{\large{=}}}
\newcommand{\msbeq}{\stackrel{\textnormal\small{msb}}{\large{=}}}

\newtheorem{theorem}{Theorem}[section]
\newtheorem*{unreferredaxiom}{Axiom}
\newtheorem{axiom}{Axiom}[section]

\newtheorem{requirement}[axiom]{Requirement}

\newtheorem{definition}{Definition}[section]

\newcounter{mylabelcounter}
\makeatletter
\newcommand{\labelText}[2]{
#1\refstepcounter{mylabelcounter}
\immediate\write\@auxout{
  \string\newlabel{#2}{{1}{\thepage}{{\unexpanded{#1}}}{mylabelcounter.\number\value{mylabelcounter}}{}}
}
}
\makeatother

\makeatletter
\renewcommand\subsubsection{\@startsection{subsubsection}{3}{0pt}
  {1ex plus 0.5ex minus .2ex}
  {0.5ex plus .2ex}
  {\normalfont\normalsize\itshape\textbf}} 
\renewcommand\paragraph{\@startsection{paragraph}{4}{\z@}%
  {1.5ex \@plus .2ex \@minus .3ex}
  {-1em}
  {\normalfont\normalsize\bfseries}}
\makeatother

\DeclareMathOperator{\var}{var}

\DeclareMathOperator{\precision}{prec}

\DeclareMathOperator{\Diag}{Diag}
\DeclareMathOperator{\Card}{Card}
\DeclareMathOperator{\suchthat}{:}

\DeclareMathAlphabet{\mathdutchcal}{U}{dutchcal}{m}{n}
\SetMathAlphabet{\mathdutchcal}{bold}{U}{dutchcal}{b}{n}
\DeclareMathAlphabet{\mathdutchbcal}{U}{dutchcal}{b}{n}

\definecolor{lightbeige}{rgb}{1,.95,.85}
\definecolor{beige}{rgb}{.8,.53,.1}
\sethlcolor{lightbeige}

\newcommand{\mysubsubsection}[1]{\paragraph{#1}}

\begin{document}

\title{Bounded information as a foundation for quantum theory}
\author{Paolo Ferro}
\maketitle

\begin{abstract}
The purpose of this paper is to formalize the concept that best captures our intuitive
understanding of quantum mechanics---that the information carried by a system is
limited---and, from this principle, to construct the foundations of quantum theory.
In our discussion, we also introduce a second important hypothesis:
if a measurement closely approximates an ideal one in terms of experimental precision, the information it provides
about a physical system is independent of the measurement method and, specifically, of the physical 
quantities being measured. 
This principle can be expressed in terms of metric properties of a manifold whose points represent the state of the system.
These and other reasonable hypotheses provide the foundation for a framework of quantum reconstruction. 

The theory presented in this paper is based on a description of physical systems in terms of their statistical properties,
specifically statistical parameters, and focuses on the study of estimators for these parameters. 
To achieve the goal of quantum reconstruction, a divide-and-conquer approach is employed, wherein the space of two discrete
conjugate Hamiltonian variables is partitioned into a binary tree of nested sets. This approach naturally leads to the
reconstruction of the linear and probabilistic structure of quantum mechanics.

\end{abstract}

\tableofcontents

\section{Introduction}

\subsection{General approach and related work}
The belief in a strong relationship between quantum mechanics and the concept of information dates
back to the early days of the development of the theory~\cite{Bohr1}.
However, the idea of founding quantum mechanics on strict informational principles has only been
seriously considered in recent decades, beginning with~\cite{Wheeler1}.

In this framework, an emerging approach, related to the more general program of quantum
reconstruction, is based on the principle that the amount of information a system can
carry is limited. 
Some lines of research following this approach have been developed by
Brukner and Zeilinger~\cite{BruknerZeilinger1,Zeilinger,BruknerZeilinger2},
Daki\'c and Brukner~\cite{BorivojeBrukner}, Rovelli~\cite{Rovelli1}, and
H\"ohn~\cite{H_hn_2017}. Among these works, in \cite{Rovelli1} and \cite{H_hn_2017}, there is an explicit 
axiom postulating the maximum amount of information
that can be extracted from a single system. A noteworthy contribution is Brukner and Zeilinger's 
article~\cite{BruknerZeilinger1}, which offers a general discussion about 
the connection between quantum physics and the concept of information. In that article, 
beyond the question of the boundedness of the information a system can carry, they discuss how this quantity scales
with system size. This topic is also examined in the current study.

This work shares a bounded-information principle with the authors cited above, but follows a different route to
quantum reconstruction.
More specifically, our framework is not based on General Probabilistic Theories (GPTs) or Operational Probabilistic
Theories (OPTs): we do not assume a convex state space with effects and transformations as primitives. 
In this, it differs from articles such as~\cite{BruknerZeilinger2, BorivojeBrukner, H_hn_2017}, which adopt GPT/OPT ingredients 
to varying degrees.

Furthermore, in contrast to the cited works, this study does not postulate an underlying algebraic structure for the 
mathematical
objects involved in the theory, such as a $\mathrm{C}^\ast$-algebra or a Hilbert space.
This point of view is shared with 
Masanes and M\"uller~\cite{MasanesMuller}, who do not postulate such structures but start from physical
hypotheses about preparations, transformations, and measurements. Although their work does not specifically adopt a
bounded-information postulate, it follows an information-based approach to reconstruction.
By contrast with the present work, their framework is formulated within GPT.

A similar approach has been followed by D'Ariano, Chiribella, and Perinotti~\cite{Chiribella1}, who do not postulate a
Hilbert space (or a $\mathrm{C}^\ast$-algebra) structure of the theory, and suggest that quantum theory can be derived
``from purely informational principles''.  Their principles, however, differ from those considered here (they do not take a
bound on the maximum amount of information of a system as a foundational axiom), and their work follows an OPT-based approach.

Among the above works, \cite{BorivojeBrukner}, \cite{MasanesMuller}, and \cite{Bub3} take Hardy's article~\cite{Hardy1}
as a key reference for developing their results, which is nevertheless framed within a general informational 
approach to quantum reconstruction.

As a final remark, the work of Goyal~\cite{Goyal_2010} must be mentioned, because it includes, among others, a postulate of metric invariance under transformations of the state space.
In the present work, this property plays a central role; however, it is not taken as a postulate but is derived from an axiom within a different theoretical framework.

\subsection{Overview}
The present work focuses on the reconstruction of quantum mechanics from informational principles,
and while it shares some assumptions with the above-cited works, it develops in different directions.
The central idea of our approach can be summarized as follows:
there are mathematical properties commonly used to describe a physical system,  but in certain cases,
a system may lack some of the information regarding these properties.

As an illustration of this idea, consider the following  computational analogy:
imagine a calculator in which numbers are represented with a limited number of digits,
say $N$ digits. This does not mean that the digits beyond the $N$th behave randomly, 
but simply that the calculator has no information about them. This scenario aligns with the
``Universe as a computer'' metaphor, discussed by some authors~\cite{Jaeger}.

We refer to this principle as the {\textit{Limited Information}} axiom,
which can be interpreted as a bound on the amount of information a system can carry, where
``amount of information'' refers to the number of elementary statements, or predicates, about some
property of the system. We assume that this quantity is measured in {\it bits}.

A postulate expressed in such a vague way may seem difficult to apply in constructing a theory,
since it does not specify the magnitude of the information bound.
However, by introducing additional requirements related to invariances of the physical model,
the extent of this information constraint can be naturally determined.

One of the fundamental assumptions of this work is that measurements are macroscopic processes
and are therefore always affected by random errors---a fact independent of quantum theory. 
Despite that, we can consider the best measurement possible, which minimizes  such errors,
regardless of their origin or characteristics.
The second axiom asserts that the amount of information about a system, provided by
this type of measurement, is independent of the measured quantity, the
measurement method, and the system's state.
We refer to this postulate as the {\textit{Precision Invariance}} axiom.

It should be emphasized that the random nature of the measurement process enters for the first time when this
postulate is introduced.
Thus, on one hand, we have a principle that addresses merely informational properties and does not imply any intrinsic 
randomness of physical systems, and on the other hand, the standard random behavior of the measurement process.

These principles are introduced in section~\ref{sec:defctx}.
We also introduce the concepts of states and observables;
we assume, essentially, that there exist certain mathematical properties that allow us to make statistical predictions about 
measurement outcomes, and that all such properties can be synthesized in a single object, which we call the
{\textit{physical state}}.  
									
In our quantum reconstruction program, we must also include other reasonable hypotheses.
A third fundamental hypothesis, which is rather a class of axioms, has to do with symmetries: if we  describe the
dynamics of a system on a phase space of conjugate Hamiltonian variables, we must 
require  the bare theory (not including dynamics) to respect some intrinsic symmetry, such as
translation and scale invariance.

The fourth fundamental requirement states that, in constructing the theory, we refer to a model that depends on parameters 
(e.g., size, resolution, coordinate system) whose role is only to define the model itself and which should have no real 
physical relevance.
To be more specific, our model represents a subsystem that is part of a universe. We want our theory to be independent of
the size of our model or its position relative to the surrounding world.

Up to this point, the theory presented in this work is fundamentally probabilistic---the predictions of the
measurements are expressed in terms of probability distributions. The appropriate framework for
addressing such problems relies on statistical parameters. The physical state, as we initially defined
it, can be represented by a collection of statistical parameters, and the set of all physically
allowed states can be viewed as a manifold in the space of these parameters.
This kind of approach is discussed in
section~\ref{sec:physstaeest}.

A key point of this work is the formulation of the Precision Invariance axiom in terms 
of the metric properties of the state manifold, which is introduced in section~\ref{sec:precinvfishmetr}.
In short, we can define sets of statistical parameters, 
closely tied to different choices of observables, that uniquely identify a point on the
state manifold. We demonstrate that, when we move from one of these sets to another, there
is a metric, defined on the space of these parameters, which is preserved.
This metric is an extension of the concept of Fisher metric~\cite{FisherClassic}
for some special \textit{estimators} of the statistical parameters.

In section~\ref{sec:elemsys}, all the above principles are applied to some models of 
elementary systems.
We consider, for instance, a system with three two-valued
variables. Assuming that this system carries exactly one bit of information, we  demonstrate that the probability distributions,
associated with the measurement of these variables, belong to a two-dimensional
manifold, and we derive the shape of this manifold. 
This shape corresponds to the Bloch sphere representation of a quantum two-state system (qubit).
This establishes a first correspondence with standard quantum mechanics, leading us to the
first important result of this study.

It is important to emphasize that this theory includes the measurement as a constitutive element.
In our framework, the outcomes of measurements of observables are treated as random variables. Since we started from the
Precision Invariance axiom, which concerns the precision (inverse of the variance) of distributions,
the condition obtained in our model is imposed directly on probability distributions and not on the squared modulus
of an abstract mathematical object (i.e., the wavefunction).
This implies that the Born rule is not introduced as an additional axiom.

In section~\ref{sec:nstatesyst}, we extend the model to a higher-dimensional system characterized by two
discrete canonical variables $p$ and $q$ of a phase space.
Our goal is to find constraints on the probability distributions of $p$ and $q$ that generalize those
found for the simpler three-variable system. 
In order to determine such constraints, we follow a divide-and-conquer strategy, decomposing the system into smaller subsystems,
structured as a binary tree of nested subsets in the $p$ or $q$ space.
The set of predicates on the values of $p$ and $q$, which we use to measure the amount of information carried by the 
system, reflects this binary structure. 
In the same section, we discuss what values to assign to particular characteristic parameters of the theory (the ``signature''
of the model), such as the dimensionality of the state space and the number of extra observables, in addition to $p$ and $q$, 
introduced in order to obtain a more general theory.

In this approach, the symmetry and model independence requirements 
play a central role: requiring the theory to be invariant under certain transformations and
independent of the system's size restricts the domain of physically admissible states.
We also demonstrate that a decomposition invariant under scale and translation symmetry is unique.
Thus, we obtain a decomposition in the form of a network of elementary (two-value)
three-variable systems, organized as a ``butterfly diagram'' (see Fig.~\ref{butterfly8}).

Based on the binary structure of our model and of certain invariance assumptions, we can derive an expression for the extended  
Fisher metric. If we reparametrize this quantity in terms of complex functions, whose squared moduli represent the 
probabilities, we find  that this metric corresponds to the Fubini-Study metric on complex discrete functions.

This allows us to demonstrate, in section~\ref{sec:lingen}, via a reinterpretation of Wigner's theorem~\cite{Wigner1}, that the 
metric conservation property introduced above leads naturally to the linear structure of the theory. That is, in the 
representation of statistical parameters in terms of complex functions, transformations (i.e., changes of statistical 
parameters) take the form of unitary linear operators.
This result is closely related to what was demonstrated---albeit in a different context and under different 
assumptions---by Wootters~\cite{Wootters_1981} and other authors~\cite{Braunstein_1994}.

The binary tree decomposition scheme and the Precision Invariance axiom (expressed as a metric-preserving
condition), lead us to identify the set of admissible probability distributions 
of $p$ and $q$. The relation we obtain for these functions has the form of a pattern
of increasing resolution transformations, which replicates the scheme of the Cooley-Tukey
algorithm~\cite{CooleyTukey} for the Fast Fourier Transform (FFT).

These results are shown in sections \ref{sec:fourptsys} and \ref{sec:traslinv}, and they represent the second key
result of this study: if we use complex functions---now identifiable as wavefunctions---to represent the
state, the function associated with one of the conjugate variables, $p$ or $q$, is the discrete
Fourier transform of the function associated with the other.
This result aligns with the basic formulation of standard quantum theory.

\subsection{Achievements and scope of the work}

In this paper, a new reconstruction scheme for quantum mechanics is proposed, based on the definition of the state of a system 
as the mathematical object that determines all the statistical properties of its observables. The main results of this 
reconstruction can be summarized as follows:

\begin{enumerate}[label=(\roman*)]

\item \label{itm:qbit}
The elementary one-bit quantum system is reconstructed on the basis of the theory's axioms. Both the real case (rebit) and the 
complex case (qubit) are examined.

\item \label{itm:nbits}
The construction is then extended to an elementary system with two conjugate variables defined over a discrete phase space, 
carrying a larger but finite number of bits.

\item \label{itm:cplxfunrep}
Within this framework, the state admits a representation in terms of complex vectors whose squared moduli are probabilities.
We establish this both for the one-bit system and for larger systems, thereby providing a direct link to the standard formalism 
of quantum mechanics.

\item \label{itm:qmlin}
The Precision Invariance axiom can be formalized as a metric-preserving property of the state-space manifold.
This result entails the linear structure of the complex state representation and of the transformations between different 
representations, as in the standard formulation of quantum mechanics.

\item \label{itm:butterfly}
For systems with two conjugate variables $p$ and $q$, it is shown that the complex functions representing the state, defined on 
the domains of these variables, are related by the Fourier transform, as in standard quantum mechanics.

\end{enumerate}

It must be pointed out that the reconstruction proposed in this study concerns the foundational pure-state structure of
quantum mechanics and is deliberately limited to isolated discrete systems. The density-operator formalism and mixed states are 
not reconstructed here; they may instead be introduced as subsequent extensions of the pure-state, Dirac-like formulation 
developed in this work.
Likewise, although the continuum limit and the extension to composite systems are not developed explicitly in this work, the 
models considered here suggest that these generalizations should preserve the main qualitative results of the reconstruction. 
They can therefore be regarded as natural extensions of the present framework and may be the subject of future work.

Throughout this work, we assume that the physical quantities and the mathematical objects used to represent them are 
sufficiently well behaved for the standard analytical, statistical, and differential-geometric results employed in the 
derivations to apply. In particular, we assume that the required continuity, differentiability, regularity, and asymptotic 
conditions are satisfied whenever they are needed.

\section{The framework of the theory} \label{sec:defctx}

\subsection{Basic definitions} \label{sec:funddef}

In this section, we state the foundations of the theory.
In order to establish a connection with our understanding of reality, we introduce
the basic definitions and hypotheses by following an informal approach.
This especially applies to the definitions that can be considered primitive notions,
such as {\textit{Measurement, Observable}} and  {\textit{Information}}.
Later in our analysis, most of these concepts will be
presented in a formal way to provide a consistent basis for the theory.

Several of the definitions introduced in this section are basic and well-established
concepts from general physics or quantum mechanics. However, they will be presented herein to
establish a conceptual framework relevant to the context of this article. There are two primitive notions that establish a 
connection between our framework and reality. The first is
\begin{definition}\label{def:measurement}
{\bf{(Measurement)}}
A {\textnormal{measurement}} is a process that transfers information about
a physical system to the observer. This information is usually in the form of a numerical value. 
\end{definition}

\noindent A second important notion is the following:
\begin{definition}\label{def:observable}
{\bf{(Observables)}}
An {\textnormal{observable}} is a measurable property of a system. 
\end{definition}

\noindent Even the smallest system may admit one or more
observables. Every observable has a set of possible measurement outcomes.
In this study, we always consider observables that can take finite sets of values.
Thus, for an observable $o$, we can define a set $\Omega = \{v_1, v_2, ...\}$ of
possible outcomes of its measurement. 

\mysubsubsection{Measurement randomness}
Measurement outcomes are generally affected by random errors.
This leads us to consider the measurement as a random process that can be described in terms of probability functions.

Throughout this paper, we view this hypothesis as fundamental for any physical system. We will always consider
classical Kolmogorov probability spaces. 
Thus, we suppose that any observable $o$ is equipped with a probability 
space $( \Omega, \mathcal{A}, \rho )$, with $\Omega$ the set of all possible outcomes of $o$, 
$\mathcal{A}$ a (finite) algebra of subsets of $\Omega$, and $\rho$ a probability distribution on 
$\Omega$.

\mysubsubsection{Parametrizations}
In this work, the probability distributions of the observables 
will often be defined in terms of statistical parameters~\cite{Lehmann}. Namely,
the distribution of an observable $o$ takes the form $\rho(\xi; \boldsymbol{\theta})$,
where $\xi$ is a possible measurement outcome of $o$, and $\boldsymbol{\theta}$ is
a vector of parameters (in this paper, vectors are written in bold). 

This definition, together with other related notions introduced below, belongs to the standard framework of statistics and 
estimation theory.
They are recalled here because they play a central role in the theory developed in this study.

For an observable $o$, we can identify the set of all its probability distributions, 
$\Delta_o = \{ \rho_o(\cdot) \in \mathbb{R}_{\ge 0} \, \suchthat \, \sum_{\xi \in \Omega_o} \rho_o(\xi) = 1 \} $, which is 
commonly referred to as the {\textit{probability simplex}} of the observable. 
A parametrization defines a \textit{statistical model}, namely, a
set of probability distributions 
$\mathcal{M} =
\{ \rho_o(\xi ; \boldsymbol{\theta}) \suchthat \boldsymbol{\theta} \in \Phi_{\boldsymbol{\theta}}\}$,
with $\Phi_{\boldsymbol{\theta}}$ representing the space where the parameters are defined.
A parametrization $\boldsymbol{\theta}$ is said to be \textit{full} if
it can represent all possible probability distributions (the whole simplex) of a given observable $o$, namely if:
\begin{equation} \nonumber
\forall \rho_o(\cdot) \in \Delta_o \, \exists \boldsymbol{\theta} \suchthat  \rho_o(\xi)  = \rho_o(\xi ; \boldsymbol{\theta}).
\end{equation}
\noindent We can also require a parametrization to  be \textit{identifiable}, namely that:
\begin{equation}\nonumber
\rho_o(\xi ; \boldsymbol{\theta}_1) =
\rho_o(\xi ; \boldsymbol{\theta}_2)  \, \Rightarrow \,\boldsymbol{\theta}_1 = \boldsymbol{\theta}_2.
\end{equation}
\noindent If a parametrization is full and identifiable, namely, if there is a one-to-one
correspondence between the space of all possible probability distributions and the
parameter space, we can study the properties of the probability distributions
by studying their parameters.

A system may consist of many observables, each equipped with its probability distribution.
We can think of unifying the parameters of the distribution of all observables in a single mathematical object. 
This scenario leads us to the following definition:
\begin{definition} 
\label{def:state}
{\bf{(State)}}
In a system with a set of observables $ \{o_1, o_2, ... \}$, the 
\emph{state} is an abstract mathematical object that uniquely defines the 
statistical properties of all observables of the system. 
That is, all the physically valid distributions of the observables of the system can be put in the form $\rho_{o_i}(\cdot; x)$, 
where $x$ is the state.
\end{definition}
\noindent
We assume that the state identifies the statistical properties of the
system in a one-to-one way: two states are identical if and only if they
determine the same statistical properties for all observables. This
does not prevent us from adopting mathematical representations of the
state that are not one-to-one, so that different mathematical objects
may represent the same physical state.

The notation $\rho_{o_i}(\cdot; x)$ is clearly the manifestation of the idea that the state plays the role of a 
(multidimensional) statistical parameter that controls the distributions of all the observables of a set.

The set of all states is the {\it{state space}} and is denoted by $\Phi$. At this stage, we do not assign any property to 
$\Phi$: determining its structure will be one of the purposes of this work.

\mysubsubsection{Independent observables}
Let us consider a set of $D$ observables  $\mathcal{D} = \{o_1, o_2, ... o_D\}$ of a physical
system in which these observables are {\it{independent}}, in the sense that the measurement of one of them does not imply 
the knowledge of another one. Equivalently, we suppose that there are no functional relations of the form
 $o_i = f(o_1, \dots, o_{j\ne i}, \dots)$ connecting them.
It must be pointed out that independence is a merely physical property, i.e. it is related to the physical meaning of the 
objects $o_1, o_2, ... o_D$.

This point becomes clearer if we think of the present formulation as starting from a well-established classical model, before 
any ``quantum infrastructure'' has been constructed upon it.
Thus, independence can be understood as the absence of functional relations
 $o_i = f(o_1, \dots, o_{j\ne i}, \dots)$ already present at the classical level.
For example, in a simple classical system with two variables (observables) $p$ and $q$, a third variable like the energy
$E = \frac{p^2}{2 m} + V(q)$ is not independent from $p$ and $q$.

\mysubsubsection{Information} By  \textit{information carried by a system} we mean the maximum
amount of information a system can transfer to the observer during a measurement. It
is measured in bits. 
As stated before, one of the central points of this study is the hypothesis that the amount of information that can be carried 
by a system is limited. 

As stated before, in this work, we will always consider observables that take a finite set
of values.
Under this hypothesis, the amount of information obtained from the measurement of
an observable is $n_{o_i} = \log_2 N_{o_i}$, where $N_{o_i}$ is the number of values the observable $o_i$ can take. If we could measure all the observables of the system,
we would obtain an amount of information $\sum_i n_{o_i}$. 

A further clarification concerns the classical model used as a reference. For the systems considered in this work, we assume 
that there exists a set of independent observables that is sufficient to specify the state of the system at the classical 
level; for example, a complete set of Hamiltonian canonical variables. Therefore, when we say that a set of observables 
``describes'' a system, we mean that these observables are sufficient to specify its classical state.

With these assumptions, the principle of Limited Information can be put in this form:

\begin{unreferredaxiom}  {\bf{(Limited Information)}} In a system described by $D$ independent observables, let $n_{o_i}$ 
denote the amount of information (number of bits) obtained by measuring each observable $o_i$ (e.g., $n_{o_i}=\log_2 N_{o_i}$ 
if $o_i$ has $N_{o_i}$ outcomes). 
Let $n_I$ be the maximum amount of information we can obtain in any single set of measurements on the system. Then
\begin{equation}\nonumber
 n_I \;<\; \sum_{i=1}^{D} n_{o_i}.
\end{equation}
\end{unreferredaxiom}

\noindent In the cases discussed in this work, we have that
$N_{o_i}$ is equal for all the observables, so we can say that  $\sum_i n_{o_i} = D n_o$.
The obvious consequence of this axiom is that we can never measure all the observables of 
$\mathcal{D}$ but only a subset.

In the spirit of this study, the hypothesis that measurements of observables are, in general, affected by random errors and the 
Limited Information axiom are the main features that distinguish a classical system from a quantum one. They can be grouped 
together into the following definition:
\begin{definition}
{\bf{(Quantum System)}}\label{def:quantumsystem}
A quantum system is a system where:
\begin{itemize}
 \item All observables have a random behavior described by probability distributions.
 \item The system obeys the Limited Information axiom.
\end{itemize}
\end{definition}

\noindent One of the purposes of this study is to demonstrate, with the help of other reasonable and physically acceptable 
hypotheses, that this definition matches the definition of a quantum system in standard quantum mechanics.

\mysubsubsection{Preparations and one-point distributions} A preparation is a special type of measurement.
Suppose we measure an observable $o$, and its measurement produces the value $v_o\in\Omega_o$:

\begin{definition} {\bf{(Preparation)}}
\label{def:preparation}
A measurement is a \emph{preparation} if it leaves the measured observable
with a definite value; that is, immediately after the measurement takes place, an observable $o$ has
\emph{with certainty} the value $v_o$.
\end{definition}
\noindent In our theory, we assume that an observable is always equipped with a probability
distribution. Therefore, we can regard a preparation as a process that leaves the observable $o$, with probability 1 for the 
value $v_o$ that comes from the preparation, and with probability 0 for all the other values---that is, a preparation generates 
a one-point distribution.
Throughout this work, we will often refer to this correspondence between definite
values and one-point distributions.

The definition of preparation can be generalized to the case where we measure a set of observables. In this case, each 
observable of the set is left, after the preparation, with a definite value. 

Definition \ref{def:preparation}
can lead us to formulate the following assumption, which will be involved in many derivations in this work:

{\textit{A measurement process with the characteristics of a preparation is, in general, physically realizable.}}

We can make the definition of preparation stronger with the following extension.
\begin{definition} 
\label{def:completepreparation}
{\bf{(Complete Preparation)}}
A preparation is \emph{complete} if it not only leaves an observable (or a set of observables)
with a definite value (one-point distributions), but if it fixes the system's state. Namely if, after it takes place, the 
system is left in a state $x$ defined
by a function $X(v_{o_1}, v_{o_2}, ...)$ (where $v_{o_1}, v_{o_2}, ... $ are the measured values of the set of observables). 
\end{definition}

\noindent This implies that, right after a complete preparation takes place, the statistical behavior of
the whole system is defined. 

We can put this definition into another form: a preparation is complete if the distributions of all observables
are uniquely determined by the value(s) of the observable(s) measured in the preparation.


A consequence of the definition of preparation is that, 
if we prepare the system in a given state, there must be some information that describes the new state and an observable 
corresponding to that information. 

It is worth noting that there is an underlying hypothesis 
in the definition of preparation, namely that there exist measurements that are not affected by random errors, 
because we want to determine precisely the value taken by the measured observable $o$.
This is, of course, an abstraction, but we can think of getting as close as possible to this
ideal condition.
See the definition of an ideal measurement provided in the pages that follow. 

\mysubsubsection{Complete observables and Information} \mbox{}\newline
\noindent Let us consider a set $\mathcal{C}$ of ${D_{\mathcal{C}}}$ observables $o_1, o_2, \dots, o_{D_{\mathcal{C}}}$.  We now 
introduce the following definition, closely related to the definition of complete preparation.
\begin{definition} 
\label{def:completesetofobservables}
{\bf{(Complete set of observables)}}
A set $\mathcal{C}$ of ${D_{\mathcal{C}}}$ observables is \emph{complete}\footnote{Not to be confused with a
complete set of commuting observables in standard quantum mechanics.} if:
\begin{itemize}
    \item Their joint
measurement is consistent:  they are independent of each other (they are not related by dependences in the form $o_i = f(o_1, 
\dots, o_{j\ne i}, \dots,o_{D_{\mathcal{C}}})$) 
and the amount of information we obtain from their measurement does not exceed the prescriptions
of the Limited Information axiom ($\sum_{o_i}^{\mathcal{C}} n_{o_i} \le n_I$).
    \item The values obtained from their joint measurement identify a state:  
    \begin{equation}\nonumber
      x=X(v_{o_1}, ..., v_{o_{D_{\mathcal{C}}}}),
    \end{equation}
    \noindent where $v_{o_1}, ..., v_{o_{D_{\mathcal{C}}}}$ are the measured values of the set. 
\end{itemize} 
\end{definition}
\noindent In this definition, by joint measurement of two observables, we simply 
mean the measurement of an observable $o_1$ and an observable $o_2$ in the same context.

In other words, in a set of jointly measurable observables, the information about all of them can be obtained without the 
measurement of one observable invalidating the information obtained about the others.

With the help of the definition of complete preparation, the above definition can be put in a shorter form: a set of 
observables is complete if and only if its joint measurement can be used as a complete preparation for a system.

We can prove the following theorem, which establishes a connection between the above definition and the axiom of Limited Information:
\begin{theorem}\label{th:compset}
A complete set of observables is maximal, i.e., it cannot 
be extended to a larger complete set by adding one or more observables.
\end{theorem}
\begin{proof}
Consider a complete set $\mathcal{C}$ of ${D_{\mathcal{C}}}$ observables $o_1, o_2, \dots, o_{D_{\mathcal{C}}}$.
Suppose, by contradiction, that there exists a larger complete set $\mathcal{C}' = \mathcal{C} \cup \{o_{D_{\mathcal{C}}+1}\}$, 
on which we make a measurement.
The value of the additional observable $o_{{D_{\mathcal{C}}}+1}$ obtained from such a measurement  fixes its distribution
as a one-point distribution.
This distribution must also be the only one
compatible with the measurement results (and with the one-point distributions)
of the other observables $o_1, o_2, \dots,  o_{D_{\mathcal{C}}}$, because their completeness implies that the distribution of 
any other observable is determined.
This leads to a contradiction, because one-point distributions are equivalent to having a single,
deterministic value, and this would imply that for fixed values
of $o_1, o_2, \dots, o_{D_{\mathcal{C}}}$ we get a single value of $o_{{D_{\mathcal{C}}}+1}$. In other words, we would have
a relation of the form $o_{{D_{\mathcal{C}}}+1} = f(o_1, \dots, o_{D_{\mathcal{C}}})$, which contradicts the independence condition required by the definition
of complete set of observables~\ref{def:completesetofobservables}. 
The same applies if we try to enlarge $\mathcal{C}$ by more than one
observable, leading us to the conclusion that $\mathcal{C}$ is the maximal complete
set of observables. Note that this is also the largest set of observables for which we
can identify states with one-point distributions only. 
\end{proof}

Suppose that all the complete sets of observables of a system carry the same information.
By simplifying the notation of definition \ref{def:completesetofobservables},
we express this condition in terms of sums of the observables in a complete set:
$\sum_{o \in \mathcal{C}} n_{o}=\textnormal{const.}$  for any complete set $\mathcal{C}$. Under this condition, by the above 
theorem, it is easy to show that the amount of 
information provided by the measurement of any complete set of observables is maximal, i.e., no further information can be 
obtained by measuring additional observables.

This implies that, if in a system there are more observables than those necessary to form a complete set, and if we assume that 
$\sum_{o \in \mathcal{C}} n_{o}=\textnormal{const.}\,\forall \mathcal{C}$, we must necessarily 
introduce a hypothesis that limits the information. In terms of the previous notation, for any complete set with 
${D_{\mathcal{C}}}$ observables $o_i$, this information limit is $n_I = \sum_{i=1}^{D_{\mathcal{C}}} n_{o_i}$.
According to the definition above, the observables belonging to a complete set $\mathcal C$ are themselves independent. 
However, in a quantum system, they do not constitute the whole set $\mathcal D$ of independent observables required to describe 
the system at the classical level, which is referred to by the Limited Information axiom.

In the derivations presented in this work, we are almost always interested in complete sets of observables consisting of a 
single observable. Thus, we can consider the following 
\begin{definition} 
\label{def:completeobservable}
{\bf{(Complete observables)}}
An observable is {\emph{complete}} 
if it forms a complete set (a singleton) of observables. 
\end{definition}

We can conclude the discussion about complete sets of observables with the following theorem:

\begin{theorem}\label{th:indvscomplobsset}
Every set of jointly measurable independent observables can be extended to a complete set of observables.
\end{theorem}

\begin{proof}
Let $\mathcal{J}$ be a set of jointly measurable independent observables.
If $\mathcal{J}$ is complete, there is nothing to prove.
Suppose instead that $\mathcal{J}$ is not complete. Then the information obtained by measuring the observables in $\mathcal{J}$ 
does not completely determine the state. Therefore, there exist different states that are compatible with the same measurement 
outcomes of $\mathcal{J}$ but differ in the statistical properties of at least one other observable.
In the framework adopted in this work, there is no second statistical level describing a probability distribution over 
different physical states: all the statistical uncertainty associated with a preparation is already contained in the 
probability distributions of the observables. Therefore, either there is a single distribution corresponding to a single state, 
or there are many distributions corresponding to many states; the system is in one such state (although the measurement has not 
identified which one).

Thus, there is information about which of these states is occupied by the system, which is not affected by the measurement of 
$\mathcal{J}$. This information represents a component of the state that has not been identified by the measurement of 
$\mathcal{J}$. 
Since a complete preparation, by definition, determines the state, this component of the state must have been determined by the 
preparation and, therefore, must correspond to an observable $o'$.

Since the measurement of $o'$ is not affected by the measurement of $\mathcal{J}$, the enlarged set 
$\mathcal{J}'=\mathcal{J}\cup{o'}$ is jointly measurable.

Moreover, $o'$ cannot be functionally determined by the observables belonging to $\mathcal{J}$, since otherwise its value would 
contain no information capable of distinguishing states that give the same outcomes for all the observables in $\mathcal{J}$.
Hence $\mathcal{J}'$ is also a set of independent observables.

If $\mathcal{J}'$ is not complete, the same argument can be repeated. At each step, an additional independent observable 
contributes new information about the state. Since, according to the Limited Information axiom, the total amount of information 
obtainable from a set of jointly measurable observables is bounded by the finite quantity $n_I$, this extension process cannot 
continue indefinitely. It must therefore terminate at a jointly measurable independent set whose measurement completely 
determines the state, namely, at a complete set of observables.
\end{proof}

\mysubsubsection{Measuring the state} Is it possible to \textit{measure} the state of a system?
As we stated above, the state is the mathematical object that describes the 
statistical properties of \textit{all} the observables of a system.

Hence, the state is a property that a physical system is equipped with. The question is whether
the state can be measured like an ordinary observable.
If we try to define the observable that represents the state, 
and we assume that the measurement of this ``state'' observable undergoes the random behavior of
regular observables, we may expect self-referential or paradoxical scenarios.
This is because we assume that the measurement of an observable is governed by a probability
distribution, which is represented by a point in the state space; on the other hand, a point in the state space
corresponds, by definition, to a value of the state, not to a distribution of its values.

The resolution of this problem is to assume that the state can have only one-point distributions and, therefore,
it has no random behavior.
It should be noted that this has to do with  the possibility of preparing a system in a
well-defined state, namely, with the existence of a special kind of measurement that we
identified as preparation.

\noindent  In the case of a system that admits a finite number of distinct 
states, we can demonstrate  the following theorem:
\begin{theorem} \label{th:statenotobservable}
The state of a quantum system cannot be measured
\end{theorem}
\begin{proof}
If the system can be described by an infinite number of states, the demonstration is trivial because measuring the state would 
imply the knowledge of an unbounded amount of information, while the Limited Information implicitly requires a finite $n_I$.

If the system can be described by a finite number $N$ of distinct 
states, then, when measuring the state, we get $\log_2 N$ bits of information.
Since we are dealing with a quantum system, 
the Limited Information axiom must be obeyed, and we have that $\log_2 N \le n_I$, where
$n_I$ is the maximum information allowed by this axiom.

Let us now consider, among all possible complete sets of observables, the set $\mathcal C$ that globally takes the maximum 
number $N_{\mathcal{C}}=\prod_i N_{o_i}$. 
By theorem \ref{th:compset}, a complete set cannot be enlarged to a larger complete set. Moreover, by theorem 
\ref{th:indvscomplobsset} any larger jointly measurable independent set, even if not complete, could itself be extended to a 
complete set. Therefore no jointly measurable independent set containing $\mathcal{C}$ can provide additional information. 
Since $\mathcal{C}$ has been chosen to maximize $N_{\mathcal C}$ among all complete sets, its information $n_{\mathcal{C}} = 
\log_2 N_{\mathcal{C}}$ 
is the maximum information obtainable from a joint measurement on the system.

If we had a system where only one-point distributions are allowed, all the observables of the set can take globally 
$N_{\mathcal{C}} = \prod_i N_{o_i}$ different distributions.
But in this kind of system, where all the states have zero uncertainty, there would be no random behavior. This contradicts 
the definition of a quantum system \ref{def:quantumsystem}. 
Hence $N > N_{\mathcal{C}}$. 
Since $\log_2 N \le n_I$, we can also say that $\log_2 N_{\mathcal{C}} < n_I$, 
which contradicts what was stated before, i.e., that
the information obtained by a measurement on $\mathcal{C}$ is maximal.
\end{proof}

\mysubsubsection{Quantum tomography - The ensemble hypothesis  } 
In spite of the theorem just established, there is still a way to indirectly measure 
the state of a system.
As we know, this problem has been widely studied in the framework of what is 
commonly referred to as quantum tomography~\cite{DArianoParisSacchi1,FanoQuantTom,VogelRisken1}. 
In this study, we focus on this problem from a statistical and estimation-theoretical point of view.

Suppose we can build an ensemble of identical replicas of a system.
Here ``identical'' means that they are prepared exactly in the same state.

If we perform a set of measurements on a given observable of a system, we can 
reconstruct its probability distribution. If we parametrize all 
the distributions of the observables with full and identifiable parametrizations, 
we estimate the statistical parameters. These estimates are calculated through some
\textit{estimators} of the parameters that depend on the sample data.
Typically, a maximum likelihood estimation method can be used to
define such estimators. 

Nevertheless, this can still be insufficient to determine the state: 
by asserting that the probability distribution of an observable
is defined by the state, we mean that there is a mapping from the set
of states to the set of distributions for that observable.
And this mapping is not necessarily invertible.
In order to determine the state, in general, we must consider many observables and 
estimate their probability distributions; then, we should be able to define a mapping from the space
of the probability distributions to the state space, which gives us the best estimation for the
state itself. 
Again, we can use a maximum likelihood estimation scheme to find the function that
makes the observed data the most probable.
We assume that the estimators we use are unbiased and that, for large samples of data, their asymptotic behavior satisfies
reasonable regularity conditions.
It must be pointed out that the validity of the approach we have previously outlined---deriving statistical properties by 
considering large ensembles of replicas of a system---constitutes a general principle of the theory. Therefore, we can state 
the following hypothesis:

\noindent {\textit{It is always possible to build an ensemble of identical replicas
of a system that are prepared in the same state.}}

\mysubsubsection{Ideal measurements}
Roughly speaking, an ideal measurement is one that uses the best measurement
device, which introduces minimal random noise and
has the lowest systematic error. 
When dealing with the measurement of the quantum state,
in which we employ statistical estimators to infer the measured value of the state, 
the choice of each estimator impacts the final precision of the measurement.
If we consider the estimation as part of the measurement process of the state,
an ideal measurement is a process that uses the best 
(i.e., minimum-variance and unbiased) estimators for the quantum state.

Under these conditions, we choose a quantity that represents the contribution of information about the state provided by a 
single measurement. This quantity is expected to be additive (at least asymptotically for a large number of trials), because 
its magnitude is assumed to result from the contribution of many measurements. The most natural choice is the 
{\emph{precision}} (the reciprocal of the variance).
This provides the basis for the following axiom:

\begin{unreferredaxiom} 
{\bf{(Precision Invariance)}} Consider a set of observables $\mathcal{D} = \{ o_1, \dots , o_D\}$, such that each observable 
carries the same amount of information ($n_{o_i} = \textnormal{const}$). The amount of information (in terms of precision) 
about the state given by an ideal measurement on
any observable of the set is independent of the observable being measured.
\end{unreferredaxiom}

\noindent This axiom can be better understood if we imagine having an ensemble of replicas
of a physical system that are assumed to be in the same state.
We ideally estimate the state of the system by measuring the observables for each system of the
ensemble. 
The above axiom asserts that the gain of information (expressed in terms of increased
precision) about the state provided by a single measurement is independent of the choice of the
observable we measure. 

To compare the impact that observables have on the state estimate, they must be placed under the same conditions. For this 
reason, the condition that each observable carries the same amount of information has been introduced in the axiom.

The Precision Invariance axiom can be justified in terms of a ``principle of indifference'':
if a system has a set of observables that are equivalent in their underlying
statistical and informational properties, it is reasonable to assume that the observables
contribute equally to the knowledge of the state.
This must hold independently of the unknown physical state of the system.

Consider, for example, a pair of canonical phase-space variables 
$p$ and $q$. In the ``bare'' system, with no dynamics assigned to it, the variables 
$p$ and $q$ are assumed to be interchangeable. This implies that they 
have, \textit{a priori},  the same statistical and informational properties. If these variables
are discrete, we also assume that they can take the same number, $N$, of
values and, consequently, the amount of information provided by a measurement
on any of them is the same ($\log_2 N$ bits).
Thus, there is no reason why the contribution to the knowledge of the state should be 
different if we measure, for example, $p$ instead of $q$.

As stated before, the application of an estimator is part of the measurement process of the state
and, since the above axiom refers to ideal measurements, we can assume that only minimum-variance
unbiased estimators are involved. In classical statistics, the Cram\'er-Rao bound provides a lower bound on the variance
of an estimator of a given statistical parameter, which implies an upper bound on its precision (with precision defined as
the reciprocal of the variance).
This upper limit corresponds to the Fisher information~\cite{FisherClassic}.
In many sections of this article, we refer to maximum likelihood estimators, which reach
the Cram\'er-Rao bound when the sample size tends to infinity, and the connection with the Fisher metric is the basis for
exploring other formulations of the Precision Invariance axiom. 

It is important to clarify the practical applicability of the axiom we have just introduced: the physical
state is a theoretical construct that is not directly accessible and meaningful, and it is used only to derive other quantities 
(the probability distributions of observables).
Nonetheless, the conditions we impose on the state indirectly affect the precision with which we
estimate the statistical parameters of these distributions.
Therefore, the axiom of Precision Invariance should be interpreted as a condition on the
parameters, which do have a direct physical significance.

A question may arise concerning the interpretation of this axiom and, in general, the concept of measurement precision.
It is evident that our theory incorporates measurement as a constituent element. For this reason, it is unnecessary to 
introduce the Born rule as an additional axiom.
One might ask whether the Precision Invariance axiom can be understood independently of the measurement process.
Rather than interpreting it
in terms of the precision with which the state is estimated through experimental measurements, we might instead consider
an intrinsic precision associated with the state itself.

Alternatively, we might assert that the amount of information (in terms of precision) that can potentially be
transferred---through measurement or any other interaction---is fixed. That is, it does not depend on the particular
process or on how the system interacts with others.

Clearly, this places us in the realm of interpretations. Framing the discussion in terms of measurement remains the most
straightforward way to treat issues related to estimation theory and statistical parameters.

\mysubsubsection{Random measurements -- Interpretation}
As stated before, we can improve the quality of our measurements
by using the best measurement setups, which minimize measurement errors or,
in other words, which maximize the precision. This condition applies to
the pair system+measurement setup, and we require it to be compatible with the 
Limited Information axiom.
What we will discover is that this requirement, and other postulates of the theory,
lead us to non-vanishing values of measurement uncertainty.
We will not analyze the mechanism that generates such a random behavior,
but we must observe that we started by thinking of a system that 
has less information than what we expect by counting its observables. 
This kind of system is not necessarily an ``aleatory'' system. 
The only aleatory element of the chain is the measurement apparatus.

A possible interpretation is that the measurement apparatus introduces some
aleatory behavior when the actual system does not have such information.
This mechanism must be compatible with the axioms of the theory.

\mysubsubsection{Quantization} 

Consider a set of observables $\mathcal{D} = \{o_1, o_2, ... o_D\}$ with 
distributions $\rho_{o_i}(\xi)$, where $\xi$ is a possible measurement outcome of the observable
$o_i \in \mathcal{D}$.
According to our definition of state, the probability distribution of any observable depends on the system's state $x$ and can 
be written in the form $\rho_{o_i}(\xi ; x)$.

Now consider the $D$-dimensional vector $\boldsymbol{\rho}$, whose components
are the probability functions $\rho_o(\xi)$, for $o\in \mathcal{D}$. 
According to this definition, $\boldsymbol{\rho}$ belongs to the Cartesian product of the simplices associated with the observables in $\mathcal{D}$:
$\boldsymbol{\rho}(\cdot) \in \mathlarger{\mathlarger{\mathlarger{\times}}}_{o \in \mathcal{D}}\, \Delta_o$.
We can assume that there exists a set of all the physically admissible 
values of the vector function $\boldsymbol{\rho}$, which we denote by
$\Phi_{\boldsymbol{\rho}}$.
More precisely, $\Phi_{\boldsymbol{\rho}}$ is a subset of
$\mathlarger{\mathlarger{\mathlarger{\times}}}_{o \in \mathcal{D}}\, \Delta_o$,
and is  the image of the state space $\Phi$ under the vector function with components $\rho_{o_i}(\xi ; x)$:
\begin{equation}\label{eq:rhospace}
    \Phi_{\boldsymbol{\rho}} = \{ \left( \rho_o(\cdot; x) \right)_{o \in \mathcal{D}}  \, \suchthat \, x \in \Phi\},
\end{equation}
\noindent where, by $\left( \rho_o(\cdot; x) \right)_{o \in \mathcal{D}}$, we denote the vector with components 
$\rho_o(\cdot; x)$, with $o \in \mathcal{D}$ labeling the components. According to our definition,
the state \textit{uniquely} defines the probability distributions of 
the observables of a system. Consequently, it must be possible
to define a one-to-one mapping between the state space $\Phi$ and the set 
$\Phi_{\boldsymbol{\rho}}$. From a physical point of view, what matters 
is the set of physically admissible distributions of the observables, i.e.,  
$\Phi_{\boldsymbol{\rho}}$. 
This implies that the state space can be expressed in any form (numbers or any other object),
provided that we can state a one-to-one relationship between this set and
each possible value of $\Phi_{\boldsymbol{\rho}}$.

The purpose of this work is to define a quantization, namely, 
for a given set of observables, to understand how the set of
physically admissible distributions is made. 
Or, alternatively, to define a state space and a bijective mapping between it and the set
$\Phi_{\boldsymbol{\rho}}$.

Now consider that, for a given set of observables  $\mathcal{D}$, we have
a set of statistical parameters, or a \textit{parametrization} for that set of observables: 
\begin{equation}  \label{eq:pardef}
\boldsymbol{\theta}_{\mathcal{D}} =
\{\boldsymbol{\theta}_{o_1}, \boldsymbol{\theta}_{o_2} ... , \boldsymbol{\theta}_{o_D}\}.
\end{equation}
\noindent As stated before, if the parametrization of an observable is full and identifiable, 
we can state a one-to-one correspondence between any distribution of the probability
of each observable and a value of its distribution's parameters.
This means that, if we have a space of full and identifiable parameters, we can describe a
system both in terms of physically admissible distributions and in terms of 
sets of values for the parameters that correspond to such admissible distributions.

In the above discussion, we have expressed the probability distributions of an observable $o$
both as dependent on the point $x$ of the state space  and in terms of the
observable's parameters. These two representations correspond, respectively,  to the forms $\rho_o(\xi; x)$ and 
$\rho_o^{\theta}(\xi; \boldsymbol{\theta}_{o})$.

As will become clear in the following sections, it is useful to define an injective mapping between the spaces $\Phi$ and 
$\Phi_{\boldsymbol{\theta}}$. Such a mapping can be introduced because the parametrization $\boldsymbol{\theta}_{o}$ is full 
and identifiable: any value of $\boldsymbol{\theta}_{o}$ determines the functional form of the corresponding distribution, and 
vice versa. Moreover, if the probability distributions associated with each state $x$ are known, that is, if $\rho_o(\xi;x)$ is 
known for each observable $o$, then one can define the functions
 $\boldsymbol{\Theta}_{o} : \Phi \rightarrow \Phi_{\boldsymbol{\theta}}$:
\begin{equation}  \label{eq:xtotheta}
    \boldsymbol{\theta}_o = \boldsymbol{\Theta}_o (x), \, o \in\mathcal{D},
\end{equation}
\noindent that must obey the condition  
$\rho_o(\xi; x) = \rho_o^{\theta}(\xi; \boldsymbol{\theta}_{o})$.  If an observable $o$ is defined on an $N_o$-point discrete
space, the distributions defined on this space have $N_o-1$ degrees of freedom (because of 
normalization condition on the distributions).
Consequently, if a parametrization of an observable $\boldsymbol{\theta}_{o}$ is full and
identifiable, it can be defined in an $N_o - 1$-dimensional space.

\subsection{Foundations of the theory} \label{sec:principlesoftheory}

In the previous section, we introduced the first two axioms of the theory.
In addition to those axioms,
which play a key role in the development of the theory, we are going to introduce other general and
reasonable hypotheses. The full list is presented here:

\begin{axiom} \label{ax:precisioninvariance}
{\bf{(Precision Invariance)}} Consider a set of observables $\mathcal{D} = \{ o_1, \dots , o_D\}$, such that each observable 
carries the same amount of information ($n_{o_i} = \textnormal{const}$). The amount of information (in terms of precision) 
about the state given by an ideal measurement on
any observable of the set is independent of the observable being measured.
\end{axiom}

\begin{axiom} \label{ax:limitedinfo}
 {\bf{(Limited Information)}} In a system with $D$ independent observables, let $n_{o_i}$ denote the amount of information (number of bits) obtained by measuring each observable $o_i$ (e.g., $n_{o_i}=\log_2 N_{o_i}$ if $o_i$ has $N_{o_i}$ outcomes). 
Let $n_I$ be the maximum amount of information we can obtain in any single set of measurements on the system. Then
\begin{equation}\nonumber
 n_I \;<\; \sum_{i=1}^{D} n_{o_i}.
\end{equation}
\end{axiom}
\begin{requirement}  \label{ax:symmetries}
{\bf{(System Symmetries)}}
The mathematical model of a physical system must be invariant under all the natural
symmetries of that system.
\end{requirement}

\begin{requirement} \label{ax:scalability}
{\bf{(Model Independence)}}
If a model uses auxiliary parameters (e.g., the system's position, discretization,
or size relative to the external physical space), the theory's predictions must be invariant under changes in those parameters 
(up to the induced relabeling and rescaling). 
\end{requirement}

\begin{axiom} \label{ax:infostatequivalence}
{\bf{(Information/Statistics Equivalence)}}
Given a system and an observable with a finite number $N$ of values, a one-point distribution
represents the maximum amount of information the observable can carry.
Conversely, if all the values of the observable are {\emph{a priori}} equivalent, that is, if there are no privileged values, the uniform distribution $\rho = 1/N$ corresponds to the case in which the observable carries zero information.
\end{axiom}

\mysubsubsection{Remarks}
We now discuss several points about the principles just presented.

Regarding Axiom~\ref{ax:precisioninvariance}, previously introduced in section~\ref{sec:funddef}, a clarification is necessary: the {\textit{precision}} is
a notion that must not depend on the particular scale we use to represent the state.
The comparison among estimation uncertainties on the state coming from the measurements of different
observables only makes sense if we use the same scale/range on the state space. 

In requirements \ref{ax:symmetries} and \ref{ax:scalability}, we did not use the term {\text{axiom}} because they can be better 
regarded as general principles for constructing the theory.
For instance, \ref{ax:symmetries}  does not specify which symmetries the mathematical model of quantization developed here must 
satisfy. Instead, it can be regarded as a class of axioms.
Nonetheless, this requirement is necessary if we want to build a theory that describes a system with
two canonical variables $p$ and $q$ defined in a phase space. In this case, the 
``bare'' system, with no dynamics, is expected to be invariant under some transformations, e.g., 
translations, scalings, and reflections. Thus, we must require that a quantization 
procedure does not weaken these symmetries.

Requirement \ref{ax:scalability} is similar to the previous one on symmetries, but differs
from it conceptually.
Predictions depend only on the physical system, not on the model-building choices. Any two models that differ only by external 
embedding (translations/rotations), uniform rescaling, or uniform refinement/coarsening of the partitions (including the 
apparatus) must be equivalent.

The model of a physical system on which we base our theory must be viewed as a subsystem embedded
within a larger system (the universe).
Describing physics in terms of a subsystem is just a way of constructing our theory, and the choice of one subsystem over 
another should have no effect on the predictions of that theory.

For this reason, we require the behavior of such a subsystem  to be independent of transformations
(relative to the larger system) such as translations, rescaling, etc.
To be more precise, this requirement concerns symmetries of the model with respect to the external world
while the previous one concerns symmetries of the states within the model itself.

The two requirements just introduced are preceded by axiom \ref{ax:limitedinfo}, discussed earlier,
which  tells us that there is a limitation on the amount of information
a system can carry, with respect to the total information we could get from its observables. 
In the formulation of that axiom, the strength of such a limitation is not specified.
In section \ref{sec:nstatesyst}, we will show that there is only one possible choice for the value of the limitation
that is compatible with the scalability of the model required by \ref{ax:scalability}.
That value will depend on the physical dimensions of the model and establishes 
a connection with the standard quantum theory because it can be expressed 
in terms of Planck's constant $h$.

Axiom  \ref{ax:infostatequivalence} represents a very strong hypothesis: we began our analysis
without any assumption on the non-deterministic nature of a physical system, except for the basic
one, namely, from the existence of a probability distribution.
In this framework, a connection between the form the probability functions can take and
the amount of information carried by the system is not defined. This connection is given by this
last axiom in a very reasonable way. 

As stated before, we use a very basic definition of the ``amount of
information'': it is simply the number of bits we need to define the state of a system.
More precisely, when we define the state of a system by a set of assertions (which we  
refer to as ``predicates''), each assertion counts as a bit of information.
An assertion whose trueness is left undefined counts as zero bits of information.

\mysubsubsection{Free-varying distributions}\label{sec:freevardist}
Consider a system where all the observables are complete and carry the same information, i.e., by following the notation of
definition \ref{def:completesetofobservables}, $n_{o_i}=\textnormal{const.}$ for any observable $o_i$ in the system. As stated 
before, under these conditions, theorem \ref{th:compset} implies that the amount of 
information provided by the measurement of any observable is maximal, i.e., no extra
information is left to the measurement of other observables.
Thus, we cannot perform a joint measurement on $o_0$ and another observable $o_1$.

Moreover, suppose that $o_0$ and $o_1$ are independent, i.e., by measuring $o_0$,
we have no information about the value of $o_1$; 
in other words, once the value of $o_0$ is determined, all the values of the other observable $o_1$ are equivalent.
According to axiom \ref{ax:infostatequivalence}, this implies that the state created by measuring $o_0$ identifies
a probability distribution for $o_1$ that is a uniform distribution.

In a hypothetical world in which all of the measurements/preparations behave in this way,
there would be only one-point and uniform distributions, and  quantum theory would be a very simple
(and perhaps less interesting) theory.
But, as we will show, a richer theory can be built and we will find  
that the complexity of quantum mechanics arises from what takes place between these two cases.

Of course, we are interested in a wider class of systems, where probability distributions take any value. Thus, we must assume 
that, beyond sets of independent observables, we have some observables that are not
completely independent. In other words, we require that the measurement of an observable may provide some
(but not all) information about
another observable. From a classical point of view, this makes sense:
for example, if we measure the energy of a classical harmonic oscillator, we know that its
conjugate Hamiltonian variables $p$ and $q$ must be confined in a region of the phase space, and
this represents some information about $p$ and $q$. 
If we consider this kind of ``partially independent'' observables, we can imagine that there are some
states of a system where an observable is determined and  others where it is not, but it does not
have a uniform distribution.

At this stage, we can neither state which of these types of distributions are allowed, nor under
what conditions.
However, we can require our theory to take these distributions into account.

Therefore, in order to involve the widest class of models in our theory, we require that, for each observable considered in 
isolation, any probability distribution over its outcomes is admissible. This does not imply that every collection of such 
distributions across different observables is jointly realizable by a single state; inter-observable constraints may apply.

This requirement can be expressed as follows. For an observable $o$, consider the set of all its admissible probability 
distributions, i.e. the simplex 
$\Delta_o$.  We can assert that:
\begin{requirement} \label{ax:freevaryingdistr}
{\bf{(Free-varying Distributions)}}
Unless otherwise restricted by the other axioms,
any $\rho_o$ of the probability simplex $\Delta_o$ of an observable $o$ admits a state that corresponds to $\rho_o$.
When the distribution of an observable, taken individually, satisfies this requirement, we say that it is 
\emph{free-varying}.
\end{requirement}

If we refer to the definition \eqref{eq:rhospace} of the space $\Phi_{\boldsymbol{\rho}}$ and use the same notation,  
requirement  \ref{ax:freevaryingdistr}  can also be put in the form: 
{\textit{The projection of the space $\Phi_{\boldsymbol{\rho}}$ onto the probability simplex $\Delta_{o_i}$ of an observable $o_i$ is $\Delta_{o_i}$}} itself. 

This postulate should not be regarded as a restriction on our theory. Rather, it should be viewed as a prescription to be 
followed in our quantum reconstruction program in order to avoid restricting the class of physical systems that the theory can 
describe.

\section{The framework of the Precision Invariance axiom} \label{sec:precinvfranework}

\subsection{Estimation of the Physical State}\label{sec:physstaeest}
\paragraph{Repeated measurement scheme}
Let us now consider a repeated measurement experiment for an ensemble of identical replicas
of a system: suppose that the system has $D$ observables that we collect in the
set $\mathcal{D} = \{o_1, o_2, ... o_D\}$.
For each observable $o_i$ we perform $M_{o_i}$ measurements, and the outcome of a single
measurement is denoted by $\xi_{o_i}^{(h)}$, where $h$, which identifies the measurement, 
ranges in the interval $[1, M_{o_i}]$. This implies that we have a total of $M = \sum_i M_{o_i}$ 
measurements on the ensemble.
The set of measurements for the single observable $o_i$  is denoted by 
$\Xi_{o_i} = \{\xi_{o_i}^{(1)} ... \xi_{o_i}^{(M_{o_i})}\}$, and the set of all measurements for
all the observables of $\mathcal{D}$ is denoted by $\Xi_{\mathcal{D}}$.

If we consider a single observable $o_i$ and a parametrization $\boldsymbol{\theta}_{o_i}$
of its probability distribution, we can define an estimator for
$\boldsymbol{\theta}_{o_i}$, namely, a function of the sample data of the measurement
on $o_i$:  $\hat{ \boldsymbol{\theta}}_{o_i}(\Xi_{o_i})$.

\mysubsubsection{State estimation through a set of observables}
The aim of the repeated measurement experiment described above is to provide an estimate of the
state. Let us consider a statistical model $\mathcal{M} = \{ \rho_o(\xi_o ; x) \suchthat x \in \Phi\}$ 
for a set $\mathcal{D}$ of observables. Note that, by writing distribution in the form
$\rho_o(\xi_o ; x)$ we are simply applying the definition of {\it state} $x$, 
namely, the mathematical object all probability distributions depend on.
This point of view implies that the state $x$ can be handled as a statistical parameter and, for example, we can apply the same 
definition of estimability we commonly use for statistical parameters.
For the purposes of state estimation, we assume that the state admits a representation in a finite-dimensional parameter space 
in which expectations of state estimators are well defined.

Thus, we can say that the state is  \textit{estimable} through the set of observables $\mathcal{D}$ if, 
for an ensemble of $M$ replicas of the system, we can choose
a set of measurements $\Xi_{\mathcal{D}}$ for the observables of
$\mathcal{D}$ and build a function  $\hat{x}(\Xi_{\mathcal{D}})$ of these measurements 
such that $E(\hat{x}(\Xi_{\mathcal{D}})) = x$. In other words, the state is estimable  through
$\mathcal{D}$  if it has an unbiased estimator that is a function of the measurements 
$\Xi_{\mathcal{D}}$. 
In the context of quantum tomography,  the measurements of sets of observables
that make a state estimable are called tomographically complete, and the minimal set of observables is
sometimes called the ``quorum''~\cite{FanoQuantTom}.

\subsection{Local independence of observables} \label{sec:locindep}

Let us consider the functions $\boldsymbol{\Theta}_{o_i}$, defined in
\eqref{eq:xtotheta}, that return the values of the parameters $\boldsymbol{\theta}_{o_i}$ 
starting from the state. 
We can introduce the following definition.

\begin{definition}\label{def:localindependence}{\bf{(Local Independence)}}
The observables of a set $\mathcal{L}$ are
said to be \emph{locally independent} if, for every state $x_T$, there exists a neighborhood  
$T$ of the point $(\boldsymbol{\Theta}_{o_1}(x_T), \dots, \boldsymbol{\Theta}_{o_d}(x_T) )$ in the parameter space,
for which the system
\begin{equation}  \label{eq:xtothetal}
    \boldsymbol{\theta}_{o_i} = \boldsymbol{\Theta}_{o_i} (x), \, i \in [1,d]
\end{equation}
\noindent has a single solution in $x$ for all values of the tuple 
$(\boldsymbol{\theta}_{o_1} ... \boldsymbol{\theta}_{o_d})$ that belong to $T$. 
\end{definition}

Roughly speaking, the observables of the collection $\mathcal{L}$ are locally independent if
the system \eqref{eq:xtothetal} is locally solvable in $x$.

In this definition, the term ``independent'' refers to the fact that the values 
of the parameters can be chosen independently of each other, and for every 
choice we obtain a corresponding valid state $x$.

Note that the definition just introduced does not make any assumptions about the nature (i.e., the topology) of the state 
space. At this stage, a state $x$ can be considered as an element of an arbitrary set, whether numeric or not. The collection 
of all such points $x$, on which the $\rho$ distributions depend, defines the state space. To develop a theory based on these 
foundations, we assume that the state space possesses a ``manageable'' structure, i.e. that the $\Theta$ functions are well-behaved.
In particular, we assume they are continuous and, if necessary, differentiable.

The concept of local independence leads us to introduce the following theorem: 
\begin{theorem} \label{th:locindobs}
If a physical system has a set of $d$ locally independent observables $\{o_1, o_2 ... o_d\}$, which provides a single
local solution of the system
$\boldsymbol{\theta}_{o_i} = \boldsymbol{\Theta}_{o_i} (x) $ (where $i \in [1,d]$)
within a neighborhood of every state $x$, then this is the highest-$d$ set of locally independent
observables.
\end{theorem}
\begin{proof}
Suppose we try to extend a set $\mathcal{L}$  of $d$ locally independent observables
by adding one more observable $o_{d+1}$.
It can be shown that, with this new observable, the local independence is lost: 
since the first $d$ observables are locally independent by hypothesis, once we choose their
parameters (within the image of the neighborhood of a point $x$),
we can solve the system  \eqref{eq:xtothetal} and find a unique solution $x'$.
Consequently, the new observable $o_{d+1}$ must take the value
$\boldsymbol{\Theta}_{o_{d+1}}(x')$ of its parameters. This implies that, for fixed
$\boldsymbol{\theta}_{o_1}...\boldsymbol{\theta}_{o_{d}}$, any value of
$\boldsymbol{\theta}_{o_{d+1}}$ different from $\boldsymbol{\Theta}_{o_{d+1}}(x')$ does not 
correspond to a solution of the state $x$, and we can conclude that the extended
collection of $d+1$ is not a set of independent observables.
The same applies if we try to extend  $\mathcal{L}$ to any number $d' > d$ of
independent observables.
\end{proof}

It should be noted that, in the above theorem, we specified that the observables are
{\it locally} independent.
The reason for this is that, if a function is locally invertible, it is not necessarily fully
invertible (i.e., for all its domain of definition). 
In the theory discussed in this paper, we will always consider systems in which the 
statistical parameters, in general, depend on the state through a multi-valued function 
$\boldsymbol{\Theta}_{o_i} (x)$. 

The concept of locally independent observables must not be confused 
with the concept of state estimability through a set of observables.
If a set of observables were {\it globally} independent, namely, if the function  
$\boldsymbol{\Theta}(x)$ were invertible  in all the state space, once we had estimated the
observables' parameters, we would be able to get the state $x$.
This implies that we could consider the state estimable through that set.
However, in general, we do not have sets of globally independent observables, and we expect the
state to be estimable through a set of observables that is larger than a set of independent ones.

The following theorem provides some information about the structure of the state space:

\begin{theorem}  \label{th:manifolddim}
Suppose that: (i) for each point of the state space, a system has a set $\mathcal{L}$
of $d$  locally independent observables, and
(ii) the parameter vectors $\boldsymbol{\theta}_{o}$ of its observables have equal 
dimensions $N_{\theta}$.
Then the state space is an $N_{\theta} d$-dimensional topological manifold.
\end{theorem}  
\begin{proof}
Let $\Phi_{\boldsymbol{\theta}_o}$ be the space of parameters of each observable, and 
let $N_{\theta}$ be its dimensionality. The space given by the Cartesian product of the
parameter spaces of each observable of the set $\mathcal{L}$, which we indicate by
$\Phi_{\mathcal{L}} = \times_{o \in \mathcal{L}} \Phi_{\boldsymbol{\theta}_o}$, 
has dimension $N_{\theta} \cdot d$.

Since the vectors of this product space $\Phi_{\mathcal{L}}$ are nothing but
$N_{\theta} d$-dimensional tuples of real numbers, we can assert that $\Phi_{\mathcal{L}}$ has the
structure of a Euclidean space. 
According to the definition of local independence, the functions
$\boldsymbol{\Theta}_{o_{i}}(x)$ are invertible in the images of the neighborhood $X$.
Thus, we can locally establish a homeomorphism between an $N_{\theta} d$-dimensional Euclidean
space and $\Phi$.
This is consistent with the definition of an $N_{\theta} d$-dimensional {{topological}} manifold.
\end{proof}

It is worth noting that local independence is a rather strong condition for a set of observables.
A consequence of the above theorem is that if, for every state of a system, we can define a set of 
locally independent observables, then the state space is locally compact because local compactness
is a property of topological manifolds.

\begin{figure}[!ht]
\centering
\includegraphics[scale=0.96]{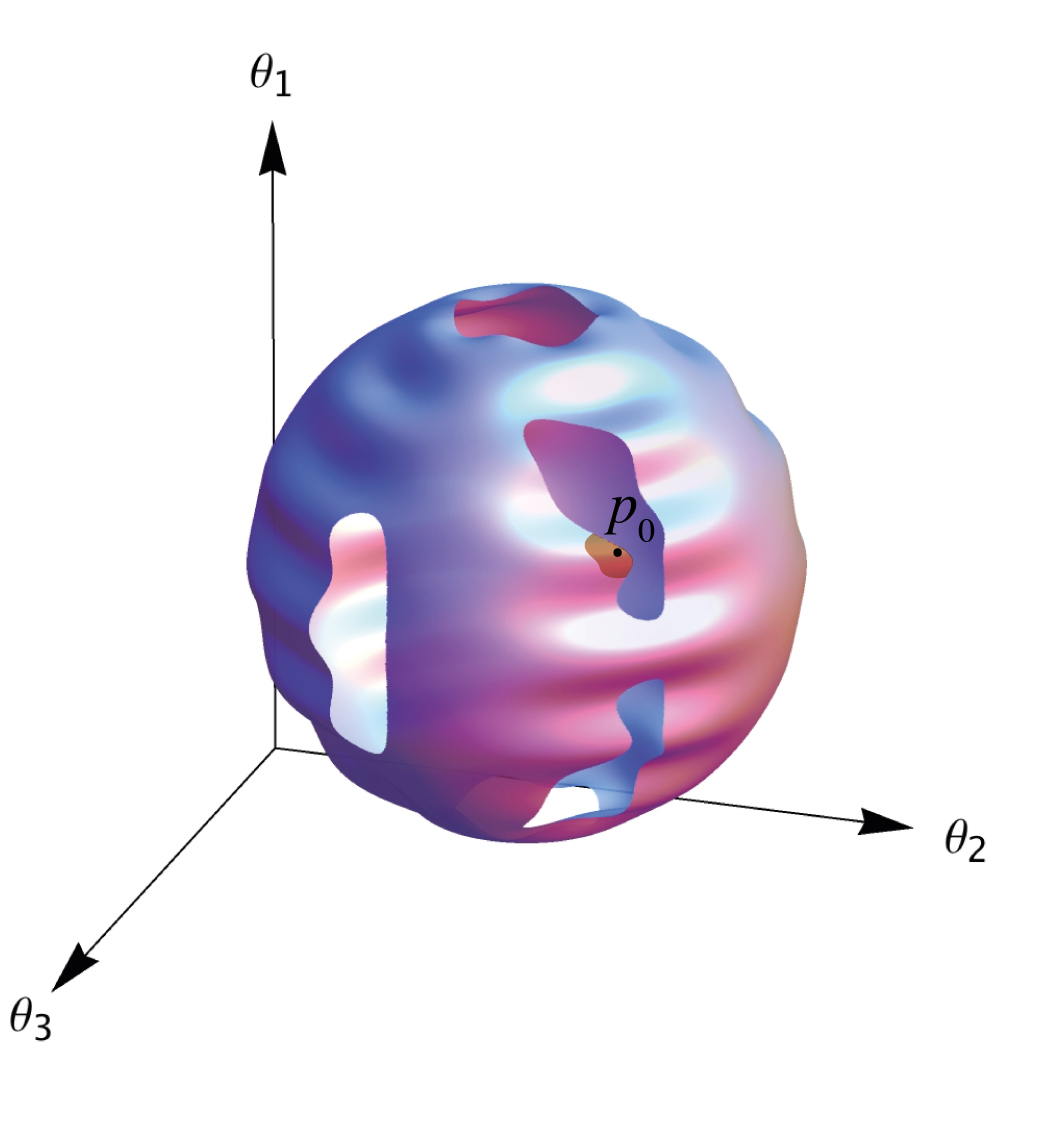}
\caption
{
Example of a non-compact two-dimensional state manifold, immersed in a three-dimensional
parameter space. Any state point $P_0$, resulting from an experiment, must have
a finite neighborhood (represented by a brown spot around $P_0$) belonging to the state manifold, 
which represents the experimental uncertainty; 
we cannot place $P_0$ arbitrarily close to an open boundary, such as
the one that encloses the ``hole'' on the manifold. This excludes the existence of open boundaries, which typically occur in non-compact sets.
}
\label{noncompactstatemanif}
\end{figure}

Since every observable considered in this work has a finite set of outcomes, its probability distributions belong to a finite-
dimensional probability simplex $\Delta_o$, which is compact. Therefore, for a finite set of observables $\mathcal{D}$,
the product $\prod_{o\in\mathcal{D}}\Delta_o$ is compact.

Let us now recall that the probability representation of the state space satisfies
\begin{equation}\label{eq:phirhosubsetofsimplex}
\Phi_{\boldsymbol{\rho}}
\subseteq
\prod_{o\in\mathcal{D}}\Delta_o,
\end{equation}
\noindent and introduce the following assumption, which is physically quite natural: a convergent sequence of physically 
admissible states cannot converge to a nonphysical probability vector. This is equivalent to requiring that the set of 
physically admissible probability vectors $\Phi_{\boldsymbol{\rho}}$ be closed. Therefore, since, by 
Eq.\eqref{eq:phirhosubsetofsimplex}, $\Phi_{\boldsymbol{\rho}}$ is a closed subset of a compact space, it is itself compact.

We can therefore conclude that, since the state space $\Phi$ and its probability representation $\Phi_{\boldsymbol{\rho}}$ are 
homeomorphic, the state space is a {\it{compact}} manifold.
An intuitive representation of the compactness property of the state space is
depicted in Fig.\ref{noncompactstatemanif}.

\subsection{Applications and consistency issues of the Precision Invariance axiom}\label{sec:consistencyissues}

The guiding thread of this study is the idea that the Precision Invariance axiom, together with other hypotheses, imposes 
conditions on the State Space that allow us to determine its structure. The application of the Precision Invariance axiom is 
based on the estimation scheme presented in the previous subsection.

The state space $\Phi_{\mathbf{x}}$ is a construct that we introduced primarily to map the probability
distributions of sets of observables, and there are no special requirements on its structure.
In pursuing our goal of deriving the structure of the state space from the precision invariance axiom,
this freedom of choice leads us to a first problem, which can be synthesized as follows.
Assume that we can define a system of local coordinates $\mathbf{x}$ on
a chart of the manifold $\Phi_{\mathbf{x}}$ (we denote these coordinates in the same way 
as in the state).
The mathematical structure of $\Phi_{\mathbf{x}}$ is not defined {\it{a priori}}: 
we can choose among different representations of the coordinates $\mathbf{x}$ that lead to 
the same set of physically valid values for the statistical parameters of the observables of
$\mathcal{D}$.

Suppose, for example, that we have two observables $a$ and $b$ and, as a result of a repeated
measurement experiment on an ensemble of identical replicas of the system, we  derive
an estimate of the state $\mathbf{x}$ from their measurements.
If $a$ and $b$ are locally independent, then, in the neighborhood of a state space point,
the function $\Theta(\mathbf{x})$ can be inverted, and this allows us to obtain an estimate of
the state.
This estimate has an uncertainty that depends on the estimate precision of the statistical
parameters of $a$ and $b$. The estimate of the parameters of each observable contributes to the
precision along a particular direction on the manifold $\Phi_{\mathbf{x}}$.

The problem with the Precision Invariance axiom is that we do not know how to compare the 
precision contribution of the measurement of $a$ with that  of another observable $b$.
We can imagine stretching or contracting $\Phi_{\mathbf{x}}$
along the direction associated with the observable $a$, and making the precision contribution
of the observable $a$ smaller or larger and even equal to that of the other observable $b$.
This would make the Precision Invariance requirement worthless.
In other words, when we talk about precision in estimating the state, we must define the scale
that the precision refers to, for each observable. 

This problem can be overcome through symmetry-related arguments: 
in many physical systems, we can assume that there is a symmetry that makes the observables of the
set $\mathcal{D}$ interchangeable. 
In order to have a theory that is manifestly invariant under this kind of symmetry, we must choose
the parametrizations that make the dependence of $\rho(\zeta^o, \boldsymbol{\theta}_o)$ on
each  parameter equivalent for all $o \in \mathcal{D}$.
If that symmetry holds, we also expect the range of the parameters
$\boldsymbol{\theta}_o$ to be the same for all the observables in $\mathcal{D}$.
In other words, if there is a requirement for the exchange symmetry
of the observables, we cannot arbitrarily redefine the scale of  $\Phi_{\mathbf{x}}$
along a particular direction to fulfill the conditions for Precision Invariance.
We must take these conditions on the parametrization into account when investigating the implications of the Precision Invariance axiom.

\subsection{A basic example: the two-observable elementary system}\label{sec:twoobsbasic}

Let us examine the case of an elementary system described by two one-bit
variables, which can carry at most one bit of information.
The most interesting case is when only one observable is locally independent.
By following the notation used in the previous section, in this case, we have $d = 1$  locally
independent observables and the state is estimable through a set of two observables (thus $D =2$). 
Furthermore, we have $N_{\theta}=1$, because each observable can take two values, which implies that the probability 
distribution has one degree of freedom and can be defined by a single one-dimensional statistical
parameter.
Since there is only one independent observable, according to theorem
\ref{th:manifolddim}, the state space must be a one-dimensional manifold.

Let $p$ and $q$ be our observables and $0$ and $1$ the values they can take.
Their probability distributions are denoted by $\rho_p$ and $\rho_q$.

The state is, by definition, the mathematical object that uniquely defines the 
statistical properties (i.e., the distributions) of all the observables of a system.
Thus, once we have found the set of the allowed probability distributions, we can say that we have found how the state space is made.

Let us apply the hypothesis that the system carries one bit of information, in conjunction with
axiom \ref{ax:infostatequivalence} on the information/statistics equivalence: if, for instance, 
$q$ is determined and has the value 0, its distribution is
$\rho_q(0) = 1, \rho_q(1) = 0$. According to axiom \ref{ax:infostatequivalence}, the amount of
information used for $q$, is 1. Since the total amount of information we can use is,
by hypothesis, 1, we have no more bits available for $p$. Thus, $p$ is totally
undetermined and, again, according to axiom \ref{ax:infostatequivalence}, its probability
distribution is uniform: $\rho_p(0) = 1/2, \,\, \rho_p(1) = 1/2$.

We can apply the same arguments for the value 1 of  $q$ and consider the case in which $p$ is
determined and $q$ is undetermined. We have four cases:
\begin{equation}\label{eq:knownpqvals}
\begin{array} {lll}
 \bf{p}_1: \, & \rho_q = (1, 0),     & \,  \rho_p = (1/2, 1/2) \\
 \bf{p}_2: \, &  \rho_q = (0, 1),     & \,  \rho_p = (1/2, 1/2)  \\
 \bf{p}_3: \, & \rho_q = (1/2, 1/2), & \,  \rho_p = (1, 0)     \\
 \bf{p}_4: \, & \rho_q = (1/2, 1/2), & \,  \rho_p = (0, 1). 
\end{array}
\end{equation}
\noindent From this point on, we refer to these special sets of probability distributions as {\it{cardinal points}}. 
We can define the parametrizations $\theta_{p}$ and $\theta_{p}$ for the distributions of $p$ and
$q$ through the equations
\begin{equation} \label{eq:rhovstheta}
\begin{array} {llc}
\rho_{\nu}(0; \theta_{\nu}) & = & \cos^2 \frac {\theta_{\nu}}{2}  \\
\rho_{\nu}(1; \theta_{\nu}) & = & \sin^2 \frac {\theta_{\nu}}{2}  
\end{array}
\end{equation}
\noindent where $\nu = q, p$. The above expressions fulfill the normalization conditions and make
the parametrizations  full and identifiable because they cover all the possible values of the
distributions $\rho$.
At this stage, we intentionally do not specify  the intervals where the parameters $\theta_{\nu}$ are defined; 
we will take advantage of this indefiniteness to manage the multivalued behavior of the mappings $\boldsymbol{\theta}_{o_i} (x)$, 
stated in the previous section.
Furthermore, for $\theta_{\nu} = 0,\pi$ we obtain the cardinal points \eqref{eq:knownpqvals}. It must be pointed out that 
Eq.~\eqref{eq:rhovstheta} does not assume, in any way, that the state space is a circle, because at this stage, the 
probabilities $\rho_p(0; \theta_p)$ and $\rho_q(0; \theta_q)$ depend on two independent parameters $\theta_p$ and $\theta_q$.

The measurements of each observable of the system behave as a Bernoulli trial process whose
probabilities  are $\rho_{\nu}(0)$ and $\rho_{\nu}(1)$, where $\nu = p,q$.
The variance of the estimator of $\theta$ for a repeated trial scheme
(namely, an ensemble of $N$ replicas of our elementary system) is expected to be
\footnote
{
The identity \eqref{eq:vartheta} can be demonstrated by treating, for instance,
$\rho_{\nu}(0; \theta_{\nu})$ as an ordinary statistical parameter.
By simplifying the notation with $\rho = \rho_{\nu}(0; \theta_{\nu})$, it can be shown that
the maximum likelihood estimator of $\rho$ for a Bernoulli process is $\hat{\rho} = N_{\xi=0}/N$ 
(where $N_{\xi=0}$ is the number of measurements whose result is 0). 
Its variance is, up to a factor of $1/N^2$, the variance of $ N_{\xi=0}$ itself, which is simply
the variance of the binomial distribution, $N \rho (1-\rho)$. 
Therefore, we have $\var \hat{\rho}  = 1/N \rho (1-\rho)$. 
\\
If we apply definition \eqref{eq:rhovstheta}, this can be expressed in terms of $\theta$:
$ \var \hat{\rho} = \frac {1}{N}  \cos^2 \frac {\theta}{2}  \sin^2 \frac {\theta}{2}$
(the subscript $\nu$ has been omitted).
By substituting this variance and the derivative 
$\frac{\partial \rho}{\partial \theta} = \cos \frac {\theta}{2}  \sin \frac {\theta}{2}$
into the variance propagation formula 
$\var (\rho) = \lvert \frac{\partial \rho}{\partial \theta} \rvert^{2} \var(\theta)$, 
we can obtain the variance of the estimator of $\theta$ given by equation
\eqref{eq:vartheta}.
}:
\begin{equation} \label{eq:vartheta}
\var \hat{\theta}_{\nu} = \frac{1}{N}, \,\, \nu = p,q.
\end{equation}
\noindent Both $\theta_p$ and $\theta_q$ depend on the state $x$ through the functions 
$\Theta_p(x)$ and $\Theta_q(x)$ introduced by \eqref{eq:xtotheta}.
Since, in our system, the maximum number of locally independent observables is 1, the state space
is expected to be one-dimensional.
If we consider only {\it{local}} independence of the observables, the functions 
$\Theta_p(x)$ and $\Theta_q(x)$ are expected to be invertible only within the neighborhoods of
state space points.

In order to get the form of the functions $\Theta_{\nu}(x)$ and understand the structure of the
state space, we can establish more conditions on the state estimation obtained from the
measurements. 
This can be done with the help of the Precision Invariance axiom. Let us build a
``thought experiment'' for the two-observable system: 
 we perform a first set of measurements on $p$ and $q$, which allows us to approximately 
estimate their values. If we knew the form of the functions $\Theta_{\nu}(x)$, we could derive the
estimated value of the state $x$ and its estimation error $\Delta x$. We assume that our
sample is large enough to make $\Delta x$ so small that $\Theta_p(x)$ and $\Theta_q(x)$ are
invertible in the range $\Delta x$, as required by the definition of local independence.

After this, we complete our thought experiment by performing a second set of measurements, and
obtain a further estimate of the state.
If we assume that the functions $\Theta_{\nu}(x)$ are well-behaved and differentiable, 
the experimental uncertainty of this estimate is given by the standard rule
for the propagation of experimental errors:
\begin{equation} \label{eq:stderrprop}
\var(x)\bigm\vert_{ \nu}  = \lvert \frac{\partial \Theta_{\nu}(x)}{\partial x} \rvert^{-2} 
      \var (\theta_{\nu}), \,\, \nu = p,q,
\end{equation}
\noindent where $\var(x)\bigm\vert_{ \nu}$ is the variance of the state obtained by measuring $p$ or
$q$. According to the axiom of Precision Invariance, the contribution to the precision 
of the estimation of $x$, given by a single measurement, must be the same for $p$ and $q$.
From expression \eqref{eq:vartheta}
we know that the estimation precision of $\theta_{\nu}$ for  the Bernoulli trial process is
$1/\var \hat{\theta}_{\nu} = N$, which implies that the contribution of a single measurement is 1.
Suppose that we can express this condition in the form 
$\precision_{\mathit{sing.}} (\hat{\theta}_{\nu}) = 1$.

Since the precision is the inverse of the variance, we can apply  \eqref{eq:stderrprop} to this
condition, and obtain the contribution to the precision \textit{on the state} given by a single
measurement on $p$ or $q$:  ${\precision_{\mathit{sing.}}} (x)\bigm\vert_{\nu} =
\lvert \frac{\partial \Theta_{\nu}(x)}{\partial x} \rvert^{2}\,(\nu = p,q) $.

With the Precision Invariance axiom, we assert that the measurement of different
observables contributes in the same way to the estimation
precision of the state. This can be expressed in terms of the notation we have just introduced:
${\precision_{\mathit{sing.}}}(x)\bigm\vert_{p} = {\precision_{\mathit{sing.}}}(x)\bigm\vert_{q}$.
Therefore, we can write the equation:
\begin{equation} \label{eq:errpropcomp}
\lvert \frac{\partial \Theta_{p}(x)}{\partial x} \rvert =  
\lvert \frac{\partial \Theta_{q}(x)}{\partial x} \rvert, 
\end{equation}
\noindent which leads to the condition:
\begin{equation} \label{eq:thetapvsthetaq}
 \theta_p =  \pm \theta_q + {\textnormal const.}
 \end{equation}
\noindent With the parametrization \eqref{eq:rhovstheta},
this solution is consistent with the cardinal points  \eqref{eq:knownpqvals} only if:
\begin{equation}\label{eq:solutionfor2d}
\theta_p = \pi/2 - \theta_q.
\end{equation}
It should be noted that we have chosen the parametrization \eqref{eq:rhovstheta} 
just because it leads to a simpler expression of the metric tensor $g$. 
It can be shown that the mutual dependence of $\rho_p$ and $\rho_q$, which defines the 
structure of the state space and represents our final result,
does not depend on what parametrization we have chosen. 

Since $\rho_p$ and $\rho_q$ are the probabilities of two full and identifiable observables,
they can be used to represent the system's state. We can also use a different system state
representation  through the variables
\begin{equation}\label{eq:spqdef}
\begin{array} {ll}
S_q =  \rho_q(0) - \rho_q(1) = \cos \theta_q \\
S_p =  \rho_p(0) - \rho_p(1) =  \cos \theta_p.
\end{array}
\end{equation}
\noindent Since $\theta_p = \pi/2 - \theta_q$, we can also write $ S_p = \sin \theta_q $.
Thus, in the  $S_q$, $S_p$ vector space, the system state manifold is represented by a circle:
\begin{equation} \nonumber
S_q^2 + S_p^2 = 1.
\end{equation}
\noindent In fact, the system we have just described is not novel. It is easy to recognize in it a
standard quantum-mechanical two-state system with real coefficients,
usually referred to as \textit{rebit} or ``real qubit''~\cite{Qbit1} because it carries one bit of
information.
The amplitudes associated with the eigenstates of $p$ and $q$ are $\pm\sqrt{ \rho_{\nu}(0)}$ and
$\pm\sqrt{ \rho_{\nu}(1)}$, with $\nu=p,q$.  Such amplitudes are confined on a unit circle.

In conclusion, we have shown how to apply the Precision Invariance axiom in a simple case,
in which we directly compare the contribution to the 
estimation of the state coming from two observables. This allowed us to derive how the observables' distributions depend on the state. 

If we go to higher-dimensional state spaces, the sets of locally independent
observables are not single-element sets. Consequently, when we apply the Precision Invariance
axiom, we must compare the measurement precision on {\it sets} of observables. 
This implies that we cannot obtain equations in the simple form
of \eqref{eq:errpropcomp} and \eqref{eq:stderrprop}.
In the following subsections, we introduce a formalism suitable for expressing the
Precision Invariance axiom, even in the general multi-dimensional case.

\section{The Fisher metric formulation of the Precision Invariance axiom}\label{sec:precinvfishmetr}

This section presents a framework, based on the Fisher information matrix, that allows the Precision Invariance axiom to be 
applied to a wide class of physical systems. The central result introduced here is the Fisher metric conservation theorem, 
which is a consequence of the Precision Invariance axiom. 

The Fisher metric conservation theorem will be the key tool for establishing the structure of the state space. 

In order to demonstrate it, in this section we introduce some mathematical background: 
we introduce a formalism based on the Fisher information matrix, which allows us to express the Precision 
Invariance axiom in a form suitable for complex problems.
Furthermore,  we study the problem of deriving the maximum likelihood estimators 
for a set of statistical parameters with the constraint that such parameters
represent a physically admissible point in the state space.
After that, we identify a form of the Fisher information that represents
the precision of the estimate of the state for this constrained problem and allows us to  
isolate the contribution to the state estimation 
given by each observable.
On the basis of this contribution, we investigate a way to express the 
Precision Invariance axiom in a formal way.
Then we introduce the {\textit{orthogonal extension}}
of an observable's statistical parameter as a system of local (tangent to the state
space) orthogonal coordinates.
In this context, the Precision Invariance axiom takes the form of a (Fisher) metric-preserving
property among the extended parameters corresponding to the observables.
The special case of a subset of the state space, for this metric-preserving
property, is also considered.

In what follows,  we will adopt the notations introduced in the previous sections:
$\mathcal{D} = \{o_1, o_2, ... o_D\}$ is a set of $D$ observables, where $D$ is large enough to make 
the state estimable through such a set.
Moreover, we suppose that the amount of information (measured in bits) carried by a measurement's
result is the same for all the observables in $\mathcal{D}$.
The distribution of each observable of $\mathcal{D}$ depends on a set of $N_{\theta}$ parameters,
and only $d$ of $D$ observables can be locally independent.
In other words, we assume that the state space  $\Phi_{\mathbf{x}}$ is $N_{\theta} d$-dimensional.
A set of $d$ locally independent observables will be denoted by $\mathcal{L}$.

\subsection{The connection with the Fisher information}
\label{sec:formalepiaxiom}

In order to express the precision invariance axiom in a more general form,
which can be applied to cases other than the one illustrated in the previous subsection,
we should establish how the ``contribution in terms of precision'',
referred to in that axiom, can be measured. The Cram\'er-Rao bound, in the form 
\begin{equation} 
    \frac{1}{\var \hat{\theta}} \le \mathcal{I}(\theta),
\end{equation}
\noindent indicates an upper limit for 
the precision $\frac{1}{\var \hat{\theta}}$ on the estimation of a statistical parameter $\theta$
obtained from a set of measurements. 
This upper limit is represented by the Fisher information
$I(\theta)$ of the set of measurements with respect to $\theta$ and,
under the usual regularity conditions, is asymptotically attained by
the maximum likelihood estimator.

The state of a system, by definition, determines the values of all its statistical parameters.
Conversely, a set of statistical parameters does not, in general, determine the state, unless we
are dealing with a set of locally independent observables, which allows (locally) to define a map
from the parameter space to the state space. 

In short, our idea is to apply the Cram\'er-Rao bound to determine how much 
a measurement improves the estimation precision of a set of parameters.
Furthermore, since our goal is to estimate the state with the help of a 
mapping from the parameter space to the state space, we must understand how the estimation precision
for the parameters impacts the estimation precision of the state.
Of course, since we are dealing with more than one statistical parameter and many observables,
we are interested in the multivariate version of the Cram\'er-Rao bound, which involves the
covariance matrix and the Fisher information matrix.

Now suppose that we have a set $\mathcal{L}$ of $d$ observables whose parameters are locally
independent in the neighborhood of a point $\mathbf{x}_0$ of the state space. 
Furthermore, suppose that the state is fully estimable through a larger set of observables
$\mathcal{D}$, and that an extremely large set of measurements on $\mathcal{D}$ can  provide
a first estimate $\mathbf{x}_0$ of the state. In this context, the 
Precision Invariance axiom asserts that the improvement (in terms of precision)
around the point $\mathbf{x}_0$ of the state estimation that we obtain from further
measurements is the same for all the observables of the collection  $\mathcal{D}$. 

In what follows, we will see how to express the state estimation precision coming from a specific observable, 
and from this formulate the Precision Invariance axiom in a special form.

\subsection{Constrained maximum likelihood estimation and \allowbreak\hspace{0pt} Fisher information} \label{sec:constrmle}

\mysubsubsection{The log-Likelihood function}
For large samples of data, the most efficient unbiased estimators, which make the
Cram\'er-Rao inequality be fulfilled on its limiting value, are the maximum likelihood estimators.
In our case, if the measurements of observables are statistically independent, the likelihood
function of the whole set of measurements can be factorized into the product of the likelihood
functions of each measurement. Thus, by following the notation introduced in the previous section,
if the state is estimable through a set $\mathcal{D}$  of $D$ observables, the log-likelihood
function can be expanded into the sum
\begin{equation}  \label{eq:mle}
\sum_{o \in \mathcal{D}} \sum_{i = 1}^{M_o}  \ell (\boldsymbol{\theta}^o, \zeta^o_i),
\end{equation}
\noindent where $M_o$ indicates how many measurements are performed on $o$, and
$\ell (\boldsymbol{\theta}^o, \zeta^o_i)$ is the log-likelihood of the single measurement
$\zeta^o_i$ on the observable $o$, which is defined as the
logarithm of the probability of the measurement outcome $\zeta^o_i$ for a given parameter value
$\boldsymbol{\theta}^o(\zeta^o_i)$:  $\ell (\boldsymbol{\theta}^o, \zeta^o_i) =
\log \rho (\zeta^o_i, \boldsymbol{\theta}^o) $.
Note that, in  \eqref{eq:mle}, the sum over the whole set of measurements has been divided into a
sum over the observables and a sum over the measurements of each observable.

Let us denote by $\Psi_{\boldsymbol{\theta}}$ the space of the parameters of the observables of
$\mathcal{D}$. 
$\Psi_{\boldsymbol{\theta}}$ is $N_{\theta} D$-dimensional because there are
$N_{\theta}$ parameters for each observable.
Now consider that, by hypothesis, we can define the map $\boldsymbol{\Theta}(x)$ from the state
space to the space $\Psi_{\boldsymbol{\theta}}$, and it is reasonable to suppose that this map is differentiable. Thus, if we 
have at most $d$ out of $D$ locally independent observables,
we can view the map $\boldsymbol{\Theta}(x)$ as an immersion of the $N_{\theta} d$-dimensional state
space $\Phi_{\mathbf{x}}$ into the $N_{\theta} D$-dimensional space $\Psi_{\boldsymbol{\theta}}$. 

It must be pointed out that the differentiability of the map $\boldsymbol{\Theta}(x)$ (and of its inverse) is the underlying 
assumption behind all the derivations presented in this subsection.

\mysubsubsection{The State Space constraint}
Suppose we can express the property that a point of the parameter space
belongs to the set of physically valid points, or, in other words, that it corresponds to a point 
of the state space, through a condition of the form 
\begin{equation} \label{eq:phicond}
\varphi_{a, k}(\boldsymbol{\theta}^{o_1}, ... , \boldsymbol{\theta}^{o_D}) = 0,
\end{equation}
\noindent where we assume that the Jacobian of the constraint function $\varphi_{a, k}$ has full rank. Actually, we must have 
$N_{\theta}(D-d)$ conditions like this, because they reduce the
space from $N_{\theta} D$ to $N_{\theta} d$ dimensions. This is why
$\varphi$ depends on two indices, whose ranges are $1...D-d$ and $1...N_{\theta}$.

We can introduce this constraint in the derivation of the maximum likelihood estimator by adding,
to the expression  \eqref{eq:mle}, a term of the form
$\sum_k \lambda_{a, k} \varphi_{a, k}(\boldsymbol{\theta}^{o_1}, ... , \boldsymbol{\theta}^{o_D})$, 
where $\lambda_{a, k}$ represents a Lagrange multiplier. Thus, the constrained log-likelihood reads:
\begin{equation}  \label{eq:mlec}
\sum_o \sum_{i = 1}^{M_o}  \ell (\boldsymbol{\theta}^o, \zeta^o_i) + 
\sum_{a, k} \lambda_{a, k} \varphi_{a, k}(\boldsymbol{\theta}^{o_1}, ... \boldsymbol{\theta}^{o_D}).
\end{equation}
\noindent In its standard definition, the Fisher information can be expressed in terms of the
log-likelihood function. Thus, the above discussion provides the ground for the introduction of a
Fisher information-based approach to the problem, which takes into account the state space constraint.

\mysubsubsection{Extension of the state space}
Now, Suppose that we can extend the $N_{\theta} d$-dimensional space $\Phi_{\mathbf{x}}$ into an
$N_{\theta} D$-dimensional one, by adding $N_{\theta} (D-d)$ arbitrary components.
In other words, we extend $\Phi_{\mathbf{x}}$ to match the dimensionality of the space in which it is immersed.
We indicate by $\mathbf{x}^{\perp}$ these additional components and  by
$\overline{\mathbf{x}} = (\mathbf{x}, \mathbf{x}^{\perp})$ a vector of the 
space generated by this extension. There are just two conditions that we require for this extended
space:
(i) it admits  an invertible  map onto the $N_{\theta} D$-dimensional
parameter space $\Psi_{\boldsymbol{\theta}}$, and (ii) the manifold  $\Phi_{\mathbf{x}}$
must be identified by a constant value of the extension component, say  
$\mathbf{x}^{\perp}_{\Phi}$.
We denote the map required by condition (i) as $\overline{\boldsymbol{\Theta}}(\overline{\mathbf{x}})$, or equivalently as 
$\overline{\boldsymbol{\Theta}}(\mathbf{x}, \mathbf{x}^{\perp}) $, and its inverse as $\overline{\mathbf{X}}(\boldsymbol{\theta})$. 
Requirement (ii) can be expressed in terms of the mapping
$\Phi_{\mathbf{x}} \rightarrow \Psi_{\boldsymbol{\theta}}$ we just introduced:
$\overline{\boldsymbol{\Theta}}(\mathbf{x}, \mathbf{x}^{\perp}_{\Phi}) =
\boldsymbol{\Theta}(\mathbf{x})$. 

It can be shown that, if we consider only points within a region close to the manifold $\Phi_{\mathbf{x}}$ and can
identify the constraint functions $\varphi_{a, k}(\boldsymbol{\theta})$,
then the inverse mapping $\overline{\mathbf{X}}$ can be derived;
note that we have $N_{\theta} (D-d)$ functions $\varphi_{a, k}(\boldsymbol{\theta})$.
We can define the following $\Phi_{\mathbf{x}} \rightarrow \Psi_{\boldsymbol{\theta}}$ mapping,
which involves the extension part in the neighborhood of a point $\mathbf{x}_0$:
\begin{equation} \label{eq:xtndxtothetamap}
    \theta^{o\, j} = \Theta^{o\, j}(\mathbf{x})  +    
    \frac{\partial \varphi_{a, k}}{\partial \theta^{o\, j}}\biggm\vert_{\mathbf{x}_0}
    ({x^{\perp}}^{a, k} - {x^{\perp}_{\Phi}}^{a, k}).
\end{equation}
\noindent This definition is consistent with requirement (ii) because, for
${x^{\perp}}^{a, k} = {x^{\perp}_{\Phi}}^{a, k}$, we get
$\theta^{o\, j} = \Theta^{o\, j}(\mathbf{x})$, which are the values of the parameters corresponding
to state space points.  
Assuming that the functions $\Theta^{o\, j}(\mathbf{x})$ are differentiable, we can replace the first term of the RHS of 
Eq.~\eqref{eq:xtndxtothetamap} with the power series expansion 
in the neighborhood of  $\mathbf{x}_0$. Thus, we obtain 
\begin{equation}
\begin{aligned}
\label{eq:xtndxtothetamap2}
    \theta^{o\, j} =
    \frac{\partial \Theta^{o\, j}}{\partial x^i}\biggm\vert_{\mathbf{x}_0} (x^i - x^i_0) + 
    \Theta^{o\, j}(\mathbf{x}_0) 
     \twocolbreak +
    \frac{\partial \varphi_{a, k}}{\partial \theta^{o\, j}}\biggm\vert_{\mathbf{x}_0}
    ({x^{\perp}}^{a, k} - {x^{\perp}_{\Phi}}^{a, k}),
\end{aligned}
\end{equation}
\noindent where, in the RHS, we can recognize an extended state-space vector
$\Delta \overline{\mathbf{x}} = (\Delta \mathbf{x}, \Delta \mathbf{x}^{\perp})$,
whose components are $\Delta x^i =  x^i - x^i_0$ and
$\Delta {x^{\perp}}^{a, k} = {x^{\perp}}^{a, k} - {x^{\perp}_{\Phi}}^{a, k}$. 
We can also define an $N_{\theta} D \times N_{\theta} D$ matrix 
$\overline{\boldsymbol{\Theta}}$, which operates on the space of $\Delta \overline{\mathbf{x}}$,
whose first $N_{\theta} d$ columns are
$\frac{\partial \Theta^{o\, j}}{\partial x^i}\bigm\vert_{\mathbf{x}_0}$ (being $o, j$ the row indices),
and $\frac{\partial \varphi_{a, k}}{\partial \theta^{o\, j}}\bigm\vert_{\mathbf{x}_0}$
are the remaining $N_{\theta} (D-d)$ columns.
With this notation, equation \eqref{eq:xtndxtothetamap2} can be synthesized as 
\begin{equation} \label{eq:xtndxtothetamapsynt}
\boldsymbol{\theta} = 
\overline{\boldsymbol{\Theta}}\cdot\Delta \overline{\mathbf{x}} + 
\boldsymbol{\Theta}(\mathbf{x}_0). 
\end{equation}
\noindent It can be shown that the matrix $\overline{\boldsymbol{\Theta}}$ is non-singular:
its first $N_{\theta} d$ columns $\frac{\partial \Theta^{o\, j}}{\partial x^i}\bigm\vert_{\mathbf{x}_0}$ 
are clearly $\Psi_{\boldsymbol{\theta}}$-space vectors  tangent to the state manifold
$\Phi_{\mathbf{x}}$, while the remaining columns 
$\frac{\partial \varphi_{a, k}}{\partial \theta^{o\, j}}\bigm\vert_{\mathbf{x}_0}$
are the gradient of the constraint that defines the $\Psi_{\boldsymbol{\theta}}$
space, which is supposed to be normal to the surface $\Phi_{\mathbf{x}}$.
This implies that these two sets of vectors are linearly independent and that the matrix
$\overline{\boldsymbol{\Theta}}$ is non-singular.
Thus we can define its inverse and rewrite equation \eqref{eq:xtndxtothetamapsynt} as
\begin{equation} \label{eq:thetatoxtndxmap}
\Delta \overline{\mathbf{x}} =
{\overline{\boldsymbol{\Theta}}}^{-1} \cdot
(\boldsymbol{\theta} -
\boldsymbol{\Theta}(\mathbf{x}_0)). 
\end{equation}
\noindent In other words, we have obtained, for points sufficiently close to the manifold 
$\Phi_{\mathbf{x}}$, the inverse of the
mapping $\overline{\boldsymbol{\Theta}}(\overline{\mathbf{x}})$, as required by condition (i).  

This new representation can be used to handle the constrained problem \eqref{eq:mlec}.
By relying on the existence of the map $\overline{\boldsymbol{\Theta}}(\overline{\mathbf{x}})$,
we can think of determining a maximum of  
the likelihood \eqref{eq:mlec} in terms of the extended coordinates $\overline{\mathbf{x}}$.
We just need to calculate the partial derivatives with respect to $\overline{\mathbf{x}}$
and find the values of $\overline{\mathbf{X}}$ that make such derivatives equal to zero. 
The partial derivative with respect to $\overline{\mathbf{x}}$ can be put in the form 
$\frac{\partial }{\partial \overline{x}^j } = 
\sum_i \frac{\partial  \overline{\Theta}^{o \, i}}{\partial \overline{x}^j} 
\frac{\partial}{\partial {\theta}^{o \, i}}$.

This discussion leads us to the conclusion that, if we assume that the state space constraint
can be expressed in the form \eqref{eq:mlec}, a maximum likelihood  estimation scheme can be used 
for a model defined in terms of sets of locally independent observables.

\mysubsubsection{The Fisher information for the constrained problem}
The great benefit of expressing the maximum likelihood condition in terms of the coordinates introduced above
is that we know for sure that, for a solution that satisfies the constraint \eqref{eq:phicond},  the 
value of the coordinate extension is $\mathbf{x}^{\perp} = \mathbf{x}^{\perp}_{\Phi}$.
Furthermore, if we fix $\mathbf{x}^{\perp} = \mathbf{x}^{\perp}_{\Phi}$, we always have
$\phi_{a,k}(\boldsymbol{\Theta}(\mathbf{x})) = 0$, and the maximum likelihood condition for the
state space  component $\mathbf{x}$ of the coordinates $\overline{\mathbf{x}}$ reads
\begin{equation}  \label{eq:mlxcond}
 \frac{\partial  \Theta^{o \, i}}{\partial x^j} 
\frac{\partial}{\partial {\theta}^{o \, i}}
\sum_o \sum_{m = 1}^{M_o}  \ell (\boldsymbol{\theta}^o, \zeta^o_m) = 0, 
\end{equation}
\noindent where we have dropped the overline from $\overline{\Theta}$ and the summation in $i$ is
implied. The solutions of the above equation are the maximum likelihood estimators for the state
$\mathbf{x}$, and the variance of these estimators 
is related to the Fisher information containing the same partial derivatives occurring in
\eqref{eq:mlxcond}, namely, by expressing the log-likelihood function in terms of probabilities:
\begin{equation}\begin{aligned} \label{eq:xfisher}
         & \mathcal{I}_{i j} (\mathbf{x}) = - \sum_o
         \textnormal{E} \left(
      \frac{\partial^2}{\partial x^i \partial x^j} 
      \log \rho(\zeta^{(o)};\boldsymbol{\Theta}_{o}(\mathbf{x}) ) 
      \biggm\vert\mathbf{x}
      \right) \\ & = 
   - \sum_o \frac{\partial \Theta_o^{k}}{\partial x^i} 
      \textnormal{E} \left(
      \frac{\partial^2}{\partial \theta_o^k \partial \theta_o^l} 
      \log \rho(\zeta^{(o)};\boldsymbol{\theta}_o  )
      \biggm\vert\boldsymbol{\theta}_o
    \right)
    \frac{\partial \Theta_o^{l}}{\partial x^j} \twocolbreak & = 
    \sum_o \frac{\partial \Theta_o^{k}}{\partial x^i} 
    \mathcal{I}_{kl}^{(o)} (\boldsymbol{\theta}_o)
    \frac{\partial \Theta_o^{l}}{\partial x^j},
\end{aligned}\end{equation}
\noindent where we made use of the equation $\textnormal{E}\left[ \frac{\partial\log\rho}{\partial\theta^k} \right]=0$ and $\zeta^{(o)}$ represents the random variable
for which we compute the expectation value. The third equality comes from the expression 
 \begin{equation} \label{eq:thetafisheri}
         \mathcal{I}_{kl}^{(o)} (\boldsymbol{\theta}_o) = 
     - \textnormal{E} \left(
      \frac{\partial^2}{\partial \theta^k_{o} \partial \theta^l_{o}} 
      \log \rho(\zeta^{(o)};\boldsymbol{\theta}_o)
      \biggm\vert\boldsymbol{\theta}_o
    \right),
\end{equation}
\noindent
which is the Fisher information for a single measurement of an observable $o$
with statistical estimators  $\hat{\theta}^k_{o}$. Since the range of the index $k$ is $1 ... N_{\theta}$, 
the above equation defines an $N_{\theta} \times N_{\theta}$ matrix. 
In order to simplify our derivations, we require
$\mathcal{I}_{kl}^{(o)}(\boldsymbol{\theta}_o)$ to be diagonal. This condition is satisfied by the parametrizations used in 
this paper. In particular, it is obtained when the probability distribution is factorized into normalized conditional 
probabilities depending on different components of the parameter vector $\boldsymbol{\theta}_o$. For observable spaces whose 
cardinality is a power of two, an explicit construction of this type is provided by the binary-tree decomposition of 
Eq.~\eqref{eq:partitionedrho}.

The Fisher information matrix defined in \eqref{eq:xfisher} is related to the estimation precision
of the state, obtained from the measurement of all the observables of $\mathcal{D}$, but we can
isolate what comes from a single observable:
\begin{definition} 
{\bf{Single observable Fisher information matrix on the state space}}\newline
Contribution to the state 
estimation of a single observable
\begin{equation}\label{eq:fisheriterm}
         \mathcal{I}_{i j}^{(o)} (\mathbf{x})  = 
    \frac{\partial \Theta_o^{k}}{\partial x^i} 
    \mathcal{I}_{kl}^{(o)} (\boldsymbol{\theta}_o)
    \frac{\partial \Theta_o^{l}}{\partial x^j}.
\end{equation} 
\end{definition}

\noindent This quantity is what we expect to be involved by the Precision Invariance axiom, 
but refers to a single measurement on an observable $o$, which is not sufficient to provide an
estimate of the state.
We have already encountered this kind of limitation in \ref{sec:twoobsbasic},
when we discussed the elementary two-observable system.
In that case, we built a thought experiment with a large number of measurements. Now we can 
follow a similar idea and think of the state estimation as a process that
works by further refinements, as a result of a large number of measurements:
once we have obtained a first estimate of the state, say $\mathbf{x}_0 $, we make
further measurements to improve this estimate, 
which is expected to be in a neighborhood of $\mathbf{x}_0$. Thus we have to move the estimated
point on the manifold $\Phi_{\mathbf{x}}$ to reach this new value. We can use a system of local
coordinates on $\Phi_{\mathbf{x}}$, centered in $\mathbf{x}_0 $
to represent that refinement. By keeping in mind the Cram\'er-Rao inequality,
we can state the following: 

\noindent{\textbf{In case of (maximum likelihood) state estimations performed with a large set of
measurements on some observables, the Fisher information defined in  \eqref{eq:fisheriterm}
is the contribution, to the estimation precision, of a single measurement on an observable $o$. }
\labelText{\ddag}{plhold:fishifestconnection}}

It should be noted that, although $\mathcal{I}_{i j}^{(o)} (\mathbf{x})$ is an
$N_{\theta} d \times N_{\theta} d$ matrix, its actual rank is 
at most $N_{\theta}$ because, according to expression \eqref{eq:fisheriterm}, 
$\mathcal{I}_{i j}^{(o)} (\mathbf{x})$ is the result of a product 
of three matrices whose dimensions are
$N_{\theta} d \times N_{\theta}$, $N_{\theta} \times N_{\theta}$ and 
$N_{\theta} \times N_{\theta} d$, respectively.
This implies that there is at most an $N_{\theta}$-dimensional subspace where the matrix
$\mathcal{I}_{i j}^{(o)} (\mathbf{x})$ is non-singular.
Actually, this subspace is the one generated by the vectors
$\frac{\partial \Theta_o^{k}}{\partial x^i}$ occurring in expression \eqref{eq:fisheriterm},
and it is expected to be different for each observable $o \in \mathcal{D}$.

\mysubsubsection{Cotangent and tangent spaces}

 Equation \eqref{eq:fisheriterm} contains matrix terms of the form
$\frac{\partial \Theta_o^{k}}{\partial x^i}$, which actually are the first-order coefficients of the
power expansion of the functions $\boldsymbol{\Theta}_{o}(\mathbf{x})$.
If we study these functions in a local domain, i.e., in a small neighborhood of a state space point,
we can approximate them by linear maps with coefficients
$\frac{\partial \Theta_o^{k}}{\partial x^i}$.
Such a coefficient can be regarded as an element of a set of $N_{\theta} d$-dimensional vectors,
where $i$ denotes the component index of a vector and $k$ identifies the vector within the set.
We identify such vectors as ${\mathbf{h}^*}_{\boldsymbol{\theta}_{o}}^k$, where
${{h}_{\boldsymbol{\theta}_{o}}^k}_i = \frac{\partial \Theta_o^{k}}{\partial x^i}$ is their
component, and define them  as {\textit{cotangent vectors}} at the point $\mathbf{x}$ of
$\Phi_{\mathbf{x}}$.
We indicate the space they belong to as the cotangent space $T^*_x$, but in 
${{h}_{\boldsymbol{\theta}_{o}}^k}_i$ we drop the asterisk ``*'' and put the index $i$ in lower
position; this agrees with the standard convention adopted for cotangent vectors (covariant components).
The normalized version of the vectors ${\mathbf{h}^*}_{\boldsymbol{\theta}_{o}}^k$ is 
${\mathbf{b}^*}_{\boldsymbol{\theta}_{o}}^k =
{\iota^*}_{\boldsymbol{\theta}_{o}}^k {\mathbf{h}^*}_{\boldsymbol{\theta}_{o}}^k$, where
$\left( {\iota^*}_{\boldsymbol{\theta}_{o}}^k \right)^{-2} =
\vert{{\mathbf{h}^*}_{\boldsymbol{\theta}_{o}}}^k \vert^2$.

Since, for an observable $o$,  we have $N_{\theta}$  parameters ${\boldsymbol{\theta}_{o}}$, while the dimensionality of the 
state space is $N_{\theta} d$, the vectors ${\mathbf{h}^*}_{\boldsymbol{\theta}_{o}}^k$ do not represent a complete basis in 
the cotangent space of $\Phi_{\mathbf{x}}$ in $\mathbf{x}$.
On the other hand, if we consider a set of $d$ locally independent observables, we can assert that the vectors
${\mathbf{h}^*}_{\boldsymbol{\theta}_{o}}^k\, \forall o \in \mathcal{L}$ form a basis because we have $N_{\theta} d$ vectors
and the dimensionality of $\Phi_{\mathbf{x}}$ is  $N_{\theta} d$ by definition.
Note that this basis is not necessarily orthogonal. 
In principle, we can define on a chart of a $\Phi_{\mathbf{x}}$-covering atlas a system of curvilinear coordinates 
$\boldsymbol{\omega} = \boldsymbol{\omega}(\mathbf{x})$  and from this an $N_{\theta} d$-dimensional cotangent space.
Furthermore, under the hypothesis of local 
independence, we can define the inverse $\mathbf{X}(\boldsymbol{\omega})$ of the mapping
$\boldsymbol{\omega}(\mathbf{x})$, and from this the dual of the cotangent space. 

The notation for the vectors of these spaces is shown in table \ref{tab:tangentcotangentvectors}.
By putting $\theta_o^{k} = {\omega}^k$ and $\Theta_o^{k}(\mathbf{x}) = \omega^{k}(\mathbf{x})$, with $k \le N_{\theta}$,
we obtain the previous definitions of the cotangent vectors for the parameters $\boldsymbol{\theta}$.

\begin{table}[h]
\centering
\setlength{\tabcolsep}{2pt}
\begin{tabular}{c|c|c}
Meaning & Symbol & Definition \\
\hline
\multicolumn{3}{c}{\parbox{6.5cm}{\centering 
\vspace{0.6ex}
\bf Cotangent space $T^*_x \Phi_{\mathbf{x}}$ / \\ Covariant vectors
\vspace{0.6ex}
}} \\
\hline
\parbox{3.2cm}{\centering Cotangent\\ vector}
&
${\mathbf{h}^{*}}_{\boldsymbol{\omega}}^{k}$
&
$\displaystyle
h^{k}_{\boldsymbol{\omega} i}
=
\frac{\partial \omega^k}{\partial x^i}
$
\\[3.0ex]

\parbox{3.2cm}{\centering Normalization\\ factor}
&
${\iota^{*}}_{\boldsymbol{\omega}}^{k}$
&
$\displaystyle
\left({\iota^{*}}_{\boldsymbol{\omega}}^{k}\right)^{-2}
=
\lvert{\mathbf{h}^{*}}_{\boldsymbol{\omega}}^{k}\rvert^2
$
\\[3.0ex]

\parbox{3.2cm}{\centering Normalized\\ vector}
&
${\mathbf{b}^{*}}_{\boldsymbol{\omega}}^{k}$
&
$\displaystyle
{\mathbf{b}^{*}}_{\boldsymbol{\omega}}^{k}
=
{\iota^{*}}_{\boldsymbol{\omega}}^{k}
{\mathbf{h}^{*}}_{\boldsymbol{\omega}}^{k}
$
\\[3.0ex]

\hline
\multicolumn{3}{c}{\parbox{6.5cm}{\centering 
\vspace{0.6ex}
\bf Tangent space $T_x \Phi_{\mathbf{x}}$ / \\ Contravariant vectors
\vspace{0.6ex}
}} \\
\hline
\parbox{3.2cm}{\centering Tangent vector}
&
${{\mathbf{h}}_{\boldsymbol{\omega}}}_{k}$
&
$\displaystyle
{h_{\boldsymbol{\omega}}}_{k}^i
=
\frac{\partial X^i}{\partial \omega^k}
$
\\[3.0ex]

\parbox{3.2cm}{\centering Normalization\\factor}
&
${{\iota}_{\boldsymbol{\omega}}}_{k}$
&
$\displaystyle
\left({{\iota}_{\boldsymbol{\omega}}}_{k}\right)^{-2}
=
\lvert{{\mathbf{h}}_{\boldsymbol{\omega}}}_{k}\rvert^2
$
\\[3.0ex]

\parbox{3.2cm}{\centering Normalized\\ tangent vector}
&
${{\mathbf{b}}_{\boldsymbol{\omega}}}_{k}$
&
$\displaystyle
{{{\mathbf{b}}_{\boldsymbol{\omega}}}}_{k}
=
{{\iota}_{\boldsymbol{\omega}}}_{k}
{{\mathbf{h}}_{\boldsymbol{\omega}}}_{k}
$
\end{tabular}
\caption{Definitions of the tangent and cotangent vectors, and of their normalization factors, associated with a generic 
curvilinear coordinate system $\boldsymbol{\omega}$.}
\label{tab:tangentcotangentvectors}
\end{table}
\subsection{Orthogonal parameter extensions} \label{sec:orthoparext}
\mysubsubsection{The fixed-$\boldsymbol{\theta}$ manifold and its orthogonal complement}
For each observable $o$ of $\mathcal{D}$, we can define 
$\Phi_{\mathbf{x}}^{\bot \boldsymbol{\theta}_{o}}$ as the submanifold of $\Phi_{\mathbf{x}}$
obtained by  keeping  $\boldsymbol{\theta}_{o}$ fixed.
In other words, $\Phi_{\mathbf{x}}^{\bot \boldsymbol{\theta}_{o}}$ is the set of $\mathbf{x}$
such that $\boldsymbol{\Theta}_{o}(\mathbf{x}) = \boldsymbol{\theta}_{o}$. The dimension of
these submanifolds is $N_{\theta}(d-1)$ because $N_{\theta} d$ is the dimension of
$\Phi_{\mathbf{x}}$ and $N_{\theta}$ is the dimension of the vector $\boldsymbol{\theta}_{o}$.
If we consider the tangent space at a point $\mathbf{x}$, we can write the following decomposition: 
\begin{equation} \label{eq:ststespacepart}
T_x \Phi_{\mathbf{x}} = V \oplus
T_x \Phi_{\mathbf{x}}^{\bot \boldsymbol{\theta}_{o}},
\end{equation}
\noindent where $V$ is the orthogonal complement of $ T_x \Phi_{\mathbf{x}}^{\bot \boldsymbol{\theta}_{o}}$,
with respect to $T_x \Phi_{\mathbf{x}}$. Of course, there is a linear space $V$ for each point of 
$\Phi_{\mathbf{x}}^{\bot \boldsymbol{\theta}_{o}}$.
In this local setting, we can hypothesize that the linear space $V$ is a space tangent to a submanifold of 
$\Phi_{\mathbf{x}}^{\boldsymbol{\theta}_{o}}$, i.e.,  $V = T_x \Phi_{\mathbf{x}}^{ \boldsymbol{\theta}_{o}}$.
Of course, there is a linear space $V$ and a submanifold $\Phi_{\mathbf{x}}^{ \boldsymbol{\theta}_{o}}$ (if it exists) for each 
point of $\Phi_{\mathbf{x}}^{\bot \boldsymbol{\theta}_{o}}$.
In what follows, we will use an auxiliary
Euclidean metric $\delta_{ij}$ on the ambient coordinate space to define
norms and orthogonality of the extension directions; it should not be
confused with the Fisher or extended Fisher metric derived below from
the Precision Invariance axiom.

\mysubsubsection{Existence of the orthogonal complement $\Phi_{\mathbf{x}}^{\boldsymbol{\theta}_{o}}$}
We can ask ourselves whether the above decomposition and the existence of the submanifold 
$\Phi_{\mathbf{x}}^{\boldsymbol{\theta}_{o}}$ can be extended to a global setting (at least on a chart of a $\Phi_{\mathbf{x}}$-
covering atlas), namely, whether we can build a submanifold whose tangent vectors are orthogonal to the submanifold  
$\Phi_{\mathbf{x}}^{\bot \boldsymbol{\theta}_{o}}$, 
that is, to its tangent space  $T_x \Phi_{\mathbf{x}}^{\bot \boldsymbol{\theta}_{o}}$.

The vectors orthogonal to the constant-$\boldsymbol{\theta}_{o}$ submanifold $\Phi_{\mathbf{x}}^{\bot \boldsymbol{\theta}_{o}}$ 
are generated by the gradients of the $\theta_o^k$s, i.e., $\frac{\partial \Theta_o^{k}}{\partial x^i}$. More precisely, at a 
given point, the space generated by a linear combination of these gradients is always orthogonal to
$\Phi_{\mathbf{x}}^{\bot \boldsymbol{\theta}_{o}}$. Thus, we must build the integral manifold from this vector field of the 
gradients or, in other words, we must build the manifold that passes through the field lines of the vectors $\frac{\partial 
\Theta_o^{k}}{\partial x^i}$, for all $k$.

The existence of this manifold is not guaranteed a priori.
Indeed, we have to assume that, for instance, if we choose two parameters $\theta_o^a$ and
$\theta_o^b$, moving first along the flow generated by $\frac{\partial \Theta_o^{a}}{\partial x^i}$
and then along the flow generated by $\frac{\partial \Theta_o^{b}}{\partial x^i}$, or vice versa,
leads locally to the same point of the space $\Phi_{\mathbf{x}}$.

This requirement can be expressed by defining the differential operator
$\nabla \Theta_o^a = \frac{\partial \Theta_o^a}{\partial x^i} \frac{\partial}{\partial x_i}$,
i.e., the directional derivative along the gradient of $\Theta_o^a$.
If the order in which two such directional derivatives are applied is not
relevant, then moving along the two corresponding gradient flows in different
orders leads locally to the same point.
A sufficient condition for the existence of the integral manifolds is that the directional derivatives of two parameters commute:
\begin{equation} \label{eq:nablacommutator}
[ \nabla \Theta_o^a, \nabla \Theta_o^b ] = 0
\end{equation}
\noindent This is the standard problem of the integrability of a distribution, which is addressed
by Frobenius' theorem\footnote{Actually the integrability condition \eqref{eq:nablacommutator} is stronger than the ones 
required by Frobenius' theorem}; see, for example, \cite{LeeSmoothManifolds}.
The above condition can be justified in the context of our repeated trial scheme of statistical estimation.

In order to apply the Precision Invariance axiom and make the estimation possible, we must be able to define an estimated value 
with a statistical error, namely, $\theta_o^l \pm \Delta \theta_o^l$.
We can consider the statistical error as the possibility that a displacement, say $\delta \theta$ (smaller than the statistical 
error $\Delta \theta$), applies to a starting estimated value ${\theta^l}^0$. If we take two parameters $\theta_o^a$ and 
$\theta_o^b$, we must assume that the order we apply this displacement does not change the estimated point of the state space 
we have moved to.

This corresponds to the requirement that, by changing the order in which a point on the state space is moved along two gradient 
lines of two different parameters, we do not change the final point. 

We can therefore conclude that the  submanifolds $\Phi_{\mathbf{x}}^{ \boldsymbol{\theta}_{o}}$ can be built from the gradient 
lines indicated by $\frac{\partial \Theta_o^{k}}{\partial x^i}$. 
We can use the parameters $\boldsymbol{\theta}_{o}$ as curvilinear coordinates on each manifold
$\Phi_{\mathbf{x}}^{ \boldsymbol{\theta}_{o}}$: if we vary a given $\boldsymbol{\theta}^k_{o}$ and keep the other 
$\boldsymbol{\theta}^l_{o}$ with $l\ne k$, we obtain a coordinate line.

It must be remarked that the gradients are supposed to be linearly independent,  by the local independence hypotheses. These 
submanifolds are orthogonal to $T_x \Phi_{\mathbf{x}}^{\bot \boldsymbol{\theta}_{o}}$ (because the gradients are orthogonal to 
that space) and its dimensionality is $N_{\theta}$ (because there are $N_{\theta}$ components ${\theta}_{o}$ of which we 
calculate the gradients). 
If we sum this dimensionality with the dimensionality of the space
$T_x \Phi_{\mathbf{x}}^{\bot \boldsymbol{\theta}_{o}}$, we 
get the dimensionality of the state space. We are therefore guaranteed that the tangent space
$T_x \Phi_{\mathbf{x}}^{ \boldsymbol{\theta}_{o}}$ is actually the orthogonal complement of 
$T_x \Phi_{\mathbf{x}}^{\bot \boldsymbol{\theta}_{o}}$. 
In what follows we will define a system of coordinates on the manifold  $\Phi_{\mathbf{x}}^{\bot \boldsymbol{\theta}_{o}}$
(on the manifold $\Phi_{\mathbf{x}}^{\boldsymbol{\theta}_{o}}$ we already have $\boldsymbol{\theta}_{o}$), and write an 
orthogonality condition based on that system.

\mysubsubsection{The extended coordinates} 
The above discussion allows us to introduce, for each observable $o$ with parameter vector $\boldsymbol{\theta}_{o}$,
 a system of
coordinates in the submanifolds $\Phi_{\mathbf{x}}^{\bot \boldsymbol{\theta}_{o}}$.
We denote these coordinates by $\boldsymbol{\alpha}_o =
(\alpha_{o}^1, ...,  \alpha_{o}^{N_{\theta}(d-1)})$. Their definition depends on the value of
$\boldsymbol{\theta}_{o}$ that fixes the submanifold.
We also assume that each submanifold $\Phi_{\mathbf{x}}^{ \boldsymbol{\theta}_{o}}$ is identified by a value of 
$\boldsymbol{\alpha}_{o}$ or, equivalently, $\boldsymbol{\alpha}_{o}$ is constant on each
$\Phi_{\mathbf{x}}^{ \boldsymbol{\theta}_{o}}$.

The $\alpha$ and the $\theta$ coordinates can be grouped into a unique tuple of coordinates: 
\begin{definition}\label{def:orthogonalextension}{\bf{Orthogonal Extension}}\newline
In the system of coordinates
\begin{equation}\nonumber
\omega_{o} = (\theta_o^1, ..., \theta_o^{N_{\theta}},
\alpha_{o}^1, ..., \alpha_{o}^{N_{\theta}(d-1)}),
\end{equation}
\noindent the components
$\alpha_{o}^k$ are defined as the {\emph {orthogonal extension}} of the parameters 
$\boldsymbol{\theta}_{o}$, and the whole tuple $\omega_{o}$
is referred to as {\emph {extended coordinates}}.
\end{definition}

There is an implicit hypothesis in this definition:  that the mapping $\boldsymbol{\omega} \rightarrow {\mathbf{x}}$ is 
``well-behaved'', namely that it is differentiable and that there is an inverse mapping $X(\boldsymbol{\omega})$.
If these conditions hold, then the Jacobian determinant is nonzero:
\begin{equation}\label{eq:nonzerojaco}
\det\left(
\frac{\partial \omega^k_{o}}{\partial x^i}
\right)\neq 0 .
\end{equation}
\noindent
We can apply the definitions and the notation introduced in table \ref{tab:tangentcotangentvectors} to the
extended coordinates $\omega_o$.

\paragraph{Orthogonality condition}
As stated before, there is a submanifold $\Phi_{\mathbf{x}}^{ \boldsymbol{\theta}_{o}}$ and a tangent space
$T_x \Phi_{\mathbf{x}}^{ \boldsymbol{\theta}_{o}}$ for each point of $\Phi_{\mathbf{x}}^{\bot \boldsymbol{\theta}_{o}}$.
The coordinates $\boldsymbol{\alpha}_o$ can be chosen to label the transverse integral manifolds 
$\Phi_{\mathbf{x}}^{\boldsymbol{\theta}_o}$, and therefore they are constant on such manifolds. On the other hand, they provide 
coordinates on the leaves $\Phi_{\mathbf{x}}^{\bot\boldsymbol{\theta}_o}$, which are labelled by $\boldsymbol{\theta}_o$.

In conclusion we have two families of manifolds, $\Phi_{\mathbf{x}}^{\bot \boldsymbol{\theta}_{o}}$ and 
$\Phi_{\mathbf{x}}^{\boldsymbol{\theta}_{o}}$, with coordinates $\boldsymbol{\alpha}_o$  and $\boldsymbol{\theta}_o$, 
respectively, whose tangent spaces are orthogonal. If we express this orthogonality property in terms of the vectors of the two 
tangent spaces, we can write:
\begin{equation}\label{eq:alphathetaortho}
\frac{\partial X^i}{\partial \alpha_o^{k}}
\frac{\partial X_i}{\partial \theta_o^{l}} = 0, \ \ \  k \le N_{\theta}(d-1), \   l \le N_{\theta} .
\end{equation}

\subsection{The Precision Invariance axiom in terms of Fisher information} 
\mysubsubsection{Sufficient conditions for the state space structure vs.  Precision Invariance axiom}
As stated before, our purpose is to derive the structure of the state space $\Phi_{\mathbf{x}}$,
which, in our framework, can be described by functions $\boldsymbol{\Theta}(\mathbf{x})$ and by their domain. 

The Precision Invariance axiom can provide some constraints on the form of the functions
$\boldsymbol{\Theta}(\mathbf{x})$.
This leads us to the question of whether we have enough conditions to derive the functions
$\boldsymbol{\Theta}(\mathbf{x})$.

If we examine this problem in a local 
context, i.e., in the neighborhood of a point of the state space, the functions
$\boldsymbol{\Theta}(\mathbf{x})$ can be approximated by linear functions plus a constant term.
These linear functions are related to the matrix
$\frac{\partial \Theta_o^{k}}{\partial x^i}$, which has $N_{\theta} d \times N_{\theta}$ elements
(see the discussion that follows equation
\eqref{eq:fisheriterm}). Since there is one such matrix for each of the $d$ observables of  $\mathcal{D}$, 
the total number of degrees of freedom required to determine the functions
$\boldsymbol{\Theta}(\mathbf{x})$,  {\it locally}, is $N_{\theta}^2 D d$.

On the other hand, when we try to identify the structure of the state space,
our approach is to establish a relation among different parametrizations of the observables
of $\mathcal{D}$; the state coordinate, which we denote by $x$, is just a support.
In other words, our purpose is to identify the set of all physically valid parameters of the
observables of $\mathcal{D}$, up to a remapping of the state space, i.e., a transformation of the
form $x \rightarrow x'$. A linear approximation of this mapping
is represented by an $N_{\theta} d \times N_{\theta} d$ matrix.
If we want to neglect the degrees of freedom related to this arbitrary remapping of (the linear approximation of) functions 
$\boldsymbol{\Theta}(\mathbf{x})$, we must decrease the effective
number of degrees of freedom to be determined by $N_{\theta}^2 d^2$.
This leads to an amount of $N_{\theta}^2 (D-d) d$ degrees of freedom.

If we consider the contribution of a single observable to the estimation of the state,
the knowledge of its statistical parameters fixes $N_{\theta}$ degrees of freedom.
This also applies if we act in a local context (in the neighborhood of a state space point).
When we apply the Precision Invariance axiom to a set of $D$ observables, we
require the contribution to the state estimation precision to be equal for all observables,  which leads to at most $D-1$ 
independent comparisons (we do not need to compare this quantity for all pairs of observables, i.e. $D(D-1)/2$ times, because 
they would not be independent). Actually, since the ``precision'' refers to the inverse 
of the variance, which in the multivariate ($N_{\theta} > 1$) case is an
$N_{\theta} \times N_{\theta}$ matrix, we can conclude that the Precision Invariance axiom can
provide at most $N_{\theta}^2 (D-1)$ constraints, which, for specific values of $d$ and $D$,
may be insufficient to fix $N_{\theta}^2 (D-d) d$ degrees
of freedom of the family of functions $\boldsymbol{\Theta}_{o}(\mathbf{x})$. 
In the cases examined in this study, we always have enough conditions to determine the functions
$\boldsymbol{\Theta}_{o}(\mathbf{x})$.

\mysubsubsection{A consistent formulation of the Precision Invariance axiom}
As anticipated in section \ref{sec:consistencyissues}, by arbitrarily stretching or contracting the
state space, we can change how much each observable contributes to the system
state's estimation precision.

These deformations of the state space are represented by a change
of the functions $\boldsymbol{\Theta}_{o}(\mathbf{x})$. 
On the other hand, as we will see in the sections that follow, there can be some exchange symmetries
between the observables; this implies that we cannot arbitrarily reshape the state space in any
direction, because the functions
$\boldsymbol{\Theta}_{o}(\mathbf{x})$ must be defined in a way that does not invalidate such
symmetries.
A way to formulate the Precision Invariance axiom independently of deformations of the state space
is to assume that the spaces of the parameters have equal size.

\paragraph{Precision Invariance axiom as an invariance property of the Fisher information} If we assume that the above 
requirements are satisfied and recall the connection \nameref{plhold:fishifestconnection}
between the Fisher information and the contribution to the state estimation precision of a single
measurement on an observable, the idea is to compare different contributions to the estimation precision of the state in the 
form \eqref{eq:fisheriterm} and express the Precision Invariance axiom in terms of Fisher information.

In its basic form, the Precision Invariance axiom says that the contribution of each observable to the estimation precision of 
the state must be the same; actually, the precision estimate is not represented by a number but by a matrix
(the inverse covariance matrix or Fisher Information matrix). 

As shown in the previous subsections, the Fisher matrix representing the contribution of a single observable to the state 
estimation (Eq.~\eqref{eq:fisheriterm}) is degenerate, i.e., its actual rank is less than its size, and 
$\mathcal{I}_{i j}^{(o)} (\mathbf{x})$ represents a precision estimation only in an $N_{\theta}$-dimensional subspace.
Thus, we have to compare different matrices, associated with different observables, that operate on different subspaces and 
enforce their equality.
Suppose that this can be expressed as the requirement that these matrices, reprojected onto a common subspace,
must be equal to a common matrix $\gamma^{i j}$.
However, a consistency issue, which we are going to fix, arises because there is no predetermined mapping to dictate which 
elements of one observable's Fisher matrix correspond to those of another, i.e., how to map those matrices to a common 
subspace. 

In order to apply the Precision Invariance axiom consistently, the comparison of the contributions of different observables to 
the precision of the state estimation must therefore be independent of arbitrary relabelings of the parameters within each 
observable's parameter set. More precisely, since no physical criterion singles out a particular correspondence between these 
parameters, the Fisher matrix in Eq.~\eqref{eq:fisheriterm} must be invariant under arbitrary permutations of its parameter 
components.

In order to complete the picture and explicitly write an equation that describes this comparison, we
have to reproject the Fisher matrices 
\eqref{eq:fisheriterm}, relative to different observables, with their own subspace, to the same space.

We can consider the matrix $\mathcal{I}_{i j}^{(o)} (\mathbf{x})$ as an ordinary symmetric matrix, which can be diagonalized 
through a set of orthogonal (normalized) eigenvectors ${\mathbf{e}^k}$, with components ${e^{k}}^i$. Since we know that the 
actual rank of that matrix is 
$N_{\theta}$, we can say that only $N_{\theta}$ eigenvalues differ from zero. We assume that the eigenvectors   
${\mathbf{e}^k}$ are ordered in such a way that the first $N_{\theta}$ of them correspond to nonvanishing eigenvalues.

By multiplying left and right   $\mathcal{I}_{i j}^{(o)} (\mathbf{x})$ by the unit vectors ${\mathbf{e}^k}$, 
\begin{equation}\label{eq:reprojectedfisheri}
{e^{k}}^i \mathcal{I}_{i j}^{(o)} (\mathbf{x}){e^{l}}^j, 
\end{equation}
\noindent we obtain a diagonal matrix that represents the same estimation precision as the original
$\mathcal{I}_{i j}^{(o)} (\mathbf{x})$ (because it has the same eigenvalues), but it takes non-vanishing values only for its 
first $N_{\theta}$ rows and columns. This expression can be derived for all the observables of the collection $\mathcal{D}$, 
where each observable has its set of unit vectors ${\mathbf{e}^k}_{(o)}$

Now that the Fisher matrices associated with the different observables
have been represented on the same $N_\theta$-dimensional subspace, the
comparison among them must be independent of the arbitrary correspondence
chosen among their nonzero precision directions. In particular, it must
be invariant under arbitrary permutations of these directions.
Since the only real diagonal matrix invariant under arbitrary permutations of its components is a multiple of the identity, we 
can express the Precision Invariance axiom through the following

\noindent\textbf{Condition (i)}
Equality to the same constant value $\delta_{i j} \gamma$ of the Fisher matrix components for all the observables:
\begin{equation}\label{eq:precindaxe}
{e^{k}}^i_{(o)} \mathcal{I}_{i j}^{(o)} (\mathbf{x}){e^{l}}^j_{(o)} =  
       \delta^{k l} \gamma \  \   k,l \le N_{\theta}, \ \forall o \in \mathcal{D}.
\end{equation} 
\noindent By taking into account  Eq.~\eqref{eq:fisheriterm}, if we rescale all the parameters $\boldsymbol{\theta}_{(o)}$ by a 
constant $\sqrt{\gamma}$, we can remove the factor $\gamma$ from the above equation. 

\paragraph{The Precision Invariance axiom in other forms.}
We now present several equivalent, or conditionally equivalent, formulations of the Precision Invariance axiom in the form 
\eqref{eq:precindaxe} with $\gamma = 1$. They all reflect the idea that the observable-specific component of the state space  
Fisher information is equal for all the observables.

A first condition is in terms of the cotangent vectors 
${\mathbf{b}^*}^k$, for $k \le N_{\theta}$
\newline\textbf{Condition (ii)}
\begin{equation}\label{eq:precindb}
{{b^*}^{k}}^i_{\omega_{o}} \mathcal{I}_{i j}^{(o)} (\mathbf{x}){{b^*}^{l}}^j_{\omega_{o}} =  
       \delta^{k l} \, \, \forall o \in \mathcal{D} 
\end{equation} 
\noindent This form and the other formulations that follow are proved in appendix \ref{apx:pialtconditionsx}. It can also be 
demonstrated that, as a consequence of the Precision Invariance axiom and of the definition \eqref{eq:fisheriterm}, the 
cotangent vectors
${\mathbf{b}^*}^k$ are orthonormal. Therefore, Eq.~\eqref{eq:precindb} also tells us that they are eigenvectors of $\mathcal{I}_{i j}^{(o)} (\mathbf{x})$, with eigenvalue 1.

If the covariant basis is orthonormal, its dual must be orthonormal too and coincide with the covariant basis.
This applies to subspaces as well. In our case we have the covariant basis ${\mathbf{b}^*_{\omega_o}}_k$ in the subspace
$k \le N_{\theta}$.  Its dual basis is ${\mathbf{b}_{\boldsymbol{\omega}_o}}_k$, whose vectors, for  $k \le N_{\theta}$, form 
a basis in the tangent space $ T_x \Phi_{\mathbf{x}}^{\omega_{o}}$ and are expected to be orthonormal  and equal to 
${\mathbf{b}^*_{\boldsymbol{\theta}_o}}_k$. Thus, Eq.~\eqref{eq:precindb} can also be written in the form:
\newline\textbf{Condition (iii)}
\begin{equation}\label{eq:precindbcontrav}
{{b}_{\boldsymbol{\theta}_{o}}^i}_k \mathcal{I}_{i j}^{(o)} (\mathbf{x}) {{b}_{\boldsymbol{\theta}_{o}}^j}_l =  
       \delta_{k l} \, \, \forall o \in \mathcal{D}, 
\end{equation}
\noindent As stated before, we can assume that, in equation \eqref{eq:fisheriterm} the Fisher matrix
$\mathcal{I}_{kl}^{(o)} (\boldsymbol{\theta}_o)$ is diagonal and, since the Fisher matrix is positive semidefinite, we can 
identify its eigenvalues as ${\sigma^k}^2$.
We can prove (see appendix \ref{apx:pialtconditionsx}) a fourth formulation of the Precision Invariance axiom can be 
represented by two joint conditions:
\newline\textbf{Condition (iv)}
The ${\mathbf{b}^*_{\omega_o}}_k$ are orthonormal eigenvectors of  $\mathcal{I}_{i j}^{(o)} (\mathbf{x})$ and the square root 
of the eigenvalues of $\mathcal{I}_{kl}^{(o)} (\boldsymbol{\theta}_o)$ are equal to the normalization factors of 
${\mathbf{b}^*_{\omega_o}}_k$: 
$\sigma^{k} = {\iota^*}^k $ 

Since the normalization factors of the covariant basis are the reciprocal of the contravariant ones, we can also write:
\newline\textbf{Condition (vi)}
\begin{equation}\label{eq:normfaceigvalidentity}
{{\iota}_{\boldsymbol{\omega}}}_{k} = \sigma_k^{-1}, \ \ k \le N_{\theta}.
\end{equation}
\noindent In appendix \ref{apx:pialtconditionsx} is also demonstrated the following 
\newline\textbf{Condition (vii)}
\begin{equation}\label{eq:fisheripiaxalt}
 \mathcal{I}_{kl}^{(o)} (\boldsymbol{\theta}_o)
= 
\frac{\partial X^i}{\partial \omega^k} \delta_{i j} \frac{\partial X^j}{\partial \omega^l},
\end{equation} 
\noindent which will be the basis to express the Precision Invariance axiom in metric form.

\subsection{The Fisher metric preservation theorem} \label{sec:fishinfxt}

\mysubsubsection{The extended Fisher metric}Let us define a special type of Fisher information matrix on the extended 
coordinates introduced above:

\begin{definition}\label{def:extendedfishermetric} \textbf{(Extended Fisher metric)}
The \emph{extended Fisher metric} $g^{(o)}_{kl}$, which acts on the space of extended coordinates 
$\omega_{\boldsymbol{\theta}_{o}}$, is defined so that its components along $\boldsymbol{\theta}_o$ coincide with the Fisher 
matrix \eqref{eq:thetafisheri}, while its components along the orthogonal extension $\boldsymbol{\alpha}_o$ are given by the 
tensor $\delta_{i j}$ reprojected onto the $\omega_{\boldsymbol{\theta}_{o}}$ coordinates: 
\begin{equation} \label{eq:extfishinfo}
\begin{array}{ll}
g^{(o)}_{kl} = \mathcal{I}_{kl}^{(o)} (\boldsymbol{\theta}_o) &
\textnormal{if}\,\, l  \le N_{\theta} \,  \textnormal{and} \,  k  \le N_{\theta} \\
g^{(o)}_{kl} = 
\frac{\partial {X_o}^i}{\partial \omega_o^k}
\delta_{i j} 
\frac{\partial {X_o}^j}{\partial \omega_o^l}
&
\textnormal{if}\,\,  k  > N_{\theta} \, \textnormal{and} \,  l  > N_{\theta} \\
g^{(o)}_{kl} =  0   &  \textnormal{otherwise}.\\
\end{array}
\end{equation}
\end{definition}

If we define the following matrix:
\begin{equation} \label{eq:extfishinfonormx}
\mathcal{I}_{i j}^{\omega_o} (\mathbf{x})  =
\frac{\partial \omega_o^k}{\partial x^i}
g^{(o)}_{kl} 
\frac{\partial \omega_o^l}{\partial x^j},
\end{equation}
\noindent it can be shown (see appendix \ref{apx:fisherinformationidmtx}) that condition \eqref{eq:normfaceigvalidentity}, 
which is a necessary condition for
the Precision Invariance, can be expressed in the form
\begin{equation} \label{eq:constfishermatrixcond}
\mathcal{I}_{i j}^{\omega_o} (\mathbf{x}) = \delta_{i j}\,\, \forall o\in \mathcal{D}.
\end{equation}
\noindent Let $d \mathbf{x}$ be an arbitrary vector of the tangent space $T_x \Phi_{\mathbf{x}}$.
By multiplying left and right both sides of equation \eqref{eq:constfishermatrixcond} by $d \mathbf{x}$ we obtain
\begin{equation} \label{eq:precinfotgspace}
d x^i \mathcal{I}_{i j}^{\omega_o} (\mathbf{x}) d x^j = d \mathbf{x}^2 \,\, \forall o\in \mathcal{D},
\end{equation}
\noindent which, like Eq.~\eqref{eq:constfishermatrixcond}, is an alternative form of the Precision 
Invariance axiom. 
In the bilinear differential form \eqref{eq:precinfotgspace}, the matrix $\mathcal{I}_{i j}^{\omega_o}$
(or the corresponding $g^{(o)}_{kl}$), can play the role of a metric tensor.

In many cases, the orthogonal extension coordinates $\boldsymbol{\alpha}_o$ are orthogonal, i.e.
$\frac{\partial {X_o}^i}{\partial \alpha_o^k} \delta_{i j} \frac{\partial {X_o}^j}{\partial \alpha_o^l} = 0$ when $k \ne l$, or 
consist of a single component. When this happens, the second equation of \eqref{eq:extfishinfo} takes a simpler form:
\begin{equation} \label{eq:extfishinfosimpl}
g^{(o)}_{kl} = 
({\iota_{\alpha_o}}_k)^{-2}
\delta_{k  l}
 \ \  \textnormal{if} \  k  > N_{\theta} \, \textnormal{and} \,  l  > N_{\theta}.
\end{equation}

\mysubsubsection{Preservation of the  Fisher metric}
Equation \eqref{eq:precinfotgspace} and the definition \eqref{eq:extfishinfonormx}, imply that, since 
$d \omega_{o}^k = \frac{\partial \omega_o^k}{\partial x^i} d x^i$, 
the Precision Invariance axiom can be put in terms of a condition on the metric:

\begin{theorem}[metric-preserving property of the orthogonal parameter extension]\label{th:fmprese}
If the measurements of the observables of a complete set \allowbreak\hspace{0pt} $\mathcal{D} = \{o_1, o_2, ... o_D\}$ 
provide the same amount of information in terms of precision and if, for their parameters
$\boldsymbol{\theta}_{o_i}$, we can define the orthogonal parameter extensions
$\boldsymbol{\omega}_{o_i}$ then, for any two observables $a$ and $b$ of $\mathcal{D}$,
the following condition holds:
\begin{equation} \label{eq:metricpreservationtheorem}
d \omega_a^k {{g^{(a)}}}_{kl} d \omega_a^l = d \omega_b^k {{g^{(b)}}}_{kl} d \omega_b^l
\end{equation}
\noindent where $d \boldsymbol{\omega}_a$ and $d \boldsymbol{\omega}_b$ are the
representation of the same vector of the state tangent space $T_x \Phi_{\mathbf{x}}$ in the
parameter extensions $\boldsymbol{\omega}_a$ and $\boldsymbol{\omega}_b$ respectively, and
${g^{(a)}}$ and ${g^{(b)}}$ are the extended Fisher metrics of $\boldsymbol{\omega}_a$
and $\boldsymbol{\omega}_b$.
\end{theorem}

In Fig.~\ref{visualmetrprestheor},
a graphic representation of how this theorem works is provided, which starts from
the comparison among measurement uncertainties. 
An intuitive explanation of the theorem, which works well in the single-parameter case (scalar Fisher information), can be 
provided by thinking in terms of measurement
uncertainties. Roughly speaking, the Fisher metric represents a notion of distance that takes, as unit of measurement, the 
estimation uncertainty (assuming we are using maximum likelihood estimators).
Suppose we have two uncertainties, coming from two different parameter estimations;
instead of requiring the invariance of uncertainty (i.e., the inverse of precision), we require that the
measurement of distances in the state space remains unchanged regardless of whether the uncertainty comes
from one parameter or another.
In simple terms, if we have two rulers, we can check whether they are of equal length by measuring the
same object with both and comparing the results.

In the case of a multi-dimensional state space, the uncertainty of the estimate can be defined along
different directions of the state space manifold, making the problem more complex and leading to the
introduction of the concept of parameter extension.

Now that we have equipped the state space with a metric, the following point should be emphasized: the state space manifold is 
not a statistical manifold in the strict sense, because it is
not equipped with the Fisher metric, but rather with what we refer to as the extended Fisher metric.
Moreover, each point of the state space does not represent a single probability distribution but a
family of distributions, each associated with a different parametrization and observable.

\begin{figure}[!ht]
\centering
\includegraphics[scale=0.5]{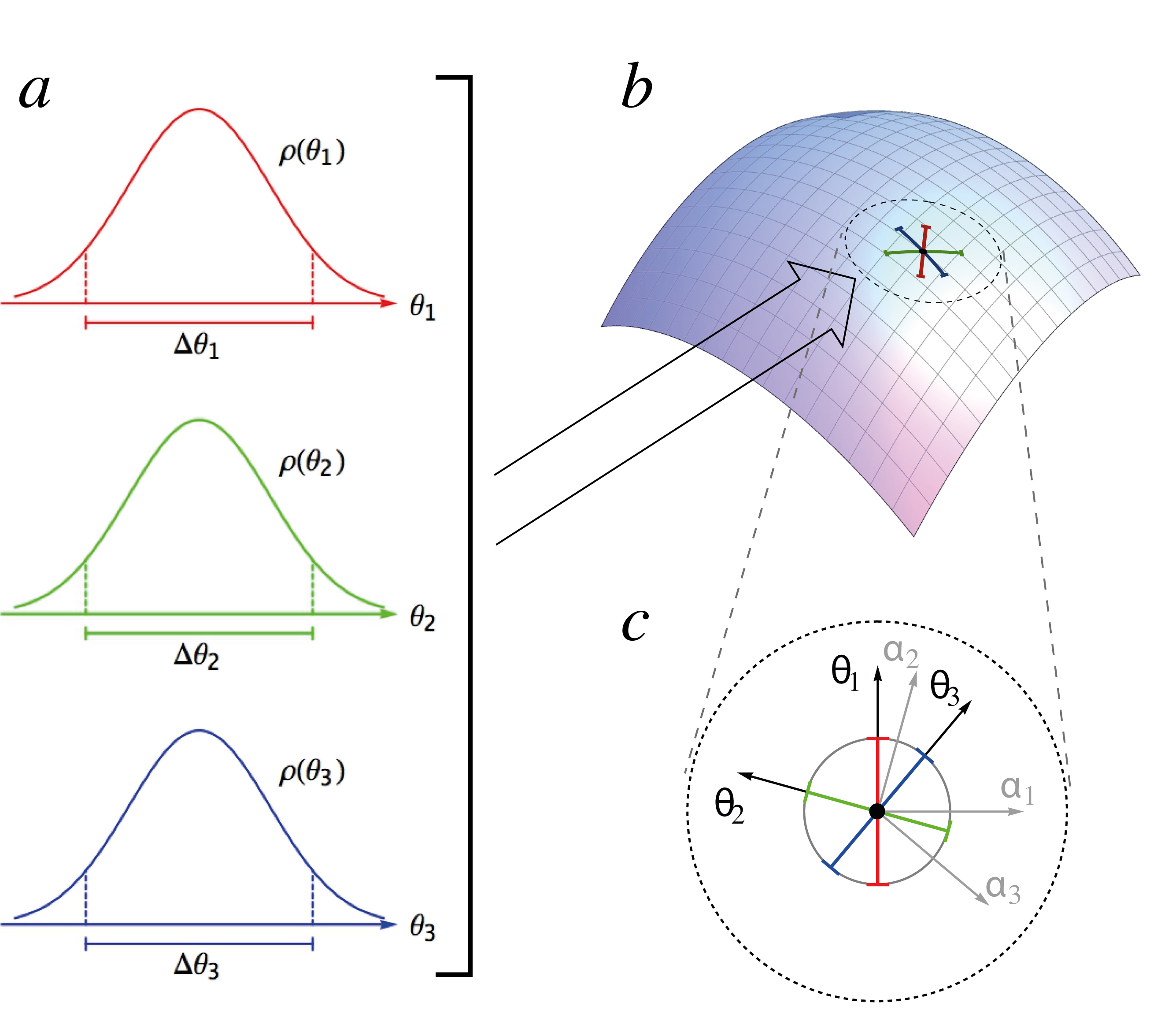}
\begin{subcaptiongroup}
\captionlistentry{a}
\label{visualmetrprestheora}
\captionlistentry{b}
\label{visualmetrprestheorb}
\captionlistentry{c}
\label{visualmetrprestheorc}
\end{subcaptiongroup}
\caption{
A graphic representation of the Precision Invariance axiom in terms of measurement uncertainties.
\subref*{visualmetrprestheora}
Parameters $\theta_1, \theta_2, \theta_3$ are  estimated with experimental uncertainties
$\Delta\theta_1, \Delta\theta_2, \Delta\theta_3$.
\subref*{visualmetrprestheorb}
Their uncertainties lead to an estimation error of the state in the direction of 
the red, green, and blue segments on the state manifold $\Phi_s$.
\subref*{visualmetrprestheorc}
According to the Precision Invariance axiom, since precision is the inverse of uncertainty, 
these segments must have the same magnitude.
The comparison among their magnitudes can be made with the help of the ``extended'' orthogonal 
coordinate systems $(\alpha_1 \theta_1), (\alpha_2 \theta_2), (\alpha_3 \theta_3)$.
The Fisher metric allows us to measure the distances in parameter space by considering the uncertainty
of the estimate as the unit of measurement (assuming the use of maximum likelihood estimators).
Thus, we can use the Fisher metric to compare uncertainties across different parameters.
This can be done, for example, by defining a contour (represented by a gray circle)
around a point that maintains a constant Fisher distance from it. }
\label{visualmetrprestheor}
\end{figure}

\mysubsubsection{Precision Invariance axiom in a subspace}\label{sec:estprecsubsp}
Now consider the case in which the number of measurements of one of the observables
of the set $\mathcal{D}$ is much larger than the number of measurements of the others.

In section \ref{sec:constrmle} we introduced the constrained log-likelihood function  \eqref{eq:mlec}.
The maximum likelihood estimators of the statistical parameters are given by 
the stationary points of that function. 
The first term of expression \eqref{eq:mlec}
is the sum of the log-likelihood functions of the measurement on a single observable.
If we had to maximize only this term, we would obtain, by applying the first 
order derivative, a system of $D$ decoupled equations because each  term of that 
sum depends only on the parameters of a single observable $o$. 

The derivation of the second term of \eqref{eq:mlec} leads to a term that breaks
the independence of the equations of that system. In what follows, we will show that,
if one of the observables of the set $\mathcal{D}$ is being measured for a large number of times,
the expression \eqref{eq:mlec}  leads to a condition of maximum likelihood, relative to that
observable, which is decoupled with respect to the other observables' one.

Suppose that the observable $o_1$ is being measured for $M_1$ times, where $M_1 >> M_k$ for
$ k>1 $. 
The terms of the constrained log-likelihood function  \eqref{eq:mlec}, which depend on 
$o_1$, are:
\begin{equation}  \label{eq:mlecsubspace}
\sum_{i = 1}^{M_1}  \log  \rho(\zeta^{o_1}_i;\hat{\boldsymbol{\theta}}_{o_1} ) + 
\sum_{a, k} \lambda_{a, k} \varphi_{a, k}(\hat{\boldsymbol{\theta}}^{o_1}, ... \hat{\boldsymbol{\theta}}^{o_D}).
\end{equation}
\noindent If the estimator of a statistical parameter is unbiased, its expected value is the value of
the parameter $\boldsymbol{\theta}$ itself. By the (weak) law of large numbers, as $M_1$ gets larger,
the  parameter estimator  
tends {\it{in probability}} to its expected value. Namely, for any $\epsilon > 0$: 
\begin{equation}\nonumber
 \lim_{M_1 \rightarrow \infty} \rho (\lvert\hat{\boldsymbol{\theta}}_{o_1} - \boldsymbol{\theta}_{o_1}\rvert < \epsilon) = 1. 
\end{equation}
\noindent This allows us to consider $\hat{\boldsymbol{\theta}}_{o_1}$ and 
$\boldsymbol{\theta}_{o_1}$ close enough to replace the first sum in  \eqref{eq:mlecsubspace} by 
its second order power series expansion in the neighborhood of $\boldsymbol{\theta}_{o_1}$.
This term can therefore be rewritten as:
\begin{equation}\begin{aligned}\label{eq:mletermsubsopace1}
& \sum_{i = 1}^{M_1}  \log  \rho(\zeta^{o_1}_i;    \boldsymbol{\theta}_{o_1} )  \twocolbreakamp + 
\frac{\partial}{\partial {\theta}_{o_1}^k} \sum_{i = 1}^{M_1} 
\log  \rho(\zeta^{o_1}_i; \boldsymbol{\theta}_{o_1} ) 
(\hat{{\theta}}_{o_1}^k - {\theta}_{o_1}^k)\\ 
& \qquad + 
\frac{1}{2} \frac{\partial^2}{\partial {\theta}_{o_1}^k \partial {\theta}_{o_1}^l} \sum_{i = 1}^{M_1} 
\log  \rho(\zeta^{o_1}_i; \boldsymbol{\theta}_{o_1} ) 
(\hat{{\theta}}_{o_1}^k - {\theta}_{o_1}^k)(\hat{{\theta}}_{o_1}^l - {\theta}_{o_1}^l).
\end{aligned}\end{equation}
\noindent In order to find the maximum in $\hat{\boldsymbol{\theta}}_{o_1}$ of the above expression, 
we have to calculate  its first derivative with respect to $\hat{\boldsymbol{\theta}}_{o_1}$.
The first term can be neglected because it does not depend on $\hat{\boldsymbol{\theta}}_{o_1}$. 
The second and third terms of Eq.~\eqref{eq:mletermsubsopace1} contain a summation,
in the range $i \in (1 ... M_1)$, of quantities that depend on the measurements' outcomes 
$\zeta^{o_1}_i$. 
Consequently, the terms of the summations can be treated as ordinary random variables. 
The sum in the range $i \in (1 ... M_1)$, up to a factor $M_1$, is actually equal to the mean
value of its terms.
Again, by the law of large numbers, in the large $M_1$ limit, we can replace in
\eqref{eq:mletermsubsopace1} the mean value of the term in the summation by its expected
value. Hence, the expression  \eqref{eq:mletermsubsopace1} yields:
\begin{equation}\begin{aligned} \label{eq:mletermsubsopace2}
 & M_1 \frac{\partial}{\partial {\theta}_{o_1}^k} \textnormal{E} \left(
\log  \rho(\zeta^{o_1}_i; \boldsymbol{\theta}_{o_1} ) 
      \bigm\vert\boldsymbol{\theta}_{o_1} \right) 
(\hat{{\theta}}_{o_1}^k - {\theta}_{o_1}^k) \\
& \qquad + \frac{1}{2}  M_1
\frac{\partial^2}{\partial {\theta}_{o_1}^k \partial {\theta}_{o_1}^l}\textnormal{E} \left( 
\log  \rho(\zeta^{o_1}_i; \boldsymbol{\theta}_{o_1} )   \bigm\vert\boldsymbol{\theta}_{o_1}  \right) \twocolbreakamp \cdot
(\hat{{\theta}}_{o_1}^k - {\theta}_{o_1}^k)
(\hat{{\theta}}_{o_1}^l - {\theta}_{o_1}^l) .
\end{aligned}\end{equation}
\noindent In this equation, the first term is zero because, in shorthand notation, we have that
$ \textnormal{E} (\frac{\partial}{\partial \theta} \log \rho(\zeta) ) = 
 \textnormal{E}(\frac{1}{\rho(\zeta)}  \frac{\partial \rho(\zeta)}{\partial \theta}) = 
 \int  \frac{\rho(\zeta)}{ \rho(\zeta)}\frac{\partial \rho(\zeta)}{\partial \theta} = 
\frac{\partial}{\partial \theta}  \int  \rho(\zeta) = 0$.
The second term, up to a factor
$-\frac{M_1}{2}(\hat{{\theta}}_{o_1}^k -
{\theta}_{o_1}^k)(\hat{{\theta}}_{o_1}^l - {\theta}_{o_1}^l)$, 
is nothing but the single-measurement Fisher information defined in Eq.~\eqref{eq:thetafisheri}.

Thus, we obtain a term proportional to $M_1$, which, in order to get the
maximum likelihood estimator, must be added to the constraint, i.e., the second term
of equation  \eqref{eq:mlecsubspace}. We assume that the constraint behaves ``regularly'', i.e., the Lagrange multiplier associated with the constraint it does not grow with $M_1$ in the
large $M_1$ limit; hence, its contribution will be negligible with respect to the term proportional to $M_1$.
Consequently, we can  
maximize the likelihood for the $\boldsymbol{\theta}_{o_1}$ component of the state separately, 
independently of the constraint term.

This allows us to calculate the other components of the maximum likelihood estimator $\hat{\boldsymbol{\theta}}^{o_k}$ with
$k>1$, by taking the first component, $\hat{\boldsymbol{\theta}}^{o_1}$, as a fixed parameter. This introduces a simplification 
because we have reduced the dimensionality of the state space.

\section{Applications of the metric-preserving property} \label{sec:elemsys}

\subsection{Two-observable elementary system}\label{sec:twoobs}
In what follows, we will study the two-variable one-bit elementary system, anticipated in  
\ref{sec:twoobsbasic}, in terms of the Fisher-metric formalism introduced in the previous section.
We will show that the metric-preserving property of theorem \ref{th:fmprese} leads to the same results obtained in section 
\ref{sec:twoobsbasic}.

Our program is to impose the Precision Invariance  axiom \ref{ax:precisioninvariance}
in the form of the metric-preserving property  \eqref{eq:metricpreservationtheorem}. 
As a first step, we must build the orthogonal extension for all the parameters of the system.
In the case of the two-observable one-bit elementary system, we can identify two sets of locally
independent observables, each represented by a single observable ($p$ or $q$). If we
consider the definition of orthogonal extension \ref{def:orthogonalextension}, we can see that, in the
case $d=1$, the extended set of coordinates corresponds to the
parameter itself because its dimensionality corresponds to the dimensionality of a set of locally
independent observables, which is $d$. Therefore, we have
$\boldsymbol{\omega}_p = \boldsymbol{\theta}_p$ and
$\boldsymbol{\omega}_q = \boldsymbol{\theta}_q$.
Since each statistical parameter consists of a single
component, condition  \eqref{eq:metricpreservationtheorem} simplifies to the form:
\begin{equation} \label{eq:gpreservationtwodim}
g_{p} d \theta_p^2 = g_{q} d \theta_q^2,
\end{equation}
\noindent where $g_{p}$ and $g_{q}$ are the metric tensors obtained by extending 
the Fisher information. In the case $d = 1$, since there are no $\alpha$ components in the
orthogonal extension, according to definition \eqref{eq:extfishinfo}, the
metric $g$ is nothing but the Fisher information metric
of the distributions $\rho_{\nu}(i; \theta_{\nu})$, which is defined as:
\begin{equation}\label{eq:fishermetrictensor}
g_{\nu} = 
\sum_{i=0}^1 \left( \frac{\partial \log \rho_{\nu}(i; \theta_{\nu})}{\partial \theta}  \right)^2
\rho_{\nu}(i; \theta_{\nu}).
\end{equation}
\noindent By applying the logarithm to both sides of the definitions \eqref{eq:rhovstheta}, and
calculating the first-order derivative with respect to  $\theta$, we obtain:
\begin{eqnarray}
\frac {\partial \log \rho_{\nu}(0; \theta_{\nu})}{\partial \theta_{\nu}} & = &
- \frac {\sin \theta_{\nu}/2}{\cos \theta_{\nu}/2 }\nonumber \\
\frac {\partial \log \rho_{\nu}(1; \theta_{\nu})}{\partial \theta_{\nu}} & = &
\frac {\cos \theta_{\nu}/2}{\sin \theta_{\nu}/2 },\nonumber
\end{eqnarray}
\noindent which can be substituted in \eqref{eq:fishermetrictensor}, giving 
the expression
\begin{equation}  \label{eq:gderivation}
g_{\nu} = 
 \left[ \left(\frac {\sin \frac{\theta_{\nu}}{2}}{\cos \frac{\theta_{\nu}}{2}} \right)^2 \cos^2\frac{\theta_{\nu}}{2} +
    \left(\frac {\cos \frac{\theta_{\nu}}{2} }{\sin \frac{\theta_{\nu}}{2}} \right)^2 \sin^2\frac{\theta_{\nu}}{2}  \right] 
 = 1
\end{equation}
\noindent If the lower bound of the Cram\'er-Rao inequality is reached, the Fisher metric
is equal to the inverse of the variance of the estimator  $\hat{\theta}$.
Thus, from \eqref{eq:vartheta}, we can conclude that the Fisher metric of a single system is 1.
With $g_{\nu} = 1$, equation \eqref{eq:gpreservationtwodim} now reads:
\begin{equation} \label{eq:fishermetrpres2d}
 d \theta_p^2 = d \theta_q^2.
\end{equation}

\noindent The solution of this equation is the same as the one we obtained in section 
\ref{sec:twoobsbasic}, (see equation \eqref{eq:thetapvsthetaq}).

\subsection{Three-observable elementary system} \label{sec:threeobs}

In this section, we study a system with three one-bit observables. As in the system studied in the
previous section, the total amount of information carried by this system is 1.
Let  $q$, $p$, and $r$ be the three observables, which can take the values 0 and 1.
We hypothesize that only two of these observables can be locally independent; thus, by theorem
\ref{th:manifolddim}, we expect a two-dimensional system state.
Thus the ``signature'' of the problem is $d=2$, $D=3$, $N_{\theta}=1$.

As we will see, this model system is closer to an actual quantum system than the previous one:
suppose we identify the variables $q$, $p$ and $r$ with the spin components 
along the $z$-, $x$- and $y$-axes, respectively, of an elementary spin-$1/2$ particle.
Here, for instance, $q = 0$ means  $s_z = -1/2$  and  $q = 1$ is $s_z = 1/2$.
It will be shown that, with the above hypotheses and by applying axioms 
\ref{ax:precisioninvariance} (in its metric-preserving condition form
\eqref{eq:metricpreservationtheorem}),
\ref{ax:limitedinfo},  \ref{ax:symmetries}  and \ref{ax:infostatequivalence}, 
the three-observable system exactly matches a quantum spin $1/2$ particle system.
 
The three-observable system studied in this section is an  extension of the 
previous section's system: 
we can express the distributions $\rho_p, \rho_q, \rho_r$, in terms of three 
parameters similar to those defined in  \eqref{eq:rhovstheta}.
At the end of the calculations, which, unfortunately, are much more complex than the two-observable
case, we will discover that these variables lie on a spherical surface.

Our approach is based on the definition of orthogonal parameter 
extensions \ref{def:orthogonalextension}. This will allow us to
express the Precision Invariance in terms of the metric property \ref{th:fmprese}
on these parameter extensions.

\mysubsubsection{Parametrizations and Orthogonal Extensions}

Let us introduce the parameters  $\theta_{\nu} \,\, (\nu = p, q, r)$ by using the same
definition we adopted for the two-observable system case \ref{sec:twoobsbasic}. 
Now we have three distributions $\rho_{\nu}$, which  
depend on the parameters $\theta_{\nu}$ by means of the same equations \eqref{eq:rhovstheta},
with $\nu = p, q, r$.
From these parametrizations, we can define the orthogonal
extensions $\alpha_{\nu}$ and the extended parameters $\omega_{\nu}$:
\begin{equation}\label{eq:threeparamtzn}
\boldsymbol{\omega}_q = \binom{\theta_q}{\alpha_q} \,\, \,
\boldsymbol{\omega}_p = \binom{\theta_p}{\alpha_p} \,\, \,
\boldsymbol{\omega}_r = \binom{\theta_r}{\alpha_r}.
\end{equation}

\noindent 
Suppose we choose a pair of parameters, for instance $\theta_q$ and $\theta_r$, 
to represent the state space. Of course, this choice allows us to identify a state
space point only in a neighborhood of a reference point,
because $q$ and $r$ are supposed to be just {\textit{locally}} independent.
This implies that, globally, the state $x$ can depend on these parameters through 
a multi-valued function; therefore, in general, any other parameter
is expected to depend on $\theta_q$ and $\theta_r$ through a multi-valued function.
In what follows, we will suppose that the coordinates
defined in \eqref{eq:threeparamtzn}, 
specifically the extensions $\alpha_{\nu}$, will be built to incorporate the information  
about the branch of the multi-valued function that maps the parameters onto the state space.
This means that extended parameters $\boldsymbol{\omega}_{\nu}$
can identify each point of the {\textit{whole}} state space.

In order to identify the parameter extension $\alpha_{\nu}$,
we need to know the relationship that connects $\alpha_{\nu}$ to the actual parameters
$\theta_{\nu}$. This relationship cannot be in the form $\alpha_q(\theta_q, \theta_r)$, because 
$\theta_q$ and $\theta_r$ cannot identify  globally a point on the state space
(while $\boldsymbol{\omega}_{\nu}$ is required to identify each point of the state space).
Thus, we can put
that relationship in an implicit form like
\begin{equation}\label{eq:threeparamfun}
\begin{array} {l}
 \theta_r =  \Theta_{q \rightarrow r}(\theta_q, \alpha_q)\\
 \theta_p =  \Theta_{r \rightarrow p}(\theta_r, \alpha_r).
\end{array}
\end{equation}

\noindent In our derivations, it will be useful to extend the definition
\eqref{eq:spqdef} to the three-observable case and introduce the following variables:
\begin{equation}\label{eq:nonisoparm}
\begin{array} {l}
 S_q =  \rho_q(0) - \rho_q(1) = \cos \theta_q ,\\
 S_p =  \rho_p(0) - \rho_p(1) = \cos \theta_p ,\\
 S_r =  \rho_r(0) - \rho_r(1) = \cos \theta_r.
\end{array}
\end{equation}

\mysubsubsection{Outline of the Demonstration}
 The demonstration, whose details are provided in appendix \ref{apx:statman},
 can be synthesized in the following points:
\begin{enumerate}
\item Following the approach adopted for the two-observable system (see conditions \eqref{eq:knownpqvals}),
we define the cardinal points
as those values of the parameters for which the observables $q$, $p$ and $r$
are either completely determined or completely undetermined.
These points are shown in Fig.~\ref{3Dsystem1a}.

\item \label{itm:subsp} We also identify some special subspaces of the state space,
in which the system can be reduced to the two-observable system studied in
\ref{sec:twoobs}; these subspaces are depicted in Fig.~\ref{3Dsystem1}.
Note that this simplification comes from what we have stated in \ref{sec:estprecsubsp}, 
when we discussed the application of the Precision Invariance axiom in a subspace.

\begin{figure}[!ht]
\centering
\includegraphics[scale=0.8]{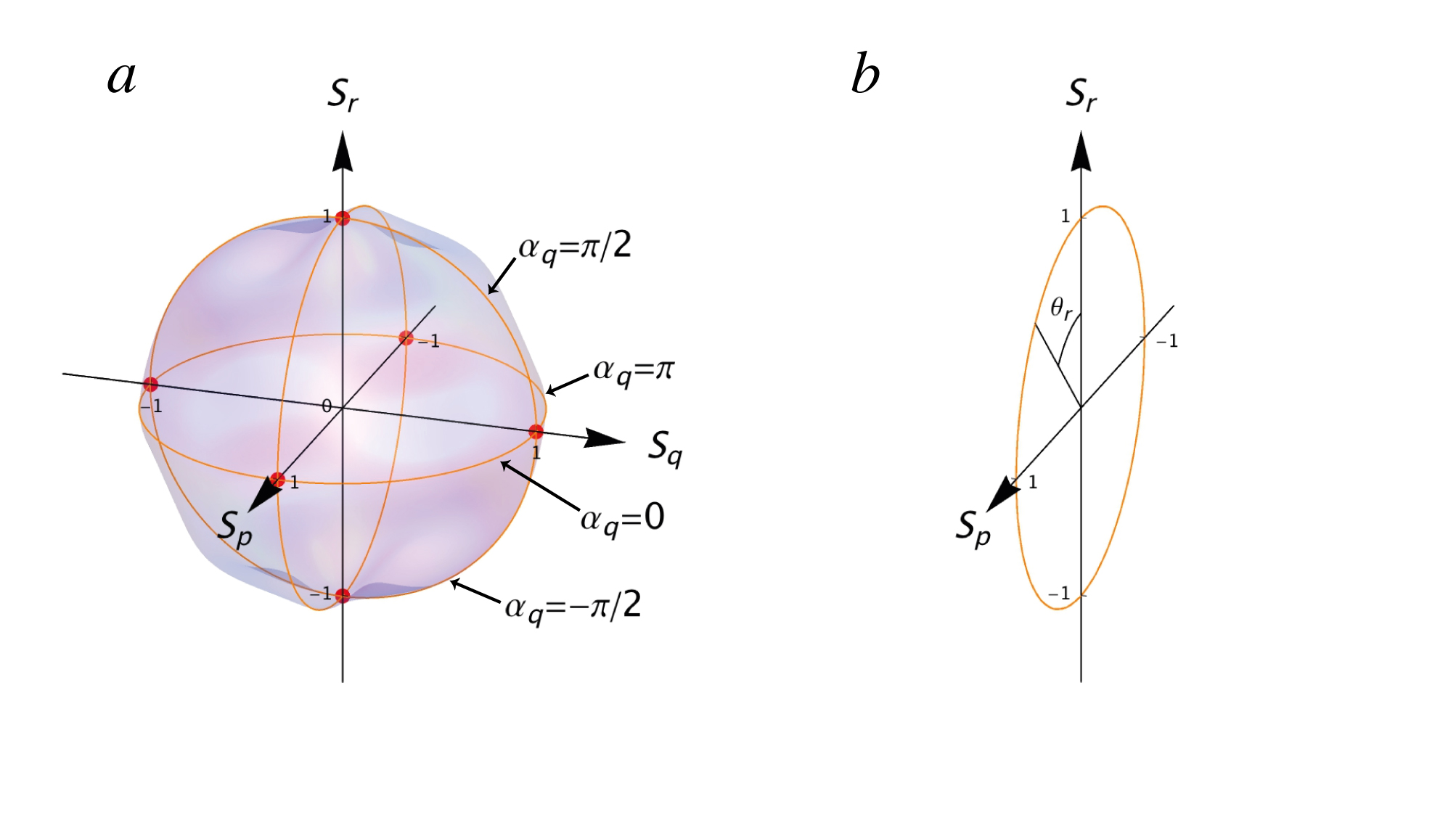} 
\begin{subcaptiongroup}
\captionlistentry{a}
\label{3Dsystem1a}
\captionlistentry{b}
\label{3Dsystem1b}
\end{subcaptiongroup}
\caption{
\subref*{3Dsystem1a} A generic state-space manifold, which passes through the cardinal points
(red dots) and through the $S_p, S_q, S_r = 0$ subspaces (orange circles).
\subref*{3Dsystem1b} 
The parametrization $\theta_r$ in the $S_q = 0$ space.
}
\label{3Dsystem1}
\end{figure}

\item There is an even larger class of subspaces where our system behaves like a two-observable
system, which corresponds to the cases where a parameter $\theta_{\mu}$ (or the corresponding
$S_{\mu}$), of any of the  $q$, $p$ and $r$ observables, is constant. See Fig.  \ref{3Dsystem2a}. 

\item We build the solution for the larger three-observable case by joining 
all the two-observable system submanifolds. Furthermore, we require such submanifolds to pass
by a special set of points (which is the subspace identified in step \ref{itm:subsp} and shown in Fig.~\ref{3Dsystem1b}); by 
varying the value of the parameter
$S_{\mu}$, we can derive the shape of the state manifold.
This process is shown in Fig.~\ref{3Dsystem2b}.
\end{enumerate}

\begin{figure}[!ht]
\centering
\includegraphics[scale=0.8]{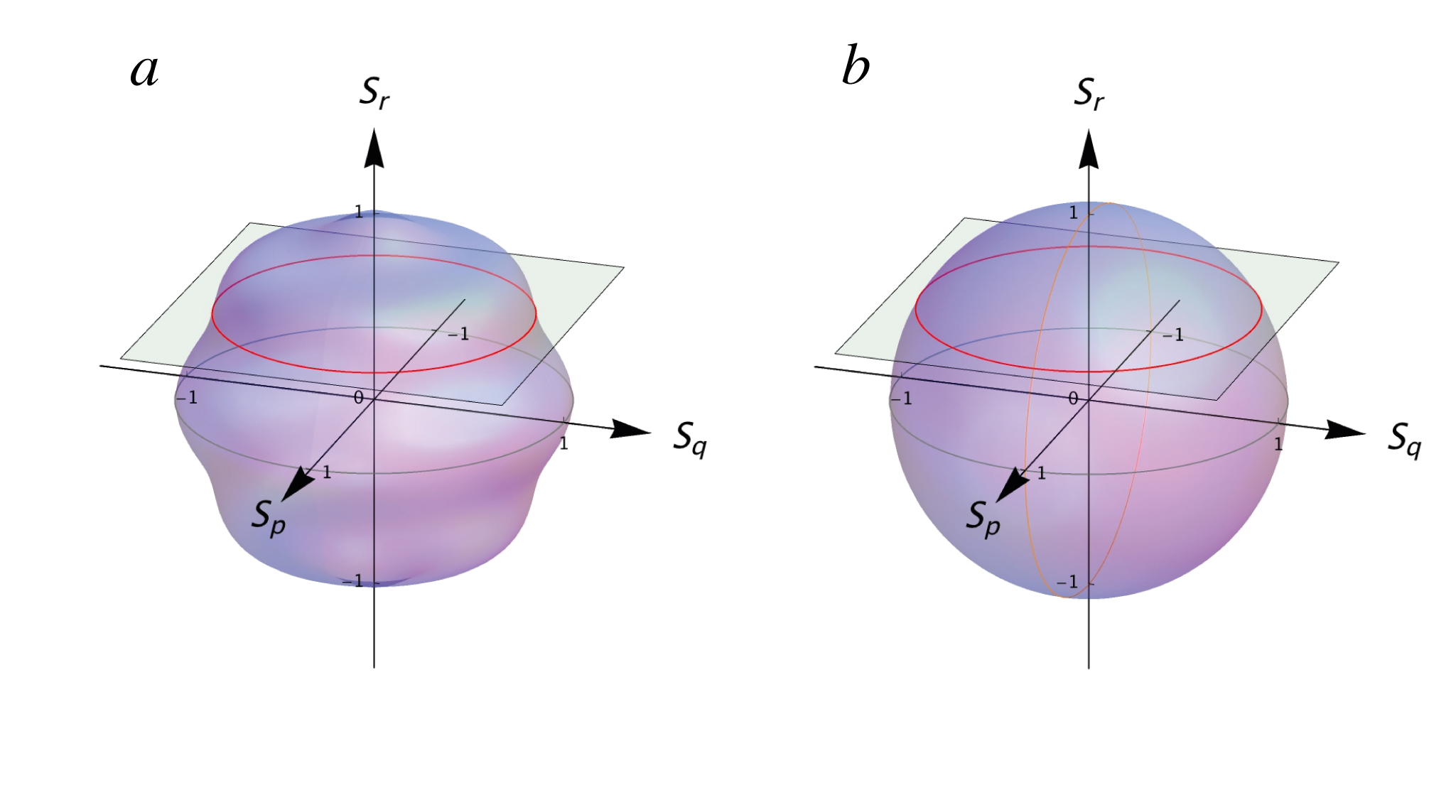}
\begin{subcaptiongroup}
\captionlistentry{a}
\label{3Dsystem2a}
\captionlistentry{b}
\label{3Dsystem2b}
\end{subcaptiongroup}
\caption
{
A visual explanation of how the state space manifold takes a spherical shape. 
\subref*{3Dsystem2a}
The state space manifold and its intersection with the 
$\theta_r = {\textnormal{const}}$ plane, which gives rise to the subspace 
identified by the red curve. By applying the same arguments of the two-observable
elementary system, it can be shown that this curve is a circle.
\subref*{3Dsystem2b}
The $\theta_r = {\textnormal{const}}$ curve must pass through 
the $S_q = 0$ subspace, which is a circle as well (identified by the orange curve). 
This implies that, in the 
$(S_p,S_q,S_r)$ space, the state space manifold takes the form of a sphere.
}
\label{3Dsystem2}
\end{figure}

\mysubsubsection{The solution}\label{par:thesolution}

It should be noted that, in our system, there is no preferred observable
among $p, q, r$. This implies that the solution is expected to be invariant 
under exchange of these observables.
Furthermore, as explained at the end of appendix \ref{apx:statman}
(Eq.~\eqref{eq:alphafredofch}), the parameter extensions
$\alpha_{\mu}$, occurring in the definitions \eqref{eq:threeparamtzn},
are not uniquely defined because they can be redefined through
the transformation $\alpha_{\mu}  \rightarrow \pi/2 - \alpha_{\mu}$. 
By following the notation introduced in that appendix (starting from Eq.~\eqref{eq:fishmprespq}),
we can define $\alpha_r =\alpha_{r p}$,  $\alpha_q =\alpha_{q r}$,
and for $\alpha_p$ we have two options:
(i) if we put $\alpha_{p} =\alpha_{p q}$, the set of parametrizations will be invariant
under rotations of the triplet $p, q, r$, namely, transformations like 
$(p, q, r) \rightarrow  (r, p, q)  \rightarrow  (q, r, p) $;
(ii) if, instead, we put   $\alpha_{p} =\alpha_{p r}$, 
the set of parametrizations will be invariant under the exchange
of observables  $p \leftrightarrow q$.
These invariances can be demonstrated by applying the corresponding index
permutations to the definitions of parameters $\alpha_{\mu}$ in terms of  $\alpha_{\mu \nu}$.
The extension of the parameter $\theta_{p}$ can be identified
as $\boldsymbol{\omega}_p^{(1)}$ in the first case and 
as $\boldsymbol{\omega}_p^{(2)}$ in the second one.

Now we can rewrite the parametric form \eqref{eq:parstatemanifold} of the state manifold obtained 
in the appendix \ref{apx:statman} for all the parametrizations \eqref{eq:threeparamtzn}:
\begin{eqnarray} \label{eq:euclimmersion} 
\begin{array}{l}
S_{\mu} = \cos \theta_{\mu}  \\
S_{\nu} = \sin \theta_{\mu} \cos \alpha_{\mu}  \\
S_{\xi} = \sin \theta_{\mu} \sin \alpha_{\mu},
\end{array}
\end{eqnarray}
\noindent where the triplet $\mu, \nu, \xi$ can take the values
 $(q,p,r)$, $(r,q,p)$ and $(p,r,q)$ (in the $\boldsymbol{\omega}_p^{(1)}$ case) or $(p,r,q)$ 
 (in the $\boldsymbol{\omega}_p^{(2)}$ case).
Regardless of how we choose the triplet  $\mu, \nu, \xi$, the mapping \eqref{eq:euclimmersion}
represents the transformation from the spherical coordinates
$\theta_{\mu},  \alpha_{\mu}$, with radius 1, to the Cartesian coordinates $S_{\mu} S_{\nu} S_{\xi}$. 

This is the result we were looking for: we have found the shape of the space of the physically
valid states. 
In this case, this set is expressed in terms of the coordinates $S_p, S_q$ and $S_r$,
which can represent the space of the parameters of the set $\mathcal{D}$.

We have a common system of coordinates $S_{p} S_{q} S_{r}$
and three systems $\boldsymbol{\omega}_{i}$ of curvilinear coordinates, which are orthogonal with
respect to the system $S_{p} S_{q} S_{r}$ 
(in the radius 1 manifold). We can therefore conclude that the systems
$\boldsymbol{\omega}_{i}$ are (extended) orthogonal coordinates because they are orthogonal
with respect to a common state coordinate system.

\mysubsubsection{The Fisher metric}
Now we can write the Fisher metric for the extended coordinate $\boldsymbol{\omega}_{i}$
by using the definition \eqref{eq:extfishinfo}.
Let us consider, for instance, the $\boldsymbol{\omega}_{r}$ coordinates: according to
the definition \eqref{eq:extfishinfo}, the component $g_{(r)}^{11}$ is given by
$\mathcal{I}^{(r)}(\theta_r)$ and, since the parametrizations $\theta_{\nu}$ have the same form \eqref{eq:nonisoparm}  of the 
two-dimensional case, we can apply the same steps from equation \eqref{eq:fishermetrictensor} to   \eqref{eq:gderivation}, and 
demonstrate that $g_{(r)}^{11}=1$.

As shown in section \ref{sec:fishinfxt}, in case of orthogonal or single-component parameter extension, the component 
$g_{(r)}^{22}$ can be calculated through the simplified equation  \eqref{eq:extfishinfosimpl} and is given by the normalization factor
$(\iota_{\alpha_r})^{-2}$ of the second vector $\mathbf{h}_2$ of the local basis 
(i.e., in the tangent space) of the system $\boldsymbol{\omega}_{i}$.
As shown in section \ref{sec:orthoparext}, such vectors take the form of the directional derivative 
$\frac{\partial X^i}{\partial \omega^k}$, which, in our case, reads
$h_2 = \frac{\partial S_i}{\partial \alpha_r}$. The derivatives of $S_i$ with respect to $\alpha_r$
can be calculated directly from  \eqref{eq:parstatemanifold} and we obtain 
$(\iota_{\alpha_r})^{-2} = \lvert\mathbf{h}_2\rvert^2 = \sin^2 \theta_r$.
We can use the same arguments for the other extensions
$\boldsymbol{\omega}_{p}$ and  $\boldsymbol{\omega}_{q}$, leading us to the result:
\begin{equation} \label{eq:gmudef}
g_{(\mu)} = 
\begin{pmatrix}
1 & 0 \\
0 & \sin^2(\theta_{\mu}),
\end{pmatrix}
\end{equation}
\noindent It is worth noting  that the above Fisher metric tensor is based
on the definition of extended Fisher metric \ref{def:extendedfishermetric},
where only the first $N_{\theta}$ components
(the first, in our case) correspond to a true Fisher information, while the others are defined 
in a way that permits us to obtain the condition on the metric \eqref{eq:metricpreservationtheorem}.
The expression of the extended Fisher metric in differential form is
\begin{equation} \label{eq:finalfishermetric}
d s^2 = d \theta_{\mu}^2 + \sin^2(\theta_{\mu})  d \alpha_{\mu}^2, \,\,\, \mu = p,q,r.
\end{equation}

\noexpand It can easily be shown that the expression we have obtained is equal to the
Cartesian metric $d s^2 = d S_{p}^2 + d S_{q}^2 + d S_{r}^2$, calculated on the unit sphere
and expressed in spherical coordinates $\theta, \alpha$ 
(the radius coordinate $\rho$ is constant and equal to 1).
In Eq.~\eqref{eq:finalfishermetric}, different  values of the triplet $\mu, \nu, \xi$ 
correspond to different choices of the polar axis of the spherical coordinate system.
Since the value of the Cartesian metric is independent of how we choose this axis,
we can conclude that the metric  \eqref{eq:finalfishermetric} 
is left unchanged when we go from a parametrization $\omega$ to another.
This satisfies what the metric-preserving theorem states.

\mysubsubsection{The Bloch sphere}\label{sec:blochsphere}

As stated before, the coordinates $S_q$, $S_p$, and $S_r$ are required to belong to the unit sphere.
This representation is commonly referred to as the Bloch sphere.
There is a well-known homomorphism between the groups SU(2) and SO(3), which permits us to state a 
relation between these coordinates and the two-dimensional space of complex vectors $ \vert\psi^{q}>$:
\begin{subequations}\label{eq:so3vssu2}
\begin{align}
S_q & = < \psi^{q} \vert \sigma_{3} \vert \psi^{q} >
      = \vert \psi_0^q \vert^2 - \vert \psi_1^q \vert^2 \label{eq:so3vssu2a}  \\
S_p & = < \psi^{q} \vert \sigma_{1} \vert \psi^{q} >
      = {\psi_1^q}^* \psi_0^q + {\psi_0^q}^* \psi_1^q 
      \twocolbreakamp
    = 2  \vert \psi_1^q \vert  \vert \psi_0^q \vert  \cos (\phi_1^q - \phi_0^q) \label{eq:so3vssu2b} \\
S_r & = < \psi^{q} \vert \sigma_{2} \vert \psi^{q} >   
      = i ({\psi_1^q}^* \psi_0^q - {\psi_0^q}^* \psi_1^q ) 
      \twocolbreakamp
     =  2  \vert \psi_1^q \vert  \vert \psi_0^q \vert  \sin (\phi_1^q - \phi_0^q) \label{eq:so3vssu2c}
\end{align}
\end{subequations}
\noindent where  $\sigma _{1}$, $\sigma _{2}$ and $\sigma _{3}$ are the Pauli matrices.
Note that, in the above equations,  $\sigma_{1}$, $\sigma_{2}$ and $\sigma_{3}$ are used
to obtain the coordinates $S_p$, $S_r$ and $S_q$ respectively. 
This correspondence, and the fact that $\psi^q$ is represented in the basis of the eigenvectors of
$\sigma_{3}$, justify the superscript $q$ in $\psi^q$.
$\phi_i$ is the phase of the $i$-th component of $\psi$.

We can define the vectors  $\psi^p$ and  $\psi^r$ in a similar way:
$\psi^p$ are the components of the state in the basis of the $\sigma_{1}$'s eigenvectors and
$\psi^r$ are components of the state in the basis of the $\sigma_{2}$'s eigenvectors.
It is easy to demonstrate that, in these representations, the probabilities $\rho$ satisfy  
the following relations
\begin{equation} \label{eq:psi2expr}
  \vert \psi_i^{\nu} \vert^2 = {\rho_i}^{\nu} , \, \, \nu = q,p,r
\end{equation}
and  that the phases $\phi_i$ are related to the $\alpha$'s through the identity
\begin{equation}  \label{eq:phi2expr}
   \alpha_{\nu} = \phi_1^{\nu} - \phi_0^{\nu}.
\end{equation}
In this representation, a  transformation from the parameters $\boldsymbol{\omega}_q$ to
$\boldsymbol{\omega}_p$ can be viewed as a change of basis from $\sigma_3$'s to $\sigma_1$'s
eigenvectors and, in the space of the two-dimensional complex vectors $\psi$, this transformation takes
a simple linear form given by the matrix
\begin{equation} \label{eq:linpsi}
\frac{1}{\sqrt{2}} \begin{pmatrix}1&1\\1&-1\end{pmatrix}.
\end{equation}
\noindent The system we have just described can be recognized as a
standard quantum-mechanical two-state system that carries one bit of information,
usually referred to as the \textit{qubit}, and the above Bloch sphere representation is commonly 
used to visualize its states.
This result is highly significant because it demonstrates, in an elementary case, the equivalence between 
the system of axioms proposed in this study and the standard quantum mechanics.

\section{$n$-bit systems}\label{sec:nstatesyst}

The purpose of this section is to provide the basis for the model of a system with two
canonically conjugate variables that obeys the  postulates \ref{ax:precisioninvariance} to \ref{ax:freevaryingdistr}---above 
all, the requirement that the amount of information we can obtain about the system is limited. 

There are some strong but reasonable assumptions that are also required by the model
proposed in this paper: 
\begin{itemize}
    \item Physical quantities (like canonical variables) can be approximated by discrete variables
    that can take a finite number of values. 
    \item The number of values a physical quantity can take is assumed to be a power of 2. 
    For example, we assume that the position variable $q$ can take $N_q = 2^{n_q}$ values.
      In the $n_q \rightarrow \infty$ limit, $q$ will approximate its corresponding continuous
physical quantity.
\end{itemize}

\mysubsubsection{Notation and conventions}

The binary representation of a number $x$ is indicated by the square brackets: 
$x = [b_1 ... b_n]$, where $b_1$ is the most significant bit and  
$b_n$ is the least significant one.

We also introduce some notation to indicate that some bits of two numbers are equal:
$x \lsbeq [ b_1 \dots b_n]$  indicates that the least significant bits of $x$, in its binary 
representation, take the values $b_1 \dots b_n$. Likewise, the notation $x \msbeq [ b_1 \dots b_n]$ 
indicates that the most significant bits of $x$ take the values $b_1 \dots b_n$.
Furthermore,  $x \stackrel{i}{=}  [ b_i ]$ indicates that the $i$-th bit of $x$ is $ b_i$,
and $x \stackrel{l\dots l'}{=}  [ b_l \dots b_{l'}]$ indicates that the bits that range from the 
$l$-th to the $l'$-th take the values $ b_l \dots b_{l'}$.
If the symbols in the square brackets explicitly refer to the bit position, we can also omit
the range $l\dots l'$ over the $=$ sign. For instance, we can write:
$x \stackrel{...}{=}  [ b_l \dots b_{l'}]$. 

We call this kind of relations {\it{digit-subset}} equality relations.
Sometimes they will be used to define sets; for example $\phi_{q, x} = \{ q \lsbeq [x_1 ... x_l]\}$
is the set of all the values of $q$ whose least significant bits are  $x_1 ... x_l$.

In this section, probability distributions will often be expressed in terms of the
bit-based notation introduced above. For example, the notation $\rho(q  \lsbeq [x_1 ... x_l])$
denotes the probability that the least significant bits of $q$ take the values $x_1 ... x_l$.
In the sections that follow, we will make extensive use of conditional probabilities that 
will often be expressed in terms of bit-based relations. Thus, for instance, the conditional
probability of the $l$-th bit of $q$ with respect to the rightmost ones
(less significant than the $l$-th) is
\begin{equation} \nonumber
\rho_q (q \stackrel{l}{=} [k] \vert q \lsbeq  [x_{l+1} ... x_n]).
\end{equation}

\subsection{Definition of the model signature}\label{sec:modelsignature}

In this subsection, we discuss how the values of the parameters $d, D, N$ must be chosen.
This point is central in the framework of our quantum reconstruction program.

Let us consider the parameter $N$, which indicates the number of values that each observable can take. 
Since a probability distribution over $N$ values is described by $N-1$ independent statistical parameters,
we have $N = N_{\theta} + 1$.
$N$ depends on how large the model we want to describe is and, to ease the construction of our model, it is assumed to be a
power of two, $N = 2^n$. It can be considered a free parameter. 

\mysubsubsection{Additional observables}
As stated at the beginning of this section, our classical reference system is made of two canonically conjugate variables. 

If our model were limited to these two observables, its applicability would be severely restricted. Indeed, if the only 
available observables were $p$ and $q$, the one-bit information bound would allow us to define complete preparations only by 
measuring either $p$ or $q$. In this sense, there would be no further independent observable through which additional families 
of states could be operationally characterized.

By postulating the existence of additional observables, the set of physically admissible states is enlarged to include states 
in which neither $p$ nor $q$ is fully determined. In order to obtain a model that is as general as possible, we also require 
that the additional observables be able to form locally independent sets with the original observables $p$ and $q$. For 
example, in the elementary one-bit case, the two-observable construction leads to a boundary relation in the $S_p,S_q$ plane, 
whereas the introduction of a third locally independent observable $r$ leads to a state space represented by a 2-sphere in the 
$S_p,S_q,S_r$ space. When projected onto the $S_p,S_q$ plane, this sphere gives a filled circle. Thus, states in which both $p$ 
and $q$ are not fully determined can be understood as states whose complementary information is encoded in the third observable.

Note that we can also add a fourth observable $s$, which leads to a 3-sphere
(we will not demonstrate this, but we can think of extending the three-observable system
\ref{sec:threeobs} to obtain this result). 
By varying the parameter $S_s$, we obtain a filled 2-sphere in the $S_p\,S_q\,S_r$ subspace and,
again, by varying $S_r$ we obtain a filled circle in the $S_p\,S_q$ subspace; this tells us that
introducing more than one extra observable does not improve the
generality of our theory to any significant extent.

The above ideas can be extended to cases where $N_{\theta} > 1$. We assume that introducing just one
independent observable is sufficient to give rise to states where both $p$ and $q$ are undetermined,
without introducing unnecessary complexity to our model.

In conclusion, we add a third observable $r$ to the pair $p$ $q$,
and assume that a set of locally independent observables is made
by two observables.
If, instead,  we had 3 locally independent observables, we would have freedom of choice for all the distributions
$\rho(p)$, $\rho(q)$ and $\rho(r)$; whereas, for dimensionality 1, for example, the probability
$\rho(p)$ would determine (locally) $\rho(q)$, and $\rho(r)$ would be irrelevant.

The number of values that the third observable
can take,
$N_r$, is related to the
maximum information of the system by an equation like $n_I = \log_2 N_r$
(which implies $N_r = N_p = N_q = N$, being $n_I = \log_2 N_{p,q}$).
If we assume that we can prepare the system in states that any observable is fully determined, there must be at least one
state where $r$ is fully determined, and by the Limited Information axiom \ref{ax:limitedinfo} and axiom 
\ref{ax:infostatequivalence}, $p$ and $q$ are completely undetermined.

It must be pointed out that we will never explicitly express this additional $N$-dimensional observable,
but, in what follows, especially in section \ref{sec:sepsubs}, we will consider three-observable
one-bit subsystems, which have two locally independent observables. Therefore, we will have
many one-bit additional observables that play the role of the observable $r$.
\mysubsubsection{Dimensionality of the state space}
In the previous paragraphs, we have stated that having a model with $D = 3$ observables
is a good choice to ensure a sufficient degree of generality to our theory.
Now we have to identify the parameter $d$ of the signature, namely the number of locally independent observables needed to estimate the state.

For a $N_{\theta} d$-dimensional state space manifold, in order to determine the state (i.e., a point on the state space 
manifold), we have to measure the distributions of $d$ observables.
This is, in fact, a statistical estimate, which can be obtained through a repeated measurement scheme on an ensemble of 
replicas of the system.
In order to make the Precision Invariance axiom applicable, we have to make this kind of measurement (on a subset of $d$ 
observables) and compare the precision on the state space estimation with the precision of an estimation coming from another 
set of $d$ observables. And, if we want more than one set of $d$ observables to exist, we must require that $D > d$. 

We can also demonstrate more restrictive conditions on the parameter $d$. Suppose we estimate $N_{\theta} d$ statistical 
parameters of $d$ locally independent observables.
Since the state space manifold is $N_{\theta} d$-dimensional, we can define multivalued functions $f_i$ that determine the 
values of the remaining $N_{\theta} (D-d)$ parameters: 
$\boldsymbol{\theta}_i = f_i(\boldsymbol{\theta}_1, \dots,  \boldsymbol{\theta}_d), \, i =  d+1 \dots D$. 

Under suitable regularity conditions on the state space $\Phi$, we may assume that, for a large subset of $\Phi$, the functions 
$f_i()$ have finitely many branches, namely, that the remaining parameters are determined by $\boldsymbol{\theta}_1, ... 
\boldsymbol{\theta}_d$ up to a finite number of possible values.

If we accept the general idea of an exchange symmetry among observables, i.e., that all observables are physically equivalent, 
we can state that the functions $f_i()$ are equal for all the observables. 
This implies that the corresponding observables behave as if they were the same observable.

In order to avoid this kind of redundancy, we can assume that a single observable that exceeds the $d$ locally independent ones 
is sufficient, namely, $D=d+1$.

\subsection{Partitions of the phase space}\label{sec:setpartitions}

Let us consider a generic observable $o$ defined on a set $\Omega$, which can take $N = 2^{n}$ values.
We can define a partition on the set $\Omega$, namely, a family of $N_{\mathcal{A}}$
non-intersecting sets $\mathcal{A} = \{ A_1, ... A_{N_{\mathcal{A}}} \}$
whose union is $\Omega$.  The largest (highest resolution) partition we can define
consists of $N$ single-element sets.
In this paper, we consider only observables defined on sets whose cardinality
is a power of 2 and partitions that contain a number of sets that is a power of 2.

In general, a partition can be expressed as the {\textit{product}}
of lower-resolution partitions, where by product of two partitions, we mean the family of sets made
by all non-empty intersections of the sets taken from the two partitions. We refer to these
lower-resolution partitions as {\textit{subpartitions}}.

For example, if we consider the partition 
$\mathcal{A} = \{ A_1, ... A_{N_{\mathcal{A}}} \}$, we can say that each set $A_i$ can be expressed
as an intersection of larger sets: by denoting by  $a_i^l$ these larger sets, we can write each set
of $\mathcal{A}$ as
\begin{equation} \label{eq:gendecomp}
 A_{\mu} = a_{\mu_1}^1 \cap a_{\mu_2}^2 \cap ... \cap a_{\mu_m}^m,
\end{equation}
\noindent where the tuple of indices $\mu_j$ depends on the index $i$ of the set $A_i$.

As an example, let us consider the following partition: we divide the set 
$\Omega$ into $N_{\mathcal{A}} = 2^{n_{\mathcal{A}}}$ adjacent sets $A_{\mu}$, which
have the same cardinality $N/N_{\mathcal{A}}$. With the help
of the binary representation of its elements, $A_{\mu}$ can be expressed as: 
\begin{equation} \label{eq:partitionsetexample}
  A_{\mu} = \{ o \msbeq [\mu_1 ... \mu_{n_{\mathcal{A}}}] \},
\end{equation}
\noindent where $[\mu_1 ... \mu_{n_{\mathcal{A}}}]$ is the binary representation of the
index $\mu$ of $A_{\mu}$.
The first set $A_0$ of $\mathcal{A}$ contains the elements that range
from $[0...0] = 0$ to $[0 ...0 \underbrace{1 ... 1}_{n-n_{\mathcal{A}}}] = N/N_{\mathcal{A}} -1$. 

We can also define the smallest (most coarse) partitions of $\Omega$, which are made just by two
elements: we divide the set $\Omega$ into the subsets $a_0^l$ and $a_1^l$, whose elements have their
$l$-th bit equal to 0 and 1, respectively.  We indicate as $\mathdutchcal{a}^l = ( a_0^l,  a_1^l )$
this binary
partition and its sets can be written in the form 
$a_0^l = \{ o \stackrel{l}{=} [0] \}$ and $a_1^l = \{ o \stackrel{l}{=} [1] \}$.
A set $A_{\mu}$ can be expressed as an intersection of $n_{\mathcal{A}}$ of these sets:
\begin{equation} \label{eq:binsetpartition}
  A_{\mu} =  \bigcap_{l=1}^{n_{\mathcal{A}}}  a_{\mu_l}^l,
\end{equation}
\noindent being $[\mu_1 ... \mu_{n_{\mathcal{A}}}]$ the binary representation of $\mu$.
Figure~\ref{partitions4} shows an elementary example of a product of partitions for $N = 4$.
Although we have decomposed a partition whose sets have the form \eqref{eq:partitionsetexample}, 
the same kind of decomposition can be applied to more complex sets, defined
through any digit-subset equality relation. For example, we can combine set partitions that
refer to two or more observables.

\begin{figure}[!ht]
\centering
\includegraphics[scale=0.5]{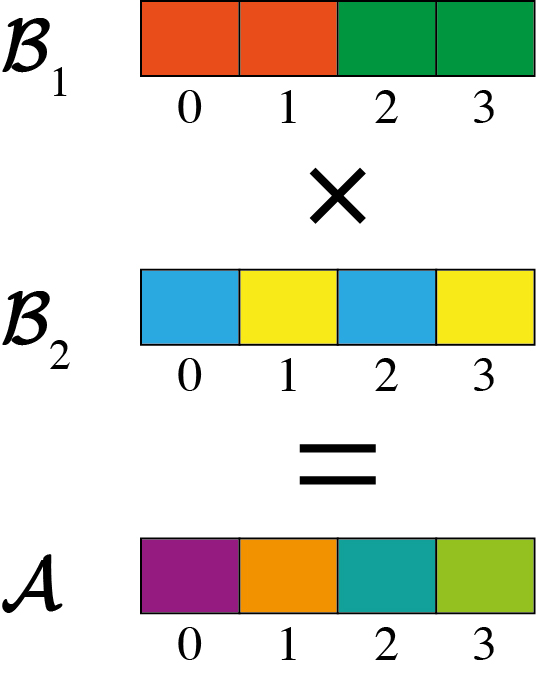}
\caption
{
The four-set partition $\mathcal{A}$ can be expressed in terms of the 
product of two coarser partitions $\mathcal{B}_1$ and 
$\mathcal{B}_2$. The rectangles represent the finest grain sets $A_0, A_1, A_2, A_3$.
Different colors mean different sets of the three partitions $\mathcal{A}$,  $\mathcal{B}_1$
and  $\mathcal{B}_2$.
} 
\label{partitions4}
\end{figure}

\mysubsubsection{Conditional probabilities with respect to set partitions}\label{sec:condprobpart}

The decomposition  of a partition as a product of coarser partitions allows us to define the
conditional probability with respect to a coarser family of sets.
Suppose that each set $A_{\mu}$ of a partition $\mathcal{A}$ can be decomposed into the intersections 
of the sets of two coarser partitions $\mathcal{B}_1$ and 
$\mathcal{B}_2$: $A_i =  B_{i_1 1} \cap B_{i_2 2}$. 
According to the definition of conditional probability, we can write the probability of a set $A_i$
of the partition $\mathcal{A}$ as 
\begin{equation} \label{eq:twopartitiondec}
\rho_{\mathcal{A}}(A_i) = \rho_{\mathcal{B}_1\vert\mathcal{B}_2}(B_{i_1 1} \vert B_{i_2 2}) 
                        \rho_{\mathcal{B}_2}(B_{i_2 2}).
\end{equation}

\noindent In this study, we follow a divide-and-conquer approach, which is based on dividing a
system into simple binary sets. These sets belong to a family of binary partitions that forms a
binary tree (hereafter, the terms ``family'' and ``tree'' are used interchangeably to refer
to this structure).

In developing our model, we will impose certain conditions on the partitioning of the space of some observables. By using the 
binary partitions introduced in the previous paragraphs, which reflect the binary representation of the value of an 
observable, we would make a strong hypothesis on how the system of conditions is structured. 
Instead, we start from the most general family of binary partitions. We split the space 
$\Omega$ into two arbitrary subsets, which are not necessarily defined as
$a_{\mu}^l = \{ o \stackrel{l}{=} [\mu] \}$, and then we split again. After splitting the
sets $l$ times, we obtain a family of increasing resolution partitions 
$\mathcal{A}^l = \{ A_0^l, ... A_{2^l-1}^l \}$. 
Again, we can define the binary partitions $\mathdutchcal{a}^l = ( a_0^l,  a_1^l )$ in order to 
express the sets $A_k^l$ in a form similar to that of Eq.~\eqref{eq:gendecomp}:
\begin{equation} \label{eq:binsetpartitionarb}
  A_k^l =  \bigcap_{i=1}^{l}  a_{k_{i}}^i,
\end{equation}
\noindent where $k = 0 ... 2^l-1$ and each of the indices $k_1 ... k_l$ can take the values 0,1.
Note that each value of the tuple $k_1 ... k_l$ must correspond to a value of $k$,
for instance $k = [k_1 ... k_l]$; nonetheless,  we must keep in mind that at this stage the partitions 
$\mathcal{A}_l$ and $\mathdutchcal{a}_{k_l, l}$ are arbitrary and do not reflect any particular
order in the space $\Omega$.
The last expression can also be put in the recursive form
\begin{equation} \label{eq:binsetpartitionarbrec}
  A_k^l =  a_{k_{l}}^l \cap A_{[k_1...k_{l-1}]}^{l-1},
\end{equation}
\noindent being $A^0 = \Omega$ the zero-level of the recursion. 
This decomposition allows  us to write the conditional 
probability on a binary partition $\mathdutchcal{a}^{l}$ with respect to the lower
resolution partition $\mathcal{A}^{l-1}$, in a form that resembles Eq.  
 \eqref{eq:twopartitiondec}:
\begin{equation}\begin{aligned}
\label{eq:partitionedrho}
\rho_{\mathcal{A}^l}(A_k^l) = &
\rho_{\mathdutchcal{a}^l\vert\mathcal{A}^{l-1}} \left(  a_{k_{l}}^l \biggm\vert 
A_{[k_1...k_{l-1}]}^{l-1} \right)\twocolbreak & \cdot 
\rho_{\mathcal{A}^{l-1}}\left(  A_{[k_1...k_{l-1}]}^{l-1} \right). 
\end{aligned}\end{equation}
\noindent We can decompose the partitions  $\mathcal{A}^{l-1}$ 
into a product of coarser partitions and reiterate the process until  
we obtain the lowest resolution binary partition. By keeping in mind
that, in the last step of the iteration (i.e., $l = 1$), equation \eqref{eq:binsetpartitionarb} reduces 
to the form $A_k^l = A_{k_l}^l$, the last expression
can be expanded into the following form:
\begin{equation}\label{eq:partitionedrhofull}
\rho_{\mathcal{A}^l}(A_k^l) =  \prod_{i=1}^{l}
\rho_{\mathdutchcal{a}^i\vert\mathcal{A}^{i-1}} \left(  a_{k_{i}}^i \biggm\vert 
A_{[k_1...k_{i-1}]}^{i-1} \right). 
\end{equation}
\noindent The separation \eqref{eq:gendecomp} of the set of an observable's values into a
family of increasing resolution partitions can be viewed as the decomposition of the predicate ``the observable $x$ takes the 
value $v$'' into a family of elementary predicates;
clearly, each of these predicates carries less information than the starting predicate.

\noindent Each elementary predicate can be viewed as the outcome of a coarser measurement
of an observable, whose probability distribution is given by one of the conditional probabilities
occurring in the factorization \eqref{eq:partitionedrho} of $\rho$.
These conditional probabilities are referred to as
{\it{components}} of the highest resolution probability $\rho_{\mathcal{A}}$.

\mysubsubsection{Generalization of free varying distributions}
\label{sec:freevardistgen}
The idea of considering the partitions of the observables' values domains as sets of predicates
suggests a way to count the amount of information carried by the observables:
each binary subpartition counts as a bit of information. In other words, for a partition
$\mathcal{A}$, the amount of information depends on its cardinality through an expression like
$I = \log_2 \Card (\mathcal{A})$.

If the values of a complete set of observables can be fully determined by their measurement,
we can say that the amount of information calculated from the cardinality of their domains'
partitions fixes the amount of information carried by the whole system.
In section \ref{sec:freevardist} we also stated that the distributions of these complete 
observables are free-varying, because they allow one-point distributions.

Now suppose we have the following scenario: we have many choices to form a complete set of observables. We identify these 
possible sets by $\mathcal{D}_1, \mathcal{D}_2$, etc.;
we take some observables from one set, another observable from another set and so on---we
just require all these observables to be independent of each other in the meaning stated by definition 
\ref{def:completesetofobservables}, i.e., we require that fixing the
value of one of them does not impose any restriction on the possible
values of the other.
We have therefore generated a new set of observables; if the amount of information corresponding to the partitions generated by 
this set of observables does not exceed the limit we
imposed to the system, we can assert that this new set of observables is free varying because it
has been built from free varying observables. 

The point is that, when we deal with partitions and families of subpartitions, the concepts of
observable and of set of observables are not fully distinguishable: for instance, if we have a binary partition, the 
information about the set of the partition in which the value of an observable falls is itself another observable.
As a consequence, the above scenario applies also to the case in which we consider a complete
observable as complete set of two-value observables; thus we can ``build'' new (complete and free
varying) observables by mixing other complete observables.

In the following sections, this way of building complete sets of observables, which have the property
of being free-varying, will be the foundation for a divide-and-conquer approach to study some
properties of the state space.

\mysubsubsection{The state space constraint $\chi$}

Suppose we have a system made by the set of observables $o_1 ... o_D$ and a collection of partitions
$\mathcal{A}^{o_i}$ of the domains of these observables. The state space is, by definition, the set
of physically admissible distributions on such partitions, and we will indicate the constraint that
identifies this set by the symbol $\chi$.
The constraint $\chi$ can be viewed as a boolean-valued function depending on the probability 
distributions, $\chi(\rho_{\mathcal{A}^{o_1}}, ..., \rho_{\mathcal{A}^{o_D}})$,
or on their parameters $\chi(\boldsymbol{\theta}^{o_1}, ..., \boldsymbol{\theta}^{o_D})$.
Its value is {\textit{true}} only if the tuple
$\rho_{\mathcal{A}^{o_1}}, ..., \rho_{\mathcal{A}^{o_D}}$
corresponds to a physically admissible state.

At this stage, we should consider the most general system of conditions $\chi$ that 
bounds all the components of the probability functions of phase space partitions.
Since, by Eq.~\eqref{eq:partitionedrho}, the probability function on
a partition can be expressed in terms of probability functions on lower-resolution
partitions, the function $\chi$ can be expressed in terms of
elementary probability components. 

In the following sections, we will discover how to rearrange the dependence of $\chi$ on these
components in a way that reflects the organization of the phase space into families of
nested partitions, such as those defined by \eqref{eq:binsetpartitionarb}.

\subsection{Binding two partitions} \label{sec:partitionsbinding}

\mysubsubsection{Finest common partitions}
Let us consider two partitions $\mathcal{A} = \{A_1, ... A_{N_{\mathcal{A}}} \}$ and
$\mathcal{A}' = \{A'_1, ... A'_{N_{\mathcal{A}'}} \}$, which
can be decomposed into products of coarser partitions in the form of expression \eqref{eq:gendecomp}:
$ A_i = a_{i_1}^1 \cap a_{i_2}^2 \cap ... \cap a_{i_m}^m$ and
$ A'_{\mu} = {a'}_{{\mu}'_1}^1 \cap a_{{\mu}'_2}^2 \cap ... \cap a_{{\mu}'_{m'}}^{m'}$.
We are still considering the simpler case where the sets $a_{...}$ and $a'_{...}$ 
form binary partitions like the ones introduced in section \ref{sec:setpartitions},
which are defined as $\mathdutchcal{a}^j = ( a_{0}^j,  a_{1}^j  )$
and ${\mathdutchcal{a}'}^j =  ( {a'}_{0}^j,  {a'}_{1}^j  )$. 

Now suppose that some of the partitions in the family $\mathdutchcal{a}^j$ are actually the same as in
the family ${\mathdutchcal{a}'}_j$. In other words, there are some values $j_1 ... j_{m_c}$ and 
$ j'_1 ... j'_{m_c}$ of the indices of $\mathdutchcal{a}^j$ and ${\mathdutchcal{a}'}^{j'}$, such that 
$\mathdutchcal{a}^{j_1} = {\mathdutchcal{a}'}_{j'_1}, ... , \mathdutchcal{a}_{j_{m_c}} =
\mathdutchcal{a}'_{j'_{m_c}}$.
We identify these common partitions as $\mathdutchcal{c}^1 = \mathdutchcal{a}^{j_1} =
{\mathdutchcal{a}'}^{j'_1}, ... , \mathdutchcal{c}^{m_c} = \mathdutchcal{a}^{j_{m_c}} =
{\mathdutchcal{a}'}^{j'_{m_c}}$.
We can denote by $\mathcal{D}$ the partition generated by the product of the binary partitions
$\mathdutchcal{c}^j$:
\begin{equation} \nonumber
   \mathcal{D} = \mathlarger{\mathlarger{\times}}_j  \mathdutchcal{c}^j = 
   \left\{  \bigcap_{k=1}^{m_c}  a_{\mu_k}^{j_k},\, \forall \mu_k = 0,1 \right\}.
\end{equation}

\noindent In other words, after decomposing  $\mathcal{A}$ and $\mathcal{A}'$ into binary
partitions, we take the binary partitions that $\mathcal{A}$ and $\mathcal{A}'$, have in common.
The new family $\mathcal{D}$ obtained from the intersections of all possible choices of such
sets is defined as the {\textit{finest common partition}} of $\mathcal{A}$ an $\mathcal{A}'$. 

This concept is referred to in the literature as the  {\textit{meet}},
a term used in the framework of order-theoretic lattice structures;
actually, the families of nested partitions can be seen as partially ordered sets. 
In this context, the meet denotes the greatest lower bound of two partitions,
while the {\textit{join}} denotes the least upper bound.
Fig.~\ref{finestcommonpart} shows a visual representation of 
how a common partition is identified.

\begin{figure}[!ht]
\centering
\includegraphics[scale=0.47]{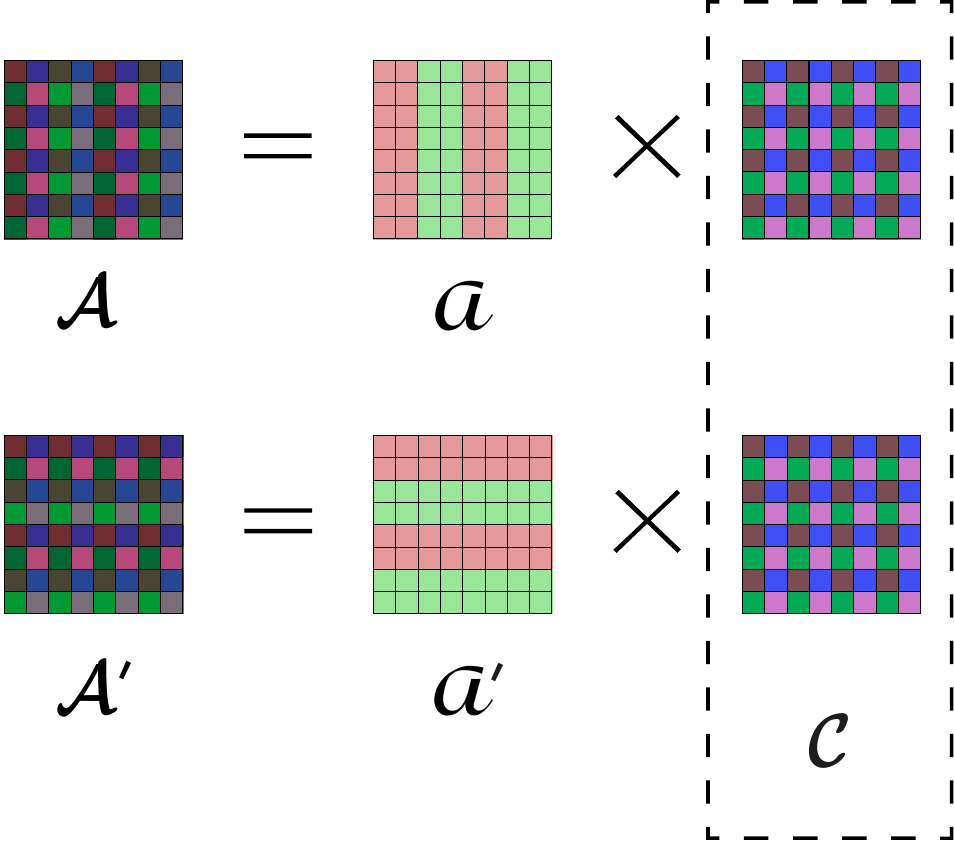}
\caption{ 
A visual representation of how the finest common partition between two 
partitions is defined. The sets of a decomposition are indicated with squares of different colors, 
and the same color indicates the same set.
The partitions $\mathcal{A}$ and $\mathcal{A}'$ can be decomposed into two lower-resolution
(lower number of sets) sub-partitions. We discover that, in this decomposition, 
$\mathcal{A}$ and $\mathcal{A}'$ share the same sub-partitions $\mathcal{D}$.
}
\label{finestcommonpart}
\end{figure}

Why is the finest common partition useful to us? 
Suppose we have a system of one or more observables, with partitions $\mathcal{A}$
and $\mathcal{A}'$ defined on their respective domains. Both $\mathcal{A}$ and 
$\mathcal{A}'$ can be decomposed into families of lower-resolution (for instance, binary) partitions.
If we define probability distributions on these partitions,
it is easy to show that the probability on a lower-resolution partition can always
be derived from the probability on the higher-resolution one. 
Consequently, if $\mathcal{A}$ and $\mathcal{A}'$ share a common subpartition 
$\mathcal{D}$, the probability on $\mathcal{D}$ can be derived from the probabilities on
$\mathcal{A}$ and $\mathcal{A}'$.
In other words, we are considering a case where the probability distributions on 
$\mathcal{A}$ and $\mathcal{A}'$ are not fully independent.

Our strategy to simplify the problem of identifying the state space is to decompose the probabilities
on $\mathcal{A}$ and $\mathcal{A}'$ into components, determining  which are mutually dependent and 
which are independent (and, therefore, related to $\mathcal{D}$).

This can be expressed in terms of conditional probabilities, where the decomposition into 
probability components takes the form outlined by equation \eqref{eq:partitionedrho}. We denote by
$\rho_{\mathcal{A}}$ and $\rho_{\mathcal{A}'}$ the probability distributions on
$\mathcal{A}$ and $\mathcal{A}'$ and by $\rho_{\mathdutchcal{a}_j}$ and  
$\rho_{\mathdutchcal{a}'_j}$ their probability components (conditional probabilities)
on the binary partitions. The probability components on the common binary partitions 
$\mathdutchcal{c}_j$ must be the same, and can be identified as $\rho_{\mathdutchcal{c}_j}$. 
The largest number of common probability components we can pick is given by the finest common
partition.

According to what was stated in subsection \ref{sec:freevardistgen}, there is a correspondence between partitioning of the 
domain of observables and the identification of complete sets of observables: each binary subpartition counts as a bit of 
information in identifying the value of an observable. If we take the partition $\mathcal{A}$ to be sufficiently small (i.e., 
coarse) so that the corresponding observables do not exceed the maximum number of locally independent observables, then we can 
say that such a set of observables is free to vary. 
We have denoted by  $\chi(\rho_{\mathcal{A}}, \rho_{\mathcal{A}'})$ the state space constraint that bounds 
$\rho_{\mathcal{A}'}$ and $\rho_{\mathcal{A}}$.
If we require that $\rho_{\mathcal{A}}$ can be free to vary, this constraint must involve only the
components of $ \rho_{\mathcal{A}'}$ that are actually independent of $\rho_{\mathcal{A}}$.
Consequently, the condition $\chi(\rho_{\mathcal{A}}, \rho_{\mathcal{A}'})$ will not constrain the
probability distribution on the finest common partition ${\mathcal{D}}$. If we factorize out the
finest common partition-dependent component from the probabilities on
$\mathcal{A}$ and $\mathcal{A}'$, we can express this concept in the form:
\begin{subequations}\label{eq:partbounds}
\begin{align}
 \rho_{\mathcal{A}}(A_{\mu})   & = 
\rho_{\mathcal{A}\vert{\mathcal{D}}}( A_{\mu} \big \vert C_{\mu''} ) \, 
\rho_{\mathcal{D}}( C_{\mu''} )
 \label{eq:partbounds1} \\
 \rho_{\mathcal{A}'}(A'_{\mu'}) &  =
\rho_{\mathcal{A}'\vert{\mathcal{D}}}( A'_{\mu'} \big \vert C_{\mu''} ) \, 
\rho_{\mathcal{D}}( C_{\mu''} )
\label{eq:partbounds2}  \\
\chi 
(\rho_{\mathcal{A}\vert{\mathcal{D}}} & ( A_{\mu} \big \vert C_{\mu''} ),
\rho_{\mathcal{A}'\vert{\mathcal{D}}}( A'_{\mu'} \big \vert C_{\mu''} ) 
).
\label{eq:partbounds3}
\end{align}
\end{subequations}
\noindent where $A_{\mu}$ and $A'_{\mu'}$ are, respectively, sets of the partitions 
$\mathcal{A}$ and $\mathcal{A}'$, and $C_{\mu''}$ is a set
from the finest common partition $\mathcal{D}$ that contains both $A_{\mu}$ and $A'_{\mu'}$.

If we have a family of partitions, in some cases we can identify a sequence of partitions that can 
be connected by a chain of conditions in the form \eqref{eq:partbounds}.
The advantage of this type of condition is that it may involve fewer probability components than
would be required if conditions were imposed directly on  partitions $\mathcal{A}$ and $\mathcal{A}'$, and is therefore 
expected to take a simpler form.

\subsection{State Space conditions and symmetries} \label{sec:statespacesymm}

A central concept of this study is that, in order to describe the behavior of a physical system, we
just have to identify which states are physically valid.
As shown in the previous subsection, we can define the  probability distributions on the partitions of
the space of the observables,
and then require some conditions on these probability distributions.
Such conditions will be identified with the help of the axioms introduced in section
\eqref{ax:precisioninvariance}. 

The requirement that a physical system obeys certain symmetries is one of the hypotheses that can help us
determine the conditions $\chi$, on the basis of requirement \ref{ax:symmetries}.
The invariance of the theory under a given symmetry can be put in  the form
$\chi(\rho'_1 ... \rho'_n) = \chi(\rho_1 ... \rho_n)$, where $\rho'_i $ are the transformed
probabilities, and the conditions  $\chi(\rho_1 ... \rho_n)$ are treated as boolean-valued  functions.
A fundamental tool for building a theory that is invariant under a given transformation
comes from the following definition:

\begin{definition}\label{def:invariantpartition}
{\bf{Invariant partitions}}
A partition is said {\emph{invariant}} with respect to a transformation if
each of its sets is transformed into itself or into another set of the same partition. 
\end{definition}

In what follows we will show how this definition applies to the system of binary partitions introduced 
in the previous sections.
   
Let us start from a family of binary partitions $\mathdutchcal{a}^1 ... \mathdutchcal{a}^n$, whose
sets generate the family of increasing resolution
partitions $\mathcal{A}^1 ... \mathcal{A}^n$ (through the formula \eqref{eq:binsetpartition}).
Suppose we have a generic transformation $T$ of the system and, for any $i$, the partitions
$\mathcal{A}^i$ are invariant under $T$; namely, any set of $\mathcal{A}_{\mu}^i$
is transformed by $T$ into 
another set of the same partition: $A_{\mu}^i \stackrel{T}{\rightarrow} A_{\mu'}^i$.
This hypothesis is supposed to hold for any $i$; for example, we identify two invariant partitions
$\mathcal{A}^i$ and $\mathcal{A}^{i+1}$, which differ for the degree of resolution on the space of
an observable.
The invariance of a partition allows us to say that, both $A_{\mu}^i$ and the transformed 
set  $A_{\mu'}^i$ belong to $\mathcal{A}^i$. This implies that both can be decomposed into higher-resolution subsets that 
belong to the same partition $\mathcal{A}^{i+1}$:
$A_{\mu}^i = A_{0 \mu}^{i+1} \cup A_{1 \mu}^{i+1}$ and
$A_{\mu'}^i = A_{0 \mu'}^{i+1} \cup A_{1 \mu'}^{i+1}$.
Note that, in these formulas and in some of the formulas in this section, we  
use notations like $0\mu$ or $1\mu$, which must be read as the ``the number obtained by putting 0 (or 1) 
in front of the binary representation of $\mu$''. Hence, we can write:
\begin{equation}\label{eq:invpartmultires}
A_{0 \mu'}^{i+1} \cup A_{1 \mu'}^{i+1} = A_{\mu'}^i = T A_{\mu}^i = T A_{0 \mu}^{i+1} \cup T A_{1 \mu}^{i+1}.
\end{equation}
\noindent Since, as stated above, the higher-resolution decomposition $\mathcal{A}^{i+1}$ is invariant under $T$ as well,
both $T A_{0 \mu}^{i+1} $ and $ T A_{1 \mu}^{i+1}$ are expected to be sets of the partition $\mathcal{A}^{i+1}$.
This implies that a set $A_{\nu}^{i+1}$ is transformed by $T$ into itself or into another set
of  $\mathcal{A}^{i+1}$. Namely, in the last equation, we have that either
$T A_{0 \mu}^{i+1} = A_{0 \mu'}^{i+1}$ and $T A_{1 \mu}^{i+1} = A_{1 \mu' }^{i+1}$, or
$T A_{0 \mu}^{i+1} = A_{1 \mu'}^{i+1} $ and $ T A_{1 \mu}^{i+1} = A_{0 \mu'}^{i+1}$. 
 
If we group $A_{0 \mu}^{i+1}$ and $A_{1 \mu}^{i+1}$ in an ordered pair, we can write
$(A_{0 \mu}^{i+1}, A_{1 \mu}^{i+1}) \stackrel{T}{\rightarrow}
 {\large{\pi}}_T (A_{0 \mu'}^{i+1}, A_{1 \mu'}^{i+1})$, where  
${\large{\pi}}_T$ is an optional permutation of the ordered pair's elements. 

\mysubsubsection{State Space conditions on invariant partitions}
The basic idea is to break the state space condition into a family of weaker
conditions, each depending on a subset of a set taken from a partition
$\mathcal{A}$, which is invariant under a transformation $T$.
When we apply $T$, the sets of $\mathcal{A}$ are not broken into
subsets, which implies that the conditions $\chi$ on such sets are not mixed.
Rather, the transformation $T$ leads to an exchange among different $\chi$s
on different sets of $\mathcal{A}$. This makes our goal of building a theory invariant
under some symmetries easier.

We will start from the scheme outlined by Eqs.\eqref{eq:partbounds}, in which we consider 
only some components of the state, represented by conditional probabilities of partitions 
related to two observables. In this scheme, the conditions $\chi$ are mainly related to 
such conditional probabilities.

In the framework of a nested binary partitions scheme, we can express the state
space conditions in a form that emphasizes the dependencies on the partitions related to a
single observable (the full two-observable dependence will be reintroduced later):
\begin{equation}\begin{aligned}
\label{eq:chiinvariancegen}
\chi \left( \rho(A_{0 0...0}^{i+1}\vert A_{0...0}^i), \rho(A_{1 0...0}^{i+1}\vert A_{0...0}^i), \right. \twocolbreak
 \left. ... , \rho(A_{1 1...1}^{i+1}\vert A_{1...1}^i); \boldsymbol{\phi}(A ...) \right),
\end{aligned}\end{equation}
\noindent where the sets $A_{\nu}^{i+1}$ refer to the domain of one observable and belong
to the $i+1$-th level of binary partitions;
the term $\boldsymbol{\phi}(A ...)$ encloses the dependencies on other observables or on
other partition levels and, in general, can be multi-dimensional.
As anticipated above, we suppose that it is possible to rewrite the 
state space condition \eqref{eq:chiinvariancegen} as a family of 
distinct $\chi$ conditions, each operating on a different set of a partition:
\begin{equation}
\label{eq:chidecoupled}
 \bigwedge_{\mu} \chi^{\mu} \left( \rho(A_{0 \mu}^{i+1}\vert A_{\mu}^i),
 \rho(A_{1 \mu}^{i+1}\vert A_{\mu}^i); \boldsymbol{\phi}^{\mu}(A ...) \right),
\end{equation}
\noindent where the symbol $\wedge$ indicates the logical conjunction between different 
conditions $\chi$ (we assume that they are boolean-valued functions). 
This way of writing the state space condition may seem
somewhat tricky because it can be the result of a different grouping 
of the parameters: the probabilities $\rho(A_{0 \mu}\vert A_{\mu}), \rho(A_{1 \mu}\vert A_{\mu})$ 
first, and then the other probabilities, embedded within the generic function
$\phi^{\mu}$. But, as we will see in the last part of this study, 
we can identify a decomposition of the state space condition
into simpler $\chi$ conditions, each depending on a few components
of the probability distributions.
We will still have a dependence on other components that, however, is weak and manageable.

It should be noted that, in equations  \eqref{eq:chiinvariancegen} and \eqref{eq:chidecoupled},
we are free to express the conditional probabilities either in terms of the
intersection $A_{k \mu}^{i+1} = a_{k \mu}^{i+1} \cap A_{\mu}^i$ or in terms
of the binary sets $a_{k \mu}^{i+1}$ because, according to the definition of
conditional probability, for any two sets $A$ and $B$, we can formally write
$\rho(A \vert B) = \rho(A \cap B \vert B)$.
Therefore we can make $\chi$ depend either on the probabilities
$\rho(A_{k \mu}^{i+1}\vert A_{\mu}^i)$ or $\rho(a_{k \mu}^{i+1}\vert A_{\mu}^i)$ .
Furthermore, in expression \eqref{eq:chidecoupled}, we can group the first two $\rho$'s to form a
two-component probability vector $\boldsymbol{\rho}(A_{\bullet \mu}^{i+1} \vert A_{\mu}^i) = 
(\rho(A_{0 \mu}^{i+1}\vert A_{\mu}^i), \rho(A_{1 \mu}^{i+1} \vert A_{\mu}^i))$.
Thus, the expression \eqref{eq:chidecoupled} can be rewritten in the shorter form:
\begin{equation}\label{eq:chielembin}
\bigwedge_{\mu} \chi^{\mu} \big( \boldsymbol{\rho}(A_{\bullet \mu}^{i+1} \vert A_{\mu}^i) ;
\boldsymbol{\phi}^{\mu} \big).
\end{equation}
\noindent Now we are going to see what happens to the above expression of
the state space conditions, when we apply a transformation $T$ 
on the sets $A_{k \mu}^{i+1}$. Since we suppose the $i$th-level partition $\mathcal{A}^{i}$
to be invariant under $T$, we can assert that a set  $A_{\mu}^{i}$ of 
$\mathcal{A}^{i}$ is transformed into another set $A_{\mu'}^{i}$,
and use the same symbol $T$ for the operator that transforms 
the indices of the sets: $\mu' = T \mu$. Thus, by applying $T$ to \eqref{eq:chielembin},
we obtain:
\begin{equation}\label{eq:chielembintransf}
\bigwedge_{\mu} \chi^{\mu}
\big( {\large{\pi}}_T  \boldsymbol{\rho}(A_{\bullet \mu'}^{i+1} \vert A_{\mu'}^i) ;
\boldsymbol{\phi}^{\mu} \big),
\end{equation}
\noindent where the operator ${\large{\pi}}_T$ takes into account the fact that the
transformation $T$ may produce an exchange between the higher-resolution
sets  $A_{0 \mu}^{i+1}$ and $A_{1 \mu}^{i+1}$ and, consequently, 
between the two components of $\boldsymbol{\rho}$. This is an {\it{optional}} exchange operator because transformation $T$ may
or may not involve an exchange between sets.
In the last expression, the index $\mu$ ranges within all its allowed values. 
Therefore, we can replace $\mu$ by $T^{-1} \mu'$ and evaluate the conjunction
$\wedge$ for all the $\mu'$s:
\begin{equation}\label{eq:chielembintransf2}
\bigwedge_{\mu'} \chi^{T^{-1} \mu'} 
\big( {\large{\pi}}_T  \boldsymbol{\rho}(A_{\bullet \mu'}^{i+1} \vert A_{\mu'}^i) ;
\boldsymbol{\phi}^{T^{-1} \mu'} \big).
\end{equation}
\noindent We can now introduce the hypothesis that the theory is invariant under $T$ by
requiring the state space condition \eqref{eq:chielembin} and its transformed
version \eqref{eq:chielembintransf2} to be equal. In order to fulfill this equality,
it is sufficient to require that the terms with equal index $\mu$ of both expressions
are equal. Therefore, up to a renaming of the index $\mu' \rightarrow \mu$, we can express 
the $T$-invariance condition as
\begin{equation}\begin{aligned}\label{eq:chiinvariance}
\chi^{\mu} & \big( \boldsymbol{\rho}(A_{\bullet \mu}^{i+1} \vert A_{\mu}^i) ;
\boldsymbol{\phi}^{\mu} \big) \twocolbreak & = \chi^{T^{-1} \mu}
\big( {\large{\pi}}_T  \boldsymbol{\rho}(A_{\bullet \mu}^{i+1} \vert A_{\mu}^i) ;
\boldsymbol{\phi}^{T^{-1}\mu} \big).
\end{aligned}\end{equation}
It is worth noting that the choice of a particular family of partitions does not 
imply any hypothesis on the expression of the physical state conditions $\chi$.
Consequently, when we choose a partition that is invariant under $T$, we are not
implicitly requiring the theory to be $T$-invariant.
But the invariance of the family of partitions  $\mathcal{A}^1 ... \mathcal{A}^n$ under $T$
allows us to write the invariance of the theory in the form \eqref{eq:chiinvariance}, which involves each pair of sets
$(A_{0 \mu}^{i+1}, A_{1 \mu}^{i+1})$ independently.

\mysubsubsection{Translational symmetry} \label{sec:translsymm}

In a discrete one-dimensional space, the group of translations is generated by
the increment operator $S : x \rightarrow x+1$. In our case, we have a finite size 
space and we can simplify the model by assuming a periodicity of this space.
Under this hypothesis, the translation behaves as a circular shift:
\begin{equation} \label{th:circshift}
S: x \rightarrow x+1 \mod{N},
\end{equation}
\noindent where $N$ is the size of the space.
If we assume that $N = 2^n$, all partitions with equipotent sets have a cardinality equal to a power of two, and
we can demonstrate the following theorem: 

\begin{theorem}\label{th:lsbgen}
Let $X$ be the set of the $n$-bit integers and let $S: x \rightarrow x + 1 \mod{N}$ denote
the increment (circular shift) of $x \in X$.
Then, a  partition $\mathcal{A}$ with cardinality $2^c$ is invariant under the
transformation $S$ if and only if its sets are $A_{\mu} = \{ x \stackrel{l...n}{=} \mu \}$,
being $l = n - c + 1$.
\end{theorem}

The proof is presented in Appendix~\ref{sec:trasltheor}.
From this theorem, it follows that one can identify a family of $S$-invariant partitions with 
increasing resolution---the lower the value of $l$, the finer the resolution of the partition.
If we consider state space conditions $\chi$ in the form \eqref{eq:chiinvariancegen}
and \eqref{eq:chidecoupled}, we can observe that they deal with sets, denoted by 
$A^i_{\mu}$, belonging to partitions whose resolution depends on the index $i$. 
In order to have a translational invariance condition on the $\chi$'s in the simple form
\eqref{eq:chiinvariance}, which applies to decoupled $\chi$ conditions,
we can suppose that the sets  $A^i_{\mu}$ are grouped into families of partitions that
correspond to the families of translation-invariant partitions identified by the above theorem.

It should be noted that, in theorem \ref{th:lsbgen}, the requirement that the partitions have
cardinalities that are powers of 2 is justified by the assumptions made at the beginning of this
section---namely, that the number of values a physical quantity can take is a power of 2. 
Within the binary decomposition scheme adopted here, the only way to divide the space into a family of partitions with increasing resolution is to require
that the cardinality of these partitions be a power of 2.

\mysubsubsection{The discrete phase space} \label{sec:discphspace}

As stated before, the purpose of this section is to study a one-dimensional system
with two canonically conjugate variables $p$ and $q$, defined on the 
spaces $\Omega_p$ and $\Omega_q$. We require these spaces to be discretized into 
$N_p$ and $N_q$ equally spaced intervals.
In order to simplify our model, we also assume that $N_p$ and $N_q$ are powers of 2:
$N_q = 2^{n_q}$ and $N_p = 2^{n_p}$. If this system were a (discretized) classical one-dimensional
system, then the number of bits required to define the state would be
$I = \log_2 N_q + \log_2 N_p = n_q + n_p$.
Of course, the real physical system is given by the limit
$N_q \rightarrow \infty, N_p \rightarrow \infty$,
and $\Omega_p$ and $\Omega_q$ are enclosed in sub-volumes of the real space.

To clarify how this simplification into smaller volumes and
discrete spaces works, we can assume the following: we start from larger physical spaces 
$\Omega_p^{\mathit{phys}}, \Omega_q^{\mathit{phys}}$ and from an indefinitely high resolution 
discretization whose intervals are  $\epsilon_p^{\mathit{phys}}$ and $\epsilon_q^{\mathit{phys}}$.
We denote by $V_p$, $V_q$, $V_p^{\mathit{phys}}$ and $V_q^{\mathit{phys}}$ the volumes of the
spaces $\Omega_p$, $\Omega_q$, $\Omega_p^{\mathit{phys}}$ and $\Omega_q^{\mathit{phys}}$
respectively.
Then, we assume that our model describes the physics within the scales $\epsilon_q ... V_q$
and $\epsilon_p ... V_p$ of $q$ and $p$, where $\epsilon_{\nu} > \epsilon_{\nu}^{\mathit{phys}}$
and $V_{\nu} < V_{\nu}^{\mathit{phys}}$. We can express these ranges on a logarithmic scale:
$l_{\mathit{min}}^q = \log_2 \frac{V_q^{\mathit{phys}}}{V_q} +1$
and $l_{\mathit{min}}^p = \log_2 \frac{V_p^{\mathit{phys}}}{V_p} +1$ 
are the lowest resolution scales, and 
$l_{\mathit{max}}^q = \log_2 \frac{V_q^{\mathit{phys}}}{\epsilon_q}$
and $l_{\mathit{max}}^p = \log_2 \frac{V_p^{\mathit{phys}}}{\epsilon_p}$
are the highest-resolution scales.
According to these definitions, if our system occupied all the available volume (namely, 
 $V_{p,q} = V_{p,q}^{\mathit{phys}}$), then we would have
 $l_{\mathit{min}} =1$.

A classical system, represented in terms of the canonically conjugate variables $p$ and $q$,
in the absence of dynamics, has some intrinsic symmetries, such as
translational and the scale symmetry.
If we want our quantization program to be consistent with the underlying classical phase-space
representation of a system, we must require our theory to be intrinsically
invariant under these symmetries. 

By following the approach introduced in sections \ref{sec:translsymm} and \ref{sec:statespacesymm}, 
we handle the symmetries of a system by dividing the phase space $\Omega_p \otimes\Omega_q$
into families of set partitions, which are invariant with respect to translational and scale
transformations. To understand how these invariant partitions should be made, we will use the theorem
\ref{th:lsbgen}

The binary representation of an observable $x$ can be seen as a sequence of nested binary partitions,
where the $i$-th bit tells us in which set of the $i$-th partition $x$ is located.
The last interval of this subdivision, the one with the highest resolution, has a certain width,
which we previously denoted by $\epsilon$; ideally, the resolution on $x$ can be pushed
further, and we can define finer-grained binary partitions and add bits to the binary representation.
In theorem \ref{th:lsbgen}, the bit corresponding to the highest resolution is in the $n$-th position.
The statement of that theorem can be reformulated by assuming that the bit sequence continues, so we can
rename $n$ as $l_{\mathit{max}}$, which corresponds to the precision $\epsilon$, related to $l_{\mathit{max}}$ through
$l_{\mathit{max}} = \log_2 \frac{V^{\mathit{phys}}}{\epsilon}$.
Clearly, to keep the essence of the theorem unchanged, we never make statements that define $x$
at a resolution higher than $\epsilon$ (corresponding to bits beyond the $l_{\mathit{max}}$-th one),
and the shift involved in the theorem is of $2^{n - l_{\mathit{max}}}$ bits. 
We can therefore conclude that sets of the translation-invariant partitions involved by theorem
\ref{th:lsbgen} can be put in the form $A_{\mu} = \{ x \stackrel{l...l_{\mathit{max}}}{=} \mu \}$.

In a two-dimensional phase space $\Omega_p \otimes \Omega_q$, we must require  
 the translational invariance along both the $p$ and $q$ axes.  
By extending the result of theorem \ref{th:lsbgen} to the two-dimensional $p,q$ case, we find that, in order for a partition to 
be invariant under translations along both the $p$ and $q$ axes, it must have the form:
\begin{equation} \label{eq:translinvpartpq}
\mathcal{A}^{l, l'} = \{ A_{x,y}^{l, l'}, \,  \forall x,y\},
\end{equation}
\noindent where
\begin{equation}\label{eq:decomppq}
\small{A_{x,y}^{l, l'} =
    \{ q \stackrel{l... l_{\mathit{max}}^q}{=} [ x_{l} ... x_{l_{\mathit{max}}^q}] \}  \cap
    \{ p \stackrel{l'...l_{\mathit{max}}^p}{=} [ y_{l'} ... y_{l_{\mathit{max}}^p}] \}.}
\end{equation}

\mysubsubsection{Scale symmetry}\label{sec:scalesymm}
The second phase-space intrinsic symmetry  we are going to consider 
is the scale symmetry, which is defined by transformations like 
$q \to \alpha q$,  \  $p \to \frac{1}{\alpha} p$.

Since we are dealing with the discrete spaces $\Omega_p$ and $\Omega_q$, which have a number of points
equal to a power of 2, we consider  only scale transformations of the form   
$q \rightarrow 2 q$ or $q \rightarrow \frac{1}{2} q$.  
Let us define a scale operator $S_{\textnormal{sc}}$ that acts on
 probability distributions in the following way:
\begin{equation} \label{eq:scaletransdef}
\begin{split}
    S_{\textnormal{sc}} \cdot \rho(q)  &\stackrel{\textnormal{def}}{=} \rho(2 q) \\
    S_{\textnormal{sc}} \cdot \rho(p)  &\stackrel{\textnormal{def}}{=} \rho(p/2)
\end{split}
\end{equation}
\noindent
Let us now examine how the sets defined by \eqref{eq:decomppq} change under scale transformations. 
An expansion by a factor of 2 on a set like
$\{ q \lsbeq [ x_{l_{\mathit{min}}^q} ... x_{l_{\mathit{max}}^q}] \}$
is represented by a bitwise left shift.
Vice versa, a contraction by a factor of 2 is equivalent to a right shift. 
Thus, when we apply the transformation \eqref{eq:scaletransdef} to the set \eqref{eq:decomppq} we
obtain:
\begin{equation}\begin{aligned} \label{eq:decomppqtrasf}
  S_{\textnormal{sc}} \cdot A_{x,y}^{l, l'} = &
    \{ q \stackrel{...}{=} [ x_{l-1} ... x_{l_{\mathit{max}}^q-1}] \} \twocolbreakamp \cap 
    \{ p \stackrel{...}{=} [ y_{l'+1} ... y_{l_{\mathit{max}}^p+1}] \},
\end{aligned}\end{equation}
\noindent where $\stackrel{...}{=}$ is a shorthand notation for the digit-subset equality relation,
whose bit range is specified by the bit indices of the RHS.
From the above relation, we can see that the scale transformation changes a
set of the form \eqref{eq:decomppq} into another set, by keeping the 
sum $l+l'$ constant.
We can conclude that, in the case of scale transformations, the partitions
defined by \eqref{eq:translinvpartpq} and \eqref{eq:decomppq} are not invariant, 
but the effect of $S_{\textnormal{sc}}$ is to move each partition $\mathcal{A}^{l, l'}$
into the partition $\mathcal{A}^{l-1, l'+1} $. 

Thus, if we consider the family of partitions obtained by varying $l$ and $l'$ with the 
condition $l+l'=  {\mathit{const}.}$ and neglect the cases
$l,l' = l_{\mathit{min}}, l_{\mathit{max}}$ (which would make
$l$ and $l'$ go out of range), we can assert that the scale transformation
exchanges the partitions within a family, by keeping the family unchanged.

This may help us understand the conditions that make the theory invariant under scale transformations.
More specifically, it is reasonable to accept the following statement: when we apply
$S_{\textnormal{sc}}$ to a state space condition of the form \eqref{eq:partbounds3} or
\eqref{eq:chiinvariancegen}, we transform each pair of sets occurring in a condition into a pair that
occurs in another condition  $\chi$ of a different level $l$. This is equivalent to moving each $\chi$
condition from a level $l$ to another. If the whole family of $\chi$'s is left unchanged by this
transformation, then the invariance of the theory under scale transformation is guaranteed.

The choice of neglecting the cases $l,l' = l_{\mathit{min}}, l_{\mathit{max}}$ can be justified by the
following arguments: 

\noindent (i) In a realistic model of a physical system, the probability distributions can
not be fast varying on an arbitrarily small scale. There must be a scale that
is small enough to make a probability distribution function behave like a constant
function. This means that, by varying the values of the least significant (i.e., 
closest to $l_{\mathit{max}}$) bits of a physical quantity, the probability does 
not change appreciably, and we can neglect the cases where $l$ and $l'$ are close to  
$l_{\mathit{max}}$.

\noindent (ii) We can assume that the values of $p$ (or $q$) whose probability
is non-zero are confined within a region smaller than the size of the available  space 
 $V_p = V_p^{\mathit{phys}} / 2^{l_{\mathit{min}}-1}$ (or  $V_q = V_q^{\mathit{phys}} / 2^{l_{\mathit{min}}-1}$).
This allows us to neglect the cases where $l$ and $l'$ are close to  $l_{\mathit{min}}$.

\mysubsubsection{$p \leftrightarrow q$ exchange symmetry }

There is another phase-space intrinsic symmetry, which will help us build our theory.
In the absence of dynamics, the variables $p$ and $q$ are expected to be interchangeable and,
as a consequence, the spaces $\Omega_p$ and $\Omega_q$ are supposed to have the same 
property. This provides more  information on how the range 
$l_{\mathit{min}}^q ... l_{\mathit{max}}^q, l_{\mathit{min}}^p ... l_{\mathit{max}}^p$ can be chosen. 

We can reasonably require the spaces $\Omega_p$ and $\Omega_q$ to have the 
same range of resolution, namely:
\begin{equation}\label{eq:pqexcsymm}
 l_{\mathit{max}}^q - l_{\mathit{min}}^q = 
l_{\mathit{max}}^p - l_{\mathit{min}}^p .  
\end{equation}

 \subsection{The reconstruction scheme and the butterfly diagram} \label{sec:recscheme}

\mysubsubsection{Intermediate partition chain}
With equations \eqref{eq:partbounds}, we introduced the idea of connecting
probability functions defined on different partitions of the state space
through a chain of intermediate partitions. Each partition shares
a common sub-partition with the next partition of the chain. 
This leads to a family of state space conditions $\chi$, which depend on a 
smaller number of probability components.

In this section we apply this idea to the translation- and scale-invariant partition scheme
of the phase space $\Omega_p \otimes \Omega_q$, introduced in the previous sections. 
In other words, we have to rewrite equations \eqref{eq:partbounds} for the phase space partitions 
whose sets are given by the expression \eqref{eq:decomppq}. 
We will make the strong assumption that the $\chi$'s can be put in an explicit 
form that permits us to derive one of the conditional probabilities 
from the other (for instance, in equation  \eqref{eq:partboundsxpc}, $ \rho_{q\vert qp}$ from $ \rho_{p\vert qp}$). 
 
Our next goal is to show how to build a chain of intermediate probability functions that connect 
the probability $\rho(q)$ and $\rho(p)$.
The probabilities of a set $A_{x,y}$ in the form \eqref{eq:decomppq},
for the lowest $l$ and highest $l'$ (or vice-versa), are nothing but the full resolution 
probabilities $\rho(q)$ (or $\rho(p)$). In other words, 
if we consider the set $A_{x,y}^{l, l'+1}$, the case $\rho(q)$ corresponds to the values
$l=l^q_{\mathit{min}}, l'=l^p_{\mathit{max}}$, while $\rho(p)$ corresponds to the 
probability of $A_{x,y}^{l+1, l'}$, with $l=l^q_{\mathit{max}}, l'=l^p_{\mathit{min}}$.

If we had a formula that permits to obtain the probability of a set $A_{x,y}^{l+1, l'}$ from the
probability of a set $A_{x,y}^{l, l'+1}$, we could build the chain of intermediate probabilities:
we start from $\rho_q = \rho(A^{l^q_{\mathit{min}},l^p_{\mathit{max}}+1})$, 
then we get $\rho(A^{l^q_{\mathit{min}}+1,l^p_{\mathit{max}}})$, then
$\rho(A^{l^q_{\mathit{min}}+2,l^p_{\mathit{max}}-1})$ and so on,
until we obtain $\rho(A^{l^q_{\mathit{max}},1}) = \rho_p$.
At each step we have increased $l$ and decreased $l'$, starting from
$l=l^q_{\mathit{min}}, l'=l^p_{\mathit{max}}$.
Consequently, along the chain, $l$ and $l'$ are related by the equation: 
\begin{equation} \label{eq:llpbound}
    l_{\mathit{max}}^q - l
      = l'  - l_{\mathit{min}}^p.
\end{equation}
\noindent Let us rewrite the sets \eqref{eq:decomppq} for the cases $(l+1,l')$ and $(l,l'+1)$:
\begin{subequations}\label{eq:decomppqllp}
\begin{alignat}{1}
A_{x,y}^{l+1, l'} & =
    \{ q \stackrel{...}{=} [ x_{l+1} ... x_{l_{\mathit{max}}^q}] \} \cap
    \{ p \stackrel{...}{=} [ y_{l'} ... y_{l_{\mathit{max}}^p}] \}  \label{eq:decomppqllp1}\\
A_{x,y}^{l, l'+1} & =
    \{ q \stackrel{...}{=} [ x_{l} ... x_{l_{\mathit{max}}^q}] \} \cap
    \{ p \stackrel{...}{=} [ y_{l'+1} ... y_{l_{\mathit{max}}^p}] \}  \label{eq:decomppqllp2},
\end{alignat}
\end{subequations}
\noindent where $l_{\mathit{min}}^q \le l \le l_{\mathit{max}}^q$ 
and $l_{\mathit{min}}^p \le l' \le l_{\mathit{max}}^p$.
The finest common partitions we can build from the above families of sets are defined by the sets: 
\begin{equation}\begin{aligned}
\label{eq:decomppqcommon}
& C_{x,y}^{l,l'} = A_{x,y}^{l+1,l'+1} =  \twocolbreak &
    \{ q \stackrel{...}{=} [ x_{l+1} ... x_{l_{\mathit{max}}^q}] \} \cap
    \{ p \stackrel{...}{=} [ y_{l'+1} ... y_{l_{\mathit{max}}^p}] \}. \hfill
\end{aligned}\end{equation}
\noindent Equations  \eqref{eq:partbounds} can now be specialized as follows:
the expression \eqref{eq:partbounds1} applies to the partition whose sets are 
\eqref{eq:decomppqllp1}  and \eqref{eq:partbounds2} applies to the partition whose sets are
\eqref{eq:decomppqllp2}. Thus we have
\begin{subequations}\label{eq:partboundsxp}
\begin{alignat}{1}
& \rho(A^{l, l'+1} )  =  \rho_{q\vert qp}^{l l'} \, \rho(A^{l+1, l'+1} ) \label{eq:partboundsxpa} \\
& \rho(A^{l+1, l'} )   =  \rho_{p\vert qp}^{l l'} \, \rho(A^{l+1, l'+1} ) \label{eq:partboundsxpb}\\
& \chi(\rho_{q\vert qp}^{l l'} , \rho_{p\vert qp}^{l l'}),  \label{eq:partboundsxpc}
 \end{alignat}
\end{subequations}
\noindent where, in the probability functions $\rho(A)$,  the subscripts have been omitted
and the conditional probabilities are defined as 
\begin{equation} \label{eq:rhocondllp}
\begin{aligned}
\rho_{q\vert qp}^{l l'}  =  \, \rho(q \stackrel{l}{=}[x_l] \bigm \vert
  & q \stackrel{...}{=} [ x_{l+1} ... x_{l_{\mathit{max}}^q}] \twocolbreak  & \wedge
 p \stackrel{...}{=} [ y_{l'+1} ... y_{l_{\mathit{max}}^p}]) 
 \\
\rho_{p\vert qp}^{l l'}  = \, \rho(p \stackrel{l'}{=}[y_{l'}] \bigm \vert
 & q \stackrel{...}{=} [ x_{l+1} ... x_{l_{\mathit{max}}^q}] \twocolbreak  & \wedge
 p \stackrel{...}{=} [ y_{l'+1} ... y_{l_{\mathit{max}}^p}]).
 \end{aligned}
\end{equation}
\noindent Note that the conditions $\chi$ in \eqref{eq:partboundsxpc} have been put in short form too.
Actually, they depend on the level identifiers $l\, l'$ and on the lowest significant bits of $p$ and
$q$ in the following way: 
\begin{equation}
\label{eq:chiindexing}
\chi^{l,l'}_{[ x_{l+1} ... x_{l_{\mathit{max}}^q}]
[ y_{l'+1} ... y_{l_{\mathit{max}}^p}]}
\end{equation}
\noindent Now we have everything we need to build a recursive process for the derivation of probability distributions that goes 
$\rho(q)$  to  $\rho(p)$, by passing through a sequence of intermediate probability distributions defined
on the partitions of the form $\mathcal{A}^{l, l'}$.
This process can be divided into three steps.

\mysubsubsection{Step 1: from partition  $\mathcal{A}^{l, l'+1}$ to partition $\mathcal{A}^{l+1, l'+1}$}
\noindent By expressing the $x$ index of the sets $A$ in binary form, the sets of the partition 
$\mathcal{A}^{l+1, l'+1}$ can be written as a union of two smaller sets of the higher-resolution
partition $\mathcal{A}^{l, l'+1}$: 
\begin{equation} \nonumber
    A^{l+1, l'+1}_{x, y}  = 
    A^{l, l'+1}_{[ x_{l} x_{l+1} ... x_{l_{\mathit{max}}^q}] \, y}   \cup 
    A^{l, l'+1}_{[\sim  x_{l} x_{l+1} ... x_{l_{\mathit{max}}^q}] \, y}  .
\end{equation}
\noindent Since the sets in the RHS of this equation are disjoint, we have:
\begin{equation}\begin{aligned}
\label{eq:rhotadisj}
     \rho(A^{l+1, l'+1}_{x, y} )  = 
    \rho(A^{l, l'+1}_{[ x_{l} x_{l+1} ... x_{l_{\mathit{max}}^q}] \, y} ) \twocolbreak  + 
    \rho(A^{l, l'+1}_{[\sim  x_{l} x_{l+1} ... x_{l_{\mathit{max}}^q}] \, y} ) . 
\end{aligned}\end{equation}

\noindent By substituting the above expression of $\rho(A^{l+1, l'+1}_{x, y} )$ in
the RHS of equation \eqref{eq:partboundsxpa}, we get:
\begin{equation} \label{eq:rhocondfromrho}
\rho_{q\vert qp}^{l l'} \hspace{-0.1cm} = \hspace{-0.1cm}\frac{\rho(A^{l, l'+1} )}
{\rho(A^{l, l'+1}_{[ x_{l} x_{l+1} ... x_{l_{\mathit{max}}^q}] \, y} )\hspace{-0.1cm} + \hspace{-0.1cm}
    \rho(A^{l, l'+1}_{[\sim  x_{l} x_{l+1} ... x_{l_{\mathit{max}}^q}] \, y} )}
\end{equation} 
\noindent where, in case of a zero denominator (and numerator) we can put 
$\rho_{q\vert qp}^{l l'} = 1/2$.
The above equations permit us to derive the conditional probability $\rho_{q\vert qp}^{l l'} $ 
from the $l, l'+1$-level probability. 

\mysubsubsection{Step 2: from $\rho_{q\vert qp}^{l l'}$ to  $\rho_{p\vert qp}^{l l'}$} \label{par:step2}
By anticipating what will be described in the last part of this paper  
(sections \ref{sec:fourptsys} and \ref{sec:traslinv}),
we can introduce more hypotheses, which strengthen our condition $\chi$ to such an 
extent that it can be used as a transformation from  $\rho_{q\vert qp}^{l l'}$ 
(and other parameters) to  $\rho_{p\vert qp}^{l l'}$. 
Actually this implies that the $\boldsymbol{\phi}$ parameters in the condition $\chi$ 
of equation \eqref{eq:partboundsxpc} depend only on lower $l$ and higher $l'$-level probabilities.
Once we have  $\chi$ in this explicit form, we can obtain the conditional probability
$\rho_{p\vert qp}^{l l'}$.

\mysubsubsection{Step 3: from $\rho_{p\vert qp}^{l l'}$ to  $\rho(A^{l+1, l'}_{x, y} )$}
This third step can be performed by simply substituting the expression of the 
conditional probability $\rho_{p\vert qp}^{l l'}$ in the RHS of equation \eqref{eq:partboundsxpb}. 
We obtain an expression for the probability $\rho(A^{l+1, l'}_{x, y} )$, which can be used as input
for the next recursion level.

\mysubsubsection{The Butterfly Diagram}
The recursive method we have just shown can be represented in a friendly graphical 
form, shown in Fig.~\ref{butterfly8}. The rules of this representation are the following:
\begin{itemize}
  \item The probabilities  $\rho(A^{l+1, l'}_{x, y} )$ and the $\chi$ conditions
     are represented as the nodes of a graph.
  \item The dependencies of the  $\chi$'s from the probabilities and the values
     of the probabilities obtained from the  $\chi$'s are represented as lines that join nodes.
  \item The probabilities  $\rho(A^{l+1, l'}_{x, y} )$ ($\rho$ nodes, represented in
   Fig.~\ref{butterfly8} by black dots) are placed on a Cartesian coordinate plane in the following way: 
   the vertical axis represents the level identifier $l + 1$ (or $l_{\mathit{max}} - l' + 
   l_{\mathit{min}} + 1$) and the horizontal ($\mu$) axis represents the indices $x$ and $y$ of
   $A^{l+1, l'}_{x, y}$, combined according to the binary form 
\end{itemize}
\begin{equation} \label{eq:butterflydiagx}
\phantom{aa} \mu = [y_{l_{\mathit{max}}} y_{l_{\mathit{max}}-1} ... y_{l_{\mathit{max}}+l_{\mathit{min}}-l}  x_{l+1}x_{l+2} ... x_{l_{\mathit{max}}}]
\end{equation}
\begin{itemize}
  \item The $\chi$ conditions are represented by nodes (depicted in the figure as white squares)
  with four incoming lines, which are connected to four $\rho$ nodes.
  A node corresponding to the condition \allowbreak\hspace{0pt} 
  $\chi^{l,l'}_{[ x_{l+1} ... x_{l_{\mathit{max}}^q}]  [ y_{l'+1} ... y_{l_{\mathit{max}}^p}]}$
  (where $l$ runs from $l_{\mathit{min}}^q$ to $l_{\mathit{max}}^q$ and 
  $l' = l_{\mathit{min}}^p + l_{\mathit{max}}^q - l$, see \eqref{eq:llpbound}) is connected to the
  $\rho$ nodes
\begin{equation} \label{eq:nodesonetwo}
\rho(A^{l, l'+1}_{[0 x_{l+1} ... x_{l_{\mathit{max}}^q}] \, y} ), \,\,\,
\rho(A^{l, l'+1}_{[1 x_{l+1} ... x_{l_{\mathit{max}}^q}] \, y} )
\end{equation}
\noindent and
\begin{equation} \label{eq:nodesthreefour}
\rho(A^{l+1, l'}_{x  \,[0 y_{l'+1} ... y_{l_{\mathit{max}}^p}]} ), \,\
\rho(A^{l+1, l'}_{x  \,[1 y_{l'+1} ... y_{l_{\mathit{max}}^p}]} ).
\end{equation}
\noindent These connections are justified by the fact that the conditional probabilities
occurring in the condition $\chi(\rho_{q\vert qp}^{l l'} , \rho_{p\vert qp}^{l l'})$ can be related to the
probabilities $\rho(A^{l+1, l'}_{x, y} )$: according to Eq.~\eqref{eq:rhocondfromrho},
the conditional probability $\rho_{q\vert qp}^{l l'}$ can be expressed in terms of the probabilities
\eqref{eq:nodesonetwo}, while the probabilities \eqref{eq:nodesthreefour} can be obtained from
$\rho_{p\vert qp}^{l l'}$ by applying \eqref{eq:partboundsxpa}.
\end{itemize}

The expression \eqref{eq:butterflydiagx} of the horizontal position $\mu$ of the $\rho$ nodes can be explained as follows: 
we take the least significant bits of the indices $x$ and $y$ ($l_{\mathit{max}} - l$ bits from $x$
and $l - l_{\mathit{min}} + 1$ bits from $y$), then we concatenate the bits from $y$ in reversed order
with the bits from $x$.
Thus, if $l = l^q_{\mathit{min}} -1$, we have $l_{\mathit{max}}-l^q_{\mathit{min}}+1$ bits for $x$ and
zero bits for $y$, which implies that $\mu = x$. Conversely, if $l = l^q_{\mathit{max}}$, $\mu$ will
be given by the bit-reversal permutation of $y$.

The type of graphs we just introduced is known in the literature as butterfly diagram, and the case
$l_{\mathit{min}} = 1\,l_{\mathit{max}} = 3\,$ is shown in figure \ref{butterfly8}.

\begin{figure}[!ht]
\centering
\includegraphics[scale=0.72]{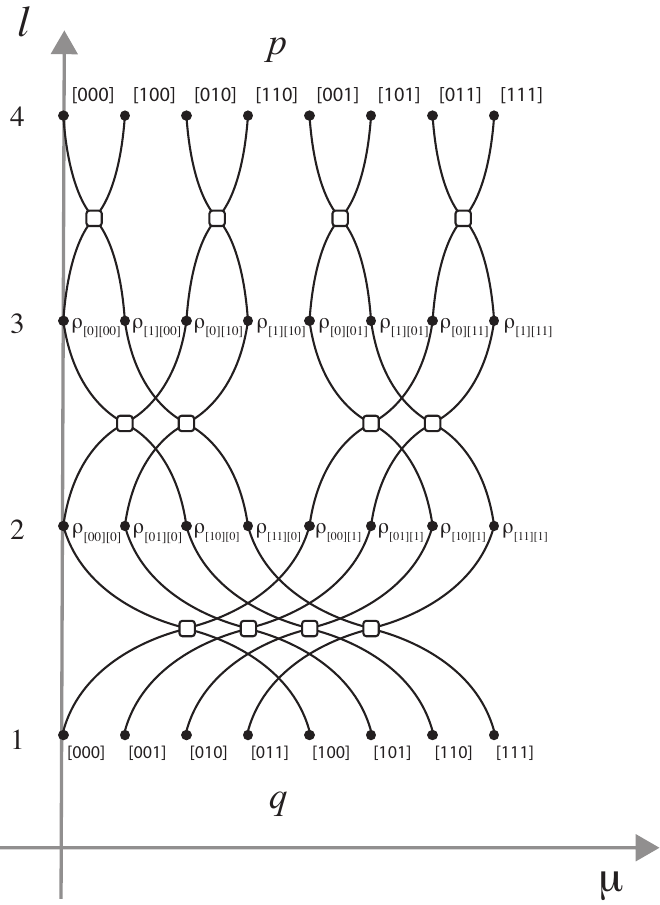}
\caption{The $l_{\mathit{max}} = 3,\, l_{\mathit{min}} = 1$ butterfly diagram.
The probabilities are represented
by black dots and the conditions $\chi$ by white squares. The 8 lowest dots  
represent the probabilities $\rho(q = 0) ... \rho(q=7)$ and the 8 highest dots
represent $\rho(p = 0) ... \rho(p=7)$, arranged in a bit-reversal order.
The second line (from the bottom)
of black dots represents the probabilities of the sets $A^{2, 3}_{x,y}$: 
$\rho_{[00][0]}, \rho_{[01][0]}, \rho_{[10][0]}, \rho_{[11][0]}, \rho_{[00][1]} ... \rho_{[11][1]}$ 
(where $\rho_{[x_2 x_3][y_3]} = \rho(q \stackrel{...}{=} [x_2 x_3] \land p \stackrel{...}{=} [y_3])$).
}
\label{butterfly8}
\end{figure}

\mysubsubsection{System ranges and information} \label{sec:separab}

The Limited Information axiom \ref{ax:limitedinfo} plays a key role in 
the development of our theory. 
A difficulty in applying this axiom to the probability-based approach introduced in this
section is that none of the principles in \ref{sec:principlesoftheory} establishes 
an explicit relation between the quantity ``amount of information'' and the probability
distributions occurring in the theory (except for the special cases of one-point and uniform
distributions involved by axiom \ref{ax:infostatequivalence}). The amount of information is
simply the number of bits needed to define some property of the system.

Nonetheless, by following qualitative but reasonable arguments, we can state how the 
amount of information carried by a system state depends on some parameters of the model.
As anticipated in section \ref{sec:discphspace}, in a discretized classical system 
whose state can take $N_p N_q$ values, the amount of information carried by the state is 
$n_p + n_q$, where $n_{p,q} = \log_2 N_{p,q}$. In the model we are going to consider,
the amount of information is smaller than or equal to $n_p + n_q$, and it will be denoted by $n_I$

In section \ref{sec:scalesymm}, we introduced  the hypothesis that the 
probability distributions are confined within a region whose size is 
$V^{\mathit{phys}}_q/2^{l_{\mathit{min}}^q}$ and 
$V^{\mathit{phys}}_p/ 2^{l_{\mathit{min}}^p}$, for $q$ and $p$ respectively.
Starting from that hypothesis, if, for example, we double the volume  $V_q$
(which implies a decrement of $l_{\mathit{min}}^q$ by 1) and keep the distributions 
unchanged, the amount of information we need to describe the system is incremented by 1.
This scenario is depicted in Fig.\ref{infoscalinga}. 
Conversely, under the hypothesis that the distributions are smooth on a scale greater than
$V^{\mathit{phys}}_q/2^{l_{\mathit{max}}^q}$, it can be shown that, by doubling the model's
resolution (which implies an increment of $l_{\mathit{max}}^q$ by 1), we need no extra bits to define
the state, and the amount of information carried by the state does not change.
The reason is the following: suppose that the distribution is confined within a region 
of the $q$ space whose size is $\Delta_q$. The amount of information we need to identify this state is
approximately $-\log ( \Delta_q / V_q)$. If we change $l_{\mathit{max}}^q$, this quantity does
not change, because neither $\Delta_q$ nor $V_q$ depend on  $l_{\mathit{max}}^q$
(see Fig.\ref{infoscalingb}).  

\begin{figure}[h]
\centering
\includegraphics[scale=0.7]{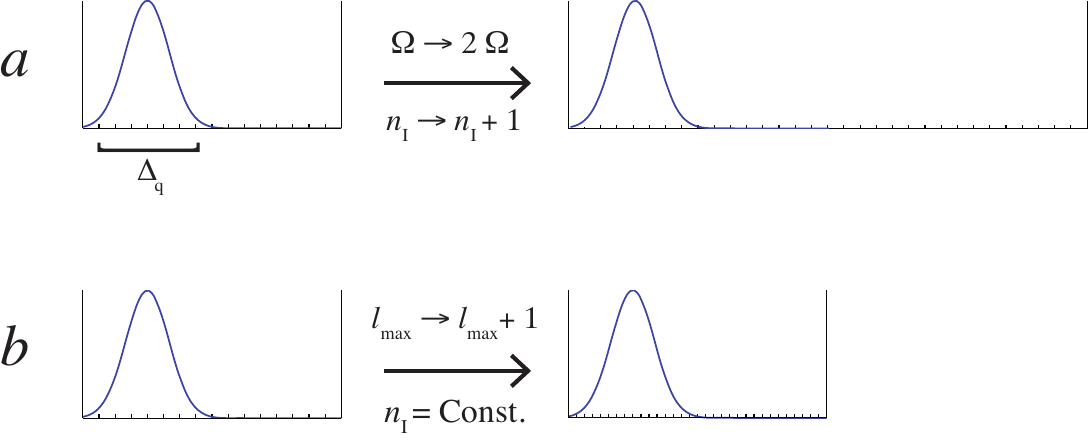}
\begin{subcaptiongroup}
\captionlistentry{a}
\label{infoscalinga}
\captionlistentry{b}
\label{infoscalingb}
\end{subcaptiongroup}

\caption{A qualitative argument to show how the information $n_I$ changes when the system's
boundaries are resized. \subref*{infoscalinga} The size of the space $\Omega_q$ doubles and the
information $n_I$ gains one bit.
\subref*{infoscalingb} The resolution doubles, but we keep the uncertainty of the distribution the
same with respect to the space size. In this case, the information $n_I$ does not change.}
\label{infoscaling}
\end{figure}
If we use the same arguments for the space $p$, we can synthesize the dependence of the information on
the system
ranges with the relation:
\begin{equation}  \label{eq:llpvsni}
n_I = {\mathit{const}} - l_{\mathit{min}}^p - l_{\mathit{min}}^q.
\end{equation}
\noindent Let us recall the type of observables introduced in section \ref{sec:freevardist},
namely, complete observables whose distributions are free to vary between the 
cases of a one-point distribution and a uniform distribution.
If the distribution of an observable, defined on a partition with a scale range
$l_{\mathit{min}}... l_{\mathit{max}}$, is free to vary, its contribution 
to the amount of information carried by the state will range between
0 (uniform probability case) and $l_{\mathit{max}} - l_{\mathit{min}} +1$
(one-point distribution). This implies that our construction satisfies axiom
\ref{ax:infostatequivalence}.

If an observable is complete, then the information provided by its measurement
must be equal to the total amount of information that can be carried by the system.
Since a measurement leads to a one-point distribution for the measured observable, we 
get the following expression:
\begin{equation} \label{eq:llptobounds}
    n_I = l_{\mathit{max}}^p - l_{\mathit{min}}^p + 1
      = l_{\mathit{max}}^q - l_{\mathit{min}}^q + 1.
\end{equation}
\noindent If we want this expression to be consistent with the dependence
of $n_I$ on the model's ranges $l_{\mathit{min}} ... l_{\mathit{max}}$, given by equation 
\eqref{eq:llpvsni}, we must require the sums  $l_{\mathit{max}}^p + l_{\mathit{min}}^q + 1$ and 
$l_{\mathit{max}}^q + l_{\mathit{min}}^p + 1$ to be equal to the ${\mathit{const}}$
term in equation \eqref{eq:llpvsni}.

It can be shown that the intermediate partitions $\mathcal{A}^{l, l'+1}$ and  $\mathcal{A}^{l+1, l'+1}$
we used to build the chain described in the previous subsection have the same cardinality as the
partitions used for the spaces of $p$ and $q$: 
if we count the number of bits involved in identifying the set \eqref{eq:decomppqllp1}
or \eqref{eq:decomppqllp2}, we obtain that the corresponding partitions $\mathcal{A}^{l, l'+1}$ and 
$\mathcal{A}^{l+1, l'}$ have the cardinality
$2^{l_{\mathit{max}}^q - l + l_{\mathit{max}}^p -l' + 1}$. 
Since we build the chain by increasing $l$ and decreasing $l'$ for each step, we can assume that 
this cardinality is constant along the chain and equal to the value taken at the first step
($l = l_{\mathit{min}}^q$ $l' = l_{\mathit{max}}^p$) of the chain, namely
$2^{l_{\mathit{max}}^q - l_{\mathit{min}}^p + 1}$, which, by Eq.~\eqref{eq:llptobounds},
is equal to $2^{n_I}$. We can therefore conclude  that, for any level of the chain,
\begin{equation} \label{eq:infobound1}
    l_{\mathit{max}}^q - l + l_{\mathit{max}}^p -l' + 1 = n_I.
\end{equation}

\subsection{Implications of the Model Independence axiom}\label{sec:modindimpl}

\mysubsubsection{Invariant partitions}
There is a connection between the idea of partitioning the space of the observables and the
measurement process. When building a measurement setup, we must divide the range
of values that one or more observables can take into intervals. This division is tied to the choice of the
measurement units and to the precision of the measuring instrument.

On the other hand, our approach is to derive the properties of the state space by imposing conditions
whose structure reflects that of a family of partitions of the space of certain observables.
Thus, the formulation of our theory may depend on how this family of partitions is defined.

Furthermore, our model of a physical system can be viewed as a subsystem embedded within a
larger system, which, in turn, can be associated with a family of partitions. 
We can imagine that our subsystem lies on a structure represented by the larger system's family of
partitions. Keeping this family of partitions fixed, and redefining the position or the scale of the
subsystem with respect to the larger system, has the same effect as applying such transformations
(in the opposite direction) to the family of partitions.
In this context, requirement \ref{ax:scalability} becomes relevant: we expect our model to be independent
of this kind of transformations. 
The necessity of this invariance becomes even more evident when we consider that our theoretical
model is composed of both the system under study and the measuring apparatus, which is inherently
linked to the notion of partition.
In subsection \ref{sec:sepsubs}, we will discuss how to isolate an elementary subsystem through an
ideal experiment involving a large set of measurements (and, consequently, a measurement apparatus).

We are thus led to conclude that the partitions used to divide the space of observables must be
invariant under translations and rescalings.
Theorem \ref{th:lsbgen} states that the only translation-invariant partitions, compatible with the
requirement that their cardinality is a power of two, are composed of sets of the form $A_{\mu} = \{ x \stackrel{l...n}{=} \mu \}$.

Thus, the choice of dividing the space of observables into partitions that are invariant under
specific transformations is not merely a convenient way to express the invariance of the state
space conditions in the compact form \eqref{eq:chiinvariance}. Rather, it becomes essential for
ensuring the internal consistency of the theory and its compliance with hypothesis \ref{ax:scalability}.

\mysubsubsection{Model Scalability}
Now we investigate how requirement \ref{ax:scalability}---specifically, the scale invariance of the
model---may impact the values of some  parameters
of the theory, such as $V_{p,q}$, $N_{p,q}$, $l_{\mathit{min}}^{p,q}$,
$l_{\mathit{max}}^{p,q}$, etc.

By keeping in mind that the model approximates the real physical system in the large
$N_p\,N_q\,V_p\,V_q$ limit, for large but finite values of these parameters, 
hypothesis \ref{ax:scalability} tells us that the predictions of the theory must be independent of how large
and fine-grained our model is. In what follows, we will determine how the values 
of the range parameters $l^{p,q}_{\mathit{min}}$ and $l^{p,q}_{\mathit{max}}$
must be chosen to make this happen.

As a first step, we are going to investigate the conditions that allow two models of different sizes
(or resolutions) to be equivalent, namely, when these models describe the same real physical system. 
Consider a model  $M_1$, whose ranges are $l_{\mathit{min}}^{p,q}$ and
$l_{\mathit{max}}^{p,q}$ and a smaller model  $M_2$,
that has the $q$-axis space reduced by a factor of $2$
with respect to that of $M_1$.
Since the range parameters $l_{\mathit{min}}$ $l_{\mathit{max}}$ indicate the size and the 
resolution of the model (on a logarithmic scale) with respect to the whole space size 
$V^{\mathit{phys}}$, the size of $M_1$ is
$V_q^{M_1} = V^{\mathit{phys}}_q 2^{-l_{\mathit{min}}^q+1}$, and the size of $M_2$ is 
$V_q^{M_2} = V^{\mathit{phys}}_q 2^{-l_{\mathit{min}}^q}$.
According to hypothesis \ref{ax:scalability}, we require  $M_1$ and $M_2$ to 
produce the same results, which implies that we have to adjust their parameters
so that, if a physical state is allowed in $M_1$, then it must be allowed in $M_2$ too.
If the size of the $q$ space of $M_1$ is twice that of $M_2$, then the same ``physical'' distribution
must appear shrunken by a factor of two in $M_1$, compared to how it appears in $M_2$.
Note that we are requiring $M_2$ to be halved with respect to $M_1$ on the $q$ axis only, which means
that $M_1$ and $M_2$ share the same $p$-axis lower range boundary $l_{\mathit{min}}^p$.

The condition that a smaller-size model, like $M_2$, leads to the same results
as a larger one can be expressed in another form: each state 
(family of distributions of the observables) that is valid in $M_1$ must have
a corresponding state, which is valid in $M_2$. This means that, up to a loss of 
resolution and to a contraction of the volume $V$, the probabilities 
$\rho_p,\, \rho_p$ that fulfill the state space condition in $M_1$ must be the same as in $M_2$.

We now apply the above condition to an approximately uniform distribution in $M_1$ that takes non-zero values only
in the lower half of the $V_q$ volume, namely that, roughly, has
$\rho(q \stackrel{l_{\mathit{min}}^q}{=}[0]) = 1$
and $\rho(q \stackrel{l_{\mathit{min}}^q}{=}[1]) = 0$.
Since we have defined the $q$ space of $M_2$ as the lower half of the $q$ space of $M_1$,
we can say that this type of distribution corresponds to a uniform distribution in $M_2$,
which is spread along the $q$ axis.

We can view this in terms of amount of information: in the model $M_1$, asserting
that the probability distribution confined on half $V_q^{M_1}$ volume,
corresponds to one bit of information; this is a consequence of the
information/statistics equivalence axiom \ref{ax:infostatequivalence},
because the predicates $q \stackrel{l_{\mathit{min}}^q}{=}[0]$ or 
$q \stackrel{l_{\mathit{min}}^q}{=}[1]$ represent one bit of information.
On the other hand, in $M_2$, a uniform distribution corresponds to zero bits,
as stated by axiom \ref{ax:infostatequivalence}.
As for the $p$ axis, we assumed that $M_1$ and $M_2$ have
the same size. Consequently, if $M_1$ and $M_2$ have the same distribution $\rho(p)$,
the amount of information related to $p$ is the same in both models. We can therefore conclude that
the maximum amount of information carried by $M_2$ is
$l_{\mathit{max}}^q - l_{\mathit{min}}^q$, one bit less than that carried by $M_1$
(see Eq.~\eqref{eq:llptobounds}). 

We are also assuming that $M_1$ and $M_2$ are built by following the same 
scheme, in which the total amount of information corresponds to the number of
points on each axis, as stated by \eqref{eq:llptobounds}. 
According to that equation, if $n_I$ and $l_{\mathit{min}}^p$ are respectively 
decremented and left unchanged in $M_2$ with respect to $M_1$, then $l_{\mathit{max}}^p$ 
is expected to be decremented. This corresponds to a loss in resolution in the 
$p$ axis, which implies that, up to a different discretization of the $p$
axis, the distributions $\rho(p)$ can be almost equal
in $M_1$ and $M_2$.
In summary, we can say that, in  $M_2$, the ranges of the $q$ and $p$ are
$l_{\mathit{min}}^q+1 ... l_{\mathit{max}}^q$ and 
$l_{\mathit{min}}^p ...  l_{\mathit{max}}^p-1$, respectively.

We can apply the above scheme again and obtain a model $M_3$ whose range boundaries are
$l_{\mathit{min}}^q+2, ... l_{\mathit{max}}^q$ and 
$l_{\mathit{min}}^p ... l_{\mathit{max}}^p-2$, and so on. Thus, we can identify
a class of equivalent (i.e., leading to the same physics) models, which have the property
that the sum $l_{\mathit{min}}^q + l_{\mathit{max}}^p$ is a constant. 
If we recall the $p \leftrightarrow q$ exchange symmetry,  we can get the same property 
(with the same constant) for $l_{\mathit{min}}^p$ and $l_{\mathit{max}}^q$.
Hence, we can define the constant $l_0$:
\begin{equation} \label{eq:lzerodefinition}  
l_{\mathit{min}}^q + l_{\mathit{max}}^p - 1 =
l_{\mathit{min}}^p + l_{\mathit{max}}^q - 1  \stackrel{\mathit{def}}{=}l^0
\end{equation}
\noindent Now that we know some of the conditions that parameters like $V_{p,q}$, $N_{p,q}$, 
$l_{\mathit{min}}$, $l_{\mathit{max}}$ must obey, we can try to 
understand how these parameters are related to real physical constants. 
Let us consider again the case of a model with ranges $l_{\mathit{max}}$, 
$l_{\mathit{min}}$, and a state where the $q$ variable is completely undetermined
and the $p$ variable is maximally determined. This implies that $q$ is undetermined
within a box of size $V^{\mathit{phys}}_q 2^{-l_{\mathit{min}}^q+1}$ and $p$
is confined within the smallest size the model can identify on the $p$ axis,
i.e. $V^{\mathit{phys}}_p 2^{-l_{\mathit{max}}^p}$.
The product of these uncertainties is
$\Delta q \Delta p =V^{\mathit{phys}}_p V^{\mathit{phys}}_q 
2^{-l_{\mathit{min}}^q - l_{\mathit{max}}^p + 1}$,
which, by the definition \eqref{eq:lzerodefinition} of the constant $l_0$, is equal to
${V^{\mathit{phys}}_q V^{\mathit{phys}}_p} 2^{-l_0}$.
This expression is independent of the size of the model (because we have proven that the
parameter $l^0$ obeys this condition) and, 
in order to satisfy the scale invariance of the theory, it must be independent of how we choose the
uncertainties of $p$ and $q$. In other words, it must be a physical constant.
If we want our theory to
lead to the same results as standard quantum theory, then, for consistency with the uncertainty principle, this quantity must be 
equal to one-half of the reduced Planck's constant: 
\begin{equation}
\frac{V^{\mathit{phys}}_q V^{\mathit{phys}}_p}{2^{l^0}} = \frac{\hbar}{2}.
\end{equation}
\noindent This connection with a real, physically 
meaningful quantity clearly reinforces our theory.
In the sections that follow, we will forget about all the constants like $V_{p,q}$,
$N_{p,q}$, $l_{\mathit{min}}^{p,q}$, $n_I$, $n_p$, $n_q$, $l_0$ 
$l_{\mathit{max}}^{p,q}$, etc., and we will simplify the model with the 
assumptions:
$l_{\mathit{max}}^p = l_{\mathit{max}}^q = n_p = n_q \stackrel{\mathit{def}}{=} n = \log_2 N$,
$l_{\mathit{min}}^p = l_{\mathit{min}}^q = 1$ and, consequently, 
$l_{\mathit{max}} - l_{\mathit{min}} + 1 = n$. According to \eqref{eq:llptobounds},
the maximum amount of information must obey the identity $n_I = n$. 
Furthermore, equation \eqref{eq:llpbound} can be simplified as:
\begin{equation} \label{eq:llpbounds}
   n - l  = l'  - 1.
\end{equation}

\subsection{The role of additional observables in the chain reconstruction scheme}\label{sec:phstatevscondprob}

In this subsection we investigate on how the degrees of freedom associated with the state of an $N$-dimensional system can be distributed
on the subparts of the modular structure described in the previous paragraphs, which is represented by
the butterfly diagram shown in Fig.~\ref{butterfly8}.

\mysubsubsection{Additional observables and symmetries}\label{sec:addobssym}
The introduction of additional observables is linked to the need to enlarge the dimensionality of the state space manifold. 
Moreover, the introduction of additional observables is inherently connected to
the introduction of parameter extensions, since a local (orthogonal) coordinate system with the same
dimensionality as the state space manifold is required. Consequently, in our derivations, we define
additional observables starting from the extensions of each component of the parameters of the base
physical observables (i.e., $p$ or $q$). More specifically, when we introduced the concept of
parameter extension \ref{def:orthogonalextension}, each base parameter was associated with a
corresponding extension. As a result, in our model, these extensions follow the modular structure of
the parameters of $p$ and $q$, as represented by the butterfly diagram.

This premise directly affects how a symmetry transformation acts on additional observables. 
A transformation can be interpreted as a change in the structure of the model. 
For instance, we can define a transformation that exchanges points along the $q$ axis; this would lead
to a reorganization of the branches of the butterfly diagram (or, equivalently, the partitions of the
set $\Omega_q$). If we maintain our framework, in which each parameter extension is associated with a
base parameter, then these extensions---and the additional observables derived from them---are
expected to transform together with the base observables, thereby preserving the model's structural
consistency.

In conclusion, in order to comply with the Model Independence axiom, we require that, if our model is
invariant under a given transformation defined on the base observables, then this transformation must
also involve the additional observables derived from parameter extensions. Although this condition and
the above considerations may seem abstract, they will become fundamental when we analyze how
permutations and translations, defined on $p$ and $q$, act on the parameter extensions.
In some cases, there is an argument that may help us establish how the symmetries impact the
additional observables, which refers to the idea that all the observables of a system are expected to
share the same properties. We have applied this argument to the two- and three-observable ($N=2$)
elementary system. For example, in the three-observable elementary system, the observables $p$, $q$,
and $r$ are fully interchangeable.

It is therefore reasonable to suppose that, if a symmetry holds for the two base observables, $p$ and $q$, it must also be 
applicable to the third observable $r$---under the hypothesis that $r$ has the same structure and properties as $p$ and $q$, as 
occurs in the $N=2$ case.

If these conditions hold, we have more instances of the same kind of symmetry. In the following
sections, it will be useful to hypothesize the invariance of our model under transformations that are
the result of a composition of two transformations acting on two different observables. For example,
we can define the simultaneous shift/rotation of $q$ and $r$, which, in the space of the
two-dimensional complex vectors $\psi_q$ defined in subsection~\ref{sec:blochsphere}, is represented
by the matrix
\begin{equation}\label{eg:shiftop2d} S^{(q)} = \begin{pmatrix} 0 & 1 \\ 1 & 0 \end{pmatrix}, \end{equation}
\noindent which produces an exchange between the probabilities $\rho_q(0)$ and $\rho_q(1)$
(i.e., a shift/rotation in $q$) and between the phases $\phi_1^{q}$ and $\phi_0^{q}$.
By definition \eqref{eq:phi2expr}, the permutation of the phases corresponds to a change in the sign
of the parameter $\alpha_q$, which in turn produces an exchange between the probabilities $\rho_r(0)$ and $\rho_r(1)$.

\mysubsubsection{Additional observables vs. control parameters}\label{sec:obsembed}
In the framework of a level-by-level reconstruction scheme, described in
section \ref{par:step2}, when we consider each $\chi$-node of the butterfly diagram,
the $l\,l'$-level component $\rho_{p\vert qp}$ depends, through the conditions $\chi$ , on the lower
level $\rho_{q\vert qp}$ and on all the other components 
of  $\rho_{q}$  of the diagram. This is  described in Step 2 of section \ref{par:step2} 
of the chain reconstruction scheme and, if we consider the form \eqref{eq:chielembin} 
of the $\chi$s, all dependencies other than $\rho_{q\vert qp}$ can be grouped into the
parameters $\boldsymbol{\phi}$.

For a given level $l$, the $\chi$ conditions can be seen as boolean-valued
functions, which confine the $l$ and $l+1$ level intermediate observables' probabilities 
on the state space manifold. The implicit form \eqref{eq:chielembin} of the $\chi$'s can
be put into an explicit form, as long as we express such intermediate probabilities in
terms of multi-valued functions like $\rho_{p\vert qp} = \chi^{l,k}(\rho_{q\vert qp}; \boldsymbol{\phi})$.
Actually, these functions must depend (also through the parameters $\boldsymbol{\phi}$)
on the probabilities of a set of independent observables, which can be defined only
locally. Consequently, this explicit form is valid only within
a neighborhood of a point of the state space.

Let us consider the contour lines of the function
$\chi^{l,k}(\rho_{q\vert qp}; \boldsymbol{\phi})$ in the space of the parameters 
$\boldsymbol{\phi}$, for fixed $\rho_{q\vert qp}$, and let us define the variable
$\alpha$, which identifies the contour line. Since the function $\chi^{l,k}$
is constant along a contour line, we can say that it only depends on the variable
$\alpha$. This implies that the dependence on $\boldsymbol{\phi}$ 
can be put in the form $\chi^{l,k}\left(\rho_{q\vert qp}; \alpha^k(\boldsymbol{\phi})\right)$,
where the superscript $k$ in $\alpha$ indicates that there may be a different $\alpha$
parameter associated with each node of the butterfly diagram.

In the next subsection, we will see that, for a node of the butterfly diagram, the condition $\chi$
is equivalent to the condition that determines the state space of an elementary three-observable
one-bit system.
Specifically, for fixed $\alpha$, the condition
$\rho_{p\vert qp} = \chi^{l,k}\left(\rho_{q\vert qp};\alpha^k\right)$ corresponds to the requirement that
parameters $S_p$ and $S_q$ belong to the subspace $S_r = {\textnormal const.}$.
This allows us to state that the parameter $\alpha$ can be considered as a  parametrization for the
probability of the observable $r$.

\mysubsubsection{Counting the independent $\chi$-controlling parameters}\label{sec:paramcount}
If we count the $\chi$ nodes in the butterfly diagram (see, for instance,
Fig.\ref{butterfly8}), we discover that they are $N/2 \log N$, which implies
that we have $N/2 \log N = n 2^{n-1}$ conditions to be controlled by parameters $\alpha$. 

By starting from the lowest level of the butterfly diagram, 
we are free to choose $N-1$ components of the distribution 
$\rho(q)$, and we need $N/2$ parameters  $\alpha$ to control
the first level of $\chi$ conditions (corresponding to $N/2$  nodes of the diagram).
Thus, we can choose the values of these parameters $\alpha$.

It must be pointed out that our aim is to find the {\it{largest}} state space that 
satisfies the principles introduced in section \ref{sec:principlesoftheory}.
Consequently, each parameter $\alpha$ can range within all the values compatible with those axioms.
If the values of other parameters that, for example, belong to the same level
$l$, have no effect on the compliance with such axioms, we can assert that each
parameter can vary independently.

When we go to the next level, again, we need $N/2$ parameters $\alpha$ to control 
the second level of  $\chi$ conditions. However, we will demonstrate later in
section \ref{sec:fourptsys} that
$N/4$ degrees of freedom of these parameters are bound by the choice
of the first-level parameters $\alpha$.
Going further along the chain, at each step, the number of parameters we can 
assign halves and, by summing the series $N/2 + N/4 + ...$, we get exactly $N-1$.
Knowing the number of parameters $\alpha$ and  how they are organized
along the butterfly diagram will be crucial in deriving an explicit form of the chain
reconstruction scheme.

\subsection{Separable subsystems }\label{sec:sepsubs}
To make the chain reconstruction scheme defined by the formula \eqref{eq:partboundsxp} and represented
by the diagram in Fig.\ref{butterfly8} actually usable, we should know how the conditions $\chi$ are
made. 

In what follows, we will show that a large system, represented by a butterfly diagram
of the form shown in Fig.\ref{butterfly8}, can be  decomposed into 
elementary one-bit subsystems like those studied in \ref{sec:threeobs};
this will lead us to the knowledge of the form of the $\chi$ conditions.
To justify this decomposition, we will identify some quantities that correspond to the probabilities
of the one-bit systems; such quantities are estimable through a repeated measurement scheme and obey
the same principles presented in section \ref{sec:principlesoftheory}
as the one-bit system.

\mysubsubsection{Identifying an estimable subsystem.}
If we follow the idea of expanding the range of observables, we can think of
building observables whose knowledge partially determines $p$ and $q$.
In other words, if we  consider a system that carries $n$ bits of
information, the state in which such a kind of observables is determined 
corresponds to the knowledge of some bits of the value of $p$ and some bits
of the value of $q$, making a total of  $n$ bits. 

If we take advantage of the simplifications in the system parameters introduced in subsection 
\ref{sec:modindimpl}, i.e. $l_{\mathit{max}}^p = l_{\mathit{max}}^q = n_p = n_q = n = \log_2 N$, 
$l_{\mathit{min}}^p = l_{\mathit{min}}^q = 1$, 
we can define the intermediate observable 
$o_l = [y_{n} y_{n-1} ... y_{n-l+2} \; x_{l}x_{l+1}...x_{n}]$, 
where $q = [ x_{1} ... x_{n}]$ and $p = [ y_{1} ... y_{n}]$ are the
binary representations of $q$ and $p$, respectively. We can build
a theory that involves this kind of observables and is consistent with the axioms
introduced in \ref{sec:principlesoftheory}.

Let us now consider an ideal experiment that involves this kind of observables.
We define the observables: 
\begin{subequations}\label{eq:partoper}
\begin{align}
o_a & = [y_{n} y_{n-1} ... y_{n-l+2} x_{l} x_{l+1}...x_{n}] \label{eq:partopera}\\
o_b & = [y_{n} y_{n-1} ... y_{n-l+1} x_{l+1} x_{l+2}...x_{n}] .\label{eq:partoperb}
\end{align}
\end{subequations}
\noindent Note that $o_a$ and $o_b$ share the common subset of $n-1$ bits
\begin{equation}\nonumber
y_{n} y_{n-1} ... y_{n-l+2} x_{l+1} x_{l+2}...x_{n}.
\end{equation}
\noindent In our thought experiment we have a large number of replicas of this $n$-bit system with
canonically conjugate variables $p$ and $q$, prepared in the same state.
Suppose that we perform two sets of measurements, one for $o_a$ and one for $o_b$.
Then we consider only two special subsets of these measurements, where
the common $n-1$ bits of $o_a$ and $o_b$ take a certain pair of values, say 
$o^0_a = y^0_{n} y^0_{n-1} ... y^0_{n-l+2}  x^0_{l} x^0_{l+1} x^0_{l+2}...x^0_{n}$ and
$o_b^0 = y^0_{n} y^0_{n-1} ... y^0_{n-l+2} y^0_{n-l+1} x^0_{l+1} x^0_{l+2}...x^0_{n}$ \allowbreak 
(from now on we will indicate a specific value of an observable by the superscript ``$^0$'').
It is apparent that these subsets of measurements can be used to estimate
the conditional probabilities
\begin{subequations}\label{eq:partoperrho}
\begin{equation}\begin{aligned}
   \rho_{q\vert qp}^{l} = \rho(q  \stackrel{l}{=} x^0_{l} \vert q \stackrel{...}{=} & [x^0_{l+1} ...x^0_{n}] \twocolbreakamp  \land
p \stackrel{...}{=} [y^0_{n-l+2} ... y_n]) \label{eq:partoperrhoa}  
\end{aligned}\end{equation}
\begin{equation}\begin{aligned}
\rho_{p\vert qp}^{l} = \rho(p \stackrel{l}{=} y^0_{n-l+1} \vert q \stackrel{...}{=} &
[x^0_{l+1} ...x^0_{n}] \twocolbreakamp \land
p \stackrel{...}{=} [y^0_{n-l+2} ... y_n] ) \label{eq:partoperrhob}
\end{aligned}\end{equation}
\end{subequations}
\noindent by counting the relative frequencies of the measurements whose results are $o^0_a$ and
$o^0_b$. 
The above conditional probabilities are exactly the ones that occur in the chain reconstruction scheme
stated by the formulas \eqref{eq:rhocondllp} and are represented by the incoming lines of each $\chi$
node in the butterfly diagram. 

We have therefore achieved our first goal, which is to show that the $\chi$-node of the butterfly
diagram corresponds to an elementary one-bit system, whose  (conditional)  probabilities are estimable 
through a repeated measurement scheme.

As discussed in section \ref{sec:modindimpl}, the only family of partitions that
preserves its structure under translations and satisfies the invariance of the model required by 
\ref{ax:scalability} is the one stated by theorem \ref{th:lsbgen}, whose sets have the form
$A_{\mu} = \{ x \stackrel{l...n}{=} \mu \}$.
It must be pointed out that 
hypothesis \ref{ax:scalability}  must also be considered in relation to the pair physical system-measurement apparatus: 
the set of measurements on the observables $o_a$ and $o_b$, which allowed us to isolate the subsystems
discussed above, is based on a choice of a set of  partitions that, according to axiom \ref{ax:scalability},
must be invariant under translations. 
The definitions \eqref{eq:partoper} of $o_a$ and $o_b$, which are based on the tree of
binary partitions, are the only ones compatible with what is stated by theorem \ref{th:lsbgen}.

This outcome represents a further refinement in our approach, which incrementally constrains a general
model through successive conditions.

\mysubsubsection{Equivalence to the three-observable elementary system.} \mbox{} \newline
To demonstrate the equivalence between the one-bit subsystem we have just described and the system
analyzed in section \ref{sec:threeobs}, 
we can assume that the conditional probabilities (or any equivalent statistical parameter) of the
subsystem are a subset of the coordinates of the larger system's state space that has observables
$o_a$ and $o_b$. Specifically, after identifying the system with observables $o_a$ and $o_b$, and
other additional observables, we can focus on the subspace of the state space that is 
defined by fixing all the conditional probabilities of the level $l$, $\rho_{q\vert qp}^{l}$
and $\rho_{p\vert qp}^{l}$, except for those with $o_a = o^0_a$ and $o_b = o^0_b$,  see
Eq.\eqref{eq:partoperrho}, and for the parameters of any
additional observables.

As discussed in section \ref{sec:estprecsubsp}, if the Precision Invariance axiom is valid for the
whole system formed by the observables, it must also be valid in the subspace whose coordinates are
represented by the conditional probabilities \eqref{eq:partoperrho}.

Thus, we have an elementary one-bit system, like those we discussed in sections \ref{sec:twoobsbasic}
and \ref{sec:threeobs}. In our case, we will primarily need two conditions to treat the problem:
the Precision Invariance axiom and the symmetry $\rho(0) \leftrightarrow \rho(1)$, which can be
considered a reasonable hypothesis.
We will adopt the parametrization  $\rho_{p,q}(^0_1) = \frac{1 \pm \mu\,cos(\alpha_{p,q})}{2}$,
which resembles the definitions \eqref{eq:spqalphadef} given in appendix \ref{apx:statman}. 
This choice does not assume the intersection of the state manifold with the cardinal points, but
allows us to follow the same steps we discussed in \ref{sec:twoobsbasic} and \ref{sec:twoobs} to apply
the Precision Invariance to the two-observable case.
Thus, we obtain a condition in the form \eqref{eq:fishermetrpres2d} or \eqref{eq:solutionfor2d},
namely, $ \alpha_p =  \pm \alpha_q + {\textnormal const.}$.
The only way to make this equation compatible with the $\rho(0) \leftrightarrow \rho(1)$ exchange
symmetry is to define the constant term as $\pi/2$ and to choose the $-$ sign in the $\pm$ term,
which implies the condition $\alpha_p = \pi/2 - \alpha_q$.
This allows us to represent the exchange symmetry as $\alpha_q \rightarrow - \alpha_q$ and
$\alpha_p \rightarrow - \alpha_p$ for the $p$ and $q$ axis respectively. 
If we compare this condition for $\alpha_q$ and $\alpha_p$ with
the condition \eqref{eq:thetapvsthetaq3d} and consider the 
definitions of the parametrizations \eqref{eq:spqalphadef},
we can state that the two-observable system we have just described is fully equivalent to the
three-observable system we presented in section \ref{sec:threeobs} and appendix \ref{apx:statman},
with the restriction that we are now considering the subspace 
$S_r =  {\textnormal const.}$  (note that we use the same symbol $r$ for both the 
$n$-bit the observable $r$ and the one-bit one; of course, in this context, $r$ is one bit).
With this assumption, the  $\mu$ parameter that occurs in \eqref{eq:spqalphadef} can be defined as a
function of the statistical parameters of $r$, as described by \eqref{eq:murvsthetar}.

This scenario is consistent with what we have remarked in the previous section,
when we put the $\chi$ condition in the form  $\rho_{q\vert qp} = \chi^{l,k}\left(\rho_{q\vert qp};
\alpha^k\right)$. In that case, the dependencies from the other nodes of the butterfly 
diagram and the additional observable's parameters can be embedded in a single
vector parameter $ \alpha^k$.

The above analysis has been applied to the case of a binary butterfly diagram, but
can be extended to cases in which the basic element of the diagram is a subsystem 
with $N$ larger than 2. This would lead us to larger separable subsystems.

\subsection{The extended Fisher metric}\label{sec:ndimxtdfishmetr}
In this subsection, we will derive an expression for the extended Fisher metric for an $N$-dimensional system. Our derivation will be based on the hypotheses 
that a large system can be divided into smaller subsystems and that the Fisher metric 
is independent of any particular ordering we choose for the domains of 
the system's observables.

\mysubsubsection{General form of the Fisher metric}
In what follows, we will introduce a parametrization for the probability distributions and the
orthogonal extension of such a parametrization. 
Let us consider the factorization of the probability distribution in the form
\eqref{eq:partitionedrhofull}, which is based on dividing the $q$ domain into a family of nested
binary partitions.
We can rewrite that expression for a $2^n$-dimensional system by using a shorthand notation:
\begin{equation}\label{eq:rhofactoriz}
\rho_{[j_1...j_n]} =  \prod_{l=1}^{n} \rho_{j_l\vert[j_{l+1}...j_n]},
\end{equation}
\noindent where the probabilities refer to the 
domain of the observable $q$: $\rho_{[j_1...j_n]} = \rho(q = [j_1...j_n])$ and
$\rho_{j_l\vert[j_{l+1}...j_n]} = \rho(q \stackrel{l}{=} j_l \bigm\vert q \lsbeq [j_{l+1}...j_n])$.
Following the notation used for the $N=2$ case, we also introduce the parametrization $\theta$
through the definition
\begin{equation} \label{eq:rhovsthetan}
\begin{array} {llc}
\rho_{0\vert[j_{l+1}...j_n]} & = & \cos^2 \frac {\theta_{j_{l+1}...j_n}}{2}  \\
\rho_{1\vert[j_{l+1}...j_n]} & = & \sin^2 \frac {\theta_{j_{l+1}...j_n}}{2}.  
\end{array}
\end{equation}
\noindent With this notation, the last term of the factorization \eqref{eq:rhofactoriz} is
$\rho_0$ or $\rho_1$, and $\theta$ is its parametrization.
For each parameter $\theta_{j_{l+1}...j_n}$, we also introduce the parameter extension
$\alpha_{j_{l+1}...j_n}$.
Let us now write the Fisher information matrix for our system:
\begin{equation}\label{eq:ndimfishmtrx}
\mathcal{I}_{jj'}=
\textnormal{E}  
\left[
\frac{\partial \log \rho}
{\partial \theta_{j_{l+1}...j_n }} 
\frac{\partial \log \rho}
{\partial \theta_{j'_{l+1}...j'_n }}\biggm\vert\boldsymbol{\theta} 
\right].
\end{equation}
\noindent We apply this definition to write the Fisher metric. By considering the factorization  \eqref{eq:rhofactoriz} for the 
probability $\rho$ and expanding into sums the logarithm of products, only the terms that actually depend on the 
differentiation variable are left. After some cumbersome but straightforward calculations, which also take into account the 
normalization $\sum_{j_l} \rho_{j_l\vert[j_{l+1}...j_n]} = 1$ of the conditional
probabilities, we can write the Fisher metric for the parametrization $\theta_{j_{l+1}...j_n}$ as:
\begin{equation}\begin{aligned}\label{eq:ndimfishmetr1}
d {s}^2_{F} =  \sum_{l=1}^{n} \sum_{j_l...j_n} \left(
\frac{\partial \log \rho_{j_l\vert[j_{l+1}...j_n] }}
{\partial \theta_{j_{l+1}...j_n }}
 \right)^2 \twocolbreak \cdot  
  d \theta_{j_{l+1}...j_n }^2  \prod_{l=1}^{n}\rho_{j_l\vert[j_{l+1}...j_n]}.
\end{aligned}\end{equation}
\noindent 
Now, suppose that we divide our $ 2^n$-dimensional system into two 
 $2^{n-1}$ subsystems, one for even values $q$ and the other for odd values.
 The Fisher metric of these subsystems can be written as:
 \begin{equation}\begin{aligned}\label{eq:ndimfishmetrsub}
d {s}^2_{F}\bigm\vert_{k} =  \sum_{l=1}^{n-1} \sum_{j_l...j_{n-1}}   \left(
\frac{\partial \log \rho_{j_l\vert[j_{l+1}...j_{n-1} k] }}
{\partial \theta_{j_{l+1}...j_{n-1} k }}
 \right)^2 \twocolbreak \cdot
  d \theta_{j_{l+1}...j_{n-1} k}^2  
\prod_{l=1}^{n}\rho_{j_l\vert[j_{l+1}...j_{n-1} k]}, 
\end{aligned}\end{equation}
\noindent where, by $k=0,1$, we indicate the even/odd subsystem. Note that, if we compare this
expression with \eqref{eq:ndimfishmetr1}, we can see that the index $j_n$, which indicates the least
significant bit, has been renamed as $k$. 
If we consider equation \eqref{eq:ndimfishmetr1}, all the terms in the summation, except for the last 
one, contain the factor $\rho_{j_n}$. If we isolate this term, rename the index $j_n$ as $k$, and 
collect the common factor $\rho_{k}$, Eq.~\eqref{eq:ndimfishmetr1} can be written in terms of the even
$q$ and odd $q$ contribution to the Fisher metric:
\begin{equation}\label{eq:ndimfishmetr2}
d {s}^2_{F} =  \sum_{k}d s^2_{F}\bigm\vert_{k} \rho_{k} +
d \theta^2.
\end{equation}
\noindent The last term of the above equation comes from the last ($l=n$) term in the summation 
in Eq.~\eqref{eq:ndimfishmetr1} and has been obtained by directly calculating the
logarithmic derivative and the derivatives of the probabilities $\rho_{k}$  given by
Eq.~\eqref{eq:rhovsthetan}.

The total Fisher metric is given by an extensionless part  $ds^2_F$, given by 
\eqref{eq:ndimfishmetr1}, and an $\alpha$-dependent part $ds^2_A$:
\begin{equation} 
{ds}^2 = ds^2_F + ds^2_A.
\end{equation}
\noindent We suppose that the system can be separated into an even $q$ and odd $q$ subsystem,
which can be handled independently of the larger system; each branch has a ``weight'' $\rho_k$.
If we observe  Fig.~\ref{butterfly8}, it is evident that the butterfly diagram
can be divided into even $q$ and odd $q$ branches. These branches are interconnected only at $l=3$.
We can assume that each subsystem obeys a Fisher metric preservation property independently of
the larger system. Therefore, for each subsystem, we can define an extended Fisher metric that
does not depend on the branch weight $\rho_k$.
In the case of the extensionless Fisher metric defined in Eq.~\eqref{eq:ndimfishmetr2}, we identified
the contribution of the even $q$ and odd $q$ subsystems and assigned the weight $\rho_k$ to each 
subsystem.
The only choice that makes the extended Fisher metric compatible with the expression
\eqref{eq:ndimfishmetr2} is to use the same coefficients $\rho_k$ to combine the 
even $q$ and odd $q$ contribution:
\begin{equation} \label{eq:FSm4gen}
{ds}^2 = \sum_k \rho_k {ds\bigm\vert_k}^2  + {d{\bar{s}}}^2,
\end{equation}
\noindent where the ${d{\bar{s}}}^2$ is the cross-subsystem contribution to the Fisher metric,
which in \eqref{eq:ndimfishmetr2} was given by the $d \theta^2$ term.
This contribution must also depend on the parameter extension $\alpha$, and we suppose that the matrix
element of this dependence has the form of a function of $\theta$:
\begin{equation} \label{eq:FSbar4}
  {d{\bar{s}}}^2 = d {\theta}^2 + h(\theta)^2 d {\alpha}^2.
\end{equation}
\mysubsubsection{Permutation invariance in a binary partition tree} 
Theorem \ref{th:fmprese} imposes conditions on the structure of the state space by establishing an
equality between the Fisher metrics of parametrizations corresponding to two different independent
observables.
Suppose that we compare the Fisher metric of two different parametrizations, with their extensions,
which refer to the probability distribution of the same observable. In that case, the theorem still
applies, and its proof becomes simpler: we merely perform a change of variables in both the Fisher metric
tensor and the differentials. This leads to a degenerate form of theorem \ref{th:fmprese}, as it does not
contain any actual condition on the structure of the state space, it is just a change of
variables.  Nonetheless, any parametrization of the same observable must preserve the Fisher metric's
invariance under reparametrization.

Consider a family of (extended) parametrizations of a given observable $q$, each defined using a distinct
binary partition tree. 
A first class of transformation under which we expect the Fisher metric to be invariant involves
permutations within each binary partition.
To be more specific, we consider the permutations that swap the sets 
$\{q \lsbeq \left[0 j_{l+1} \ldots j_n\right]\}$
and $\{q \lsbeq \left[1 j_{l+1} \ldots j_n\right]\}$, for every $\left[j_{l+1} \ldots j_n\right]$ at a given level $l$.
These transformations correspond, in group-theoretic terms,
to Sylow 2-subgroups of the symmetric group of order $2^n$\cite{Kurzweil}.  

Let us now explore how such transformations impact the parameter extensions; in the framework of our
binary tree partitions, if we start from the parameters $\theta_{j_{l+1}...j_n}$, their extension is
$\alpha_{j_{l+1}...j_n}$.
As noted in subsection \ref{sec:addobssym}, a transformation can involve both base and additional
observables, and we require the Fisher metric to be invariant under such coupled transformations.
Since the parameter extensions are a consequence of the increase in state space dimensionality resulting
from the introduction of additional observables, we expect a transformation of such  
observables be represented in terms of a transformation of the parameter extension---such as the
transformation \eqref{eg:shiftop2d} in the $N=2$ case.

We can simplify our derivations by introducing the ``phase'' parametrization $\phi$, assuming the mapping
from the parametrization $\alpha$ to $\phi$ to be differentiable. Thus, we can establish
a linear dependence on the differentials of these parametrizations:
\begin{equation} \label{eq:alpha1w}
d \alpha_{j_{l+1}...j_n} = \sum_{j_1... j_{l}} R_{j_1 ... j_n} d \phi_{j_1...j_n}.
\end{equation}
\noindent As in the $N=2$ case, there are more $\phi$ parameters than $\alpha$ parameters.
This redundancy introduces some freedom in the  choice of parameters $\phi$ and permits us to impose the condition
$\sum_{j_1 \ldots j_{l-1}} R_{j_1 \ldots j_{l-1}; 1 j_{l+1} \ldots j_n} = - \sum_{j_1 \ldots j_{l-1}} R_{j_1 \ldots j_{l-1}; 0 j_{l+1} \ldots j_n}$,
which ensures that $\alpha$ remains invariant under a global phase shift
$\phi_{j_1 \ldots j_n} \rightarrow \phi_{j_1 \ldots j_n} + \delta \phi$.
It also implies that changing the sign of $\alpha_{j_{l+1} \ldots j_n}$ corresponds to interchanging
$\phi_{j_1 \ldots j_{l-1} 0 j_{l+1} \ldots j_n}$ and $\phi_{j_1 \ldots j_{l-1} 1 j_{l+1} \ldots j_n}$.

At this point, the idea of defining permutations that simultaneously involve the
base observables and the additional ones comes into play: a permutation of the sets of the $(l+1)$-th
level binary partition effectively swaps the probabilities $\rho_{0\vert\left[j_{l+1} \ldots j_n\right]}$
and $\rho_{1\vert\left[j_{l+1} \ldots j_n\right]}$, which by \eqref{eq:rhovsthetan}, corresponds to the
transformation $\theta_{j_{l+1} \ldots j_n} \rightarrow \pi - \theta_{j_{l+1} \ldots j_n}$.
If we assume that this occurs along with a sign change in $\alpha_{j_{l+1} \ldots j_n}$
(which is supposed to represent a transformation on an extra observable), we obtain a permutation of the
phase parameters $\phi_{j_1 \ldots j_{l-1} 0 j_{l+1} \ldots j_n}$ and
$\phi_{j_1 \ldots j_{l-1} 1 j_{l+1} \ldots j_n}$.
We have therefore demonstrated that,  for
transformations that involve both base and additional observables, the $\rho$'s and the $\phi$'s permute in identical ways, 
provided that the permutations occur within the same partition of a given level
of the binary tree (or, in other words, if the permutations belong to the
Sylow 2-subgroup of the family of binary partition).

\mysubsubsection{Symmetric group invariance and the positional phase parameter}  
The properties we have just discussed refer to a special family of  binary partitions, whose sets are
$A_\mu = \{ x \stackrel{l \ldots n}{=} \mu \}$.   We can generalize our results by considering all
possible binary partition families, each inducing a distinct parameter extension
$\alpha_{j_{l+1} \ldots j_n}^{\mathcal{A}}$, with $\mathcal{A}$ labeling the partition family.
We assume that these parameter extensions map to a common parametrization $\phi_{j_1 \ldots j_n}$ via the
same form as \eqref{eq:alpha1w}, now with partition-dependent coefficients
$R_{j_1 \ldots j_n}^{\mathcal{A}}$.
Note that, if we consider a binary partition family different from that of Eq.~\eqref{eq:alpha1w},
the index $j$ no longer denotes an observable's domain position but rather a location within a binary
tree. We can therefore rename that index as $j^{(\mathcal{A})}$.
We define an integer functions $J(j^{(\mathcal{A})})$ that retrieves the actual position of the
$q$ from the index $j^{(\mathcal{A})}$.

These clarifications allow us to write in a consistent way the parameter dependence occurring in
equation \eqref{eq:alpha1w}, namely, $d \phi_{J(j^{(\mathcal{A})})}$, being $j^{(\mathcal{A})}$
the index in the summation on the RHS of Eq.~\eqref{eq:alpha1w}. 

Indexing the phase parameter $\phi$ in this way helps us understand how it transforms under
permutations within sets belonging to any binary partition family: it follows the positional
permutation of the base observable $q$, just like the probability $\rho_j$, since the function
$J(j^{(\mathcal{A})})$ recovers the correct value of $j$.

Since Eq.~\eqref{eq:alpha1w} holds for all binary partition families, we can repeat the derivation from the previous paragraph 
and conclude that, for permutations within any family of partition (i.e., permutations belonging to any Sylow 2-subgroup of the 
symmetric group), the $\phi$'s transform in the same way as the $\rho$'s.
Furthermore, since transpositions (i.e., permutations that swap two elements) generate the entire symmetric group, and every 
transposition is a permutation within a set partition of at least one
possible binary tree (i.e., belongs to some Sylow 2-subgroup), we can extend the covariance property of the $\phi$'s and the 
$\rho$'s to the entire symmetric group.

The above arguments allow us to require the invariance of the extended Fisher metric under the full symmetric
 group acting on the probabilities $\rho$ and the phases $\phi$.

\mysubsubsection{Derivation of the extended Fisher metric}
We can directly impose the permutation invariance to the Fisher metric as defined by Eqs.
\eqref{eq:FSm4gen} and \eqref{eq:FSbar4} expressed in terms our base family binary partition with sets 
$A_{\mu} = \{ x \stackrel{l...n}{=} \mu \}$, and drop the partition superscript $\mathcal{A}$.
An alternative form of the identity \eqref{eq:alpha1w} that can make our calculations easier is
\begin{equation} \label{eq:alpha1n}
\begin{split}
d \alpha_{j_{l+1}...j_n} =
& \hspace{-0.2cm} \sum_{j_1... j_{l-1}} \hspace{-0.2cm} r_{j_1... j_{l-1}; 1  j_{l+1}...j_n} d \phi_{j_1... j_{l-1} 1  j_{l+1}...j_n} \\ 
 - & \hspace{-0.2cm} \sum_{j_1... j_{l-1}} \hspace{-0.2cm} r_{j_1... j_{l-1};  0 j_{l+1}...j_n} d \phi_{j_1... j_{l-1} 0  j_{l+1}...j_n},
\end{split}
\end{equation}
\noindent where we have used the coefficients $r_{j_1 ... j_n}$ instead of $R_{j_1 ... j_n}$,
which obey the condition

\noindent $\sum_{j_1... j_{l-1}} r_{j_1... j_{l-1}; 1 j_{l+1}...j_n} =
\sum_{j_1... j_{l-1}} r_{j_1... j_{l-1}; 0 j_{l+1}...j_n}$. In the above definition  the indices
$j_{l+1}...j_n$ identify a subsystem that involves all the points of the $q$
axis, whose most significant bits are $j_{l+1}...j_n$. Therefore, $d \alpha_{j_{l+1}...j_n}$
refers to that subsystem.
Furthermore, each of these subsystems is divided into two smaller subsystems, with indices
$1 j_{l+1}...j_n$ and $0 j_{l+1}...j_n$ respectively. 

The invariance of the Fisher metric under the symmetric group discussed in this section
is the basis for the following theorem:

\begin{theorem} \label{th:hissin}
The extended Fisher metric in the form 
\eqref{eq:FSm4gen} is invariant under position permutations, only if:
(i) $r_{j_1...j_l;j_{l+1}j_n} = \rho_{[j_1...j_l]\vert[j_{l+1}j_n]}$ in the definition 
\eqref{eq:alpha1n}  and 
(ii) $h^2(\theta) = \sin^2(\theta)$ in the Fisher metric given by the equations
\eqref{eq:FSm4gen} and \eqref{eq:FSbar4}. This leads to the following expression of the 
extended  Fisher metric in terms of probabilities and phases:
\begin{widetext}
\begin{equation} \label{eq:FSm4std}
{ds}^2 = \hspace{-0.2cm} \sum_{j_1...j_n}\hspace{-0.1cm}(d \log \rho_{[j_1 \dots j_n]})^2 \rho_{[j_1 \dots j_n]} +
4 \sum_{j_1...j_n} \hspace{-0.1cm} d \phi_{j_1...j_n}^2 \rho_{[j_1...j_n]}  -
4 \left( \sum_{j_1...j_n} d \phi_{j_1...j_n} \rho_{[j_1...j_n]} \right)^2. 
\end{equation}
\end{widetext}
\end{theorem}

\noindent The proof of this theorem is presented in appendix \ref{apx:alphadepfishinftheorem}.

\mysubsubsection{The wave-function representation and the Fubini-Study metric}\label{sec:fsmetric}

If we use an integer index $j$, which ranges in the from $0 ... N-1$,
instead of a tuple of one-bit indices, Eq.~\eqref{eq:FSm4std} can be rewritten as: 
\begin{equation} \label{eq:FSmnstd}
{ds}^2 = \sum_{j} (d \log \rho_{j})^2 + 4 \sum_{j} d \phi_{j}^2 \rho_{j}
-  4 \left( \sum_{j} d \phi_{j} \rho_{j} \right)^2 \hspace{-0.2cm}, 
\end{equation}
\noindent where the symbols $\rho_{j}$ and $\phi_{j}$ now identify  the squared amplitude and the
phase of a $N$-points wavefunction.
By following the one-bit system case, we can pack the $\phi_{j_1... j_n}$
parameters and the $\rho_{[j_1... j_n]}$'s into a unique complex vector:
\begin{equation}
\psi_{[j_1... j_n]} = \sqrt{\rho_{[j_1... j_n]}} e^{i \phi_{j_1... j_n}}
\end{equation}
Since the parameters $\rho$ and $\phi$ are free to vary (with $\rho$ subject to normalization and $\phi\in[0,2\pi)$), the 
complex amplitude vector $\psi$ lies on the unit sphere of a complex Hilbert space and its overall phase factor is irrelevant. 
Thus, in this representation, states behave as rays $[\psi] \in \mathbb{CP}^{N-1}$. Although in our framework this is only one 
of several ways to represent the state, it nevertheless establishes, as in the two-dimensional case, a direct link to the 
standard formulation of quantum mechanics and is therefore a key element of our reconstruction program.

An important remark must be made about this expression: let us recall the Fubini-Study 
metric~\cite{FacchiAndVentriglia} in the form 
\begin{equation} \label{eq:fubstudybas}
d s_{FS}^2 = 
\frac{\langle \delta \psi \vert \delta \psi \rangle}
{\langle \psi \vert \psi \rangle} - 
\frac {\langle \delta \psi \vert \psi \rangle \; 
\langle \psi \vert \delta \psi \rangle}
{{\langle \psi \vert \psi \rangle}^2},
\end{equation}
\noindent  where $\psi$ is a generic complex wavefunction and  $\delta \psi$ is its differential.
If, in this expression, we consider $\psi$ as a discrete wavefunction of a pure state,
where $\psi_i = \sqrt{\rho_i} e^{i \phi_i}$, we can derive an expression of the form
\eqref{eq:FSmnstd} (for the derivation refer to~\cite{FacchiAndVentriglia}) multiplied by a factor of $1/4$. We can therefore 
conclude that
the extended Fisher metric corresponds, up to a factor of $1/4$, to the Fubini-Study metric or to the quantum Fisher 
information metric in the case of pure states (again, see~\cite{FacchiAndVentriglia}).

\section{Linearity}\label{sec:lingen}
\paragraph{Linearity from precision invariance}
In order to complete our reconstruction scheme of quantum mechanics, we need to consider
another element: roughly speaking, standard quantum mechanics, in its basic formulation, provides a systematic way to modify 
classical mechanics by replacing observable quantities with corresponding linear operators---most notably, substituting 
canonical coordinate pairs with pairs of corresponding operators.
Hence, a canonical transformation, in standard quantum mechanics, takes the form of a linear unitary
transformation: the transformed operator $O'$ of an observable $o$ 
can be put in the form $T^{\dagger} O' T$, where $T$ is the linear operator representing the canonical
transformation. 
A linear transformation can also be regarded from a different
perspective, namely, as a change of basis in a linear vector space. This is sometimes 
referred to as the ``passive'' (or ``alias'') transformation, in contrast with the ``active'' (or ``alibi'')
transformation, which changes the state itself. 

This is just one aspect of the linear structure of quantum mechanics. Linearity in quantum theory is intimately connected to 
its probabilistic interpretation.

Suppose we have a generic state $\vert \psi \rangle$ and an observable $o$. We can decompose $\psi$ by projecting it onto the 
eigenstates of $o$:
$\vert \psi \rangle = \sum_i c_i \vert i \rangle_o$.
In this representation, the squared modulus of the coefficient $c_i$ must be interpreted as the probability that a measurement 
of the observable $o$ yields the eigenvalue $o_i$, corresponding to the eigenstate $\vert i \rangle_o$.

In our formulation, this condition holds by construction. We define the wavefunction $\psi$ as a parametrization of the state 
such that, for example, $\vert\psi_q\vert^2$ represents the probability distribution of an observable $q$. Note that, if we 
define $\vert i \rangle_o$ as the states with definite values of $o$, an alternative way to write $\psi_q$  is
$\sum_i c_i \vert i \rangle_q$.

This is just a choice of representation for the state. To truly ensure linearity in our formulation, we must require that any 
state $\psi$, regardless of its original representation, can be re-expressed as a linear combination of the eigenstates of a 
different observable and that the squared moduli of the expansion coefficients in this new representation continue to represent 
measurement probabilities. That is, any change of representation must occur via a linear transformation.

In \ref{sec:fsmetric} we demonstrated that the extended Fisher metric for a system with two 
conjugate variables (Eq.~\eqref{eq:FSm4std}) corresponds to the Fubini-Study metric in the form
\eqref{eq:fubstudybas}.
It can be shown that, if the Fubini-Study metric is preserved, then the transformation between representations must be either a 
linear unitary or an antilinear antiunitary operator.

To demonstrate this, we can express the preservation condition of the Fubini-Study metric in an alternative form.
By integrating the differential $d s_{FS}$ defined by Eq.~\eqref{eq:fubstudybas}  along a
geodesic, we get the Fubini-Study distance between two states $\psi_1$ and $\psi_2$ (see, for 
example,~\cite{FacchiAndVentriglia}):
\begin{equation}\label{eq:fsfinite}
D_{FS}(\psi_1, \psi_2) = \arccos \sqrt{
\frac{\langle \psi_1 \vert \psi_2 \rangle}{\langle \psi_1 \vert \psi_1 \rangle}
\frac{\langle \psi_2 \vert \psi_1 \rangle}{\langle \psi_2 \vert \psi_2 \rangle}
}.
\end{equation}
\noindent The squared modulus of the scalar product between the states $\psi_1$ and $\psi_2$ appearing in this expression 
represents the transition probability between the states. 

We assume that the changes of representation considered in this work are well-defined and invertible over the whole state 
space, and that the metric-preserving property holds consistently throughout the corresponding coordinate charts.
Thus, it is evident that if, as stated by theorem~\ref{th:fmprese}, the Fubini-Study metric is invariant under change of 
parametrization, then the geodesic distance \eqref{eq:fsfinite} derived from this metric, and the transition probability it involves, are invariant as well.

The preservation of the transition probability places us within the hypotheses of Wigner's theorem~\cite{Wigner1,Bargmann_1964}.
Strictly speaking, in Wigner's theorem, the transformation under which the transition probability is invariant is a 
\textit{symmetry}, i.e., an active (``alibi'') transformation, while, in our case, we are dealing with a change of 
representation, i.e., a passive (``alias'') transformation. Nevertheless, the formulation of Wigner's theorem is sufficiently 
abstract that its mathematical conditions are satisfied by passive transformations as well.

For the sake of completeness, in appendix~\ref{apx:wignertheorem} a proof of Wigner's theorem, formulated in terms of 
invariance under changes of representation (``alias'' transformations) is presented. The proof is based on the direct 
application of the preservation of geodesic distance.

In the following, we restrict ourselves to the unitary branch of Wigner's theorem. The antiunitary branch is not excluded 
mathematically; in the change-of-representation context considered here, it corresponds to adopting the complex-conjugated 
convention for one of the representations and leads to an equivalent mirrored formulation. Since it does not introduce 
additional physical content relevant to the present reconstruction, we do not consider it further.

The linearity condition introduced above can be reformulated as follows: we define the linear operator associated with an 
observable by assuming it is diagonal in the basis of the observable's eigenstates:
$O = \sum_i \vert i \rangle_o\, o_i\, \langle i \vert_o$.
If linearity holds, the operator corresponding to the observable in a different representation takes the form
$O' = U^{\dagger} O U$, where $U$ is the unitary matrix of the transformation.
Therefore, we can assert that any observable corresponds, in any representation, to a linear Hermitian operator.

There is an important remark concerning the theorem just discussed.
Just as the standard form of Wigner's theorem could be phrased in terms of ``alias'' transformations (changes of
representation), the present framework admits an equivalent ``alibi'' formulation (actual state transformations).
Let $\omega_a$ be a parametrization and  $\omega_a^{\rho}$ the representation in $\omega_a$ of a state $\rho$.
Given a transformation $T:\rho\mapsto\rho'$ (e.g., time evolution), we can define a new parametrization $\omega_b$
such that $\omega_b^{\rho} = \omega_a^{T \rho} \, \forall \rho$.
Thus,  the transformation $T$  can be viewed as a change of representation and our Wigner-type result, and the concepts related 
to it, can be applied to the case of ``alibi'' transformations.

\paragraph{Linearity and reconstruction of quantum mechanics}

The demonstration of linearity we have just presented can be crucial in the development of our model. It can be shown that we 
only need a few more ingredients to complete our program of quantum reconstruction.

There is a well-known argument, presented for example by Dirac in his celebrated book on the principles of quantum 
mechanics~\cite{diracprinciples}, whereby the commutator between two quantum operators, having to satisfy a series of 
conditions (linearity, symmetry, Jacobi identity) can only be equal, up to a factor, to the Poisson brackets.
Once the canonical commutation relations have been established, our reconstruction of quantum mechanics can be regarded as 
almost structurally complete.

The above argument, however, implicitly relies on two assumptions: (i) linearity, and (ii) operator correspondence---namely, 
that for every classical observable, there exists a corresponding quantum operator, and that classical equations of motion can 
be transformed into their quantum counterparts by replacing classical quantities with the corresponding operators (as expressed 
by the Ehrenfest correspondence principle).

In this work, we do not follow this path to complete the reconstruction of quantum theory, as the operator correspondence 
hypothesis would introduce a strong additional assumption, and involving the dynamics (via Poisson brackets) would weaken the 
generality of the theory.

Furthermore, the finite-dimensional representation adopted until now cannot support the canonical commutation relations, since 
the trace of a commutator of finite-dimensional operators is zero. Moreover, there is no basis to assert that the continuum 
limit of our model is well-defined.
When, at the end of our derivations, we will prove that the discrete Fourier transform connects the 
$p$ wavefunction with the $q$ wavefunction, the continuum limit can be applied to our model, since we know that such a limit is 
applicable, under some reasonable conditions, to the discrete Fourier transform. 

In the sections that follow, we continue with our divide-and-conquer approach, consolidating the model in which the set of 
predicates on the system takes the form of a butterfly diagram.
This also offers insight into how the limitation of information affects each subsystem of the model.
In this framework, the model's intrinsic symmetries, along with the linearity of the parameter transformations just 
established, will play a central role.

\section{$N=4$ - three-observable system}\label{sec:fourptsys}
\mysubsubsection{Why is the $N=4$ case relevant?}\label{sec:fourptrelevance}

As remarked in the section \ref{sec:phstatevscondprob}, each node ($\chi$ condition) of the butterfly
diagram can be controlled by a single parameter, which we denote by $\alpha$.
This involves both a possible dependence of each node on the state of the other nodes and the dependence on
the extra observable $r$.  When we introduce an $N$-dimensional observable $r$, in addition to $p$ and
$q$, we also introduce $N-1$ extra degrees of freedom to the state space.
As stated in section \ref{sec:paramcount}, this extra observable implies that, if we have $N/2 \log N$
nodes in a butterfly diagram level, only $N-1$ parameters can be independent, as many as the extra degrees of freedom
introduced by $r$. 

To apply the chain reconstruction scheme introduced in the last section, we must understand how the
$\alpha$'s depend on each other, starting from the simplest system whose butterfly diagram 
has more than one node: the four-point system. The signature of the four-point system is 
$D =3$, $d=2$ (three variables $p$, $q$, $r$, two of which are independent) and $N_{\theta}=3$.
We have four $\alpha$ parameters, three of which are independent.

Suppose that we choose, as independent parameters, the $\alpha$'s for $l=1$. 
By following the same notation we used in \eqref{eq:chiindexing} for identifying the $\chi$
conditions, these parameters can be identified by $\alpha^{1\,2}_{[0][\,]}$, $\alpha^{1\,2}_{[1][\,]}$.
The higher level parameters, $\alpha^{2\,1}_{[\,][0]}$ and $\alpha^{2\,1}_{[\,][1]}$, will depend on
$\alpha^{1\,2}_{[0][\,]}$, $\alpha^{1\,2}_{[1][\,]}$ and on the third remaining degree of freedom.
In this section, we show how the form of this dependence can be obtained
from the principles of the theory introduced in \ref{sec:principlesoftheory}, especially from 
the Precision Invariance axiom and from the symmetries of the model.

\begin{figure}[!ht]
\centering
\includegraphics[scale=0.65]{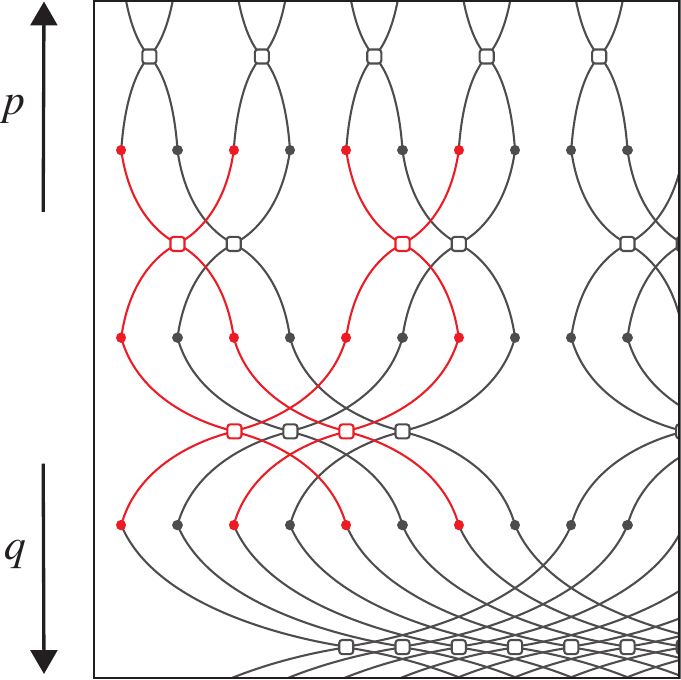}
\caption{A region of a butterfly diagram, with a four-point subdiagram highlighted in red. }
\label{subdiagcut}
\end{figure}

The butterfly diagram shown in Fig.~\ref{butterfly8} has a modular structure: we can identify smaller 
subparts that are linked together to form a larger diagram. In Fig.
\ref{subdiagcut} we can observe that, at any point of the diagram, we can identify a 
four-point subdiagram. 
This is consistent with the idea, mentioned in the previous section, that 
we can identify subsystems of the butterfly diagram with a size greater than 2 (in our case
$N=4$) and show that these subsystems are independent. We can therefore 
accept the hypothesis that a four-point subdiagram of a larger butterfly diagram obeys 
the same axioms as an isolated four-point one.

Furthermore, since a four-point sub-diagram (of a larger butterfly diagram) consists of four two-point
subsystems, once we understand the dependence of  the parameters among these subsystems,
we can apply the same dependence scheme to connect all the two-point subsystems in the larger
$N$-point diagram.
This is a crucial step in our program of reconstructing the connection between the distributions of
$p$ and $q$ by passing through a set of intermediate stages, which correspond to
the levels of the butterfly diagram.

\mysubsubsection{Notation and  parametrization of the probabilities}\label{sec:fourptnotation}

In this section we treat the four-point system by following the scheme introduced in section \ref{sec:threeobs},
namely, we introduce a parametrization for the probability distributions and an orthogonal 
extension. In the framework of our butterfly diagram-based approach, the first level of the
probabilities is the distribution on the $q$ axis. We can use shorthand notations for
the probabilities and the conditional probabilities introduced in the previous section
(see, for example, equation \eqref{eq:rhofactoriz}).
When the notation is unequivocal, the subscripts $p$ or $q$ are  omitted and, when the binary
representation is only one digit long, we can write, for example, $j$ instead of [j].
By following these conventions, in the $n=2$ case, the first level $\rho^q$ distribution can be written
as $\rho^1_{[j k]}$ and the first two equations in the chain-building relation \eqref{eq:partboundsxpa}
take the form
\begin{figure}[!ht]
\centering
\includegraphics[scale=0.85]{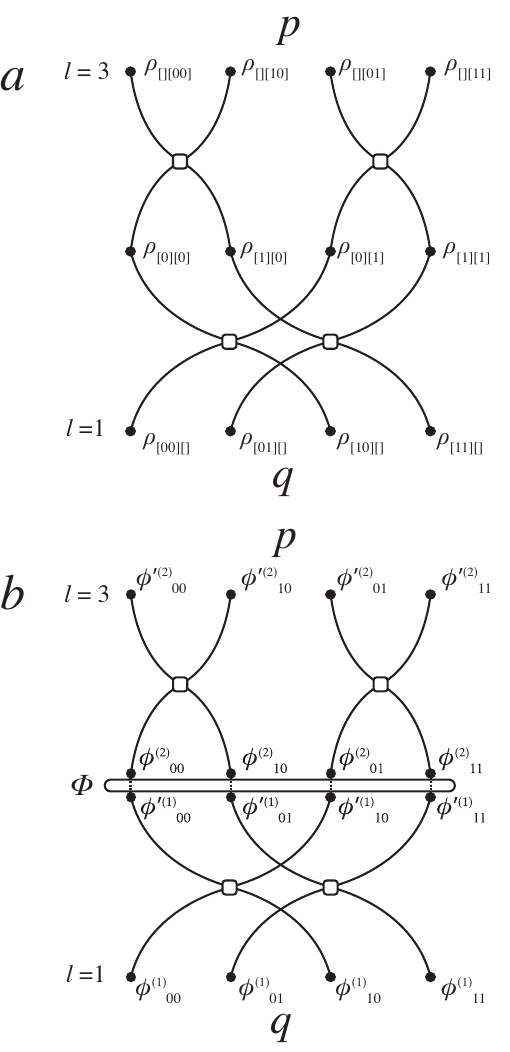}
\begin{subcaptiongroup}
\captionlistentry{a}
\label{butterfly4ab_a}
\captionlistentry{b}
\label{butterfly4ab_b}
\end{subcaptiongroup}
\caption{The butterfly diagram for the four-point system. \subref*{butterfly4ab_a}
The probabilities  $\rho_{[\,][\,]}$ of the set partitions with $l=1,2,3$ are the nodes of the
diagram.  
\subref*{butterfly4ab_b} The ``phase'' parameters ${\phi^{(l)}}_{j k}$  and ${\phi^{(l)}}'_{j k}$
are placed on the lines' endpoints of each elementary (two-point) graph; the function $\Phi_{j k}$
is represented as a transformation layer between the two levels of the diagram.
 }
\label{butterfly4ab}
\end{figure}
\begin{subequations}\label{eq:rho4full}
\begin{alignat}{1}
\rho^1_{[j k][\,]} & = {\rho^q}^1_{j \vert [k][\,]} \rho^1_{[k][\,]}
 \label{eq:rho4fulla} \\
\rho^2_{[k][j]} & = {\rho^p}^1_{j \vert [k][\,]} \rho^1_{[k][\,]}. 
\label{eq:rho4fullb}
 \end{alignat}
\end{subequations}
\noindent In Fig.~\ref{butterfly4ab_a} we can see how the above probabilities are
placed in the butterfly diagram. 
In the four-point system, we use the same parametrization \eqref{eq:rhovstheta} for the
probabilities $\rho$ that we introduced for the two-point case. Thus, we  define $\theta^1$ and 
$\theta^1_{k}$ as the parametrizations of the probability functions  $\rho^1_k\,(k=0,1)$ and
$\rho^1_{j \vert k}\, (j=0,1)$ respectively.
The first-level orthogonal extensions of the parameters $\theta^1$ and $\theta^1_k$ are indicated
as $\alpha^1$ and $\alpha^1_k$. 

\subsection{First-level transformation on the extended parameters}

Now we are going to deal with the first-level transformation of the $N=4$ system, namely,
the transformation that makes us obtain the first-level conditional probabilities
${\rho^p}^1_{j \vert [k][\,]}$ from the ${\rho^q}^1_{j \vert [k][\,]}$ (see Fig.~\ref{butterfly4ab}).
In terms of extended parameters,
we are going from ${\theta^1}, {\theta^1_k}, {\alpha^1}, {\alpha^1_k}$ to the transformed ones,
which we denote by ${\theta^1}', {\theta^1_k}', {\alpha^1}', {\alpha^1_k}'$. 
We can also use the complex vector form for these parameters.

In this subsection, we will state how the transformed first-level parameters are defined, and how the
Fisher metric can be expressed in terms of them. It should be emphasized that we have some freedom in 
defining such parameters; what has a real effect on the model is how the $\chi$ conditions (or the 
transformations to the higher level of the butterfly diagram) depend on them.

\mysubsubsection{The first-level transformation.}
As anticipated in the previous sections, we now consider conditions $\chi$ in explicit form, which
act as a transformation from a set of parameters to another. In general, we will identify these 
transformations with the symbol $F$, with some subscript or superscript that specifies the context of
the transformation.

We can consider the first level of the $N=4$ system as a pair of $N=2$ elementary subsystems:
we define the parameters ${\theta^1_k}', {\alpha^1_k}'$ as the transformed $\theta_p$ and $\alpha_p$ 
parameter of each subsystem.
Up to a change of parameters ($\boldsymbol{\psi}$ instead of $\theta$ and $\alpha$), the
transformation of each subsystem is governed by a rank 2 matrix (see Eq.~\eqref{eq:linpsi}),
which we identify as $F^0$:
\begin{equation} \label{eq:Finpsi}
F^0 = \frac{1}{\sqrt{2}} \begin{pmatrix}1&1\\1&-1\end{pmatrix}.
\end{equation}
\noindent In the first-level case, we defined ${\theta^1}$ as the parametrization 
for the probability $\rho^1_{[k][\,]}$ of each subsystem. 
Since the probability $\rho^1_{[k][\,]}$ is left unchanged by the first level of the butterfly
diagram, we define  ${\theta^1}'$ to be equal to $ \theta^1$.
 
At this stage we do not have a guideline
to state how the fourth parameter ${\alpha^1}'$ is related to the first-level parameter  ${\alpha^1}$.
The only requirement it must satisfy is to play the role of the extension of ${\theta^1}'$.
In the case of the first-level parameters (before the transformation), we defined the phase parameters
$\phi$ with the requirement of being related to $\alpha^1$ through the equation \eqref{eq:alpha1n}.
Now we are going to follow the inverse process: we firstly define the transformed phase parameters ${\phi^1}'$ 
from the corresponding complex vectors ${\psi^1}'$ of each two-point subsystem; then we define 
${\alpha^1}'$ through the four-point version of equation \eqref{eq:alpha1n}.
This scheme is just a set of definitions of parameters, it does not imply any 
physical hypothesis on the system. 

Once we have these definitions, we can derive how 
the transformed parameter ${\alpha^1}'$ depends on ${\alpha^1}$
and find an expression for the
Fisher metric of the transformed parameters.

\mysubsubsection{The complex vector form of the transformed parameters.}

As a first step, we can identify four transformed complex parameters  ${\psi_{j k}^1}'$,
which have the same properties as the corresponding ones in the $N=2$ case.
This parametrization will allow us to represent
the transformations of the chain reconstruction scheme in a friendly
linear-operators form.
By following the scheme adopted for the $\psi$'s before the transformation, 
we can start from the phases of the $\psi'$s, which correspond to the phases
of the even and odd $q$ subsystem, namely, $\phi_{j 0}' = {\phi_j}_e'$ and
$\phi_{j 1}' = {\phi_j}_o'$ (the level superscript ``1'' will be omitted
in this section). The phase parameters are related to the parameter extensions
$\alpha$ through a specialized version of formula \eqref{eq:alpha1n} for the $N=4$ case, and through the relations relative of
each subsystem, corresponding to Eq.~\eqref{eq:phi2expr}, which read $\alpha_k = \phi_{0 k} - \phi_{1 k}$ (for the first level)
and $\alpha'_k = \phi'_{0 k} - \phi'_{1 k}$ (for the transformed parameters).

By following the notation of the previous paragraphs, we identify ${\psi_j}_e$ and
${\psi_j}_o$ as the complex state vectors of the even and odd $q$ subsystems, which obey the
normalization conditions $\sum_{j}\vert\psi_{j e}'\vert^2 = \sum_{j}\vert\psi_{j o}'\vert^2 = 1$.
If we require that $\psi_{j 0}' = \sqrt{\rho_0} {\psi_j}_e'$ and
$\psi_{j 1}' = \sqrt{\rho_1} {\psi_j}_o'$, these conditions are compatible with the normalization
$\sum_{jk}\vert\psi_{jk}'\vert^2 = 1$. 
At this stage, the advantage of this type of parametrization appears clearer, because it can be easily
demonstrated that we have obtained a linear form of the first-level transformation: in the two-point
system, we have the equation $\psi_j' = \sum_i F^0_{j i} \psi_i$, which can be taken as the
transformation formula for the complex parameters ${\psi_j}_e$ and ${\psi_j}_o$ of the subsystems;
by multiplying both sides of these formulas by $\sqrt{\rho_k}$, we obtain
$\psi_{j 0}' = \sum_i F^0_{j i} \psi_{i 0}$ and $\psi_{j 1}' = \sum_i F^0_{j i} \psi_{i 1}$,
which are the linear relations we were looking for.

We must keep in mind that the  transformed parametrization $\psi_{jk}'$, or the equivalent 
pair of variables $\rho_{[j k]}'$  $\phi_{jk}'$, is expected to contain an arbitrary extra degree of
freedom, represented by a common additive phase term. 

\mysubsubsection{Derivation of the transformed cross term ${\alpha^1}'$} 
Let $\psi_k$ be the components of the two-dimensional 
complex state vector for an $N=2$ elementary system and $\bar{\psi}_k  e^{i \phi_k}$
their phase-absolute value representation. By applying the linear transformation $F^0$ defined in
Eq.~\eqref{eq:Finpsi}, we obtain the  transformed components $\psi_k'$:

\begin{equation} \label{eq:psipvspsi}
\psi_k' = \frac{1}{\sqrt{2}} \left(  \bar{\psi}_0  e^{i \phi_0} +
{(-1)}^k  \bar{\psi}_1  e^{i \phi_1}\right) .
\end{equation}

\noindent Since equation  \eqref{eq:alpha1n}, which relates the parameter $\alpha$ with the phase
parameters, is in terms of  differentials, we are interested in the differentials $d\psi$ and in how they
transform under $F^0$; this implies that we need the partial derivatives 
$\frac{\partial \phi_k'}{\partial \phi_j}$.
From the identity $\phi'_k = \mathrm{Im} \, \log \psi_k'$  we can write
\begin{equation} 
\frac{\partial \phi_k'}{\partial \phi_j} = \mathrm{Im} \frac{\partial \log \psi_k'}{\partial \phi_j} 
  = \mathrm{Im} \frac{1}{\psi_k'}\frac{\partial \psi_k'}{\partial \phi_j}.
\end{equation}
\noindent By applying equation \eqref{eq:psipvspsi} to calculate the term
$\frac{1}{\psi_k'}\frac{\partial \psi_k}{\partial \phi_j}$ of this equation, we obtain the following
expression for the differential of the transformed phase variable $\phi_k'$:
\begin{align}
d \phi_k' & = \sum_j \frac{\partial \phi_k'}{\partial \phi_j} d \phi_j \nonumber \\ & =
 \sum_j  \mathrm{Im} \left( i
 \frac{1}{\psi_0 + {(-1)}^k \psi_1} \dot {(-1)}^j \psi_j
 \right) d \phi_j,
 \end{align}
\noindent  which, if we multiply and divide the 
terms of the summation by $(\psi_0 + {(-1)}^k \psi_1)^*$, can be rewritten as
\begin{equation} \label{eq:dphiprime}
d \phi_k' = 
 \sum_j {(-1)}^j \mathrm{Im}
 \left( i \frac{(\psi_0 + {(-1)}^k \psi_1)^* \psi_j}{\vert\psi_0 + {(-1)}^k \psi_1\vert^2} \right) d \phi_j.
\end{equation}
\noindent The next step is to understand how the linear combination 
$\vert\psi_0\vert^2 d \phi_0 + \vert\psi_1\vert^2 d \phi_1$ is transformed.
By substituting the above expression and Eq.~\eqref{eq:psipvspsi}  
 in the expression $\vert\psi_0'\vert^2 d \phi_0' + \vert\psi_1'\vert^2 d \phi_1'$, we get:
\begin{equation} \label{eq:diphinvariance}
\begin{split} 
& \qquad  \vert\psi_0'\vert^2 d \phi_0' + \vert\psi_1'\vert^2 d \phi_1'  \\
& = \frac{1}{2} \mathrm{Im} \big(  
i (\psi_0 +\psi_1)^*\psi_0 d \phi_0 + i (\psi_0 +\psi_1)^*\psi_1 d \phi_1  
+ i (\psi_0 - \psi_1)^*\psi_0 d \phi_0  - i (\psi_0 -\psi_1)^*\psi_1 d \phi_1 
  \big)  \\
 & \qquad  = \vert\psi_0\vert^2 d \phi_0 + \vert\psi_1\vert^2 d \phi_1.
\end{split}
\end{equation}
\noindent This can be rewritten in the form
\begin{equation} \label{eq:diphinvariancerho}
 \rho_{0} d \phi_0 + \rho_{1} d \phi_1 = \rho_{0}' d \phi_0' + \rho_{1}' d \phi_1'.
\end{equation}

\noindent 
Let us consider the expression for the differential parameter extension obtained, by specializing
\eqref{eq:alpha1n} for an $N=4$ system and substituting $r_{ij}$ with $\rho_{i\vert j}$ (see theorem  \ref{th:hissin}):
\begin{equation} \label{eq:alpha1expl}
d \alpha^1 = (\rho_{0\vert1} d {\phi^1_{0 1}} + \rho_{1\vert1} d \phi^1_{1 1} ) 
- (\rho_{0\vert0} d {\phi^1_{0 0}}  + \rho_{1\vert0} d {\phi^1_{1 0}} ) . 
\end{equation}
\noindent As stated before, the transformed parameter ${\alpha^1}'$ must be defined in terms of the
transformed phase parameters $\phi'$:
\begin{equation}\begin{aligned} \label{eq:alpha1expltransf}
d {\alpha^1}' = & (\rho_{0\vert1}' d {\phi^1_{0 1}}'  + 
\rho_{1\vert1}' d {\phi^1_{1 1}}' ) \twocolbreakamp 
- (\rho_{0\vert0}' d {\phi^1_{0 0}}'  +
\rho_{1\vert0}' d {\phi^1_{1 0}}' ) . 
\end{aligned}\end{equation}
\noindent 
The conditional probabilities $\rho_{j\vert k}$, occurring in this equation and in Eq.\eqref{eq:alpha1expl}, are nothing but the probabilities of the two-point systems:
$\rho_{j\vert0} = {\rho_{j}}_e$, $\rho_{j\vert1} = {\rho_{j}}_o$ 
and $\rho_{j\vert0}' = {\rho_{j}}_e'$, $\rho_{j\vert1}' = {\rho_{j}}_o'$. Furthermore,  
the phases of the $N=4$ system are, by definition, the phases of the two-point subsystems, namely
$\phi_{j 0} = {\phi_j}_e$, $\phi_{j 1} = {\phi_j}_o$ and $\phi_{j 0}' = {\phi_j}_e'$,
$\phi_{j 1}' = {\phi_j}_o'$. These conditions allow us to rewrite the RHS of Eqs.~\eqref{eq:alpha1expl}
and \eqref{eq:alpha1expltransf} in terms of the parameters of the  two-point systems.
After these substitutions, we can recognize, in the sums enclosed in parentheses of the equations
\eqref{eq:alpha1expl} and \eqref{eq:alpha1expltransf}, the LHS and the RHS of the identity 
\eqref{eq:diphinvariancerho} for a two-point systems.
If apply this identity, equations \eqref{eq:alpha1expl} and \eqref{eq:alpha1expltransf}
lead us to the equation:
\begin{equation} \label{eq:alphaeqalphaprime}
d {\alpha^1} = d {\alpha^1}' . 
\end{equation}

\mysubsubsection{The Fisher metric in terms of transformed parameters.} 
We can specialize equations 
\eqref{eq:FSm4gen} and \eqref{eq:FSbar4} for the $N=4$ case
and obtain the following expression for the extended Fisher metric:
\begin{equation} \label{eq:FSm4}
{ds}^2 = \sum_k \rho_k (d {\theta_k}^2 + \sin^2\theta_k d {\alpha_k}^2) +
d {\theta}^2 + \sin^2\theta d {\alpha}^2.
\end{equation}
\noindent We can apply the arguments that led us to the above expression to all the levels of
the butterfly diagram, and obtain equivalent expressions for the Fisher metric.
For this reason, we have intentionally omitted the level index ``1'' in it. 

In order to rewrite the Fisher metric for the transformed parameters, we must
consider that the first term in \eqref{eq:FSm4}, which contains the summation $\sum_k$,
maintains the same form for the transformed parameters because the 
Fisher metric of each two-point subsystem takes the form 
$d {\theta^1_k}'^2 + \sin^2{\theta^1_k}' d {{\alpha^1_k}'}^2$; this allows us to apply  the concepts introduced in subsection
\ref{sec:ndimxtdfishmetr} and write the transformed Fisher metric in the same form as equation \eqref{eq:FSm4gen}.
The second term of equation \eqref{eq:FSm4} maintains the same form as well, 
because of the invariance $\theta^1 = {\theta_1}'$.
As we have shown above, if we define the transformed parameter ${\alpha^1}'$ in terms of the phases
${\phi_{jk}^1}'$ (Eq.~\eqref{eq:alpha1expltransf}), we obtain the identity
$d {\alpha^1} = d {\alpha^1}'$. 
From this condition, we can conclude that the only way to fulfill the metric-preserving condition
through the first stage of the butterfly diagram is to choose 
$ \sin^2{\theta^1}' d {{\alpha^1}'}^2$
as the third term of the Fisher metric in the form \eqref{eq:FSm4}.

\subsection{Linearity and metric-preserving conditions }\label{sec:linearitynpertpres}

Now we move to the second level of the butterfly diagram \ref{butterfly4ab}.
According to the parametrization we choose, the input of this second stage of the
transformation can be either in the form $\boldsymbol{\rho}, \boldsymbol{\phi}$
or $\boldsymbol{\theta}, \boldsymbol{\alpha}$.
The unextended part of the second stage input parameters, which can take the form of the probabilities
$\boldsymbol{\rho}$ or of the  parameters $\boldsymbol{\theta}$, is expected to depend directly on
the transformed $\boldsymbol{\rho}$ or $\boldsymbol{\theta}$ components of the first stage, because of the
identities \eqref{eq:rho4full}.
On the other hand, the extension part ($\boldsymbol{\alpha}$ or $\boldsymbol{\phi}$) of the 
parametrization has been introduced with the only condition of being orthogonal with respect
to the Fisher metric, which implies that there is no direct expression for them.
To go further in the chain reconstruction scheme represented by the butterfly diagram  \ref{butterfly4ab}, we must understand  
which value these components must take or, in other words,
how they depend on the first-level ``output'' parameters.
In Fig.~\ref{butterfly4ab_b} this dependence is represented in the context of a butterfly diagram, as
the processing stage $\Phi$. 
Note that, in that figure, we have used for the level 2 the same phase/amplitude parametrization scheme
introduced for the level 1, in which the parameters $\phi_{00}, \phi_{01}, \phi_{10}, \phi_{11}$
are related to the parameter extensions $\alpha, \alpha_k$
through the equations
\begin{equation} \label{eq:phi2iiangles}
\begin{split} 
\alpha^{(2)}_k & = \phi^{(2)}_{k 0} - \phi^{(2)}_{k 1}\\
d \alpha^{(2)} & = (\rho_{0\vert1} d \phi^{(2)}_{0 1} +
\rho_{1\vert1} d \phi^{(2)}_{1 1}) \twocolbreakamp - (\rho_{0\vert0} d \phi^{(2)}_{0 0} + \rho_{1\vert0} d \phi^{(2)}_{1 0} ), 
\end{split}
\end{equation}
\noindent which correspond to equations \eqref{eq:phi2expr} and  \eqref{eq:alpha1expltransf}.
In order to avoid confusion with the power elevation, only in this section, we use
the notation $ ^{(l)}$ to indicate the level $l$ to which the parameters refer.

It should be noted that the actual degrees of freedom of the parametrization are six, 
because we have six parameters $\theta, \theta_k, \alpha, \alpha_k$.
If we use the intermediate probabilities $\rho^{(2)}_{[j] [k]}$ (the LHS of equation
\eqref{eq:rho4fullb})  as parametrization and the four phases ${\phi^{(1)}}'_{j k}$, we still have six
effective degrees of freedom, because the normalization 
fixes one degree of freedom of the $\rho$'s, and a global, non-relevant, phase factor 
is involved by the $\phi$'s. Note that, according to this notation, the
transformed probabilities ${\rho^{(1)}}'_{[j] [k]}$ and the intermediate
probabilities $\rho^{(2)}_{[j] [k]}$ are the same mathematical object.
We have shown in the previous subsection that the full extended Fisher metric for the
transformed first-level parameters can be put in the same form as
equation \eqref{eq:FSm4std}, up to a substitution of $\rho_{[j k]}$ and $\phi_{j k}$
with $\rho^{(2)}_{[j] [k]}$ and ${\phi^{(1)}}'_{j k}$.

On the other hand, we can make a similar choice for the second-level 
parameters, namely, we can use again the parameters $\rho^{(2)}_{[j] [k]}$ and
$\phi^{(2)}_{j k}$ introduced in the last equation. Again,
the  Fisher metric for this parametrization has the form \eqref{eq:FSm4std}, with
$\rho^{(2)}_{[j] [k]}$ and ${\phi^{(2)}}_{j k}$ instead of $\rho_{[j k]}$ and 
$\phi_{j k}$. Note that in the expression of the Fisher metric we have obtained, the first terms,
which contain $\rho^{(2)}_{[j] [k]}$, are the same as the transformed first-level Fisher metric
discussed above.

In order to satisfy the Fisher metric-preserving property of theorem \ref{th:fmprese} for the complete transformation from
the $q$ representation to the $p$ representation, we require the Fisher metric after the
first level and before the second level to be equal. 
By following the formalism introduced in that theorem, we can write the identity 
\begin{equation} \label{eq:N4l2fishpreserv}
d {\omega^{(1)}}'_{i j} {g^{(1)}}^{ijkl} d {\omega^{(1)}}'_{kl} =
d {\omega^{(2)}}_{i j} {g^{(2)}}^{ijkl} d {\omega^{(2)}}_{kl}, 
\end{equation}
\noindent where ${\omega^{(1)}}'_{i j}$ and  ${\omega^{(2)}}_{i j}$ are the vectors formed 
by the sets of parameters $\rho^{(2)}_{[i][j]}, {\phi^{(1)}_{i j}}'$ and 
$\rho^{(2)}_{[i][j]}, {\phi^{(2)}_{i j}}$.
We can apply the arguments presented in subsection \ref{sec:ndimxtdfishmetr},
especially the invariance under permutations of the parameters,
to show that the extended Fisher metric in both the LHS and the RHS of the above equation takes the form \eqref{eq:FSm4std}.
If we re-express the parametrization $\rho \, \phi$ in terms of complex vectors $ {\psi^{(1)}_{i}}'$ and $ {\psi^{(2)}_{i}}$, 
we can also use the form \eqref{eq:fubstudybas} for such a metric, which we have recognized as the Fubini-Study metric.

Once we have this form of the metric preserving conditions, we can apply the customized version of Wigner's theorem presented 
in section~\ref{sec:lingen} and assert that the parameters ${\psi^{(1)}_{i}}'$ and $ {\psi^{(2)}_{i}}$ are related by a linear 
or antilinear dependence. As stated above, the squared modulus $\rho$ is left unchanged by this stage of transformation, i.e. 
$\vert{\psi^{(1)}_{i}}'\vert^2 = \vert {\psi^{(2)}_{i}}\vert^2$. If we neglect the antiunitary case of the transformation, this 
implies that ${\psi^{(1)}_{i}}'$ and $ {\psi^{(2)}_{i}}$ differ only by a constant phase factor, i.e.
${\psi^{(2)}_{i}} = t_i {\psi^{(1)}_{i}}'$, where  $t_i = e^{i \phi^0_i}$. 
This scheme applies to all the other levels of the butterfly diagram as well.
In section \ref{sec:traslinv} we will derive the form of the constant phase factor $t_i$.

\section{Translational invariance and twiddle factor}\label{sec:traslinv}

In this section, we carry out the full transformation between the spaces of the probability  
functions $\rho_p$ and $\rho_q$. This result will be achieved by finding an expression 
for the phase factor $t$, defined in section \ref{sec:linearitynpertpres}, which maps the 
extension parameters of each level of the butterfly diagram into the next level ones.
The factor $t$, for historical reasons that will be justified at the end of this section,
is referred to as ``twiddle factor''~\cite{GentlemanSande}. The expression of $t$ that we are going
to find is a direct consequence of the invariance under translations/rotations.

Roughly speaking, our theory is invariant under a given transformation
if, when we apply that transformation, we obtain a physically valid state from another.
The concept of invariant partition, introduced in section~\ref{sec:statespacesymm},
will allow us to express the shift/rotation transformation in a way that acts separately 
on simpler system subparts. We have identified the translation-invariant partitions through
theorem \ref{th:lsbgen}, and they are defined by the formula \eqref{eq:decomppq}.
The butterfly structure depicted in figures \ref{butterfly8} and \ref{butterfly4ab} is a consequence
of the decomposition of the state space into a binary tree of partitions.
The translation invariance of such a partition makes the structure of the butterfly diagram invariant
under translations.
This will help us to express the translation transformations in a simple form and introduce the
translation invariance of the theory in a straightforward way.

In this section, we express the physical state in terms of the complex vectors 
${\psi}_p$ ${\psi}_q$.  The
transformation from the state ${\psi}_q$ to the state ${\psi}_p$ can
be decomposed into a sequence of single-level transformations and phase factors $t$:
\begin{equation} \label{eq:fftexpln}
\psi_p = \mathcal{F}_{n}  \psi_q = 
F_{n}  \cdot  t_{n-1} \cdot ... \cdot F_2 \cdot t_1 \cdot F_{1} \cdot \psi_q.
\end{equation}
\noindent According to what was stated in the previous sections, each step of the above
transformation is linear. The above expression is the translation of the 
butterfly diagram into formulas. As shown in Fig.~\ref{butterfly8}, each term $F_l$ of equation  
\eqref{eq:fftexpln}  is a set of $N/2$ linear transformations of the form $\eqref{eq:Finpsi}$, which 
operates on each $N=2$ elementary subsystem of the butterfly diagram. The phase factors
$t$ can be expressed in the form of an $N\times N$ diagonal matrix:
\begin{equation} \label{eq:twiddleinphase}
t_l = \Diag{(e^{i \phi_l^0},  ... , e^{i \phi_l^{N-1}})},
\end{equation}
\noindent where the notation $\Diag({\mathbf{v}})$ indicates the diagonal matrix obtained from the
vector $\mathbf{v}$. 
In this section, we will consider translations on the $p$ axis; the same results
could be achieved by considering translations on the $q$ axis instead.

Let us define $r_{e \leftrightarrow o}$ as the operator that 
exchanges the even index components of the states $\psi_p(i)$ with the odd index ones.
As already mentioned in section \ref{sec:translsymm}, by translation we actually
mean a circular shift of the form  \ref{th:circshift}. Let  $S$ be the operator
that carries out such a shift.
We can also define shift operators that act only on a subset of $\psi_p(i)$
components. For example, $S^o$ is the left circular shift
operator that involves only the odd-$i$ elements of $\psi_p(i)$.

It is easy to show that the circular shift can be represented in terms of 
a chain of even-odd inversions like $r_{e \leftrightarrow o}$, operating 
on some subdiagrams of the butterfly diagram. For example, the left circular shift
can be represented as an even-odd component exchange like $r_{e \leftrightarrow o}$
and a left circular shift acting on the odd components:
\begin{equation} \label{eg:sigmavssigmaor}
S = S^o {r_1}_{e \leftrightarrow o}.
\end{equation}

\noindent This identity can be applied recursively, by expressing the sub-shift operator $S^o$ in
terms of even-odd component exchange and sub-sub-shift operator, and so on. After some iterations, we
will reach the single-element shift operator, which is nothing but the identity operator. 

In the case of the basic $N=2$ elementary system, the transformation $\psi_q \rightarrow \psi_p$ is
given by the matrix \eqref{eq:Finpsi}. The translation operator is simply the permutation operator
that exchanges $\psi_p(0)$ with $\psi_p(1)$ or, in other words, the operator that swaps even $p$
and odd $p$. It is easy to show that, in the $\psi_q$ space, this operator takes the form:
\begin{equation} \label{eq:invoptwod}
S_0 =\begin{pmatrix}
1 & 0 \\ 
0 & -1
\end{pmatrix} = \Diag(1, e^{i \pi}).
\end{equation}
\subsection{The transformation $\mathcal{F}_l$}\label{sec:fltransf}
Let us rewrite Eq.~\eqref{eq:fftexpln}, which describes the transformation from the space of $\psi_q$
to the space of $\psi_p$, for a subdiagram with $l$ levels:
\begin{equation} \label{eq:fftexpll}
\psi_p = \mathcal{F}_l \psi_q =
F_{l}^{(l)}  \cdot t_{l-1} \cdot ... \cdot F_{2}^{(l)} \cdot t_{1} \cdot F_{1}^{(l)} \cdot \psi_q.
\end{equation}
\noindent The superscript $(l)$ (which will be omitted in the derivations that follow) indicates 
that we are dealing with a $L$-dimensional subdiagram, where $L = 2^l$. 
Let us now introduce the notation
\begin{equation} \label{eq:fftpartiall}
F_{m \rightarrow m'}= F_{m'}  \cdot t_{m'-1} \cdot ... \cdot F_{m + 2} \cdot t_{m+1} \cdot F_{m+1},
\end{equation}
\noindent which allows us to rewrite equation \eqref{eq:fftexpll} as
\begin{equation} \label{eq:fftgrpl}
\psi_p = F_{1 \rightarrow l}  \cdot t_{1} \cdot F_{1} \cdot \psi_q.
\end{equation}
\noindent
In Fig.~\ref{butterflyblocks} we can identify the partial transformation $F_{1 \rightarrow l}$.
If we observe its position with respect to the larger butterfly diagram, it is apparent that
$F_{1 \rightarrow l}$ acts on the even $p$ (left subdiagram) and odd $p$ (right subdiagram) components
separately.
If we identify the blocks of the matrix $F_{1 \rightarrow l}$ that operate on the even components
and the blocks that operate on the odd components, we can write $F_{1 \rightarrow l}$
by blocks:
\begin{equation} \label{eq:fftgrpblockl}
F_{1 \rightarrow l}^{(l)} = 
\begin{pmatrix}
F_{1 \rightarrow l}^e & 0 \\
0 & F_{1 \rightarrow l}^o \\
\end{pmatrix}
\end{equation}
\noindent From Fig.~\ref{butterflyblocks}, we can also observe that
$F_{1 \rightarrow l}^e = F_{1 \rightarrow l}^o = \mathcal{F}_{l-1}$ is  the transformation of an
$l-1$-levels subdiagram, defined in a $2^{l-1}$-dimensional space.
The first (rightmost) stage of the transformation \eqref{eq:fftpartiall}
corresponds to the lowest level of the butterfly diagram and can be regarded as a combination of
$L/2$ transformations, each acting on a pair $\psi_q(i)\, \psi_q(i + 2^{l-1})$ separately.
Since the elementary transformation on one such pair is $1/\sqrt{2}\begin{pmatrix}
1 & 1 \\ 1 & -1
\end{pmatrix}$, we can write this first stage in the following 
block matrix form: 
\begin{equation} \label{eq:fftpartblockl}
F_{1} = \frac{1}{\sqrt{2}}
\begin{pmatrix}
I & I \\
I & -I \\
\end{pmatrix}.
\end{equation}
\begin{figure}[!ht]
\centering
\includegraphics[scale=0.4]{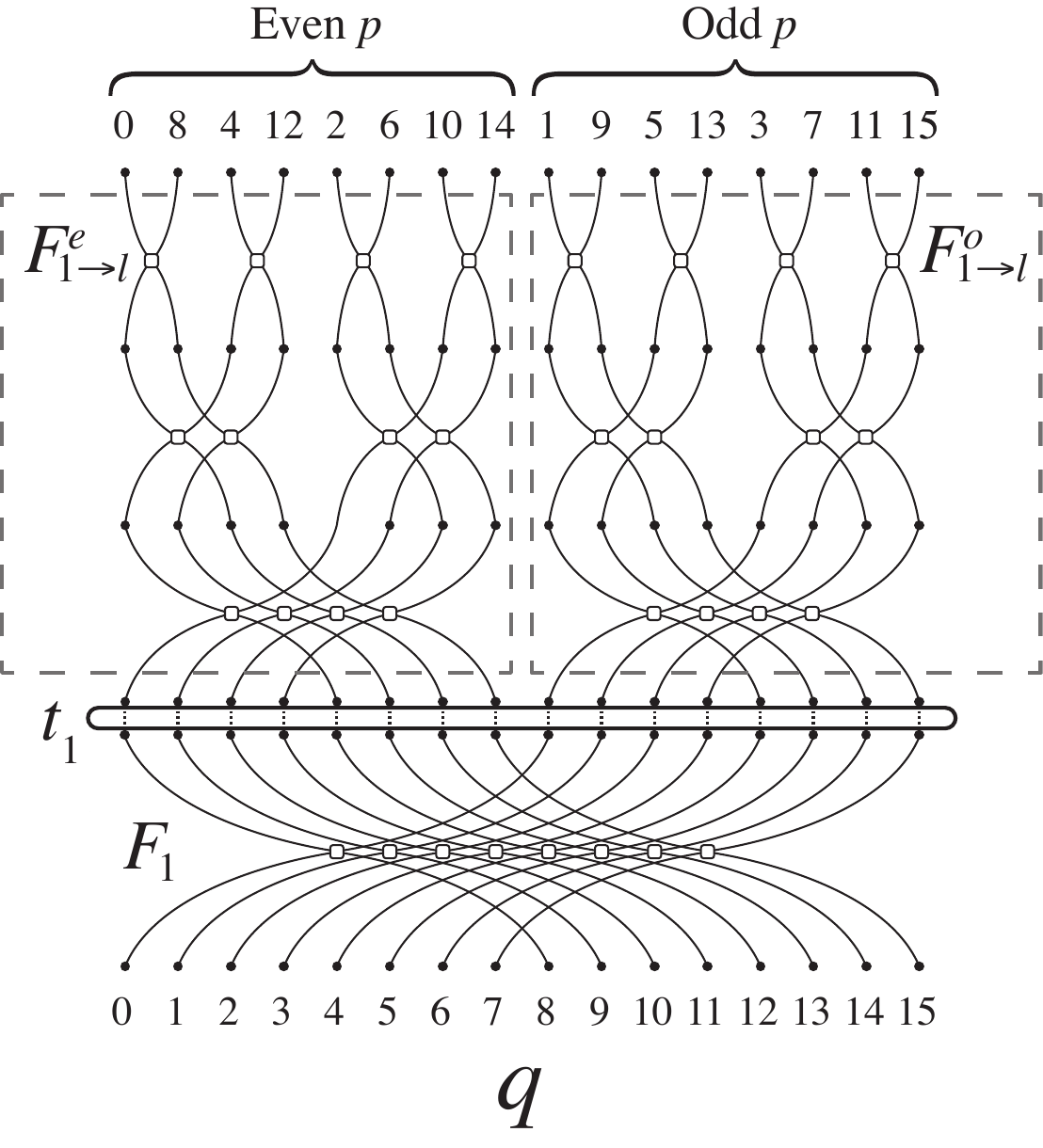}
\caption
{
The decomposition of a large transformation into two smaller partial transformations
$F_{1 \rightarrow l}^e$ and $F_{1 \rightarrow l}^o$, relative to the even $p$ and odd $p$
branch of a larger butterfly diagram.
The twiddle factor $t_1$ is applied between the first-level transformation $F_1$ and 
the blocks $F_{1 \rightarrow l}^e$ and $F_{1 \rightarrow l}^o$.
}
\label{butterflyblocks}
\end{figure}

\noindent Note that the above operator combines each $i$-th component
of $\psi_q$ with the corresponding $i + 2^{l-1}$-th one and produces two components, one for the even $p$ branch and the
other for the odd $p$ one. 

\subsection{Derivation of the twiddle factor}\label{sec:translop}
The shift/rotation operator $S$, which actually is a permutation operator on the components of the vector $\psi_p$, can be 
represented on the basis of the linear 
space of the $\psi_q$s. The change of representation from the wavefunctions $\psi_q$ to the wavefunctions $\psi_p$ is carried 
out by the transformation $\mathcal{F}$; thus, 
by identifying the shift/rotation operator in the spaces $\psi_p$ and $\psi_q$ 
as $S_l$ and $S_0$ respectively, we can write the following relation:
\begin{equation}  \label{eq:shifttrasf}
S_0 = \mathcal{F}_{l}^{-1} S_l \mathcal{F}_{l}.
\end{equation}
 
\noindent The index convention for the operators $S$ reflects the parametrization level: $S_m$ means that it is
represented in the vector space relative to the  level $m$, with $S_l$ corresponding to the
highest level $\psi_p$, and $S_0$ to the lowest level $\psi_q$.
This convention allows us to generalize  \eqref{eq:shifttrasf} for
any level $l'$:  we can write the expression of the shift/rotation operator for an ``intermediate''
state vector
$\psi^{(l)} = F_{0 \rightarrow l}\psi_q$ (where the partial transform $F_{l' \rightarrow l}$
is given by \eqref{eq:fftpartiall}): 
\begin{equation}  \label{eq:shifttrasflg}
S_{l'} = F_{l' \rightarrow l}^{-1} S_l F_{l' \rightarrow l}.
\end{equation}
\noindent Equation \eqref{eq:shifttrasf} also implies that 
$S_l  \psi_p = S_l \mathcal{F}_{l} \psi_q  =
 \mathcal{F}_{l}  {\mathcal{F}_{l}}^{-1} S_l \mathcal{F}_{l} \psi_q
= \mathcal{F}_{l} S_0 \psi_q$. 
If we use the expression \eqref{eq:fftgrpl} of the  transformation $\mathcal{F}_{l}$,
this can also be written as:
\begin{equation}  \label{eq:shifttrasfl}
S_l  \psi_p  =  F_{1 \rightarrow l}   t_{1} F_{1} S_{0}  \psi_q.
\end{equation}

\noindent Now the translation invariance hypothesis comes into play:
we suppose that the shift/rotation transformation on the $p$ axis
does not change the distribution $\rho_q$.
This corresponds to the condition that a translation transforms the physical
state of the system into another valid physical state; if a state with distributions
$\rho_q$ and $\rho_p$ is a valid one, then the state with the same $\rho_q$
and a translated $\rho_p$ must be valid too.

It should be noted that, while in section \ref{sec:nstatesyst} we dealt with generic conditions
$\chi$ on the state space, given in implicit form, we now have an object $\psi$ that embeds a set of
independent coordinates on the state space manifold and, by definition, represents a \textit{valid} system state.
The transformation $\mathcal{F}$ just transforms the representation from $q$ to $p$.
While, in the implicit $\chi$ condition case, we should require 
the $\chi$ to transform according to Eq.~\eqref{eq:chiinvariance}, now we have just to verify that 
a transformation (namely, $S_0$) actually behaves like a shift/rotation. 
The subdivision of the phase space into invariant partitions introduced in section
\ref{sec:statespacesymm}, which we used to represent the probabilities in our model, allows us to
express the shift/rotation transformation in a simple way, based on binary permutations.
Since the shift/rotation transformation operator on the $p$ must leave the $q$ distribution unchanged,
it is expected to take a diagonal form in the space of the $\psi_q$ functions:
\begin{equation} \label{eq:sdiag}
S_0 = \Diag{(e^{i s^0}, ... ,e^{i s^{L}})},    
\end{equation}
\noindent where $s^k$ are some phase coefficients. Now we have all the necessary elements to derive an expression for the 
twiddle factor $t$. In appendix \ref{apx:shrotoperformula} we
demonstrate the following formula:
\begin{equation} \label{eq:twiddlblockexpr}
t_{l'} = \Diag (\underbrace{1,...1}_{L/2},
\underbrace{e^{-2\pi i\frac{0}{L}},e^{-2\pi i\frac{1}{L}},...,e^{-2\pi i\,\frac{L/2-1}{L}}}_{L/2},
...),
\end{equation}
\noindent where $l = n-l'+1$, with $l'$ being the level to which the twiddle factor refers. Note that, since the vector space 
on which the matrix is defined  is $N$-dimensional, in the above expression the sequence\par\noindent
$1,1,...,1,e^{-2\pi i\,1/L},e^{-2\pi i\, 2/L},...,e^{-2\pi i\,(L/2-1)/L}$ is repeated $\frac{N}{L}$ times.

\subsection{Recognizing the discrete Fourier transform}
We are very close to our goal, which is to recognize the factor 
$t_l^k$ as the twiddle factors of the Danielson-Lanczos lemma and $\mathcal{F}$ as 
the discrete Fourier transform.
To do this, we need a bit more algebra. The expression \eqref{eq:fftexpln}
of $\mathcal{F}_{n}$ can be inverted as: 
\begin{equation} \label{eq:fftexplninv }
\psi_q = \mathcal{F}_{n}^{-1}  \psi_p =  F_{1}^{-1} 
\cdot  t_1^{-1} \cdot ... \cdot F_{n-1}^{-1} \cdot t_{n-1}^{-1} \cdot F_{n}^{-1} \cdot \psi_p.
\end{equation}
\noindent Actually, all the $F_l$ transforms are a composition of $N/2$ elementary systems' 
transformations \eqref{eq:Finpsi}. It is easy to see that the $F_l$ transforms are self-inverse
because  \eqref{eq:Finpsi} is self-inverse.
What differs among different level transforms is how the space components are arranged. 
For example, the first level $F_{1}$ differs from $F_{n}^{-1}$ just for the inversion of the bit order
of $\psi_p$ with respect to $\psi_q$; this is apparent from the butterfly 
diagram in Fig.~\ref{butterflyblocks}.

Furthermore, the twiddle factors $t_l$ obey the relation $t_l^{-1} = t_l^*$, because they are simple 
phase factors. Thus, we can write:
\begin{equation} \label{eq:fftexplninv2}
\psi_q = \mathcal{F}_{n}^{-1}  \psi_p =  F_{1}  \cdot  t_1^{*} \cdot
\cdot F_{2} ... \cdot t_{n-1}^{*} \cdot F_{n}\cdot \psi_p.
\end{equation}
\noindent
Let us now recall the Danielson-Lanczos lemma: the $N$-dimensional discrete Fourier transform 
on a discrete function can be decomposed into two $N/2$ dimensional DFTs as
\begin{equation}
    \begin{split} \label{eq:dft1}
\psi_q(j)    = & \frac{1}{\sqrt{N}} \sum^{N-1}_{k=0} \, e^{2 \pi i j k/N} \psi_p(k)  \\
 = &  \frac{1}{\sqrt{N}} \Big( \sum^{N/2-1}_{k=0} \, e^{2 \pi i j k/(N/2)} \psi_p(2 k) \twocolbreakamp 
  + W^j \sum^{N/2-1}_{k=0}  \, e^{2 \pi i j k/(N/2)} \psi_p(2 k + 1) \Big) ,
\end{split}
\end{equation}
\noindent where $W^j = e^{2 \pi i j/N}$. 
Strictly speaking, the above equation refers to the inverse
DFT, while the direct DFT would have a ``-'' sign in the exponents. This is the inverse version of the
Danielson-Lanczos lemma, which is, of course, conceptually equivalent to the standard one; we choose
the plus sign to match the standard quantum mechanics conventions. 
It is easy to show that $W^{j + N/2} = - W^j$. This allows us to write: 
\begin{equation}\begin{aligned}
\label{eq:dft2}
\psi_q(j+N/2)  = \frac{1}{\sqrt{N}} \Big(
\sum^{N/2-1}_{k=0} \, e^{2 \pi i j k/(N/2)} \psi_p(2 k) 
\twocolbreakamp  - \, W^j \sum^{N/2-1}_{k=0} \, e^{2 \pi i j k/(N/2)} \psi_p(2 k + 1) \Big)
\end{aligned}\end{equation}
\noindent The last two equations can be synthesized into the following system:
\begin{widetext}
\begin{equation} \label{eq:dftsyst}
\begin{array}{rl}
\psi_q(j) & = \frac{1}{\sqrt{2}} (f_{(N/2)}^{j k} \psi_p(2 k) +  W^j f_{(N/2)}^{j k} \psi_p(2 k+1)) \\
\psi_q(j+N/2) & =  \frac{1}{\sqrt{2}} ( f_{(N/2)}^{j k} \psi_p(2 k) -  W^j f_{(N/2)}^{j k} \psi_p(2 k+1) ),
\end{array}
\end{equation}
\end{widetext}
\noindent where $f_{(N/2)}^{j k}$ are the $N/2$-dimensional DFT matrix elements and the summation over
$k$ is implicit. Let us consider the first two terms of factorization
\eqref{eq:fftexplninv2}, $F_{1}  \cdot  t_1^{*}$. We can apply the same
arguments we used in appendix \ref{apx:shrotoperformula}  for equation \eqref{eq:shifteq2d}, to show that this matrix
does not mix all the $N$ components of its vector space:
each $j$-th component will be mixed with the corresponding $j+N/2$ component. 
This can also be understood from Fig.~\ref{butterflyblocks}.
We can use the expression \eqref{eq:twiddlblockexpr} for the twiddle factor 
(in the case $l'=1 \rightarrow L= N$) and write only the components $j$ and $j+N/2$
of the product  $F_{1}  \cdot  t_1^{*}$. This leads us to a relation in a two-dimensional vector space:
\begin{equation}\begin{aligned} \label{eq:firsfactors}
F_{1}  \cdot  t_1^{*} \bigm\vert_{j, j+N/2} &  = \frac{1}{\sqrt{2}}
\begin{pmatrix}
1 & 1 \\
1 & -1
\end{pmatrix}
\begin{pmatrix}
1 & 0 \\
0 & e^{\frac{2 \pi i j}{N}}
\end{pmatrix} \twocolbreak
& =
\frac{1}{\sqrt{2}} \begin{pmatrix}
1 &  e^{\frac{2 \pi i j}{N}} \\
1 & -e^{\frac{2 \pi i j}{N}}
\end{pmatrix}
\end{aligned}\end{equation}
\noindent 
The factorization on the RHS of \eqref{eq:fftexplninv2} can be put  
in the form $\mathcal{F}^{-1}_{N} = F_1  t_1^{*} F^{-1}_{1 \rightarrow n}$, where 
$F^{-1}_{1 \rightarrow n} = F_{2}  \cdot  t_2^{*} \cdot ... \cdot t_{n-1}^{*} \cdot F_{n}$
is the inverse of the partial transform defined in \eqref{eq:fftpartiall};  this form can also be
obtained directly from \eqref{eq:fftgrpl}.
Therefore, we can multiply on the right both sides of \eqref{eq:firsfactors} by
$F^{-1}_{1 \rightarrow n} \psi_p$ (by considering the $j$- and $j+N/2$-th components only) and rewrite the term on the LHS as 
$\mathcal{F}_{N}^{-1} \psi_p = \psi_q$. This leads us to the equation
\begin{equation} \label{eq:allfactors}
 \psi_q \bigm\vert_{j, j+N/2} = \frac{1}{\sqrt{2}}
\begin{pmatrix}
1 &  e^{\frac{2 \pi i j}{N}} \\
1 & -e^{\frac{2 \pi i j}{N}} 
\end{pmatrix} (F^{-1}_{1 \rightarrow n} \psi_p)
\bigm\vert_{j, j+N/2},
\end{equation}
\noindent in which we can recognize the system  \eqref{eq:dftsyst}, by  assuming that
$f_{(N/2)}^{j k}$ is the order $n-1$ DFT and that $W^j = e^{2 \pi i j/N}$ corresponds to our twiddle factor
\eqref{eq:twiddlblockexpr}. 
In conclusion, we have demonstrated that the last step $F_{1}  \cdot  t_1^{*}$ of the transformation of our model
corresponds to the last step of the DFT in the decomposition prescribed 
by the Danielson-Lanczos lemma, and that our factor $t_l^j$ corresponds to the twiddle factor $W_l^j$.
Since the Danielson-Lanczos lemma can be applied recursively, if this statement is true for the last
level, then we can apply the same scheme to all subsystems and for discrete Fourier transforms with a
dimensionality smaller than $N$. It must be pointed out that the choice we made in \ref{sec:indformterms} for the sign of the 
$\pi$ term in the twiddle factor phase's recursion formula has led 
us to the standard---positive spatial frequency---convention of quantum mechanics, but a different
choice would have led us to an equivalent formulation.

This completes the proof of the correspondence between our model's transformation operator
$\mathcal{F}$ and  the discrete Fourier transform.

\section{Conclusions}\label{sec:conclusions}

This study focuses on a purely probabilistic representation of nature and on the measurement of physical quantities, assuming 
that only a limited amount of information is available to characterize physical phenomena. Together with a set of reasonable 
axioms, these assumptions lead, in the finite-dimensional case, to a theory that reproduces several of the central structures 
of quantum mechanics in its basic formulation.

The Precision Invariance axiom, which is arguably the strongest among the theory's assumptions, has been
expressed in terms of metric properties of the manifold that represents the space of physical states.
We have shown that the standard linear-space representation of quantum mechanics arises naturally as a
consequence of these metric properties.

One of the strengths of this work is that it establishes a direct connection between the intuitive conception physicists have 
about quantum mechanics---namely, the theory that reflects the limitation of information in describing physical systems---and 
its conventional formulation.

Another key contribution is the attempt to change the perspective in multiple aspects of quantum mechanics. In the standard 
formulation of quantum mechanics, it is implicitly accepted that its probabilistic/random nature is something different from 
that of probability applied to classical physics. While, in classical mechanics, randomness is just a consequence of our lack 
of knowledge about the details of a system, in quantum mechanics, randomness and---consequently---the statistical properties of 
a system, are intrinsic features of the physics of the system itself.

In contrast, the theory proposed here adopts an agnostic viewpoint regarding the origin of quantum randomness. It may be 
intrinsic or it may result from incomplete knowledge.
Nonetheless, this randomness is still expected to satisfy certain properties, such as Precision Invariance.
This perspective can broaden the scope of application and interpretation of quantum mechanics.

An even more intriguing point that could change our conception of quantum mechanics is that its linear structure emerges not as 
a postulate, but as a consequence of the metric preserving property of the state space, which in turn is a consequence of the 
Precision Invariance hypothesis. When these conditions are not met, one can encounter scenarios where the state space is no 
longer linearizable. Linearity in quantum mechanics is a fundamental feature that affects, for example, the form of the time 
evolution operator.

The Precision Invariance axiom is based on the hypothesis that we can {\it{prepare}} a system in a well-defined state. 
Moreover, the \allowbreak\hspace{0pt} metric-preserving property, derived from this axiom, assumes the existence of a 
differentiable mapping between different parameterizations.
This is merely an assumption, based on the general idea that transformations between different parameterizations can be 
expressed through well-behaved functions. Nonetheless, we must consider that there are cases where even the behavior of the 
classical observables, related to these parameterizations, cannot be described by well-behaved functions.

This point makes the theory presented in this article generally equivalent to quantum mechanics, yet distinct in its essence. 
By removing the linear structure as a foundational requirement of quantum theory~\cite{Pearle,Weinberg}, this model opens the 
door to alternative formulations that may address long-standing challenges, such as the measurement or wavefunction collapse 
problem.

This aspect might represent a starting point for future developments and new perspectives on the interpretation of quantum mechanics.

There are other possible directions for further development of the theory presented in this work. We can extend these results 
to more complex systems, such as those involving infinitely many pairs of conjugate variables, and examine whether the axioms 
of the theory are compatible with quantum field theory. One will likely face the same problems that arise, for instance, when 
dealing with gauge invariance in quantum field theory.

If we accept the idea that the state space $\Phi$ of our theory corresponds to the standard Fock space of quantum fields, we 
must ask ourselves how negative norm states, occurring in the covariant theory of photons~\cite{Gupta}, can be interpreted in 
terms of state space points. A similar issue appears in Yang-Mills theories, where auxiliary fields 
(ghosts)~\cite{FaddeevPopov} are introduced to ensure gauge invariance of the theory.

Another aspect that can be investigated is how all the concepts introduced by this formulation of quantum mechanics can be 
adapted to the case of fermion systems. We can consider applying all of the machinery coming from the axioms presented in this 
study to Grassmann variables describing fermion fields.  

The author hopes that this study offers an intuitive perspective on quantum theory and encourages further exploration of new 
areas of discussion on its foundational principles.

\bibliography{refs}

\appendix

\section{Derivations for the metric formulation of the  Precision Invariance axiom} 
\subsection{Proofs of equivalent conditions for the Precision Invariance axiom} \label{apx:pialtconditionsx}

In the following derivations, we will lighten the notation by adopting the bra-ket formalism and neglecting the subscript 
$(o)$; thus, we can write $\vert e^k \rangle = {\mathbf{e}^k}_{(o)}$ and equation \eqref{eq:reprojectedfisheri} can be put in 
the form:
\begin{equation}\label{eq:precindaxbraket}
\langle {e^{k}}\vert \mathcal{I}^{(o)} (\mathbf{x}) \vert {e^{l}} \rangle  =  
       \delta^{k l} \, \, \forall o \in \mathcal{D}, 
\end{equation} 
\noindent where $k,l \le N_{\theta}$. As stated before, we can assume that, in equation \eqref{eq:fisheriterm} the Fisher 
matrix $\mathcal{I}_{kl}^{(o)} (\boldsymbol{\theta}_o)$ is diagonal; being positive semidefinite, we can denote its eigenvalues 
by ${\sigma^k}^2$ and write the diagonality condition as $\mathcal{I}_{k'l'}^{(o)} (\boldsymbol{\theta}_o) = \sigma^k \delta^{k l} \sigma^l$. 

With this assumption, if we apply the formalism introduced above and identify $\vert {h^*}^k \rangle$ as the cotangent vector 
$\frac{\partial \Theta_o^{k}}{\partial x^i}$, we can rewrite 
\eqref{eq:fisheriterm} as 
\begin{equation}\label{eq:fisheritermbraket}
         \mathcal{I}^{(o)} (\mathbf{x})  = \sum_{k,l \le N_{\theta}}
         \vert {h^*}^k \rangle \sigma^k \delta^{k l} \sigma^l \langle {h^*}^l \vert.
\end{equation} 
\noindent By multiplying left and right the above equation by $\langle {e^{k}}\vert$ and $\vert {e^{l}} \rangle$ and 
simplifying out $\delta^{k l}$, we have
\begin{equation}\label{eq:fisheritermbraket2}
    \langle {e^{k}} \vert \mathcal{I}^{(o)} (\mathbf{x}) \vert {e^{l}} \rangle = 
    \sum_{k' \le N_{\theta}}
         \langle {e^{k}}\vert {h^*}^{k'} \rangle \sigma^{k'} \sigma^{k'} \langle {h^*}^{k'} \vert {e^{l}} \rangle
\end{equation} 
\noindent  and, by applying \eqref{eq:precindaxbraket}, we can write
\begin{equation}\label{eq:fisheritermbraket3}
    \sum_{k'\le N_{\theta}}
         \langle {e^{k}}\vert {h^*}^{k'} \rangle \sigma^{k'}  \sigma^{k'} \langle {h^*}^{k'} \vert {e^{l}} \rangle =   
       \delta^{k l}. 
\end{equation} 
\noindent This equation states that the matrix $E^{k l} = \sigma^{k} \langle {h^*}^{k} \vert {e^{l}} \rangle$ is orthogonal. 
Since for an orthogonal matrix $E$  we have $E^{\mathrm{T}} E = E\,E^{\mathrm{T}} =I$, Eq.~\eqref{eq:fisheritermbraket3} 
implies that 
\begin{equation}\label{eq:fisheritermbraket4}
    \sum_{k'\le N_{\theta}}
         \sigma^{k} \langle {h^*}^{k} \vert {e^{k'}}\rangle \langle {e^{k'}}\vert {h^*}^{l} \rangle \sigma^{l} =   
       \delta^{k l}. 
\end{equation}
\noindent Being the eigenvectors $e^k$ orthogonal and normalized, in the subspace where the matrix
$\mathcal{I}^{(o)} (\mathbf{x})$ is nonzero, they obey the completeness condition
$\sum_k \vert{e^{k}}_{(o)} \rangle \langle {e^{k}}_{(o)} \vert = I$. Thus, the last equation allows us to conclude that
\begin{equation}\label{eq:fisheritermbraket5}
         \sigma^{k} \langle {h^*}^{k} \vert  {h^*}^{l} \rangle \sigma^{l} =   
       \delta^{k l}, \ \  k,l  \le N_{\theta}, 
\end{equation}
\noindent which tells us that the cotangent vectors $\vert {h^*}^{k} \rangle \sigma^{k} $ are orthogonal and that their 
normalization factor is $\sigma^{k}$, i.e.,  
$\sigma^{k} = {\iota^*}^k $ (we omit the subscript ${\boldsymbol{\omega}_{o}}$ in $\iota^*$ to avoid stacking of subscripts).
The identity  $\sigma^{k} = {\iota^*}^k$ and the definition
${\mathbf{b^*}}_{\boldsymbol{\theta}_{o}}^{k} = {\iota^*}^k {\mathbf{h}^*}_{\boldsymbol{\theta}_{o}}^{k}$
imply that equation \eqref{eq:fisheritermbraket3} can also be put in the form 
\begin{equation}\label{eq:hbcompleteness}
    \langle {{b^*}^{k}}\vert \sum_{k'\le N_{\theta}}
        \vert e^{k'} \rangle \langle e^{k'} \vert\ \vert {{b^*}^{l}} \rangle =   
       \delta^{k l}, \ \  k,l  \le N_{\theta},
\end{equation} 
\noindent which is a completeness condition of the basis $\mathbf{e}^k$ on the vectors ${\mathbf{b}^*}^k$. Or, in other words, 
it states that the ${\mathbf{b}^*}^k$s belong to the subspace generated by the eigenvectors $\mathbf{e}^k$.

Since, according to \eqref{eq:precindaxbraket}, the vectors of ${\mathbf{e}^k}$ are $N_{\theta}$-times degenerate eigenvectors 
of $\mathcal{I}^{(o)} (\mathbf{x})$, any other vector of the subspace generated by them is an eigenvector as well. This also 
applies to the vectors ${\mathbf{b}^*}_{\boldsymbol{\theta}_{o}}^{k}$. 
We can therefore choose the cotangent vectors ${\mathbf{b}^*}^k$, for
$k \le N_{\theta}$ as eigenvectors and  rewrite \eqref{eq:precindaxbraket} in the form of Eq.~\eqref{eq:precindb}.

Now let us apply, in turn, the diagonality of $\mathcal{I}_{kl}^{(o)} (\boldsymbol{\theta}_o)$, the identity between eigenvalue 
and normalization factors \eqref{eq:normfaceigvalidentity} and the orthogonality of the basis $b_{\omega_o}$, in the hypothesis 
that $k,l \le N_{\theta}$:
\begin{equation}\begin{aligned}\label{eq:fisheripiaxaltder}
& \mathcal{I}_{kl}^{(o)} (\boldsymbol{\theta}_o)
= \sigma_k \delta_{k l} \sigma_l =   
{{\iota}_{\boldsymbol{\omega}}}_{k}^{-1}  \delta_{k l} 
{{\iota}_{\boldsymbol{\omega}}}_{l}^{-1}  \twocolbreak & =
 {{\iota}_{\boldsymbol{\omega}}}_{k}^{-1} {b_{\boldsymbol{\omega}}}_{k}^i
 \delta_{i j} 
{b_{\boldsymbol{\omega}}}_{l}^j
{{\iota}_{\boldsymbol{\omega}}}_{l}^{-1}
= 
\frac{\partial X^i}{\partial \omega^k} \delta_{i j} \frac{\partial X^j}{\partial \omega^l}. 
\end{aligned}\end{equation} 
\noindent We have therefore proved  Eq.~\eqref{eq:fisheripiaxalt}.

\subsection{Proof of the Precision Invariance formula \eqref{eq:constfishermatrixcond}} 
\label{apx:fisherinformationidmtx}
Let us rewrite the definition \eqref{eq:constfishermatrixcond} in the following form:
\begin{equation} \label{eq:extfishinfsep}
g^{(o)}_{kl} = g^{(o)}_{kl}\bigm\vert_{\boldsymbol{\theta}} + g^{(o)}_{kl}\bigm\vert_{\boldsymbol{\alpha}}
\end{equation}
\noindent where
\begin{subequations} \label{eq:extfishinfterm}
\begin{align}
g^{(o)}_{kl}\bigm\vert_{\boldsymbol{\theta}}
  & =   \mathcal{I}_{kl}^{(o)} (\boldsymbol{\theta}_o) \delta_{k,l \le N_{\theta}} 
  \label{eq:extfishinfterm1}\\
g^{(o)}_{kl}\bigm\vert_{\boldsymbol{\alpha}} & = 
\frac{\partial {X_o}^i}{\partial \omega_o^k}
\delta_{i j} 
\frac{\partial {X_o}^j}{\partial \omega_o^l} \delta_{k,l > N_{\theta}},
\label{eq:extfishinfterm2}
\end{align}
\end{subequations}
\noindent  with $\delta_{k,l \le N_{\theta}}$ and $\delta_{k,l > N_{\theta}}$
denoting the indicator functions that take the value 1 only when, respectively, 
$k,l \le N_{\theta}$ and $k,l > N_{\theta}$, and vanish otherwise.
With this definition we can rewrite Eq.~\eqref{eq:extfishinfonormx} as
\begin{equation} \label{eq:extfishinfonormxspl}
\mathcal{I}_{i j}^{\omega_o} (\mathbf{x})  =
\frac{\partial \omega_o^k}{\partial x^i}
g^{(o)}_{kl} \bigm\vert_{\boldsymbol{\theta}}
\frac{\partial \omega_o^l}{\partial x^j} + 
\frac{\partial \omega_o^k}{\partial x^i}
g^{(o)}_{kl} \bigm\vert_{\boldsymbol{\alpha}}
\frac{\partial \omega_o^l}{\partial x^j}.
\end{equation}
\noindent By definition \eqref{eq:extfishinfterm1} and then by Eq.~\eqref{eq:fisheripiaxalt} the first term on the RHS of the 
above equation can be rewritten as 
\begin{equation}\begin{aligned}\label{eq:fisherinfofirstterm}
&\frac{\partial \omega_o^k}{\partial x^i}
 \mathcal{I}_{kl}^{(o)} (\boldsymbol{\theta}_o) \delta_{k,l \le N_{\theta}}
\frac{\partial \omega_o^l}{\partial x^j}  \twocolbreak & =
\frac{\partial \omega_o^k}{\partial x^i}
 \delta_{k \le N_{\theta}}
 \frac{\partial X^i}{\partial \omega_o^k} \delta_{i j} \frac{\partial X^j}{\partial \omega_o^l}
  \delta_{l \le N_{\theta}}
\frac{\partial \omega_o^l}{\partial x^j},
\end{aligned}\end{equation}

\noindent By definition \eqref{eq:extfishinfterm2}, the second term of \eqref{eq:extfishinfonormxspl} reads:
\begin{equation}\begin{aligned}\label{eq:fisherinfosecondtterm}
& \frac{\partial \omega_o^k}{\partial x^i}
 \mathcal{I}_{kl}^{(o)} (\boldsymbol{\theta}_o) \delta_{k,l > N_{\theta}}
\frac{\partial \omega_o^l}{\partial x^j}  \twocolbreak & = 
\frac{\partial \omega_o^k}{\partial x^i}
 \delta_{k > N_{\theta}}
 \frac{\partial X^i}{\partial \omega_o^k} \delta_{i j} \frac{\partial X^j}{\partial \omega_o^l}
  \delta_{l > N_{\theta}}
\frac{\partial \omega_o^l}{\partial x^j},
\end{aligned}\end{equation}

\noindent Of course, the definition of the extended coordinate $\omega_o$ implies that the term
$\frac{\partial X^i}{\partial \omega_o^k}$ corresponds to $\frac{\partial X^i}{\partial \theta_o^{l}}$ and
$\frac{\partial X^i}{\partial \alpha_o^{k}}$ in the cases $k \le N_{\theta}$ and $k > N_{\theta}$, respectively.
If we expand the product:
\begin{equation}\begin{aligned}
\frac{\partial X^i}{\partial \omega_o^k} 
\delta_{i j}
\frac{\partial X^j}{\partial \omega_o^l}  =
& \left(\delta_{k \le N_{\theta}}
\frac{\partial X^i}{\partial \omega_o^k} + 
\delta_{k > N_{\theta}}
\frac{\partial X^i}{\partial \omega_o^k} \right)
\delta_{i j} \cdot \\
& \left(\delta_{l \le N_{\theta}}
\frac{\partial X^j}{\partial \omega_o^l} + 
\delta_{l > N_{\theta}}
\frac{\partial X^j}{\partial \omega_o^l} \right),
\end{aligned}\end{equation}
\noindent the orthogonality condition \eqref{eq:alphathetaortho} makes the cross terms cancel, and we can recognize the sum of 
the two central terms in Eqs.~\eqref{eq:fisherinfofirstterm} and \eqref{eq:fisherinfosecondtterm}. 
Eq.~\eqref{eq:extfishinfonormxspl} can therefore be written in the form
\begin{equation} \label{eq:extfishinfosimpla}
\mathcal{I}_{i j}^{\omega_o} (\mathbf{x})  =
\frac{\partial \omega_o^k}{\partial x^i}
\frac{\partial X^i}{\partial \omega_o^k} 
\delta_{i j}
\frac{\partial X^j}{\partial \omega_o^l} 
\frac{\partial \omega_o^l}{\partial x^j},
\end{equation}
\noindent which proves Eq.~\eqref{eq:constfishermatrixcond}.

\section{The state manifold of the three-observable elementary system} \label{apx:statman}
\subsection{Cardinal Points}
The first step is to consider the cases where, in turn, one of the observables  $q$, $p$ or $r$ is
completely determined and the other two are completely undetermined. By applying the same scheme
as in the two-observable case, the limitation of the amount of information and axiom 
\ref{ax:infostatequivalence} leads us to the table
\begin{equation}\label{eq:knownpqrvals}
\def\arraystretch{1.13}
\begin{array} {llll}
\bf{p}_1: \, & \rho_q = (1, 0),    &    \rho_p = (\frac{1}{2}, \frac{1}{2}), &   \rho_r = (\frac{1}{2}, \frac{1}{2}) \\
\bf{p}_2: \, & \rho_q = (0, 1),    &    \rho_p = (\frac{1}{2}, \frac{1}{2}), &   \rho_r = (\frac{1}{2}, \frac{1}{2})  \\
\bf{p}_3: \, & \rho_q = (\frac{1}{2}, \frac{1}{2}), &   \rho_p = (1, 0), &       \rho_r = (\frac{1}{2}, \frac{1}{2})     \\
\bf{p}_4: \, & \rho_q = (\frac{1}{2}, \frac{1}{2}), &   \rho_p = (0, 1), &       \rho_r = (\frac{1}{2}, \frac{1}{2})     \\
\bf{p}_5: \, & \rho_q = (\frac{1}{2}, \frac{1}{2}), &   \rho_p = (\frac{1}{2}, \frac{1}{2}), &   \rho_r = (1, 0)     \\
\bf{p}_6: \, & \rho_q =  (\frac{1}{2}, \frac{1}{2}), &  \rho_p = (\frac{1}{2}, \frac{1}{2}), &   \rho_r = (0, 1)   , 
\end{array}
\end{equation}

\noindent which is the three-observable version of \eqref{eq:knownpqvals}. See figure \ref{3Dsystem1a}
for a graphical representation of the above points. In order to find the structure
of the state space, we will require (i) the system state to be a manifold passing
by the points defined in \eqref{eq:knownpqrvals} and (ii) that the metric-preserving condition 
\eqref{eq:metricpreservationtheorem} holds.

\subsection{Special cases}

\mysubsubsection{Two-Observable subspace case}
To simplify the problem, we consider some special subset of the state space that behaves in a way
we already know.
The first special case we are going to consider is when we know that one of the observables 
is completely undetermined. For example $ \rho_q(0) = \rho_q(1) = 1/2$
(or, by \eqref{eq:nonisoparm}, $\theta_q = \pi/2$ and $S_q = 0$) (see Fig.~\ref{3Dsystem1a} and
\subref{3Dsystem1b}).
This choice implies that the total amount of information available (one bit) can be entirely
dedicated to the other two observables. The metric-preserving condition 
\eqref{eq:metricpreservationtheorem} is still valid, and it acts on the other two parameters.
Thus, the system simplifies into the two-observable problem studied in  \ref{sec:twoobs}.
The solution of this system is given by \eqref{eq:solutionfor2d} and  \eqref{eq:rhovstheta}.

In this case, the two-observable system that we obtain 
is described in terms of the parameters $\theta_p$ and $\theta_r$.
These parameters obey condition \eqref{eq:solutionfor2d}, which now reads
$\theta_r = \pi/2 - \theta_p$. 

The same scheme can be applied to the cases $\theta_p = \pi/2$ ($S_p = 0$) and $\theta_r = \pi/2$ ($S_r = 0$),
leading to two-observable systems with parameters $\theta_q, \theta_r$ and $\theta_q, \theta_p$,
respectively, which obey conditions analogous to those of the $S_q = 0$ case.

Let us write these conditions for the cases $\theta_q = \pi/2$ and $\theta_p = \pi/2$ only,
which will help us in the derivations presented in the following paragraph:
\begin{equation}\label{eq:spqrplanes}
\begin{array} {lllll}
S_q = 0 \,\, \textnormal{plane} & \Rightarrow & \theta_p = \pi / 2 - \theta_r & \Rightarrow & S_p =  \sin \theta_r \\
S_p = 0 \,\, \textnormal{plane} & \Rightarrow & \theta_q = \pi / 2 - \theta_r & \Rightarrow & S_q =  \sin \theta_r.
\end{array}
\end{equation}

\mysubsubsection{Symmetries}
As prescribed by requirement \ref{ax:symmetries}, the system is assumed to exhibit certain symmetries, which, in our framework, 
take the form of symmetries of the state space.

A first type of symmetry arises from the assumption that, from a probabilistic point of view, no particular
value of a given observable is privileged. For example, the states $q = 0$ and $q = 1$ must be treated
equivalently. Consequently, the system must be symmetric under transformations such as
$\rho_{\nu}(0) \leftrightarrow \rho_{\nu}(1)$, where $\nu = p, q, r$,
which correspond to the inversion of any of the axes $S_q$, $S_p$, or $S_r$.

A second symmetry stems from the assumption that the three observables are expected to have identical
physical properties. This implies that the system must be symmetric under the exchange of observables.

\mysubsubsection{Other conditions}
Besides being mutually orthogonal, there are no special conditions the parametrization 
\eqref{eq:threeparamtzn} must obey. 
Taking advantage of this arbitrariness, we can require the parameter $\alpha$ to obey some
additional conditions.
One of them is that the points $\alpha = -\pi/2, 0, \pi/2, \pi$
correspond to some planes of the  $(S_q, S_p, S_r)$ space.
For example, as shown in Fig.~\ref{3Dsystem1a},  we suppose that $\alpha_q = 0,\pi$ identifies the
plane $S_r = 0$ and $\alpha_q = \pm \pi/2$ identifies the plane $S_p = 0$.

Furthermore, we also assume the functions $\Theta_{\mu \rightarrow \nu}(\theta_{\mu}, \alpha_{\mu})$, 
defined in \eqref{eq:threeparamfun}, to be periodic in $\alpha$ with a period of $2 \pi$.
This implies that $\alpha = -\pi$ and $\alpha = \pi$ represent
the same state.

\subsection{The $S_{\mu} = \textnormal{const}$ subspace}\label{sec:sconstsubspace}
\mysubsubsection{Defining the subspace and its coordinate system}
Let us now consider the three-observable problem, under the hypotheses introduced in section
\ref{sec:estprecsubsp}, namely that a large number of measurements are performed on an observable.
If the observable to be measured is $r$, we focus only on the one-dimensional submanifold
$\theta_r = {\textnormal{const}}$, shown in Fig.~\ref{3Dsystem2a}. 
On this manifold, there is only one locally independent parameter, which can be either $\theta_p$ or 
$\theta_q$. Instead of considering these parameters directly, we introduce the parameters 
$\alpha_{r q} = \alpha_{r q}(\theta_q, \theta_r)$ and
$\alpha_{r p} = \alpha_{r p}(\theta_p, \theta_r)$, which, at this stage, do not necessarily correspond
to any of the $\alpha_{\nu}$s. 
Note that a pair of parameters $\theta$ can identify a state space point only locally, which implies that
they cannot uniquely identify the third parameter $\theta$.
If we want the parameters $\alpha_{r q}$ and $\alpha_{r p}$ to represent every point of the state space, 
we must require the functions $\alpha_{r q}()$ and $\alpha_{r p}()$ to be multi-valued.

We  define how $\alpha_{r q}$ and $\alpha_{r p}$ are related to
$\theta_p$ and $\theta_q$ by using an implicit condition, namely, we require their relation with $S_p, S_q$ and
$\theta_p$ to be of the form
\begin{equation}\label{eq:alphareparams}
\begin{array} {l}
S_q = \mu_q(\theta_r)\cos \alpha_{r q} \\
S_p= \mu_p(\theta_r) \cos \alpha_{r p},
\end{array}
\end{equation}
\noindent which implies that
\begin{equation}\label{eq:spqalphadef}
\begin{array} {l}
\rho_{\nu}(0; \alpha_{\nu}, \theta_{\nu}) = \frac{1+\mu_{\nu}(\theta_{\nu}) \cos \alpha_{r \nu}}{2}
 \\
\rho_{\nu}(1; \alpha_{\nu}, \theta_{\nu}) = \frac{1-\mu_{\nu}(\theta_{\nu}) \cos \alpha_{r \nu}}{2}  
\end{array}
\,\, \nu=p,q,
\end{equation}
\noindent where $\mu_{\nu}(\theta_{\nu}) $ are some arbitrary functions. 
Note that these definitions correspond to equations \eqref{eq:rhovstheta} of
the two-observable case. In order to reach all values of $\rho_p$ and $\rho_q$, we suppose the domain
of  $\alpha_{r \nu}$ to be at least $(0, \pi)$. 
 
It should be noted that the metric-preserving condition \eqref{eq:metricpreservationtheorem} is 
expected to hold for any point of the state space and, as a consequence,
for the neighborhoods defined around any point of the state manifold.
Therefore, if we represent the state space in terms of a set of parameters that include the  
$\alpha_{\nu}$'s,  we must also consider the neighborhoods of 
the points $\alpha_{\nu} = 0$ and  $\alpha_{\nu} = \pi$, and extend the definitions 
\eqref{eq:spqalphadef} to those neighborhoods. 

The definitions \eqref{eq:spqalphadef} express the probabilities
$\rho_p\, \rho_q$ in terms of the parameters $\alpha_{r p}\, \alpha_{r q}$.
These expressions differ between the $\rho_p$ and $\rho_q$ cases only in the factors $\mu_p(\theta_r)$ and $\mu_q(\theta_r)$.
The definitions \eqref{eq:spqalphadef} also imply that the probabilities $\rho_p$ and
$\rho_q$ range in the intervals $(1/2 - \mu_p/2, 1/2 + \mu_p/2)$ and
$(1/2 - \mu_q/2, 1/2 + \mu_q/2)$. 
If we want our theory to be invariant under observable exchange symmetries
like $p \leftrightarrow q$, we must require these intervals to be equal and, 
as a consequence, the functions $\mu_p(\theta_r)$ and $\mu_q(\theta_r)$ to be the same. This allows us to rename 
$\mu_p(\theta_r)$ and $\mu_q(\theta_r)$ simply as $\mu(\theta_r)$

\mysubsubsection{Applying the Precision Invariance axiom}
According to what we explained in section \ref{sec:estprecsubsp}, for large $M_r$, the
estimate of $\theta_r$ is independent of the results of the measurements of $p$ and $q$. This means that 
we can treat $\theta_r$ as a fixed parameter and the measurements on $p$ and $q$ only impact
the estimation of the parameters we use in the submanifold $\theta_r = \textnormal{const.}$. 
We can apply the Precision Invariance axiom to the estimation of the parameters of $p$ and $q$,
namely,  $\alpha_{r p}$ and $\alpha_{r q}$ on this one-dimensional manifold. 
In other words, we have to handle a two-observable elementary system like the one we described in
section \ref{sec:twoobs}.
The metric-preserving condition, corresponding to \eqref{eq:gpreservationtwodim}, now reads
$g_{p} d \alpha_{r p}^2 = g_{q} d \alpha_{r q}^2$.

We can apply the same steps that led to equation \eqref{eq:gderivation} from 
\eqref{eq:fishermetrictensor} 
in the two-observable elementary system case and calculate, through the definitions \eqref{eq:spqalphadef}, the Fisher metric tensors $g_{p}$ and $g_{q}$
with respect to the parameters $\alpha_{r q}$ and 
$\alpha_{r p}$.  Thus, we obtain: 

\begin{equation}\label{eq:gpgqexpr}
g_p = \mu(\theta_r)^2, \,\, g_q = \mu(\theta_r)^2.
\end{equation}

\noindent The condition on the Fisher metric now reads
 \begin{equation}\label{eq:fishmprespq}
\mu(\theta_r)^2 d \alpha_{r p}^2 = \mu(\theta_r)^2 d \alpha_{r q}^2,
\end{equation}
\noindent which leads to a solution that corresponds to \eqref{eq:thetapvsthetaq}:
\begin{equation} \label{eq:thetapvsthetaqmu}
\alpha_{r p} =  \pm \alpha_{r q} + K.
\end{equation}
\noindent The value of the constant $K$ can be derived from symmetry considerations.
According to \eqref{eq:spqrplanes}, the plane $S_q = 0$ is equivalent to the condition
$S_p = \sin \theta_r$. The definitions \eqref{eq:spqalphadef} for $\alpha_{r q}$ tell us that 
$S_q = 0$ if $\alpha_{r q} = \pm \pi/2$. It should be noted that $S_q = 0$ is also the plane of 
symmetry for the transformation $S_q \rightarrow - S_q$ ($\rho_{q}(0) \leftrightarrow \rho_{q}(1)$).
In our framework, symmetry under a given transformation means that the state manifold is invariant 
under that transformation. In the case of the $S_q \rightarrow - S_q$ symmetry, 
we have that for any point $(S_q, S_p, S_r)$ there must be a point of the state manifold
that takes the value  $(-S_q, S_p, S_r)$.

Let us consider the term   $\mu(\theta_r)\cos \alpha_{r q}$, which occurs in \eqref{eq:spqalphadef}, 
as a function of $\alpha_{r q}$:  clearly, in the point of symmetry $\alpha_{r q} = \pi/2$, this
function  is anti-symmetric, which means that it can be inverted in the neighborhood of 
$\alpha_{r q} = \pi/2$. Thus, in this neighborhood, we can say that
$\alpha_{r q} = \cos^{-1} S_q/\mu(\theta_r) $.
If we substitute this expression in equation \eqref{eq:thetapvsthetaqmu}, from the second line of 
\eqref{eq:alphareparams} we get:
 \begin{equation} \label{eq:spsqmap}
S_p =  \mu(\theta_r)\cos\left(\pm \cos^{-1} (S_q/\mu(\theta_r)) + K\right).
\end{equation}
\noindent Requiring that the model has the symmetry $S_q \rightarrow - S_q$ implies that the state
space is invariant under this transformation.
Thus, for each valid point on the state space, there is a corresponding valid point
whose $S_q$ is changed in sign, while  $S_p$ and  $S_r$ do not change. 
Consequently, for fixed $\theta_r$, the mapping  \eqref{eq:spsqmap}
from $S_q$ to $S_p$ in neighborhood of the point $S_q = 0\,(\alpha_{r q} = \pi/2)$ must be a symmetric
function with axis of symmetry $S_q = 0$. 
The sign $\pm$ in \eqref{eq:thetapvsthetaqmu} is not relevant 
because, by changing it, we can redefine $\alpha_{r q} \rightarrow - \alpha_{r q}$ and leave the 
definition $S_q =  \mu(\theta_r)\cos(\alpha_{r q})$ unchanged; we choose that sign to be a minus.

Only a symmetric function applied to an anti-symmetric
one generates a symmetric function; if, for the antisymmetric function at the point of symmetry $x_0$, we have $y_0 = f_a(x_0)$, then the symmetric function must have $y_0$ as its axis of symmetry. Thus, in our case, 
the function $\mu(\theta_r)\cos(-\alpha_{r q} + K)$ in 
\eqref{eq:spsqmap} must be symmetric, with axis of symmetry $\alpha_{r q} = - \pi/2$.
This implies that $K =\pi/2$.
We 
obtain an equation that corresponds to \eqref{eq:solutionfor2d}:
\begin{equation} \label{eq:thetapvsthetaq3d}
\alpha_{r p} =  \pi/2 - \alpha_{r q}
\end{equation}
\noindent The definition \eqref{eq:alphareparams}  of $\alpha_{r p}$ tells us that, in the plane 
$S_q = 0$ ($\alpha_{r p} = 0$), the condition $S_p = \mu(\theta_r)$ must hold.
But, according to \eqref{eq:spqrplanes}, and to the condition  $S_q = 0$,  this also implies that
$S_p = \sin \theta_r$ (as shown in Fig.~\ref{3Dsystem1b}). Therefore we can write
\begin{equation} \label{eq:murvsthetar}
\mu(\theta_r) = \sin \theta_r.
\end{equation}

\mysubsubsection{The state space manifold in parametric form}
Let us consider the extended parameters $\boldsymbol{\omega}_r$  defined in \eqref{eq:threeparamtzn}.
We rename $\alpha_{r p}$ as the  $\alpha_r$ component occurring in that definition.
We can now write the full parametric form of the state manifold coordinates $S_p, S_q, S_r$
in terms of the parametrization $\boldsymbol{\omega}_r$. To do this, we just have to put
together the third equation of \eqref{eq:nonisoparm}, which defines $S_r$,
and the expressions \eqref{eq:alphareparams} of $S_p$ and $S_q$ in terms of
$\alpha_{r p}$ and  $\alpha_{r q}$; 
then, by substituting the expressions of $\alpha_{r q}$ and $\mu(\theta_{r})$ given by
\eqref{eq:thetapvsthetaq3d} and \eqref{eq:murvsthetar}, we obtain:
\begin{eqnarray} \label{eq:parstatemanifold}
\begin{array}{l}
S_{r} = \cos \theta_{r}  \\
S_{q} = \sin \theta_{r} \cos \alpha_{r}  \\
S_{p} = \sin \theta_{r} \sin \alpha_{r}
\end{array}
\end{eqnarray}

\noindent The above expression represents a unit sphere in parametric form.
The steps that led us from the solution in the $S_{\mu} = \textnormal{const}$ 
subspace (equations \eqref{eq:fishmprespq} and \eqref{eq:thetapvsthetaqmu}) to 
a spherical-shaped state space are illustrated in Fig.~\ref{3Dsystem2}. 
Note that if in the parametrization $\boldsymbol{\omega}_r$ we had chosen 
the component $\alpha_r$ as $\alpha_{r q}$ 
instead  of $\alpha_{r p}$, the expressions for $S_p$ and $S_q$ in \eqref{eq:parstatemanifold}
would have been swapped, but the system would have the same form.  Namely, we have a freedom of choice on the $\alpha_r$
parameter, which can be redefined according to the transformation:
\begin{equation} \label{eq:alphafredofch}
\alpha_r  \rightarrow \pi/2 - \alpha_r.
\end{equation}

\noindent If in section \ref{sec:sconstsubspace} we had chosen to start from the measurements of
the observable $p$ or $q$ instead of $r$, we would have considered the subspaces
$\theta_p = {\textnormal{const}}$ or $\theta_q = {\textnormal{const}}$. This would have led us to
a system of equations analogous to \eqref{eq:parstatemanifold}, in terms of parametrizations
$\omega_p$ and $\omega_q$.
Furthermore, the freedom of choice represented by \eqref{eq:alphafredofch} would
apply to the parameter extensions $\alpha_q$ and $\alpha_p$.

\section{Proof of the Translation Invariance theorem \ref{th:lsbgen}}\label{sec:trasltheor}
Let us demonstrate that a partition whose sets are  $A_{\mu} = \{ x \stackrel{l...n}{=} \mu \}$
is $S$-invariant.
Suppose we have two numbers $x$ and $y$ that belong to the same set $A_{\mu}$. By the definition of
$A_{\mu}$, $x$ and $y$ share the same bits from position $l$ to $n$, that is: 
$x \stackrel{l...n}{=} y \stackrel{l...n}{=} \mu $. When we apply the increment $S$ to $x$ and $y$, 
their least significant bits are incremented by one, and thus:
$x \stackrel{l...n}{=} y \stackrel{l...n}{=} \mu + 1$. 

This implies that both $x$ and $y$ now belong to the set $A_{\mu+1}$, which is still in the partition
$\mathcal{A}$.
Since this reasoning holds for any pair of elements of the same set, we have
demonstrated that any set of the partition $\mathcal{A}$ is mapped by $S$ into another set of the same
partition, confirming that the partition is $S$-invariant.

Different values of the lower bound $l$ index identify different partitions; thus, from now on, we use the notation
$\mathcal{A}^l$ to indicate the partitions with cardinality $2^c$, being $c = n - l + 1$. 

Let us now demonstrate that, if a partition $\mathcal{A}^l$ is $S$-invariant, then its sets must be of
the form  $A_{\mu} = \{ x \stackrel{l...n}{=} \mu \}$.
Since there are $2^n$ elements in $X$, the cardinality of the sets $A_{\mu}$  of a partition
$\mathcal{A}^{l}$ is $2^n/2^c$. The proof proceeds in four steps:

\begin{itemize}
\item assuming that $\mathcal{A}^{l}$ is invariant under $S$, we can show that
$S A \ne A$ for any set $A$ of $\mathcal{A}^{l}$. 
If we repeatedly apply the transformation $S$, we can reach any element along the set $X$.
In other words, the transformation $S^m$, generates all $X$: $\bigcup_m S^m A = X$.
If we had that $S A = A$, then we could also say that $S^m A = A$ for any integer $m$.
Consequently, we would have that $A=X$, which is impossible, since $A$ is only a subset of $X$. 

\item We can also demonstrate that, if we start from a set $A$ of a given partition
$\mathcal{A}^{l}$, the successive applications of $S$ generate the whole partition.
The reason is that, by hypothesis, $S^m A$ is always a set of  $\mathcal{A}^{l}$, and, as mentioned
before, repeated applications of $S$ make us reach all the points of $X$.

\item However, this does not mean that for every value of $m$ we obtain a different set $S^m A$,
because there are no infinite sets in the partition $\mathcal{A}^{l}$.
Therefore, as $m$ increases in the expression  $S^m A$, there must exist a smallest $m_0$, such that
we obtain a set we have already encountered, which we denote as $A'$. If $A'$ occurred for a certain value of $m$,
which we denote by $m'$ (with $m' < m_0$), we could assert that $S^{m_0} A = S^{m'} A$.
On the other hand, by applying $S^{-m'}$ to both sides of this 
equation, we get $S^{m_0-m'} A = A$, which contradicts the assumption that $m_0$ is the
smallest value of $m$ such that we obtain a set we have already encountered, unless $m' = 0$.
We can therefore conclude that $S^{m_0} A = A$. The minimal value $m_0$
must be equal to the cardinality of $\mathcal{A}^{l}$, i.e. $2^c$, because we stated that $S^m A$
generates all the sets of partitions $\mathcal{A}^{l}$.

\item The relation $S^{(2^c)} A = A$ can  be regarded as the condition that the set $A$ is
periodic with period $2^c$, under repeated applications of $S$. 
Thus, if we start from an element $x$ of $A$ and repeatedly apply the transformation $S^{(2^c)}$,
which corresponds to a shift by $2^c$, we obtain other elements of $A$. We can repeat the operation
until our total shift is $N$ (which means that we have applied the transformations $N/2^c$ times)
and we go back to our starting element $x$.
Since  $N/2^c = 2^{n-c}$ is also the cardinality of a set of the
partition $\mathcal{A}^{l}$, we can assert that we have reached all the elements of the set, namely
$S^{(2^c)}$ generates a set of the partition $\mathcal{A}^{l}$. 

Finally, we can observe that the application of $S^{(2^c)}$ leaves unchanged the $c$ least significant
bits of a number, namely the bits from position $l$ to $n$. This means that the sets generated by
$S^{(2^c)}$, which are all the sets of an invariant partition, must obey the condition
$A_{\mu} = \{ x \stackrel{l...n}{=} \mu \}$, for a given $n-l+1$-bit number $\mu$.
This completes the proof.
\end{itemize}

\section{The extended Fisher metric in higher-dimensional systems (theorem 
\ref{th:hissin}) } \label{apx:alphadepfishinftheorem}
In this section, we demonstrate theorem \ref{th:hissin} and, specifically, we demonstrate
the formula \eqref{eq:FSm4std}. Our derivations mainly refer to a subsystem of a larger
$N$-point system. 
This may lead us to overly complex index notation, like in the case of equation \eqref{eq:alpha1n};
therefore, to keep the notation simpler, we will omit the indices that specify the subsystem. 
For example, we can rewrite formula  \eqref{eq:FSm4std} for a $2^l$-point subsystem by specifying only
the index ${j_1...j_l}$ that identifies the position within a subsystem, and omitting the index
${j_{l+1}...j_n}$ that identifies the subsystem. If we introduce a further simplification by using a
single integer index $\bf{j}$ instead of a tuple of binary indices ${j_1...j_l}$, we can write
\begin{equation}\begin{aligned} \label{eq:FSmnstdjj}
{ds^{(l)}}^2 = & \sum_{\bf{j}}   (d \log \rho_{[{\bf{j}}]})^2 
\rho_{[{\bf{j}}]}  \twocolbreak &
+ 4 \sum_{\bf{j}} d \phi_{\bf{j}}^2 \rho_{[{\bf{j}}]}
-  4 \Bigl( \sum_{\bf{j}} d \phi_{\bf{j}} \rho_{[{\bf{j}}]} \Bigr)^2
\end{aligned}\end{equation}
\noindent 
To obtain the first (extensionless) term of this expression, we just have to write the generic Fisher
matrix element for a generic set of parameters $\theta_i$: 
\begin{equation} \nonumber
\frac{\partial \log \rho_{[{\bf{j}}]} }{\partial \theta_{i}}
\frac{\partial \log \rho_{[{\bf{j}}]} }{\partial \theta_{i'}}  \rho_{[{\bf{j}}]},
\end{equation}
\noindent and multiply by the differentials $ d \theta_{i} d  \theta_{i'}$.
The full demonstration of theorem \ref{th:hissin} and, specifically, the derivation of the extension 
($d \phi$-dependent) part of equation \eqref{eq:FSmnstdjj}, is based on the assumption that, starting 
from that equation,  we can apply Eq.~\eqref{eq:FSm4gen} as a recursion formula to obtain the 
Fisher metric for a larger $l+1$-level system.
If we consider only the extension part and explicitly specify (in terms of level $l$) the size of the 
subsystem to which each term refers, Eq.~\eqref{eq:FSm4gen} can be rewritten as
\begin{equation} \label{eq:FSmrec}
{ds^{(l+1)}_A}^2 = \sum_k \rho_k {{ds_A^{(l)}} \bigm\vert_k}^2  + h^2(\theta) d \alpha^2.
\end{equation}
\noindent We can rewrite equation \eqref{eq:FSmnstdjj} for one of the branches (even or odd) of a larger $l+1$-level system by 
indicating the phases and the probabilities as $\phi_{{\bf{j}}k}$ and $\rho_{[{\bf{j}}]k}$, where $k$, as in equation 
\eqref{eq:FSmrec}, labels the branch. The second and third terms of this equation represent the extension part. If we 
substitute these terms in place of  $ds_A^{(l)}\bigm\vert_k$ in \eqref{eq:FSmrec}, we obtain
\begin{widetext}
\begin{equation} \label{eq:angfisherjj}
{ds^{(l+1)}_A}^2= 4 \sum_k \rho_{k} 
\left[ 
\sum_{\bf{j}} d \phi_{{\bf{j}}k}^2 \rho_{[{\bf{j}}]\vert k}
-  \left( \sum_{\bf{j}} d \phi_{{\bf{j}}k} \rho_{[{\bf{j}}]\vert k} \right)^2
\right] + h^2(\theta) d \alpha^2.
\end{equation}
\noindent  
Let us introduce the notation $<x_{{\bf{j}} k}>_{\bf{j}} = \sum_{\bf{j}} r_{{\bf{j}} k}
x_{{\bf{j}} k}$, where, again, we use the single index convention ${\bf{j}} = {j_1...j_l}$ and $r$ is the coefficient used in 
Eq.~\eqref{eq:alpha1n}.
If we omit the subsystem index $j_{l+1}... j_n$ (see the discussion at the beginning of this 
demonstration),  this notation allows us to rewrite the expression  \eqref{eq:alpha1n} as 
\begin{equation} \label{eq:dalphaprm}
d \alpha = <d \phi_{{\bf{j}} 1}>_{\bf{j}} -  <d \phi_{{\bf{j}} 0}>_{\bf{j}}.
\end{equation}
\noindent \noindent Now we will demonstrate the following identity:
\begin{equation} \label{eq:phisqridentity2}
\sum_k  \rho_k <d \phi_{{\bf{j}} k}>_{\bf{j}}^2 - 
\left(\sum_k  \rho_k <d \phi_{{\bf{j}} k}>_{\bf{j}} \right)^2 =
 \rho_0  \rho_1 (<d \phi_{{\bf{j}} 1}>_{\bf{j}} -  <d \phi_{{\bf{j}} 0}>_{\bf{j}})^2.
\end{equation}
\noindent If, in the above equation, we expand all the sums and their squares and 
collect together the terms involving the same dependence on $d \phi_{{\bf{j}} k}^2 $,
its LHS can be written as
\begin{equation} \nonumber
 \rho_{0} (1  -\rho_{0} ) <d \phi_{{\bf{j}} 0}>_{\bf{j}}^2 +
 \rho_{1} (1  -\rho_{1} ) <d \phi_{{\bf{j}} 1}>_{\bf{j}}^2
- 2 \rho_{0} \rho_{1} <d \phi_{{\bf{j}} 0}>_{\bf{j}} <d \phi_{{\bf{j}} 1}>_{\bf{j}},
\end{equation}
\noindent where the terms in parentheses can be simplified with the help of the identity
$\rho_{0} + \rho_{1} = 1$, allowing us to rewrite  \eqref{eq:phisqridentity2} as:
\begin{equation}
\label{eq:phasesidentity}
\rho_{0} \rho_{1} (<d \phi_{{\bf{j}} 0}>_{\bf{j}}^2 + <d \phi_{{\bf{j}} 1}>_{\bf{j}}^2)  -
2 \rho_{0} \rho_{1} <d \phi_{{\bf{j}} 0}>_{\bf{j}} <d \phi_{{\bf{j}} 1}>_{\bf{j}} 
 = \rho_{0} \rho_{1} (<d \phi_{{\bf{j}} 1}>_{\bf{j}} - <d \phi_{{\bf{j}} 0}>_{\bf{j}})^2.\hfill
\end{equation}
\noindent By expanding  the square of
$<d \phi_{{\bf{j}} 1}>_{\bf{j}} - <d \phi_{{\bf{j}} 0}>_{\bf{j}}$, we can easily prove the above equation.
Since \eqref{eq:phasesidentity} is just an alternative way to write \eqref{eq:phisqridentity2}, by demonstrating it we have 
also proven \eqref{eq:phisqridentity2}. 
Let us now consider again the expression \eqref{eq:angfisherjj} of the Extension part of the
Fisher metric.
If we apply the substitution $ h^2(\theta) \rightarrow \overline{h}^2(\theta)\sin^2 \theta$ 
and use the identity $\rho_{0} \rho_{1} = \cos^2 \frac{\theta}{2}  \sin^2 \frac{\theta}{2} = 
\frac{1}{4} \sin^2 \theta$, the last term of Eq.\eqref{eq:angfisherjj}, which contains the function
$h(\theta)$, can be put in the form $4\,\overline{h}^2(\theta)  \rho_{0} \,\rho_{1} d \alpha^2$, 
which, by substituting the expression \eqref{eq:dalphaprm} of the parameter $\alpha$,
can be rewritten as: 
\begin{equation} \label{eq:angfisher3}
4 \overline{h}^2(\theta)  \rho_0  \rho_1
(<d \phi_{{\bf{j}} 1}>_{\bf{j}} -  <d \phi_{{\bf{j}} 0}>_{\bf{j}})^2.
\end{equation}
\noindent In this expression, we recognize the RHS of the identity \eqref{eq:phisqridentity2};
by applying the substitution suggested by that identity, equation \eqref{eq:angfisherjj} can be put
in the form
\begin{equation}\begin{aligned} \label{eq:angfisher4}
{ds^{(l+1)}_A}^2  = &  4 \sum_{ k}\sum_{\bf{j}} d \phi_{{\bf{j}} k}^2 \rho_k \,  \rho_{{\bf{j}}\vert k} -
4 \sum_k \rho_k \,  \left(\sum_{\bf{j}} d \phi_{{\bf{j}} k} \rho_{{\bf{j}}\vert k} \right)^2   \\
 &  + 4 \overline{h}^2(\theta)   \sum_k \rho_k <d \phi_{{\bf{j}} k}>_j ^2
 - 4 \overline{h}^2(\theta)   \left(\sum_k \rho_k <d \phi_{{\bf{j}} k}>_{\bf{j}}   \right)^2
\end{aligned}\end{equation}
\noindent Then, by expanding the terms in the form $<x>_{\bf{j}}$, and collecting the terms that have
a common factor, we have:
\begin{equation}\begin{aligned} \label{eq:angfisher5}
{ds^{(l+1)}_A}^2   = & 4 \sum_{k} \sum_{\bf{j}} d \phi_{{\bf{j}} k}^2 \rho_k \,  \rho_{{\bf{j}}\vert k} -
4 \sum_k \rho_k \,  \left(\sum_{\bf{j}} d \phi_{{\bf{j}} k} \rho_{{\bf{j}}\vert k} \right)^2  \\
&  + 4 \overline{h}^2(\theta)\sum_k \rho_k \,
 \left(\sum_{\bf{j}} d \phi_{{\bf{j}} k} r_{{\bf{j}} k} \right)^2 
 - 4 \overline{h}^2(\theta)   \left(\sum_{k} \sum_{\bf{j}} 
 d \phi_{{\bf{j}} k} r_{{\bf{j}} k}   \rho_k \right)^2.
\end{aligned}\end{equation}
\end{widetext}
\noindent In the hypothesis that the parameters $\phi$ are covariant under permutation (i.e., there is a transformation that 
exchanges the place of the probabilities 
$\rho_{\bf{j}}$ and of the parameters $\phi$ in the same way), in the above equation, only the first 
term is explicitly invariant under permutations of the positions $[{\bf{j}} k]$, because
$\rho_k \rho_{{\bf{j}}\vert k} =  \rho_{[{\bf{j}} k]}$. 

The only way to make Eq.~\eqref{eq:angfisher5} invariant under these permutations for any $\rho_k$ is to require that the 
second and third terms vanish, because they contain mixed terms like
$d \phi_{{\bf{j}} k} \, d \phi_{{\bf{j}} k'}$,  which break the invariance. This occurs only if we
impose that $r_{{\bf{j}} k} = \rho_{{\bf{j}}\vert k}$ and $\overline{h}^2(\theta) = 1$, 
that is, $h^2(\theta) = \sin^2(\theta)$, as stated in theorem \ref{th:hissin}.
With the substitution   $r_{{\bf{j}} k} \rightarrow \rho_{{\bf{j}}\vert k}$,
the fourth term becomes invariant too. If we rename the index $k$ in Eq.~\eqref{eq:angfisher5} as $j_{l+1}$,
we obtain a form for ${ds^{(l+1)}_A}^2$ that is consistent with the extension part of ${ds^{(l)}}^2$, as identified by
Eq.~\eqref{eq:FSmnstdjj}, since
 $\rho_{{\bf{j}}\vert{j_{l+1}}} \rho_{j_{l+1}} = \rho_{[j_1 ... {j_{l+1}}]}$.
Thus,  we have built the recursion formula that makes us go from an $l$-level system
to an $l+1$-level one.

In order to complete our demonstration, we need to build the first step ($l=0$) of
the recursion. If we put $l=0$ in the identity \eqref{eq:phisqridentity2}, we obtain
the equation $\sum_k \rho_k d \phi_k^2 - (\sum_k \rho_k d \phi_k)^2 = \rho_0\rho_1 d \alpha^2$.
By using this equation to replace the parameter $\alpha$ with $\phi_0 \phi_1$ in 
the Fisher metric \eqref{eq:finalfishermetric} of the one-bit three-observable
system, we obtain the first level  ($l=0$) of the recursion formula \eqref{eq:phisqridentity2}.
This completes the proof of the theorem.

\section{Proof of Wigner's theorem in the change of representation context} \label{apx:wignertheorem}

Let us demonstrate Wigner's theorem for changes of representation of the statistical parameters

\begin{theorem} \label{th:linearity}
Let $\psi$ be a generic state, i.e., a point in the state space
$\Phi$. The state $\psi$ can be represented as a wave function depending either on the coordinate
$q$ or on the coordinate $u$. 
In terms of the coefficients of these representations,
we can write $ \vert \psi \rangle^{(q)} = \sum_i a_i \vert i \rangle_q$ and
$\vert \psi \rangle^{(u)} = \sum_i b_i \vert i \rangle_u$,  where
$ \vert i \rangle_q$ and $\vert i \rangle_u$ are the states with definite values of
$q$ and $u$, respectively, and the superscripts $(q)$ and $(u)$ specify the coordinate of the representation. If these 
conditions hold and if the Fubini-Study distance \eqref{eq:fsfinite} is preserved for changes of representations,
then  this change of representation is a linear or antilinear transformation, namely, 
the coefficients $a_i$ and $ b_i$ are related by a linear---or antilinear---operator
$a_i = U^{i j} b_j$. The matrix $U^{i j}$ is also unitary---or antiunitary---because it must preserve 
the normalization of the states.
\end{theorem}

\begin{proof}
At this stage, we have two separate spaces: the space of the wavefunctions $\psi^{(q)}$ 
and the space of the wavefunctions $\psi^{(u)}$. Both of them allow a mapping on the state space and,
vice versa, the same state can be represented in both spaces.
Let us rewrite the above-introduced expression of a generic state $\psi$
in terms of eigenvectors of $u$, by specifying the space to represent 
the vectors through  the superscripts $(q)$ and $(u)$:
\begin{subequations}\label{eq:psibyegens}
\begin{align}
\vert \psi \rangle^{(q)} & = \sum_j b^{(q)}_j \vert j \rangle^{(q)}_u   \label{eq:psibyegensq} \\
\vert \psi \rangle^{(u)} & = \sum_j b^{(u)}_j \vert j \rangle^{(u)}_u.  \label{eq:psibyegensu} 
\end{align}
\end{subequations}
\noindent Here, the vector $ \vert j \rangle^{(q)}_u$ indicates the eigenvector of 
the observable $u$ represented in the space of wavefunctions $\psi^{(q)}$ and, again,
$\vert \psi \rangle^{(q)}$ and $\vert \psi \rangle^{(u)}$ are the same state, in two different
representations. By multiplying left both sides of \eqref{eq:psibyegensq} by 
${}_{\;\;q}^{(q)}\!\langle i \vert$ (the $i$-th eigenvector of $q$ in the space of the wave
functions $\psi_q$), we get:
\begin{equation} \label{eq:acoeffinu} 
 {}_{\;\;q}^{(q)}\!\langle i \vert \psi \rangle^{(q)}  =  \sum_j 
 {}_{\;\;q}^{(q)}\!\langle i \vert j \rangle^{(q)}_u b^{(q)}_j.
\end{equation}
\noindent By using the definition $ \vert \psi \rangle^{(q)} = \sum_i a_i \vert i \rangle_q$, it is evident that the LHS of the 
above equation corresponds to the coefficient $a_i$
(or, for consistency of notation,  $a^{(q)}_i$) of the representation of the state in the space of
the $\psi^{(q)}$.  If we demonstrate that: 
\begin{enumerate}[label=(\roman*)]
 \item the matrix ${}_{\;\;q}^{(q)}\!\langle i \vert j \rangle^{(q)}_u$, where $i,j$
are the row and column indices, is unitary;
 \item in different spaces, the coefficients of the representation of the state are equal or conjugate to each other
namely, $b^{(q)}_j = b^{(u)}_j$ or ${b^{(q)}_j}^* = b^{(u)}_j$,
\end{enumerate}
\noindent then, from   \eqref{eq:acoeffinu}, we get an equation of the form  
$a^{(q)}_i =  \sum_j U^{i j}  b^{(u)}_j$ or $a^{(q)}_i =  \sum_j A^{i j}  {b^{(u)}_j}^*$, which expresses the fact that a 
change in the representation
of the state is a linear or antilinear transformation.

Let us demonstrate the point (i).  The unitarity of the matrix
${}_{\;\;q}^{(q)}\!\langle i \vert j \rangle^{(q)}_u$ is equivalent to the requirement that the
eigenvectors $\vert j \rangle^{(q)}_u$ are orthogonal and normalized. 
Of course, they are orthogonal in the space $\Phi_u$ because we have chosen them as the basis of that
space. 
To demonstrate that they preserve the orthogonality when represented in the $\Phi_q$ space, 
we can apply the metric preservation theorem to the pair of observables $q$ and $u$. 
As stated before, if the statement of the preservation theorem holds, then the Fubini-Study geodesic distance defined in 
\eqref{eq:fubstudybas} is preserved between two parametrizations (in the complex wavefunction form) based on two different 
observables.

The Fubini-Study geodesic distance \eqref{eq:fsfinite} reaches its maximum, i.e $\pi/2$,
when the operand of $\arccos$ is zero. It is easy to show that this happens if and only 
if $\psi_1$ and $\psi_2$ are orthogonal.
By applying this property to the case of the eigenvectors of $u$, we can see that, since they are 
orthogonal in the space $\Phi_u$, the Fubini-Study distance between two different eigenvectors is
$\pi/2$. 
But this distance must be preserved when we represent the vectors in the space $\Phi_q$. Consequently, 
two different eigenvectors of $u$ are orthogonal also when they are represented in the space 
$\Phi_q$. This proves the assertion (i).

As for point (ii), it will be proved first in the subspace generated by a pair of eigenvectors
of the observable $u$, then we will extend the demonstration to the whole space.
Once again, the demonstration is based on the theorem of metric preservation \ref{th:fmprese}, applied to the
Fubini-Study metric in the form \eqref{eq:fubstudybas}.

Let us express the state as a linear combination of eigenstates of $u$, by separating the
absolute value of the coefficients from the phase factor:
\begin{equation}\label{eq:stateuamplphase}
\vert \psi \rangle  = \sum_j c_j  e^{i \phi_j}\vert j \rangle_u.  
\end{equation}
\noindent By differentiating the above equation, we can write \allowbreak\hspace{0pt} 
$d \vert \psi \rangle  = \sum_j (d c_j  + i\, c_j d \phi_j)e^{i \phi_j} \vert j \rangle_u$.
If we substitute the expression of the state \eqref{eq:stateuamplphase} and 
its differential into the Fubini-Study metric in the form \eqref{eq:fubstudybas}, and take into account that, 
from the normalization condition, $2 \sum_j c_j  d c_j = d \sum_j c_j^2 = 0 $, we obtain: 
\begin{equation}\label{eq:fsmetramplphase}
d s_{FS}^2 = \sum_j d c_j^2 +  \sum_j c_j^2 d \phi_j ^2 - \left(\sum_j c_j^2 d \phi_j \right)^2.
\end{equation}
The second and third terms on the RHS of the above equation can be interpreted as the mean
square position of a set of point masses (where $d \phi_j$ represents the position of the
$j$-th point and $c_j^2$ its mass) and the square of the mean position, respectively.
This analogy naturally suggests applying the parallel axis theorem (also known as the
Huygens-Steiner theorem), allowing the difference between these terms to be expressed in the form
\begin{equation}
\sum_{j=1}^{N}c_j^2d\phi_j^2 -
\left( \sum_{j=1}^{N}c_j^2d\phi_j
\right)^2 =
\sum_{j=1}^{N}
c_j^2
\left(d\phi_j-d\bar{\phi}\right)^2,
\end{equation}
where $d \overline{\phi} = \sum_j c_j^2 d \phi_j$
denotes the mean position (i.e., the center-of-mass position).

Let us consider the points along the
curve, on the state manifold, that go from an eigenstate $\vert 1 \rangle_u$
to an eigenstate  $\vert 2 \rangle_u$.
In this framework, the geodesic is the curve that goes from $\vert 1 \rangle_u$
to $\vert 2 \rangle_u$, such that the integral of $d s_{FS}$ is minimal. 

If we consider the metric \eqref{eq:fsmetramplphase}, for fixed variations of the amplitudes $c_j$, the length of a curve is 
minimized when the relative phases are constant. A common variation of all phases corresponds only to an overall phase and does 
not contribute to the Fubini-Study distance.
This implies that it is sufficient to minimize the amplitude contribution
\begin{equation}\label{eq:sfsamp} 
ds_{\mathrm{amp}}^2 =
\sum_{j=1}^{N}dc_j^2,
\end{equation}
subject to the normalization condition $\sum_{j=1}^{N}c_j^2=1$. The amplitude vector
\begin{equation}
\mathbf{c} = (c_1,\ldots,c_N)
\end{equation}
lies on the unit sphere $S^{N-1}$ equipped with its ordinary Euclidean-induced metric, and the states $\vert 1 \rangle_u$ and 
$\vert 2 \rangle_u$ correspond respectively to the points
\begin{equation}
\mathbf{e}_1=(1,0,\ldots,0), \,\,\,
\mathbf{e}_2=(0,1,0,\ldots,0).
\end{equation}
\noindent The shortest curve on the unit sphere joining these two points is the great-circle arc contained in the two
dimensional plane spanned by $\mathbf{e}_1$ and $\mathbf{e}_2$. Since
$\mathbf{e}_1\cdot\mathbf{e}_2=0$, its length is
\begin{equation}
\arccos(\mathbf{e}_1\cdot\mathbf{e}_2) = \frac{\pi}{2}.
\end{equation}
The minimizing curve can therefore be parametrized as
\begin{equation}
c_1=\cos\beta,
\,\,\,
c_2=\sin\beta,
\,\,\,
c_j=0
\,\,
\forall j>2,
\end{equation}
with
$ 0\leq\beta\leq\frac{\pi}{2}$ and the identity  $c_1^2+c_2^2 = 1$ satisfied along the whole geodesic.
The corresponding family of physical states can be therefore written as
\begin{equation}\label{eq:psibetarepr}
\vert \psi \rangle =
\cos\beta \vert 1 \rangle_u  e^{i \phi_1} + \sin \beta \vert 2 \rangle_u  e^{i \phi_2}
\end{equation}
\noindent and the metric obtained from  \eqref{eq:sfsamp} reads:
\begin{equation}\label{eq:sfsamp2} 
d s_{FS}^2 
=
d \beta^2.
\end{equation}

\noindent In the linear combination \eqref{eq:psibetarepr}, only the phase difference $\phi_{1 2} = \phi_1 - \phi_2$ is
relevant and, consequently, there is a minimal curve for each (constant) 
value of $\phi_{1 2}$. This is not surprising because we know that, in general, 
given two points on a manifold, there can be more than one geodesic. Vice versa,
since the coefficients $c_j$ for $j>2$ are zero in the linear combination,
we can also assert that in the subspace generated by
$\vert 1 \rangle_u$ and  $\vert 2 \rangle_u$, \textit{all} minimal curves can
be obtained by choosing the right values of $\phi_{1 2}$.

The demonstration of \eqref{eq:psibetarepr} is independent of the choice of space, $\Phi_u$
or $\Phi_q$, used to represent the vectors $\vert j \rangle_u$. 
On the other hand, by the metric preservation theorem, the distance between two
points (given by the Fisher/Fubini-Study metric), is independent of the space 
used to represent the vectors. Thus, if a curve is a geodesic in $\Phi_u$, it will
be a geodesic in $\Phi_q$ as well. 
This also implies that each value of $\phi_{1 2}$  identifies a geodesic in $\Phi_q$, but also in
$\Phi_u$, which corresponds to a single value of $\phi_{1 2} $ for $\Phi_q$ because, as stated before,  each geodesic 
corresponds to a value of $\phi_{1 2}$. We can conclude that
there are pairs of values of $\phi_{1 2}$, say $\phi_{1 2}^{(q)}$ and $\phi_{1 2}^{(u)}$, 
which corresponds to a curve, which is a geodesic in both $\Phi_q$ and $\Phi_u$.

If we integrate the differential defined by \eqref{eq:sfsamp2}  along a geodesic,
from the starting point $\vert 1 \rangle_u$, corresponding to $\beta = 0$, 
to another point $\vert \psi' \rangle$, corresponding to $\beta = \beta'$, 
we obtain $D_{FS} = \beta'$.
If we apply once more the metric preservation theorem, the distance between two points 
(given by the Fisher/Fubini-Study metric), is independent of the space  we used to 
represent the vectors.
The only point of the space $\Phi_q$, along a geodesic curve, which corresponds to some state
$\vert \psi' \rangle^{(q)}$  of $\Phi_q$ is the point that has the same distance
$D_{FS} = \beta'$ and, as a consequence, the same value of $\beta'$ (we assume that 
$\beta \ge 0$). Thus, the coefficients $c_1$ and $c_2$ of one space must be equal to those in the other space.
This demonstrates the point (ii) in the subspace generated by the vectors 
$\vert 1 \rangle_u$ and $\vert 2 \rangle_u$, under the condition that we
choose the right pair $\phi_{1 2}^{(q)} \, \phi_{1 2}^{(u)}$.

The next step is to extend the demonstration for complex coefficients of the eigenstates 
$\vert 1 \rangle_u$ and $\vert 2 \rangle_u$.
This means that we have to demonstrate that a state corresponds, in the spaces
$\Phi_q$ and $\Phi_u$, to equal values of the phase factors of the eigenstates
$\vert 1 \rangle_u$ and $\vert 2 \rangle_u$ as well as to the absolute values of their coefficients.
To be more specific, we have to demonstrate that in both spaces the difference $\phi_{1 2}$ has the 
same value.

We will follow a scheme similar to that adopted for the coefficients $c_1$ and $c_2$.
In order to simplify the demonstration, we will assume that we can choose the 
pair $\phi_{1 2}^{(q)}$  $\phi_{1 2}^{(u)}$, which identifies the same geodesic in both
$\Phi_q$ and $\Phi_u$, and this pair will be  $\phi_{1 2}^{(q)} = \phi_{1 2}^{(u)} = 0$. This  
is equivalent to the multiplication by a constant phase factor of any of the 
states $\vert 1 \rangle_u$ and $\vert 2 \rangle_u$.

Let $\gamma$ be a curve, in any of the spaces $\Phi_q$ or $\Phi_u$, which has constant
$\phi_i$  for $i > 2$ and constant $c_i$ for any $i$, and let us express it in parametric form with
parameter $\phi_{1 2}$. 
Suppose that we start from a point, for instance  $\phi_{1 2} = 0$ and calculate,
along this curve,  the definite integral of the Fubini-Study metric in differential form: 
$D_{\gamma}(\phi_{1 2}') = \int_0^{\phi_{1 2}'} d s_{FS}$; note that the parameter $\beta$ of the geodesic
curve introduced above, is constant.
If we assume that $\phi_{1 2}$ can take either non-negative or non-positive values only, this integral is 
a monotonic function of $\phi_{1 2}'$ and, as a consequence, $D_{\gamma}(\phi_{1 2}')$ 
is also invertible (we will justify this assumption at the end of the section).
Furthermore, according to the metric-preserving property \eqref{eq:metricpreservationtheorem},
this integral must take the same value for both spaces $\Phi_q$ and $\Phi_u$. 
In other words, if the value of the parameter
$\phi_{1 2}$ for a given point on $\gamma$ is equal in the spaces $\Phi_q$ and
$\Phi_u$, namely if $\phi_{1 2}'^{(q)} \equiv  \phi_{1 2}'^{(u)}$,
then the equation  $D_{\gamma}(\phi_{1 2}'^{(q)}) = D_{\gamma}(\phi_{1 2}'^{(u)})$ 
must hold.

If we consider either positive or negative values of its argument,
and if $D_{\gamma}$ is invertible in a semi-axis of its argument, we have that
$\phi_{1 2}'^{(u)} = \pm \phi_{1 2}'^{(q)}$, where the sign $\pm$ is due to the fact that
actually $D_{\gamma}$ depend on the absolute value of $\phi_{1 2}'$. 
This proves the assertion (ii) in the subspace generated by
$\vert 1 \rangle_u$ and $\vert 2 \rangle_u$, where the ``+'' sign corresponds to equal coefficients and the ``-'' sign 
corresponds to coefficients that are conjugate to each other.

It should be noted that the demonstration of \eqref{eq:psibetarepr} could be applied
to any set of orthonormal vectors.
This allows us to repeat the above scheme to the spaces generated by the orthogonal vectors 
$\vert \psi_{1 2} \rangle_q = c_1 \vert 1 \rangle_q + c_2\vert 2 \rangle_q$ 
and $\vert 3 \rangle_u$ and the corresponding ones in the space $\Phi_u$ (assuming the plus sign in the condition $\phi_{1 2}'^{(u)} = \pm \phi_{1 2}'^{(q)}$),
which will lead us to demonstrate the statement (ii) in these spaces. Going 
further, we can recursively demonstrate (ii) in the full spaces $\Phi_q$ and  $\Phi_u$.

We assumed the plus sign in the condition $\phi_{1 2}'^{(u)} = \pm \phi_{1 2}'^{(q)}$, but it can be shown that it is also 
possible to choose the minus sign. This leads us to the antiunitary case of the theorem's statement, but we must ensure that 
this choice is applied consistently across all terms $\phi_{i j}^{(q)}$. To demonstrate this, we can rewrite the last identity
in the form
\begin{equation}\label{eq:s12us12q}
\phi_1^{(u)} - \phi_2^{(u)} = s_{1 2}(\phi_1^{(q)} - \phi_2^{(q)}),
\end{equation}
\noindent where 
$s_{1 2} = \pm 1$. We can write a similar identity for the second and third coefficients:
$\phi_2^{(u)} - \phi_3^{(u)} = s_{2 3}(\phi_2^{(q)} - \phi_3^{(q)})$.
By adding the two identities, we get:
\begin{equation}\label{eq:s12s23s31}
\phi_1^{(u)} - \phi_3^{(u)} =  s_{1 2}\phi_1^{(q)} -  s_{2 3} \phi_3^{(q)} 
+ ( s_{2 3} - s_{1 2})\phi_2^{(q)}.
\end{equation}
\noindent
We expect this equation to be compatible with the first and the third coefficient identity:
$\phi_1^{(u)} - \phi_3^{(u)} =  s_{1 3}(\phi_1^{(q)} - \phi_3^{(q)})$. 
This is possible only if the last term in \eqref{eq:s12s23s31} vanishes, which 
implies that $s_{2 3} = s_{1 2}$ and, as a consequence  $s_{1 3} = s_{1 2}$.
We can go further with this scheme and extend it to all the terms  
$\phi_{i j}^{(q)}$. We can conclude that all the identities of the
form \eqref{eq:s12us12q} have the same sign $s_{i j}$.

The consequence of this condition is that there are two representations in 
the space $\Phi_u$ that are compatible with the metric-preserving condition,
whose phases $\phi_1$ have opposite signs, corresponding to the unitary and antiunitary case of the theorem.
\end{proof}

\section{Derivation of the twiddle factor from the shift/rotation operator} \label{apx:shrotoperformula}
\subsection{The Recursive formula for the shift/rotation operator}\label{sec:indform}
From  equation \eqref{eq:shifttrasflg}, in the case $l'= 1 $, we get the identity
$S_{l} = F_{1 \rightarrow l} S_{1} F_{1 \rightarrow l}^{-1}$.
If we multiply both sides of Eq.~\eqref{eq:fftgrpl} by $S_l$ and, in the RHS, we apply the above 
expression of $S_l$, we get:
\begin{equation}  \label{eq:shifttrasflgt}
S_{l} \psi_p =  F_{1 \rightarrow l} S_{1}  t_{1} F_{1} \psi_q.
\end{equation}
\noindent  If we apply the identity \eqref{eq:shifttrasfl}, this equation yields
\begin{equation}  \label{eq:shifteqext}
 F_{1 \rightarrow l}    t_{1} F_{1} S_{0}  \psi_q =  F_{1 \rightarrow l}  S_{1}  t_{1} F_{1} \psi_q.
\end{equation}
\noindent We can left-multiply both sides of the above equation by $F_{1 \rightarrow l}^{-1}$,
and, since it must hold for any $\psi_q$, we can also remove  $\psi_q$.
Thus, we obtain the operator expression:
\begin{equation}  \label{eq:shifteqextx}
t_{1} F_{1} S_{0} = S_{1}  t_{1} F_{1}.
\end{equation}

\noindent The above equation can play the role of a recursion formula because it expresses the $S_1$ 
rotation/shift operator in terms of the lower-level one. Our plan is now to substitute, in that 
formula, the twiddle and the translation operator
in block-matrix form. Then, thanks to condition \eqref{eg:sigmavssigmaor}, we can derive an even
more explicit form of the shift/rotation operator, which leads us to
the expression of the twiddle factor.

It is worth noting that there is a subtle and, in a way, hidden connection between the above formula 
and the \eqref{eq:chiinvariance}.
In \eqref{eq:shifteqextx} we have the transformation operator $ t_{1} F_{1}$ instead of 
the condition $\chi$, which connects two consecutive levels of the butterfly diagram.
Equation \eqref{eq:chiinvariance} synthesizes the idea that, if we divide our space into 
partitions invariant with respect to a transformation, this transformation does not mix different 
conditions $\chi$. In the following subsection, we will show that the rotation/shift operator $S_1$
actually involves $L/2$ pair of components separately, leading us to a set of $L/2$ decoupled
equations, which contain a condition on how the transformation operator is made.
Since each pair of components corresponds to a set of an invariant partition,
this scenario is very close to what expressed by equation \eqref{eq:chiinvariance}.

\subsection{The Recursive Formula in block matrix form}\label{sec:indformterms}

Let us consider the expression \eqref{eq:invoptwod} of the translation
operator in the $\psi_q$ space for the elementary $N=2$ systems.  
We can extend this operator to all the elementary systems at the 
first level of the butterfly diagram of an $L$-dimensional system.
We simply need to apply $L/2 = 2^{l-1}$ elementary system shift/rotation operators simultaneously: 
this operation, in the  $\psi_q$ space, takes the form of the matrix
\begin{equation} \nonumber
r_l = \Diag(1,... 1, e^{i \pi}, ... , e^{i \pi}),
\end{equation}
\noindent where the phase factors $1$ apply to the first $L/2$  (even $p$) elements, and 
the phase factors $e^{i \pi}$ apply to the other $L/2$  (odd $p$) elements.
Let us now rewrite Eq.~\eqref{eg:sigmavssigmaor} in block-matrix form.
We need first to write the operator $S^o$, namely, the odd $p$-branch  shift/rotation operator
\begin{equation} \nonumber
S^o = 
\begin{pmatrix}
I & 0 \\
0 & S_{l-1}^{(l-1)}
\end{pmatrix},
\end{equation}
\noindent where the superscript $(l-1)$ on $S$ indicates that it operates on a $2^{l-1}$-dimensional space of a
$2^{l-1}$-dimensional subsystem.
The subscript ${l-1}$ specifies that we are acting on the highest level of the  odd $p$ branch. We can rewrite
Eq.~\eqref{eg:sigmavssigmaor}, which involves matrices that act on the $\psi_p$ space, in the explicit form
\begin{equation} \nonumber
S_l = 
\begin{pmatrix}
I & 0 \\
0 & S_{l-1}^{(l-1)}
\end{pmatrix}
\begin{pmatrix}
0 & I \\
I & 0
\end{pmatrix}.
\end{equation}
\noindent 
By applying $F_{1 \rightarrow l}^{-1}$ to the left and $F_{1 \rightarrow l}$ to the right of 
both sides of the above equation, and then applying  \eqref{eq:shifttrasflg} (with $l' = 1$)
on the LHS we get:
\begin{equation} \label{eq:shifttrasflm1}
S_{1} = 
F_{1 \rightarrow l}^{-1}
\begin{pmatrix}
I & 0 \\
0 & S_{l-1}^{(l-1)}
\end{pmatrix}
F_{1 \rightarrow l}
F_{1 \rightarrow l}^{-1}
\begin{pmatrix}
0 & I \\
I & 0
\end{pmatrix}
F_{1 \rightarrow l}
\end{equation}

\noindent Let us consider again equation \eqref{eq:fftgrpblockl}, which expresses the condition
that, from the first level on, the butterfly diagram branches into two subdiagrams, and that the
transformation $F_{1 \rightarrow l}$ acts on each branch independently. 
In that equation, we have the condition
$F_{1 \rightarrow l}^e = F_{1 \rightarrow l}^o = \mathcal{F}_{l-1}$; hence, we can rewrite the product 
in the first half of the RHS of  \eqref{eq:shifttrasflm1} as:
\begin{equation} \nonumber
F_{1 \rightarrow l}^{-1} 
\begin{pmatrix}
I & 0 \\
0 & S_{l-1}^{(l-1)}
\end{pmatrix}
F_{1 \rightarrow l} =
\begin{pmatrix}
I & 0 \\
0 & \mathcal{F}_{l-1}^{-1} S_{l-1}^{(l-1)} \mathcal{F}_{l-1}
\end{pmatrix}.
\end{equation}
\noindent If we use the above equation and simplify the second half of the RHS  of
\eqref{eq:shifttrasflm1} through the identity
$\mathcal{F}_{l-1}^{-1} I \mathcal{F}_{l-1} = I$,  we get:
\begin{equation} \label{eq:shifttrasflm1simp1}
S_{1} = 
\begin{pmatrix}
I & 0 \\
0 & \mathcal{F}_{l-1}^{-1} S_{l-1}^{(l-1)} \mathcal{F}_{l-1}
\end{pmatrix}
\begin{pmatrix}
0 & I \\
I & 0
\end{pmatrix}
\end{equation}
\noindent We can recognize the term $\mathcal{F}_{l-1}^{-1} S_{l-1}^{(l-1)} \mathcal{F}_{l-1}$ as the
rotation/shift operator on the $p$ variable in the space of $N/2$-dimensional $\psi_q$ functions, 
which, according to our index conventions, is represented by $S_{0}^{(l-1)}$.
We have obtained an expression of the rotation/shift operator, represented in the space of intermediate
($l=2$ on the butterfly diagram) functions, in terms of a rotation/shift operator of a smaller
($L/2$-dimensions) system.

The considerations that allowed us to write $S_l$ in the diagonal form \eqref{eq:sdiag}
can be applied to the operator $S_{1}^{(l-1)}$ as well. Thus, we can write:
\begin{equation} \label{eq:s1explexpr}
S_{0}^{(l-1)} = \Diag{(e^{i {s}_{(l-1)}^0}, ... , e^{i {s}_{(l-1)}^{N/2-1}})}.
\end{equation}

\noindent Note that, in the above equation, and in the equations that follow, the subscript in the form
$(m)$, of the phases $s_{(m)}^k$, indicates the size of the subsystem to which we are applying the
translation operator $S_{0}^{(m)}$, while $k$  indicates the component of the corresponding space. 

Now we have all the necessary elements to express the recursion formula \eqref{eq:shifteqextx} in a more explicit block matrix 
form: we have the expressions of $S_1$
\eqref{eq:shifttrasflm1simp1}, with the explicit expression of $S_{0}^{(l-1)}$ \eqref{eq:s1explexpr}; 
we also have $S_0$ (equation \eqref{eq:sdiag} with $s^k = s^k_{(l)}$),
$F_1$   (equation \eqref{eq:fftpartblockl} ) and we can also use the explicit matrix element 
form  \eqref{eq:twiddleinphase} for the twiddle factor $t_1$. Note that, since all of these terms
are either diagonal or $N/2$ block matrices, equation \eqref{eq:shifteqextx} is expected to connect
only index $k$ with index $k+L/2$ components (where $k < L/2$). This implies that, when we apply all
of the above substitutions to equation \eqref{eq:shifteqextx}, we get $N/2$ decoupled equations
that involve $2 \times 2$ matrices
\begin{widetext}
\begin{equation}
\scalebox{.97}{$
    \displaystyle
\label{eq:shifteq2d}
{\begin{pmatrix}
e^{i \phi_{1}^k} & 0 \\
0 & e^{i \phi_{1}^{k+\frac{L}{2}}}
\end{pmatrix}
\begin{pmatrix}
1 & 1 \\
1 & -1
\end{pmatrix}
\begin{pmatrix}
e^{i {s}_{(l)}^k} & 0 \\
0 & e^{i {s}_{(l)}^{k+\frac{L}{2}}}  
\end{pmatrix} =
\begin{pmatrix}
1 & 0 \\
0  & e^{i {s}_{(l-1)}^{k}}
\end{pmatrix}
\begin{pmatrix}
0 & 1 \\
1 & 0
\end{pmatrix}
\begin{pmatrix}
e^{i \phi_{1}^k} & 0 \\
0 & e^{i \phi_{1}^{k+\frac{L}{2}}}
\end{pmatrix}
\begin{pmatrix}
1 & 1 \\
1 & -1
\end{pmatrix}}$}
\end{equation}
\end{widetext}
\noindent where $k < L/2$. By expanding all products and writing an equation for
each matrix element we obtain:
\begin{equation} \label{eq:shifteq2dbyterm}
\begin{array}{ll}
e^{i {s}_{(l)}^{k}} e^{i \phi_{1}^{k}}  & =  e^{i \phi_{1}^{k+L/2}} \\
e^{i {s}_{(l)}^{k+L/2}} e^{i \phi_{1}^{k}} & =  - e^{i \phi_{1}^{k+L/2}} \\
e^{i {s}_{(l)}^{k}} e^{i \phi_{1}^{k+L/2}}  & =  e^{i {s}_{(l-1)}^{k}} e^{i \phi_{1}^{k}} \\
- e^{i {s}_{(L)}^{k+L/2}} e^{i \phi_{1}^{k+L/2}} & =   e^{i {s}_{(l-1)}^{k}}  e^{i \phi_{1}^{k}} ,
\end{array}
\end{equation}
\noindent which, expressed in terms of phases, yields
\begin{equation} \label{eq:shifteq2dbyphase}
\begin{array}{ll}
{s}_{(l)}^{k} + \phi_{1}^{k}  & =  \phi_{1}^{k+L/2} \\
{s}_{(l)}^{k+L/2} + \phi_{1}^{k} & =  \phi_{1}^{k+L/2} \pm \pi \\
{s}_{(l)}^{k} + \phi_{1}^{k+L/2}  & =   \phi_{1}^{k} + {s}_{(l-1)}^{k}\\
{s}_{(l)}^{k+L/2} + \phi_{1}^{k+L/2} & =  \phi_{1}^{k} + {s}_{(l-1)}^{k}  \pm \pi.\\
\end{array}
\end{equation}
\noindent Note that the $\pm$ sign preceding the $\pi$ terms arises because the exponential function is not injective on the 
complex plane. As it will be evident at the end of our derivation, the two choices of sign lead to equivalent formulations of 
quantum mechanics. We adopt the minus sign; with this choice, the above system can be rewritten in an even more compact form:
\begin{subequations} \label{eq:shifteqrelations}
\begin{align}
  {s}_{(l)}^{k} & =  \phi_{1}^{k+L/2}-  \phi_{1}^{k}  \label{eq:shifteqrelations1}\\
  {s}_{(l)}^{k+L/2} - {s}_{(l)}^{k} & = - \pi \label{eq:shifteqrelations2}\\
 2 {s}_{(l)}^{k} & =  {s}_{(l-1)}^{k}. \label{eq:shifteqrelations3}
\end{align}
\end{subequations}
\subsection{The twiddle factor}\label{sec:twiddlefactder}
\noindent System \ref{eq:shifteqrelations} connects, for the level $l$,  the twiddle phase terms
$\phi^k$ with the shift/rotation phases $s^k$. 
The third equation \eqref{eq:shifteqrelations3}
can be used as a recursion formula for deriving the rotation/shift operator in the
$\psi_q$ space, for a set of increasing size systems. The subscripts in parentheses denote both the level and the size of the 
corresponding subsystem (for example,  $L= 2^l$). 
Note that equations \eqref{eq:shifteqrelations}
have been derived for $k < L/2$; this is consistent with the occurrence of the term ${s}_{(l-1)}^{k}$, 
which applies to a $L/2$-dimensional space. On the other hand, in order to make
\eqref{eq:shifteqrelations} applicable to all the components of the $L$-dimensional space,
the term ${s}_{(l)}^{k}$ should be determined for $ k \ge L/2$ and $k < L$ as well; 
this is possible through the second equation  \eqref{eq:shifteqrelations2}.

In the $l=2$ case, each second level branch (i.e. the even $p$ and odd $p$ branches) is an $N = 2$ elementary system. Therefore 
the operator $S_1^{(2)}$, for $0 \le k < L/2-1$,  is a $\psi_p$
component exchange, which means that in the $\psi_q$ space $S_0^{(2)}$ takes the form
\eqref{eq:invoptwod}.
Thus, the phase term occurring in \eqref{eq:shifteqrelations} takes the value $s_1 = (0, - \pi)$.
By applying Eq.~\eqref{eq:shifteqrelations3}, for the $k=0,1$  terms, and 
Eq.~\eqref{eq:shifteqrelations2}, for $k=2,3$  terms, we obtain
$s_2 = (0, - \pi/2, - \pi, - 3/2 \pi)$. 
By repeating this procedure recursively we get, for larger values of $l$:
\begin{equation} \label{eq:sigmalexpr}
{s}_{(l)}^k = - \frac{2 \pi k}{L} \,\,\,  \forall k : 0 \le k < L.
\end{equation}
\noindent From this equation, and from the first equation \eqref{eq:shifteqrelations1}
of the system \eqref{eq:shifteqrelations},  we obtain the following condition for the twiddle:
\begin{equation} \label{eq:twiddlelexpr}
\phi_{1}^{k+L/2}-  \phi_{1}^{k} = - \frac{2 \pi k}{L} \,\,\,  \forall k : 0 \le k < L/2.
\end{equation}
\noindent This still does not fully define the twiddle factor, but fixes the phase difference 
between an even $p$ twiddle factor $t_l^{k}$ and its corresponding odd $p$ factor, $t_l^{k+L/2}$.
At this stage, if we do not introduce further hypotheses, we are free to choose the overall phase
factor of each  $t_l^{k}, t_l^{k+L/2}$ pair.

The question is whether this overall phase factor is physically relevant or not. 

The common phase of each pair of twiddle components can be transferred
to a constant phase offset for the corresponding pair of components
of $\psi_q$. This redefinition leaves both the probability distribution
and the Fisher metric unchanged. As we will see, we will use this constant phase offset to drive the phase of subsequent 
levels' elementary systems:
suppose that we redefine each pair of components of $\psi_q$ with indices $k$ and
$k+L/2$ as $e^{i \phi_{1}^{k}} \psi_q$.
Since this factor does not interfere with the first-level transformation $F_1$, it can be applied to
the twiddle factor, which can now be redefined as $\phi_{1}^{k} = 0$ and
$\phi_{1}^{k+L/2} = -\frac{2 \pi k}{L}$.
This factor does not impact the Fisher metric and  does not change the probability distributions,
it is just a redefinition of the parametrization extension $\phi$ by a constant, which is
not physically relevant.
We can therefore conclude that the twiddle factor's phase can be defined as follows:
\begin{equation} \label{eq:twiddlelexpr2}
\begin{array}{lll}
\phi_{1}^{k} & = 0 \,\, & \forall k : 0 \le k <L/2 \\
\phi_{1}^{k} & = - \frac{2 \pi (k - L/2)}{L}  &  \forall k : L/2 \le k < L.
\end{array}
\end{equation}

\noindent Note that the above conditions apply to the first-level twiddle factor of an
$L$-dimensional system.
On the other hand, each subdiagram of a larger butterfly diagram must behave as the
diagram of an independent subsystem, which obeys the same hypotheses as the larger system.

Now we are going to generalize the above conditions for the twiddle factor to all levels of the system.
If we denote by $N = 2^n$ the size of a larger system,
its level index is expected to run from 1 to $n+1$ (see Fig.~\ref{butterfly8}). 
In the discussion that led us to expression \eqref{eq:twiddlelexpr2}, we 
considered only the even $p$ and the odd $p$ subsystems, shown in Fig.~\ref{butterflyblocks}.
Now we are going to consider more general subsystems, namely, those
that run from a level $m$ of the larger system, to the highest level, namely, $n+1$.
The size of these subsystems, $L=2^l$, now depends on
the difference between the highest and the lowest level index, i.e., we have that $l=n+1-m$. 

We know that each of these subsystems behaves in the same way. 
The index $k$, which occurs in the twiddle factor expression  \eqref{eq:twiddlelexpr2},
is specific for each $L$-dimensional subsystem. But it can be expressed 
in terms of the larger $N$-dimensional system's index, which we denote by $k'$. 
We organize the indices $k$ and $k'$ as follows: the values of $k'$ from 0 to $L-1$ refer to the first 
subsystem, the values from $L$ to $2L-1$ to the second, and so on. Since the range of the 
subsystem-specific index $k$ must be $0 ... L-1$, we can write  $k = k' \mod L$.

If we want to consider the twiddle factors for all the levels occurring in the
factorization \eqref{eq:fftexpln}, we must state which level of the 
$N$-dimensional system's twiddle factor corresponds to the first level of the 
subsystems' twiddle factors.
To do this, we can count all the twiddle factor levels from the lowest
to the highest. Since, for an $L$-dimensional subsystem, we have $l-1$ levels of 
twiddle factors, the position of the lowest factor $t_1$ is $l-1$ from the top. 
If we map this term onto the sequence of twiddle factors $t_1 ... t_{n-1}$ of an $N$ dimensional 
system, we can see that the lowest (within a subsystem) term of an $N$ dimensional system is
$t_{n-l+1}$; we identify this level of the twiddle factor as $l' = n-l+1$.
We can also say that given an $l'$, the factor $t_{l'}$ is the first level twiddle factor
of a subsystem with $l = n-l'+1$ levels.

Now we can generalize \eqref{eq:twiddlelexpr2} as follows:
the twiddle factor phase $\phi_{l'}^{k'}$ of the $N$ dimensional system corresponds
to the phase of an $L$-dimensional subsystem for $k = k' \mod L$ and $l = n-l'+1$. 
If we make these substitutions in \eqref{eq:twiddlelexpr2} and rename $k'$ as $k$, we obtain the 
following expressions for the phase $\phi_{l'}^{k'}$: 
\begin{equation} \label{eq:twiddlelexprgenp}
\begin{aligned}
\phi_{l'}^{k}  = \, & \, 0 \, & \forall k : 0 \le (k \mod 2^{n-l'+1}) < 2^{n-l'} \\
\phi_{l'}^{k}  = \, & - 2 \pi\frac{ (k \mod 2^{n-l'+1}) - 2^{n-l'}   }{2^{n-l'+1}}   \twocolbreak 
& \forall k : 2^{n-l'} \le (k \mod 2^{n-l'+1}).
\end{aligned}
\end{equation}
\noindent It is easy to show that the above expressions  allow us to express the twiddle factor in the matrix form given in equation \eqref{eq:twiddlblockexpr}.
\end{document}